\documentclass[prx,twocolumn,showpacs,amsmath,amssymb,superscriptaddress,floatfix,reprint, nofootinbib]{revtex4-2}
\usepackage{macros}
\usepackage[export]{adjustbox}
\usepackage{framed}
\usepackage{minted}
\definecolor{LightGray}{gray}{0.95}
\newcounter{mylemma}

\makeatletter
\def\l@subsection#1#2{}
\def\l@subsubsection#1#2{}
\makeatother

\makeatletter

\renewcommand\onecolumngrid{% <<<<<<
\do@columngrid{one}{\@ne}%
\def\set@footnotewidth{\onecolumngrid}% <<<<<<<<<<<<<<<<
\def\footnoterule{\kern-6pt\hrule width 1.5in\kern6pt}%
}

\renewcommand\twocolumngrid{% <<<<<<
		\def\footnoterule{% restore rule
		\dimen@\skip\footins\divide\dimen@\thr@@
		\kern-\dimen@\hrule width.5in\kern\dimen@}
		\do@columngrid{mlt}{\tw@}
}%

\makeatother
\usepackage{svg}

\begin{document} % MARK: TITLE / ABS

\title{
Geometry of Noisy Quantum Many-Body Dynamics with Continuous Symmetries:\\
Entanglement and Correlations
}

\date{\today}
\author{Marco Lastres}
\email{marco.lastres@tum.de}
\affiliation{Technical University of Munich, TUM School of Natural Sciences, Physics Department, 85748 Garching, Germany}
\affiliation{Munich Center for Quantum Science and Technology (MCQST), Schellingstr. 4, 80799 M\"unchen, Germany}
\author{Sanjay Moudgalya}
\email{sanjay.moudgalya@gmail.com}
\affiliation{Technical University of Munich, TUM School of Natural Sciences, Physics Department, 85748 Garching, Germany}
\affiliation{Munich Center for Quantum Science and Technology (MCQST), Schellingstr. 4, 80799 M\"unchen, Germany}
\begin{abstract}
We study unitary quantum dynamics in noisy Brownian models with global continuous symmetries, such as $U(1)$ and $SU(2)$, focusing on Rényi entanglement entropies and hydrodynamic and non-hydrodynamic correlators.
By mapping the averaged late-time dynamics to the low-energy physics of effective replica Hamiltonians, we find that the evolution is controlled by the quantum geometry of their ground-state manifolds, which is directly related to the geometry of $k$-commutants—the symmetry algebra of $k$ replicas of the system.
In interacting systems, these $k$-commutants are generically determined solely by the symmetries of the system, independent of microscopic details of the noisy evolution.
This allows us to use the time-dependent variational principle (TDVP) to provide simple geometric explanations for the sub-ballistic Rényi entanglement growth and the anomalous decay of non-hydrodynamic correlators in interacting systems with continuous symmetries.
We find this behavior to be intimately connected to singularities within the $k$-commutant manifolds, arising from frozen ``void'' states in the Hilbert space that exist due to continuous on-site symmetries.
This also demystifies the important role of voids in the dynamics of these observables, previously identified in $U(1)$ symmetric systems.
We compare these behaviors in interacting systems with Abelian and non-Abelian continuous symmetries and in free-fermion systems, which differ in the geometry of their $k$-commutants.
Ultimately, this work provides a general geometric framework for systematically studying observables in noisy systems with continuous symmetries, including Haar-random circuits.
\end{abstract}
\maketitle
\tableofcontents

\section{Introduction} % MARK: 1. INTRO

Isolated dynamical many-body quantum systems are generically expected to thermalize.
Under unitary evolution, simple initial states become highly entangled over time, and local operators spread, leading to the decay of correlation functions.
While this phenomenon has been ubiquitously observed numerically in many systems~\cite{d2016quantum, mori2018thermalization, rigol2008thermalization, polkovnikov2011colloquium}, as well as in real world experiments, the microscopic study of thermalization has only been possible in very special settings.
One particularly successful set of models are local Haar random circuits~\cite{PhysRevX.7.031016, Nahum2018, keyserlingk2018hydro, chan2018solution,  fisher2023random,Bertini2026May}, where replica theory and Weingarten calculus~\cite{Collins_2010, PhysRevX.7.031016, hunterjones2019unitary} have enabled detailed computations of randomness averaged quantities such as correlation functions \cite{keyserlingk2018hydro,Nahum2018, nahum2022realtime} and entanglement entropies \cite{PhysRevX.7.031016,Zhou_2019, zhou_nahum} in the absence of symmetries.
Another family of random circuits are noisy Brownian models, where the discrete brickwall structure is replaced by a random time-dependent Hamiltonian~\cite{lashkari2013towards, sunderhauf2019quantum, Agarwal_2022, Zhou2019May,xu2019locality, bauer2017stochastic, 10.21468/SciPostPhys.15.4.175,    vardhan_entanglement_2024}, and where a similar replica formalism can be used for computing noise-averaged physical quantities, which have been done for entanglement growth~\cite{vardhan_entanglement_2024} and decay of correlations~\cite{Zhou2019May} in systems without symmetries.
While being conceptually similar to Haar random circuits, they are more flexible in many aspects, allowing for example a larger degree of tunability for the randomness.
The introduction of symmetries to quantum many-body systems brings in new and varied behavior.
In particular, continuous symmetries lead to the phenomenon of hydrodynamics, where late-time dynamics of symmetric systems display rich universal structure associated with the conserved charges of the system.
This is well-understood for simple correlation functions of the densities of the conserved charge, both in Haar random circuits with symmetries~\cite{Rakovszky2018Sep,Khemani2018Sep}, as well as more general systems where there are mature effective field theories of hydrodynamics that capture the late-time physics of these quantities~\cite{crossley,liu2018lectures,  2026arXiv260602391D}.
The same level of understanding does not exist for other quantities that are not completely dependent on these charge densities, such as the entanglement or other non-hydrodynamic correlators.
Take for example, the case of entanglement growth.
In the absence of symmetries, there is a detailed microscopic understanding of it via the entanglement membrane picture \cite{ PhysRevX.7.031016, jonay_huse, mark_membrane, Zhou_2019,vardhan_entanglement_2024}, which provides an evolution equation for the local entanglement densities as a function of time, and ultimately leads to a ballistic $\sim t$ entanglement growth.
Continuous symmetries are known to slow down this growth in some cases, so that in $U(1)$ symmetric systems, Rényi entanglement entropies are expected to exhibit a diffusive $\sim \sqrt{t}$ growth~\cite{Rakovszky_2019}.
This is mainly known from rigorous non-perturbative bounds that show that the growth of Rényi entropies is limited by the contributions of configurations involving rare \textit{void} regions in the initial state~\cite{Huang_2020, Rakovszky_2019,Zhou2020Jul}, whose diffusively-slow melting is enforced by the $U(1)$ symmetry.
However, a microscopic picture of the behavior of the entanglement analogous to the entanglement membrane without symmetries, has not been established in symmetric systems.
While there have indeed been many attempts to compute these entropies in $U(1)$-symmetric Haar random circuits, they have only been tractable in some limits:
either by bounding the entanglement in terms of other easier quantities~\cite{Zhou2020Jul, Han2023Dec}, or in a peculiar limit~\cite{Rakovszky_2019,Turkeshi_mpemba} where each charged spin in the $U(1)$ symmetric system is coupled to a neutral $q$-dimensional qudit, which is taken to be large and leading terms are kept in a $1/q$ expansion, after which a $q \rightarrow 1$ limit is taken.
Such limits are common for the study of many observables in Haar random systems with continuous symmetries.
Haar random circuits with non-Abelian continuous symmetries such as $SU(2)$ are even more poorly understood in the literature since the statistical mechanics models they lead to exhibit an even lesser degree of tractability beyond the simplest two-point observables~\cite{PhysRevB.108.054307, Li2025Oct}.
In this work, our goal is to overcome these technical obstacles and provide a simple framework that applies very generally to the study of various observables in noisy systems with arbitrary continuous symmetry groups.
This progress is enabled using the recent insights into the structure of dynamics of Brownian evolutions and their relations to commutant algebras~\cite{zanardi2001virtual, bartlett2007reference, read2007enlarged, moudgalya2022, moudgalya2022from, moudgalya2022exhaustive, moudgalya2023numerical,moudgalya2024,vardhan_entanglement_2024, lastres2026}.
The averaged evolution of physical quantities under unitary Brownian dynamics reduces to matrix elements of imaginary time evolution under an appropriate effective Hamiltonian \cite{lashkari2013towards,bauer2017stochastic,sunderhauf2019quantum, xu2019locality, moudgalya2024,vardhan_entanglement_2024}, which lives on a replicated Hilbert space, and is sometimes also referred to as super-Hamiltonian.
The number of copies of the Hilbert space required depends on the quantity of interest, e.g., simple two-point correlation functions require two copies (sometimes referred as one replica), whereas the $k$-th Rényi entanglement entropies require $2k$ copies (or $k$ replicas).
Given the imaginary-time evolution, the late-time dynamics is determined by the low-energy physics of these effective Hamiltonians.
What enables the understanding of this evolution,  is the fact that the ground states of these effective Hamiltonians on $k$ replicas are in direct correspondence with the \textit{$k$-commutant} of the random circuit, i.e., the symmetry algebra of $k$ identical copies of the system~\cite{Gross2007Designs, Dankert2009, HarrowLow2009}.
In systems that form so-called symmetric $k$-designs~\cite{hearth2025unitary, liu2024unitary, mitsuhashi2025unitary}, which is expected of generic interacting symmetric systems for small $k$, these are in turn known to be completely specified in terms of permutations of $k$ replicas of the original symmetry operators~\cite{enriched-phases,Li2024Dec}.
Given the knowledge of these ground states, it is quite natural to understand the nature of excitations on top of them, and hence their low-energy physics.
For example, this idea has been previously applied to show that for $k = 1$ replicas, the ground states of effective Hamiltonians of $U(1)$-symmetric systems can be mapped onto those of a Heisenberg ferromagnet, and that the spin-wave excitations on top of them lead to the diffusive behavior of hydrodynamic two-point correlation functions~\cite{ogunnaike2023unifying, hearth2025unitary, moudgalya2024}.
Further, in systems without symmetries, the ground states for $k = 2$ replicas were shown to be similar to those of the Ising ferromagnet, and the natural domain wall excitations on top of them led to the entanglement membrane picture for Rényi entropies~\cite{vardhan_entanglement_2024}.
In this work, we show that in the presence of a continuous symmetry group $G$, the effective Hamiltonians for Brownian models on any number $k$ of replicas have the structure of a generalized Heisenberg ferromagnet, with a continuous manifold of ground states.
These manifolds are analogous to the Bloch sphere of spin-coherent states that parametrize the standard Heisenberg ferromagnet.
Similar ferromagnetic models naturally appeared in our study of $k$-commutants of free fermionic systems~\cite{lastres2026geometryfreefermioncommutants}, where the ground states are parametrized by appropriate (orthogonal) Grassmannian manifolds.
This structure motivates the description of excitations in terms of smooth spatially-varying ferromagnetic configurations, similar to spin-waves of spin-coherent states in a standard Heisenberg ferromagnet.
However, while traditional continuous ferromagnets are parametrized by manifolds which are either smooth or discrete, the $k$-commutant manifolds that naturally emerge in this work possess \textit{singularities}, and should mathematically be described as general complex projective varieties.
We show that these singularities are a generic feature in interacting systems, and occur due to the presence of \textit{frozen states} that are often present in systems with continuous symmetries.
They have many interesting dynamical implications, and in particular are responsible for the formation of void regions closely related to the ones have appeared in earlier lower bounds for the diffusive growth of entanglement~\cite{McCulloch_2026,mcculloch2026longlivedlocalquantumcoherences} and sub-exponential decay of non-hydrodynamic correlators~\cite{Rakovszky_2019,Zhou2020Jul}.
Given these ground state manifolds, the low-energy physics of the corresponding effective Hamiltonians can be studied using the time-dependent variational principle (TDVP)~\cite{tdvp20} formulated on the space of smooth variations on top of this ground state manifold. 
The TDVP equations are then described in terms of the quantum geometric properties of the ground state manifold itself, such as its Fubini-Study metric~\cite{Bengtsson2017Aug}.
This leads to a simple geometry picture for the imaginary-time evolution of various initial states under effective Hamiltonians derived for Brownian evolutions.
For observables of interest here, these initial states are particularly simple:  for Rényi entanglement they correspond to domain-walls between two ground states, and for correlators they correspond to single-particle excitations on top of ground states.
While TDVP is well understood on smooth manifolds~\cite{tdvp20}, here we are necessarily working with manifolds with singularities, are forced to provide additional prescriptions for the dynamics of such singularities, whose behavior has implications for the late-time entanglement saturation physics.
We primarily illustrate this TDVP framework and numerically test its performance in the ferromagnetic Heisenberg model, where the ground state manifold is smooth, and then in a toy model we construct, which has a ground state manifold with a single point-like singularity.
We then generalize the results obtained there to realistic effective Hamiltonians corresponding to two symmetries: (i) $U(1)$, where the ground state manifold generally contains two point-like singularities, and (ii) $SU(2)$ and non-Abelian symmetries, where the singularities form a continuum.
The effective description we construct is completely based on the structure of the $k$-commutant manifolds, which is mostly insensitive to the microscopic details of the Brownian models, as long as they form symmetric $k$-designs.
Hence we expect our results to apply very generally to interacting noisy symmetric  systems, and also to the study of many other kinds of observables.
The rest of this paper is organized as follows.
In Sec.~\ref{sec:brown}, we review the physics of Brownian evolutions and their connections to $k$-commutants.
Then in Sec.~\ref{sec:fgsman}, we discuss the general theory behind characterizing ferromagnetic ground state manifolds, and we discuss the manifolds that naturally appear as $k$-commutants of systems considered in this work.
In Sec.~\ref{sec:tdvp} we discuss the low-energy excitations on top of these ground state manifolds and methods for studying time-evolution of various initial states within the variational space defined by smooth variations on top of the manifold, mostly using TDVP, and we demonstrate its application to the ferromagnetic Heisenberg model.
In Sec.~\ref{sec:esn}, we introduce a toy model where the ground state manifold contains a singularity, which reveals all the features that appear more generally in effective Hamiltonians corresponding to continuous symmetries.
Based on this phenomenology, we study $U(1)$ symmetric systems in Sec.~\ref{sec:u1}, where we particularly focus on Rényi entanglement growth as well as non-hydrodynamic correlators, both of which are influenced by the singularities in the manifold.
Finally in Sec.~\ref{sec:su2}, we generalize this understanding to $SU(2)$ and other non-Abelian continuous symmetries.
We conclude in Sec.~\ref{sec:outlooks} with a summary of results and open questions.

\section{Review of Brownian Models for Replicated Systems}\label{sec:brown} % MARK: 2. BROWN
We will focus on systems with a tensor product Hilbert space $\mathcal{H}$ and study dynamics under a time-dependent \textit{Brownian} Hamiltonian~\cite{lashkari2013towards, xu2019locality, sunderhauf2019quantum, bauer2017stochastic, Bernard_2021, 10.21468/SciPostPhys.15.4.175, swann2025, ogunnaike2023unifying, moudgalya2024, vardhan_entanglement_2024}
\begin{equation}
	H(t)=\sum_{\alpha} J_\alpha(t) h_\alpha\label{eq:brownhamdef}
\end{equation}
where $\{h_\alpha\}$ is the set of Hermitian generators that defines the model, and
$\{J_\alpha(t)\}$ are i.i.d. Brownian random variables with zero mean $\langle J_\alpha(t)\rangle = 0$ and white-noise correlations $\langle J_\alpha(t)J_{\alpha'}(t')\rangle = 2\kappa\,\delta_{\alpha\alpha'}\,\delta(t-t')$.
One can then formally perform the statistical average over the randomness and write down an averaged time-evolution operator as an imaginary-time evolution of an effective positive semidefinite Hamiltonian~\cite{moudgalya2024}:
\begin{equation}
		\overline{U(t)} = \overline{ \mrm{T}\{e^{-i\int_0^t H(t')\,\dd t'}\}}=e^{-\kappa Pt},\;\;\;
		P\defeq\sum_\alpha h_\alpha^2.
\label{eq:unitaryaverage}
\end{equation}
Since $h_\alpha^2\geq 0$ for all $\alpha$, any ground state of $P$ is \textit{frustration free}, in the sense that:
\begin{equation}
	P\ket{\psi}=0\iff \forall\alpha: h_\alpha^2\ket\psi =0.
\end{equation}
In order to study observables which can generally be non-linear functions of the density matrix $\rho$, we are typically not interested in plain unitary averages as in Eq.~(\ref{eq:unitaryaverage}), but rather in averages over multiple copies of the unitary, which are given by
\begin{equation}\label{eq:superham}
	\overline{(U(t)\ot U^*(t))^{\ot k}}=e^{-\kappa P^{(k)} t},\;\;\; P^{(k)}\defeq \sum_\alpha \left(\ad{h_\alpha}^{(k)}\right)^2,
\end{equation}
where we have interpreted the replicated unitary $(U(t)\ot U^*(t))^{\ot k}$ as a Brownian evolution coming from an associated \textit{$k$-replica} (or $2k$-copy) Hamiltonian \cite{enriched-phases, swann2025} defined on $\mc H^{\ot 2k}$, given by
\begin{gather}
    H^{(k)}(t)=\ad{H(t)}^{(k)} = \sum_\alpha J_\alpha(t) \ad{h_\alpha}^{(k)}\\
	\ad{M}^{(k)} \defeq \sum_{l=1}^k \1^{\ot 2(l-1)}\ot (M\ot\1-\1\ot M^*)\ot \1^{\ot 2(k-l)},\nonumber
\end{gather}
and then we have applied Eq.~(\ref{eq:unitaryaverage}) for performing the average on the replicated Hilbert space.
$P^{(k)}$ is referred to as an \textit{effective Hamiltonian} or \textit{super-Hamiltonian}.
If we enumerate the various copies of the Hilbert space using an index $a\in\{1,...,2k\}$, then the replicated Brownian Hamiltonian $H^{(k)}(t)$ evolves the odd-numbered copies forwards in time by $H(t)$, and the even-numbered copies backwards in time by $-H(t)^*$.
For example, the time evolution of a single density matrix $\rho(t)$ is generated by the superoperator $H^{(1)}(t)$ acting on two copies.
Both the time-dependent Hamiltonian $H^{(k)}(t)$ and the averaged effective Hamiltonian $P^{(k)}$ are symmetric under permutation of the even copies and of the odd copies among themselves, described by the group $S_k\times S_k$, where $S_k$ is the $k$-th symmetric group.
In addition, if $H(t)$ is symmetric under the action of a group $G$, then the replica models also possess a replicated symmetry $G^{\times 2k}$, which describes the action of the symmetry group on each separate copy.
Notice that while $H^{(k)}(t)$ evolves all copies independently, $P^{(k)}$ will in general couple them.
In this work we will be studying models where the Hilbert space describes a local many-body structure, so that $\mc H=\mc H_\mrm{loc}^{\otimes L}$ and interactions only occur along bonds connecting neighboring sites.
In this case we may define a spatially-homogeneous set of generators $\{h_{\alpha,ij}\}$ where the indices $\langle ij\rangle$ specify the bond, and $\alpha$ specifies an operator acting on the 2-site Hilbert space $\mc H_\mrm{loc}\otimes \mc H_\mrm{loc}$.
Then the effective Hamiltonian for the $k$-replica model will decompose as a sum of bond terms:
\begin{equation}\label{eq:superham-local}
	P^{(k)} = \sum_{\langle i j\rangle} P^{(k)}_{ij},\qquad P^{(k)}_{ij}\defeq\sum_\alpha\left(\ad{h_{\alpha,ij}}^{(k)}\right)^2.
\end{equation}
It is therefore useful to picture replicated many-body systems as  having the same geometry as the original system, but composed of ``replicated sites'' whose local Hilbert space is $\mc H_\mrm{loc}^{\otimes 2k}$.

\subsection{States and Operators in Replicated Space} % MARK: *notation

We review here the conventions that will be used throughout the paper to describe states and operators on the replicated Hilbert space $\mc H^{\otimes 2k}$.
To indicate the action of an operator $O$ acting on a given replica $a\in\{1,...,2k\}$ we will use the notation
\begin{equation}
	O^{[a]}\defeq\begin{cases}
		\1^{\ot (a-1)}\ot O\ot \1^{\ot(2k-a)}\quad&\text{if }a\text{ is odd},\\
		\1^{\ot (a-1)}\ot O^*\ot \1^{\ot(2k-a)}\quad&\text{if }a\text{ is even}.
	\end{cases}
\label{eq:opconvention}
\end{equation}
This convention simplifies a few common expressions when working with Brownian models, for example, we have:
\begin{gather}
	\ad{M}^{(k)} = -\sum_{a=1}^{2k} (-1)^{a} M^{[a]},\label{eq:ad-notation}\\
	(U(t)\ot U^*(t))^{\ot k}=\prod_{a=1}^{2k}U^{[a]}.
\end{gather}
For states on a replicated  Hilbert space $\mH^{\ot 2k}$ (where $\mH$ might also be the single-site Hilbert space $\mH_{\rm loc}$ in certain situations), there are multiple convenient notations that we use interchangeably depending on context.
First, we may simply use the notation:
\begin{equation}\label{eq:1s-basis}
	\ket{{n_1}{n_2}...{n_{2k}}}\defeq\ket{{n_1}}\ot\ket{{n_2}}\ot ...\ot \ket{{n_{2k}}},
\end{equation}
where $\{\ket{n}\}$ is a basis for $\mH$.
Alternatively, we may define states in {$\mH^{\ot 2k}$} starting from a list of $k$  operators by first writing operators in $\mrm{End}(\mc H)$ as states on $\mc H\ot\mc H$ as
\begin{equation}\label{eq:liouvillianrep}
	O = \sum_{n,n'}O_{nn'}\ketbra{n}{n'}\;\;\longrightarrow\;\; \sum_{n,n'}O_{nn'}\ket{n}\ot\ket{n'},
\end{equation}
and then inserting each of the $k$ operators on a unique pair of Hilbert spaces, the first of which is odd-numbered and the second even-numbered.
For example, we can define:
\begin{equation}
	|O^{(1)}\ot \!...\!\ot O^{(k)}\rangle\defeq\!\!\!\sum_{\{n_m,n_m'\}}\!\!O^{(1)}_{n_1n_1'}...O^{(k)}_{n_kn_k'}|n_1n'_1...n_kn'_k\rangle,
\end{equation}
or more generally for $\sigma\in S_k$ we define:
\begin{equation}\label{eq:permstate}
	|O^{(1)}\ot \!...\!\ot O^{(k)};\sigma\rangle\defeq \mc S_\sigma |O^{(1)}\ot \!...\!\ot O^{(k)}\rangle,
\end{equation}
where $\mc S_\sigma$ applies the permutation $\sigma$ on the set of all even-numbered copies of the replicated Hilbert space, so that
\begin{equation}\label{eq:subgroup}
	\mc S_\sigma\ket{n_1n'_1...n_kn'_k} = |n_1n'_{\sigma^{-1}(1)}...n_kn'_{\sigma^{-1}(k)}\rangle.
\end{equation}
A similar definition will also hold when $O$ is not a local operator, but is instead an operator acting on the whole Hilbert space $\mc H$.
The pairings of the Hilbert space copies associated to the states of the form Eq.~\eqref{eq:permstate} are pictured in Fig.~\ref{fig:pairings}.
A similar definition will also hold when $O$ is not a global operator, but instead a local operator acting on the local Hilbert space $\mc H_\mrm{loc}$.
Note that we will use $e$ to represent the identity permutation and $\eta$ to represent the cyclical permutation $(12 \cdots k)$.
\begin{figure}
	\centering
	\includegraphics[width=0.9\linewidth]{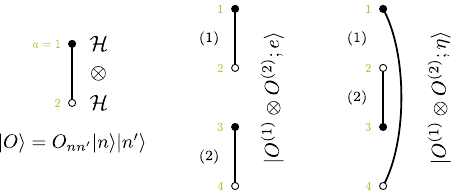}
	\caption{Pictorial representations of the vectorization of the operator $O\mapsto\ket O$ on a duplicated Hilbert space, and the states defined in Eq.~\eqref{eq:permstate} for a $k=2$ replica system ($e\in S_2$ is the identity permutation and $\eta\in S_2$ is the swap). }
	\label{fig:pairings}
\end{figure}

\subsection{Ground States of Effective Hamiltonians and $k$-Commutants}\label{subsec:gseffective} % MARK: *symms
The ground states of the effective Hamiltonians $P^{(k)}$ on $k$ replicas have deep connections to symmetries of the generating terms $\{h_\alpha\}$.
In the $k=1$ case for example, the ground states of $P^{(1)}$ [Eq.~\eqref{eq:superham}] are exactly the operators $\ket{Q}\in\mc H\otimes \mc H$ which for all $\alpha$ satisfy:
\begin{equation}\label{eq:commutant-states}
	0=P^{(1)}\ket Q = \mc L^{(1)}_{h_\alpha}\ket Q = \ket{[h_\alpha, Q]},
\end{equation}
where the second equality comes from the fact that $P^{(1)}$ is a sum of positive semi-definite terms.
This means that $\ket{Q}$ is a ground state of $P^{(1)}$ if and only if $Q$ is in the \textit{commutant algebra} of the generating set $\{h_\alpha\}$, i.e., the algebra of all operators that commute with each generator $h_\alpha$ \cite{moudgalya2022from,moudgalya2024,lastres2026}, or equivalently, with the entire set of Brownian unitaries $\{U(t)\}$ generated by the terms $\{h_\alpha\}$.
The operators $Q$ should therefore be interpreted as the symmetries or conserved charges in the original circuit model.
Similarly, $\ket {\mc Q}\in\mc H^{\ot 2k}$ is a ground state of $P^{(k)}$ if and only if the associated superoperator $\mc Q\in \mrm{End}(\mc H^{\ot k})$ belongs to the \textit{$k$-commutant}~\cite{Gross2007Designs, Dankert2009, HarrowLow2009, hearth2025unitary, liu2024unitary, mitsuhashi2025unitary, lastres2026, lastres2026geometryfreefermioncommutants}, i.e. the set of all operators in $\textrm{End}(\mathcal{H}^{\ot k})$ commute with all the replicated generators:
\begin{equation}
	0=\mc L^{(k)}_{h_\alpha}\ket{\mc Q} = \Big|\Big[\sum_{m=1}^k \1^{\ot (m-1)}\ot h_\alpha\ot \1^{\ot(k-m)},\,\mc Q\Big]\Big\rangle.
\end{equation}
In other words, $\mc Q$ belongs to the $k$-commutant if and only if it commutes with the full family of replicated Brownian unitaries $U(t)^{\ot k}$, where $U(t)$ is a Brownian unitary generated by $\{h_\alpha\}$.
It is easy to see that the ground states of the effective Hamiltonian $P^{(k)}$ are partially dictated by the ground states of $P^{(k')}$ for smaller replica number $k'<k$.
For example, we have that~\cite{vardhan_entanglement_2024, lastres2026}
\begin{equation}\label{eq:trivialgs}
	P^{(1)}\sket{Q^{(m)}} = 0\,\implies\, P^{(k)} |Q^{(1)}\otimes ... \otimes Q^{(k)};\sigma\rangle = 0,
\end{equation}
where $\sigma\in S_k$ is any permutation [cf.~Eq.~\eqref{eq:permstate}].
This means that if $\{Q^{(m)}\}$ are a set of symmetry operators of $\{h_\alpha\}$, then a large set of ground states of $P^{(k)}$ can be constructed using them.
In particular, since the identity operator $\1$ commutes with any other operator, we always have
\begin{equation}\label{eq:ssb-perm}
	P^{(k)}| \1^{\ot k};\sigma\rangle = 0,\quad \sigma\in S_k
\end{equation}
for any choice of Brownian generators $\{h_\alpha\}$.
These ground states are associated with the spontaneous breaking of the $S_k\times S_k$ replica symmetry of $P^{(k)}$ to its diagonal subgroup~\cite{enriched-phases}, and are labeled by coset elements, which we identify with permutations $\sigma\in S_k$ through the definition Eq.~\eqref{eq:subgroup}.
As we will see, the states in Eq.~\eqref{eq:ssb-perm} are relevant for the computation of replica observables such as Rényi entropies and OTOCs.
Similarly, the ground states identified in Eq.~\eqref{eq:trivialgs} for symmetric models, are associated with the spontaneous breaking of the replicated symmetry, given by the symmetry operators
\begin{equation}
    [Q^{(1)}\ot ...\ot Q^{(2k)},P^{(k)}]=0.
\end{equation}
If the symmetries of $\{h_\alpha\}$ form a group $G$,  then the replicated symmetry has the structure $G^{\times 2k}$, and it intertwines non-trivially with the replica symmetry $S_k\times S_k$, forming a semidirect product \cite{enriched-phases}.
This structure of the ground states is sometimes associated to Strong-to-Weak Spontaneous Symmetry Breaking (SW-SSB), particularly when $k = 1$~\cite{ogunnaike2023unifying, moudgalya2024,  Huang2025Mar,Chen2025Feb,Lessa2025Mar,Sala2024Oct}.
Note that not all ground states of $P^{(k)}$ need to be of the form Eq.~\eqref{eq:trivialgs}; for example, in free fermion systems $U(t)^{\ot k}$ has a larger continuous symmetry than just the discrete $S_k$, which then results in a larger ground state space for $P^{(k)}$~\cite{enriched-phases,swann2025,fava2024,lastres2026geometryfreefermioncommutants}.
If ground states of the form Eq.~\eqref{eq:trivialgs} do actually exhaust the full set of ground states of $P^{(k)}$, then the set of generators $\{h_\alpha\}$ is said to be a \textit{symmetric $k$-design}~\cite{Gross2007Designs, Dankert2009, HarrowLow2009}.
The fact that for $k\geq 2$ effective Hamiltonians for free-fermion systems have extra ground states, implies that they do not form symmetric $k$-designs~\cite{poetri2026, larocca2026, lastres2026, lastres2026geometryfreefermioncommutants}.
In this work, we will mostly focus on symmetric systems which do form symmetric $k$-designs for low values of $k$, and we will mostly be focusing on $k = 2$.

\subsection{Observables as Matrix Elements of the Evolution Operator} % MARK: * observables
\label{sec:observables-general}
A standard technique in the study of observables through replica Hilbert spaces is to map such observables to specially constructed states on the replicated Hilbert space, whose time-evolution under $P^{(k)}$ reproduces the mean value of the time-evolved observable.
Below we review these expressions for the observables we are interested in.
\subsubsection{Hydrodynamic and non-hydrodynamic correlation functions}\label{subsubsec:hydrononhydro}
For example, two-point correlation functions can for example be computed for $k=1$ as~\cite{moudgalya2024}:
\begin{equation}\label{eq:simple-correlator}
	\begin{split}
		\overline{\langle A(0)\+ B(t) \rangle} = \frac{1}{D}\mrm{tr}(A\+\overline{B(t)}) =\frac{1}{D}\langle{A}|\overline{B(t)}\rangle =\qquad \\
		= \frac{1}{D}\bra{A}\overline{U(t)\otimes U^*(t)}\ket{B} = \frac{1}{D}\bra{A}e^{-\kappa P^{(1)}t}\ket{B}
	\end{split}
\end{equation}
where $\ket{A}$ and $\ket B$ are the state representations of the operators $A$ and $B$, from Eq.~\eqref{eq:liouvillianrep}, and $D=\mrm{dim}(\mc H)$ is the dimension of the many-body Hilbert space.
Similarly, the averaged squared norm of two-point correlation functions can be written as
\begin{equation}\label{eq:squared-correlator}
	\overline{|\langle B(t) A(0)\+ \rangle|^2} = \frac{1}{D^2}\langle{A\ot A\+}|\,e^{-\kappa P^{(2)}t}\,|{B\ot B\+}\rangle.
\end{equation}
We will typically be interested in \textit{autocorrelation} functions of local operators supported on a single site $i$, and hence we usually set $A = B = O_i$.
We can write the real space expressions of the boundary states in Eqs.~(\ref{eq:simple-correlator}) and (\ref{eq:squared-correlator}) as
\begin{gather}
    \ket{O_i} = \bigotimes_{j \neq i}\ket{\mathds{1}}_j \ket{O}_i,\\
    \sket{O_i \otimes O_i^\dagger} = \bigotimes_{j \neq i}\sket{\mathds{1}^{\otimes 2}; e}_j \sket{O \otimes O^\dagger}_i.
\label{eq:replicatedhydroop}
\end{gather}
Eq.~(\ref{eq:squared-correlator}) for such autocorrelation functions then reads
\begin{align}
    \label{eq:squared-autocorrelator}
	\overline{|\langle O_i(t) O_i(0)\+ \rangle|^2} &= \frac{1}{D^2}\langle{O_i\ot O_i\+}|\,e^{-\kappa P^{(2)}t}\,|{O_i\ot O_i\+}\rangle \nonumber\\
    &=\frac{1}{D^2}\norm{e^{-\kappa P^{(2)}t/2}\,|{O_i\ot O_i\+}\rangle}^2,
\end{align}
where $\norm{\cdot}$ denotes the standard norm for quantum states.
In systems with continuous symmetry, a traceless local operator $O_i$ is referred to as being \textit{hydrodynamic} if it has a non-zero overlap on any of the non-trivial symmetry operators of the system.
In this context, this means that the operator is hydrodynamic if $\ket{O_i}$ a non-zero overlap on the $k = 1$ commutant, i.e., the ground states of $P^{(1)}$:
\begin{equation}
    \exists Q:\;\;\; P^{(1)}\ket{Q} = 0\;\;\text{and}\;\;\braket{O_i}{Q} \neq 0.
\label{eq:hydrodefn}
\end{equation}
This also implies a non-zero overlap of $\sket{O_i \otimes O_i^\dagger}$ on the states of $e$ permutation in the $k = 2$ commutant, which are the ground states of $P^{(2)}$ with $\sigma = e$ in Eq.~(\ref{eq:trivialgs}).
Conversely, a traceless operator is \textit{non-hydrodynamic} if $\ket{O_i}$ has a zero overlap on the $k = 1$ commutant.
This also implies a zero overlap of the state $\sket{O_i \otimes O_i^\dagger}$ on the $e$ permutation states in the $k = 2$ commutant. Note however that $\sket{O_i \otimes O_i^\dagger}$ generically has a non-zero overlap on the $\eta$ permutation states in the $k = 2$ commutant, i.e., the ground states of $P^{(2)}$ with $\sigma = \eta$ in Eq.~(\ref{eq:trivialgs}).
This interplay of overlaps and permutations plays an important role in characterizing the fundamental degrees of freedom and dynamics of the two-replica effective models.
\subsubsection{Entanglement entropies}
Another class of observables of interest for $k$-replica random circuits, is the averaged $k$-th Rényi purity
\begin{gather}\label{eq:renyi-purity}
	\overline{E_k(t)}\defeq\overline{\tr_A(\rho_A^k(t))}=\bra{A\!:\!\bar A}e^{-\kappa P^{(k)}t}\ket{\rho^{\ot k}}\\[.4em]
	\ket{A\!:\!\bar A}\defeq\bigotimes_{i\in \bar A,\,j\in A}|\1^{\ot k};e\rangle_{i}|\1^{\ot k};\eta\rangle_{j},\label{eq:definition-dw-id-swap}
\end{gather}
where $A$ and $\bar A$ are complementary subsystems of the full many-body system, $e\in S_k$ is again the identity permutation, and $\eta\in S_k$ is the cyclical permutation $(12...k)$, which for $k=2$ reduces to the swap permutation.
In the computation of the purity, the $e$ permutation implements the partial trace $\rho_A=\tr_{\bar A}(\rho)$, while the $\eta$ permutation implements the cyclic trace over $A$.
From this observable, one can compute the \textit{annealed} average of the $k$-th Rényi entanglement entropy for the given initial state $\rho$ by simply taking the logarithm of Eq.~(\ref{eq:renyi-purity}).
For $k > 1$, this is in general is a lower bound on the actual \textit{quenched} average:
\begin{equation}
	\overline{S_k(t)}\defeq\!\frac{1}{1-k}\overline{\log(\tr_A(\rho_A^k(t)))} \geq \frac{1}{1-k}\log\big(\overline{E_k(t)}\big)\!
\end{equation}
The discrepancy between the two averages is controlled by higher moments of the Rényi purity within the Brownian ensemble.
In practice, for random many-body unitary dynamics the difference between the two averages is expected to only amount to a subleading contribution, as is seen numerically~\cite{Rakovszky_2019, swann2025}, and also analytically justified in some cases~\cite{Zhou_2019}.
The von Neumann entanglement entropy can in principle be extracted from $E_k(t)$ using an appropriate $k \rightarrow 1$ replica limit, but we will not discuss that limit in this work.
\subsubsection{General structure}
The fact that the states $|\1^{\ot k};\sigma\rangle$ are always ground states of the effective model $P^{(k)}$ for any choice of generators $\{h_\alpha\}$ [cf. Eq.~\eqref{eq:ssb-perm}] provides a nice interpretation for the dynamics of entanglement and correlators in terms of low-energy physics of $P^{(k)}$.
If the subsystems $A$ and $\bar A$ are mostly contiguous, the state $\ket{A\!:\!\bar A}$ can be interpreted as a (set of) domain wall excitations between different ground states on the complementary subsystems $A$ and $\bar A$.
Hence from Eq.~(\ref{eq:renyi-purity}), the averaged entanglement dynamics can be understood in terms of domain wall dynamics in the effective model $P^{(k)}$~\cite{vardhan_entanglement_2024}.
Similarly, if in the correlators of Eqs.~(\ref{eq:simple-correlator})-(\ref{eq:squared-correlator}) one chooses $A$ and $B$ to be strictly local observables, then the associated states will look like localized excitation on a fully-polarized $|\1^{\ot k};\sigma\rangle$ background, and the averaged evolution of correlators is described by the imaginary-time dynamics of these localized excitation dynamics in $P^{(k)}$~\cite{moudgalya2024}.
\subsection{Accidental Symmetries of Effective Hamiltonians}\label{subsec:accidentalsym}
Finally, we discuss some symmetries of the effective Hamiltonians that can help make them more tractable.
As discussed above, in general $P^{(k)}$ has a global permutation symmetry $S_k \times S_k$, as well as a replicated symmetry $G^{\times 2k}$ if the symmetries of $\{h_\alpha\}$ form a group $G$.
In addition, for certain choices of $\{h_\alpha\}$, $P^{(k)}$ can have additional ``accidental'' symmetries that simplify its analysis.
In particular, here we will only be interested in its single-site symmetries, which are defined as single-site operators which commute with $P^{(k)}$.
Given the form of $P^{(k)}$ in Eq.~(\ref{eq:superham}), it is easy to see that any strictly local operator $A_i$ which for each $\alpha$ either commutes or anticommutes with $h_\alpha$ produces a local symmetry operator for $P^{(k)}$:
\begin{equation}\label{eq:acc-symms}
    \forall\alpha:\begin{cases}
        [\,A_i, \,h_\alpha] = 0,\ \text{or}\\
        \{A_i, h_\alpha\} = 0,
    \end{cases}
	\!\!\!\!\!\!\implies\Big[\prod_{a=1}^{2k} A_i^{[a]},\,P^{(k)}\Big]=0.
\end{equation}
Local symmetries often lead to an effective reduction of the local Hilbert space dimension, since they provide a conserved quantum number on each site.
In their presence, we will be able to study replica observables by only considering a subsector of the effective Hamiltonian $P^{(k)}$.

\begin{figure*}
	\centering
	\textbf{(a)}\includegraphics[width=0.30\linewidth,valign=t]{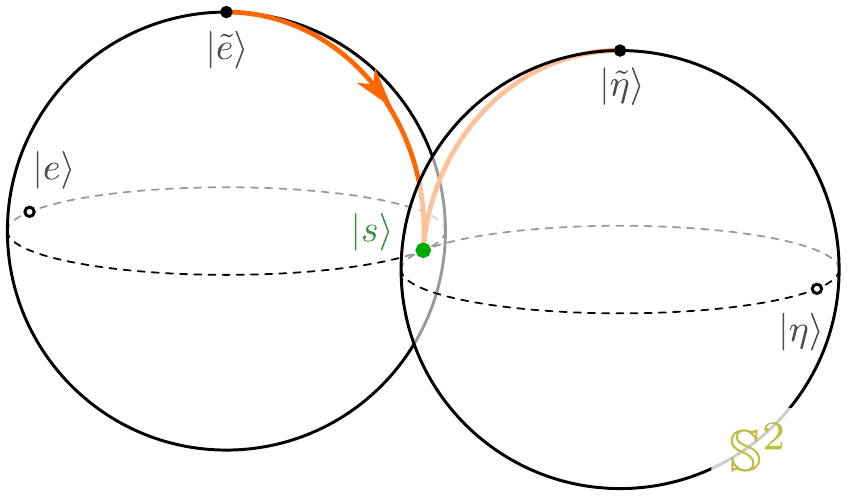}
	\textbf{(b)}\includegraphics[width=0.30\linewidth,valign=t]{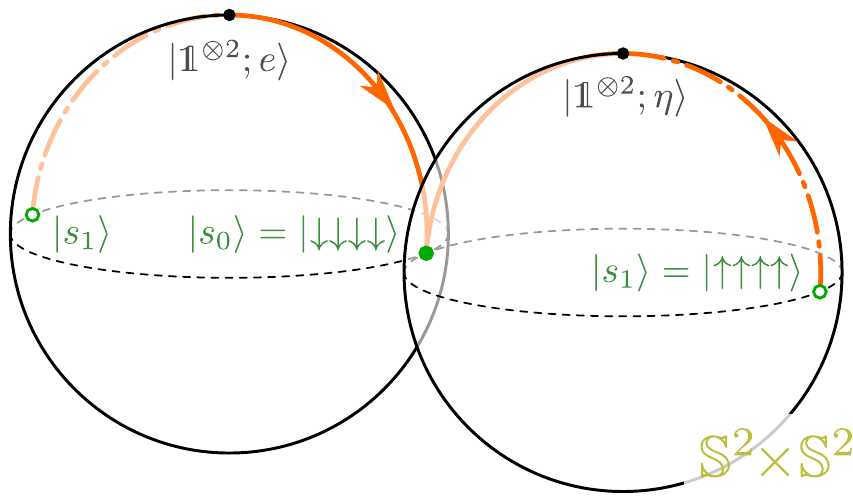}
	\textbf{(c)}\includegraphics[width=0.27\linewidth,valign=t]{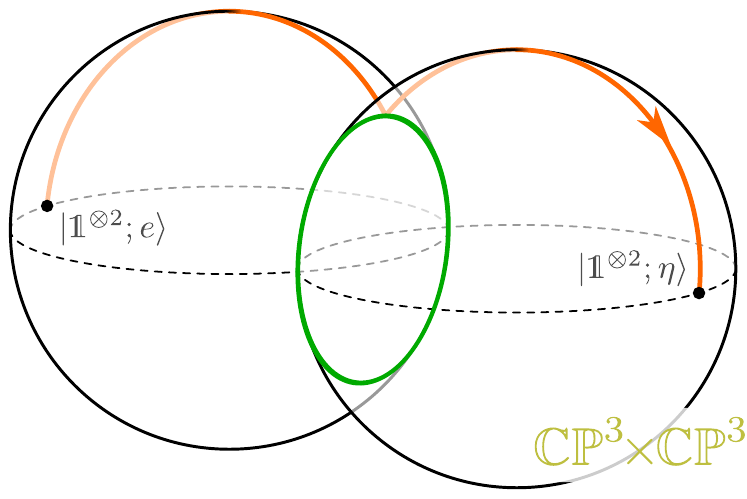}
	\caption{Illustration of the ferromagnetic ground state manifolds of the models studied here. The lines show paths between the $e$ and $\eta$ branches that are relevant for entanglement domain wall melting. (a) Ground state manifold and of the $es\eta$ toy model, composed of two Bloch spheres $\mb S^2$ which share the $s$-point. (b) Commutant manifold for a $U(1)$ symmetry at $k=2$ replicas, composed of two $\mb S^2\times \mb S^2$ manifolds that touch at two points. (c) Commutant manifold for a $SU(2)$ symmetry at $k=2$ replicas, composed of two $\mb{CP}^3\times \mb{CP}^3$ manifolds that intersect on a $(\mb S^2)^{\times 4}$ continuous submanifold.}
	\label{fig:gsman-all}
\end{figure*}

\section{Geometry of Ferromagnetic Ground State Manifolds} % MARK: 3. METHODS
\label{sec:fgsman}
In this section, we introduce the geometric notion of ground state manifold for ferromagnetic models, then show how it can be applied to the case of effective Hamiltonians for replica models.
This naturally leads to effective descriptions of the imaginary time dynamics of these models, whose semiclassical limit can be described by geometric variational methods such as the Time-Dependent Variational Principle (TDVP), which we discuss in the next section.
Central to our approach is the observation that the effective Hamiltonians for the replica Brownian models under consideration  in many cases possess \textit{ferromagnetic} ground state spaces, which, as we explain in Sec.~\ref{sec:fgsmanrep} is a consequence of the fact that the symmetries of the system are on-site group actions.
In this context, ferromagnetism is defined as the condition that the space of all ground states of $P^{(k)}$ is a vector space spanned by some set of \textit{fully-polarized} states (i.e., permutation symmetric product states):
\begin{equation}
	\ker(P^{(k)}) = \text{span}\{\ket{v}^{\ot L}\},\;\;\{\ket{v}\}\subseteq\mc H_\mrm{loc}^{\ot 2k}. \label{eq:def-ferro}
\end{equation}
For Hamiltonians which satisfy this condition, the ground state space can be fully characterized by their \textit{ground state manifold}~\cite{lastres2026geometryfreefermioncommutants}, which parametrizes all single-site states $\ket v$ which correspond to fully-polarized ground states:
\begin{equation}\label{eq:gsman-def}
	M_{\mrm{GS}}^{(k)}\defeq\mb P\big\{\!\ket{v}\in \mc H_\mrm{loc}^{\ot 2k}:P_{1,2}^{(k)}(\ket{v}\ot\ket{v})=0\big\}.
\end{equation}
Here $\mb P\{\cdot\}$ denotes the projectivization operation (the condition $P_{1,2}^{(k)}(\ket{v}\ot\ket{v})=0$ is invariant under scalar multiplication $\ket{v}\mapsto\lambda\ket{v}$, so we identify valid local states up to an overall rescaling).
The ground states of effective Hamiltonians $P^{(k)}$ are directly associated with the $k$-commutants of the generating gateset $\{h_{\alpha,ij}\}$ (see Sec.~\ref{subsec:gseffective}), we will also refer to $M_\mrm{GS}^{(k)}$ as the \textit{$k$-th commutant manifold}.
It is possible that these manifolds are smooth, and as we discuss below, many such examples are well-known in the literature.
There is however another possibility---that to our knowledge has not been explored in the literature---where these manifolds possess singularities, and here we demonstrate many such examples below.
\subsection{Smooth manifolds}
A simple example of a smooth ground state manifold is that of the spin-$\frac{1}{2}$ ferromagnetic Heisenberg model, defined as
\begin{equation}\label{eq:HeisModelHam}
		H=J\sum_{\langle i j\rangle}(1-X_iX_j-Y_iY_j-Z_iZ_j),\quad J > 0.
\end{equation}
On any connected system of $L$ sites, the ground states of this model are well-known to form a tower of states between the two fully polarized states $\ket{\up \dots \up}$ and $\ket{\dn \dots \dn}$, given by
\begin{equation}
    \mrm{ker}(H) = \mrm{span}\{(S^-_{\tot})^n\ket{\up \dots \up}\}_{n = 0}^L,
\end{equation}
where $S^-_{\rm tot} \defn \sum\nolimits_i{S^-_i}$ is the total spin lowering operator.
This ground state space is equivalently be thought of as the span of   ferromagnetic configurations of spin-coherent states parametrized by a unit vector $\vec{n} = (\sin\theta\cos\varphi, \sin\theta\sin\varphi, \cos\theta)$ on $\mathbb{S}^2$~\cite{JMRadcliffe_1971, fradkin2013field, Altland_Simons_2010}, given
\begin{equation}
    \text{span}\{\ket{\vec{n}}^{\ot L}\},\;\;\;\ket{\vec{n}} \defn \cos\frac{\theta}{2} \ket\up + e^{i\varphi}\sin\frac{\theta}{2} \ket\dn.
\label{eq:blochsphere}
\end{equation}
This means that the ground state manifold is given by the Bloch sphere $\mb S^2$:
\begin{equation}\label{eq:HeisModelGS}
		{M_{\mrm{GS}}=\{\ket{\vec{n}}: \vec{n}\in\mb S^2\}} \cong \mb S^2.
\end{equation}
This structure leads to the well-known field-theoretical description of the low-energy physics of the Heisenberg ferromagnet in terms of spin-coherent states~\cite{kochetov95, PhysRevLett.55.537, fradkin2013field, Altland_Simons_2010}, for which we review in Appendix~\ref{app:path-integral}.
\subsection{Manifolds with singularities (complex projective varieties)}\label{sec:esn-man}
For general ferromagnetic models, however, $M_\mrm{GS}$ is not always a smooth manifold since it can possess singularities at some points.
In the mathematical language, $M_\mrm{GS}$ is in general a \textit{complex projective variety}, i.e., a subset of projective space which is identified by a set of homogeneous polynomial equations.
These singular features play an important role in the dynamics of the initial states we are interested in under the effective models $P^{(k)}$.
A simple relevant example occurs in the context of the ``$es\eta$ model'' [cf.~Eq.~\eqref{eq:esn-ham}] which we introduce in greater detail in Sec.~\ref{sec:esn}.
In short, the local Hilbert space is 3-dimensional, with basis $\{\ket e,\ket s,\ket \eta\}$; the interactions between $e$-$s$ and $\eta$-$s$ are of ferromagnetic Heisenberg type Eq.~\eqref{eq:HeisModelGS}, while the $e$-$\eta$ interaction is of the gapped XXZ type.
This results in the model having ground states spanned by two $SU(2)$ towers of states, one between the polarized states $\ket{e \dots e}$ and $\ket{s \dots s}$, and the other between the polarized states $\ket{s \dots s}$ and $\ket{\eta \dots \eta}$:
\begin{equation}
    \text{span}\{(S^{(se)}_{\tot})^n\ket{s \dots s},\,(S^{(s\eta)}_{\tot})^n\ket{s \dots s}\}_{n=0}^L,
\label{eq:esetaGS}
\end{equation}
where we have defined the generalized ``lowering'' operators $S^{(\alpha\beta)}_\mrm{tot} \defn \sum\nolimits_j \ketbra{\beta}{\alpha}_j$ for $\alpha,\beta \in \{e,s,\eta\}$.
Notice that the fully-polarized $\ket{s\dots s}$ state appears in both towers.
This ferromagnetic ground state space can then be described in terms of a ground state manifold composed of two Bloch spheres  [cf.~Eqs.~(\ref{eq:blochsphere}) and (\ref{eq:HeisModelGS})] one for each of the pairs of states $\{\ket e,\ket s\}$ and $\{\ket \eta,\ket s\}$ as
\begin{equation}
    \cos\frac{\theta}{2} \ket s + e^{i\varphi}\sin\frac{\theta}{2} \ket e,\ \text{ and }\ \cos\frac{\theta}{2} \ket s + e^{i\varphi}\sin\frac{\theta}{2} \ket \eta,\label{eq:esn-gs-param}
\end{equation}
which altogether should span the same space as Eq.~(\ref{eq:esetaGS}).
However, note that both Bloch spheres contain the $\ket s$ state, which will therefore form a singular contact point between the two spheres, as shown in Fig.~\ref{fig:gsman-all}a.
Mathematically, one can also apply the definition Eq.~\eqref{eq:gsman-def} to show that the manifold is simply
\begin{equation}
	M_{\mrm{GS}} = \mathbb{P}\{x\ket e+y\ket s +z\ket \eta:\ xz=0,\ x,y,z \in \mathbb{C}\}.\label{eq:gsman-esn}
\end{equation}
This single polynomial equation precisely defines the singular $es\eta$ ground state manifold as a subvariety of the projective plane $\mb{CP}^2$.
As we will discuss in the forthcoming sections, ferromagnetic ground states whose manifolds possess ``contact'' singularities naturally arise in the context of effective Hamiltonians $P^{(k)}$.
\subsection{Characterizations of ferromagnetic ground states}\label{subsubsec:FMcharac}
The notion of ground state manifold provides a great simplification to the study of the ground states of ferromagnetic Hamiltonians, since it reduces the global problem of computing all frustration-free ground states to the computation of the local geometric datum $M_{\mrm{GS}}$ (which is independent of the lattice structure of the many-body Hilbert space and only depends on the two-site Hamiltonian).
In Ref.~\cite{lastres2026geometryfreefermioncommutants} we recently exploited the ground state manifold framework to provide a geometric characterization of the $k$-commutants of free fermion systems, and we demonstrated methods to derive the ground state manifolds using their invariance under certain Lie group symmetries in the free-fermion replica problem.
In an upcoming work Ref.~\cite{the-algebra-paper}, we provide a more general algebraic framework to rigorously prove the forms of the ferromagnetic ground states and the associated manifolds starting from the Hamiltonians, which also applies to the Hamiltonians studied in this work.
The main statement is especially simple, and there we show how it can be used to automatically check ferromagnetism in specific cases, with the help of Computer Algebra Systems.
Those results are not strictly necessary here since, as we discuss in the next section, we assume the $k$-design property for the systems we work with, which indeed holds for generic symmetric systems.
However, we believe that the techniques introduced there are useful for proving such properties by exploiting the frustration-free nature of the effective Hamiltonians, and are in general applicable to a much wider set of problems beyond $k$-designs.

\subsection{Geometry of $k$-Commutants in Generic Symmetric Systems} % MARK: * gs man replica
\label{sec:fgsmanrep}
Let us now introduce the general structure of the ground state manifolds of the effective Hamiltonians $P^{(k)}$ associated to Brownian circuits with continuous symmetries.
Consider a continuous symmetry group $G$ with an on-site representation on the many-body Hilbert space
\begin{equation}\label{eq:sym-onsite-rep}
	U_g = \prod_{i=1}^L u_i(g),\quad g \in G.
\end{equation}
We consider the Brownian model to be $G$-symmetric if and only if the commutant of the generating operators $\{h_{\alpha,ij}\}$ is the operator algebra $\llangle G\rrangle$ generated by the group action~\cite{moudgalya2022,moudgalya2022from}, which here is simply the span of the unitaries in $G$ due to the guaranteed closure under multiplication:
\begin{equation}
    \llangle G\rrangle=\mrm{span}\{U_g:g\in G\}.
\label{eq:Gcomm}
\end{equation}
\subsubsection{The $k = 1$ case}
As discussed in Sec.~\ref{subsec:gseffective}, when $k = 1$, the operators $Q$ in the commutant are exactly the ground states $\ket Q$ of the effective Hamiltonian $P^{(1)}$ [cf. Eq.~\eqref{eq:commutant-states}], we have~\cite{moudgalya2024}:
\begin{equation}
	\mrm{ker}(P^{(1)}) = \mrm{span}\big\{ \ket{u(g)}^{\ot L}: g\in G\, \big\}.
\end{equation}
For on-site $G$-symmetric Brownian circuits, $P^{(1)}$ is therefore ferromagnetic, and the group $G$ itself parametrizes a continuous family of fully-polarized ground states.
However, note that in general, the full set of ferromagnetic ground states can be larger, so that $M_G^{(1)}\supseteq G$.
For example, in the case of $U(1)$-symmetric systems [cf.~Eq.~\eqref{eq:u1-generators}], the commutant manifold is known to be expressible as the ground state manifold of a ferromagnetic Heisenberg model of Eq.~\eqref{eq:HeisModelHam} on the restricted local operator basis $\ket\up\cong\ket\1$ and $\ket\dn\cong\ket Z$~\cite{moudgalya2024}, hence the commutant manifold $M_{U(1)}^{(1)}$ is in this case the sphere of spin-coherent states
\begin{equation}
		M^{(1)}_{U(1)} = \Big\{\cos\frac{\theta}{2} \ket\1 + e^{i\varphi}\sin\frac{\theta}{2} \ket Z\Big\},
\label{eq:U1manifold}
\end{equation}
of which, only the submanifold of states with ${\varphi = \frac{\pi}{2}}$, which trace out an equatorial circle, belong to $G=U(1)$.
For a general $G$, the exact form of $M_G^{(1)}$ can be deduced using techniques discussed in Ref.~\cite{the-algebra-paper}.
\subsubsection{$k \geq 2$ and frozen states}\label{subsubsec:frozenstates}
For general $k$, if the generating gate set $\{h_{\alpha,ij}\}$ forms a $G$-symmetric $k$-design for symmetry operators of the form of Eq.~(\ref{eq:sym-onsite-rep}), then $P^{(k)}$ can be shown to be ferromagnetic.
As discussed in Sec.~\ref{subsec:gseffective}, by definition, forming a symmetric $k$-design means that \textit{all} ground states of $P^{(k)}$ are linear combinations of the ones in Eq.~\eqref{eq:trivialgs}, with $Q^{(m)}\in\llangle G\rrangle$.
For a fixed permutation $\sigma\in S_k$, the states $| Q^{(1)}\ot ...\ot Q^{(k)};\sigma\rangle $ can then all be obtained as linear combinations of fully-polarized ground states of the form
\begin{equation}
	\{|q^{(1)}\ot ...\ot q^{(k)};\sigma\rangle^{\ot L},\;\; (q^{(1)}, ..., q^{(k)})\in \big(M_G^{(1)}\big)^{\times k}\}
\label{eq:qpermstate}
\end{equation}
which implies that $P^{(k)}$ is ferromagnetic, with the manifold of these states for a fixed $\sigma$ being simply $(M_G^{(1)})^{\times k}$.
This product structure of the manifold for each permutation is important for an effective ``replica decoupling'' that occurs in the low-energy dynamics of the model (see Appendix~\ref{app:replica-decoupling}).
That leads to significant simplifications in the study of certain observables that only couple to the single branch of the manifold associated to a permutation $\sigma$, and equivalent concepts have been utilized in many previous works~\cite{McCulloch_2023,mcculloch2026longlivedlocalquantumcoherences}.
However, in this work we will be interested in more general observables for which a full characterization of the manifold that includes all permutations is needed.
In particular, the total commutant manifold including all the permutations is not always simply the disjoint union of $k!$ copies of this manifold $(M_G^{(1)})^{\times k}$.
It turns out that two submanifolds associated to different permutations $\sigma\neq\sigma'$ can possess non-trivial intersections, resulting in a more intricate global structure.
This happens wherever for some choice of local operators we have
\begin{equation}
	|q^{(1)}\ot ...\ot q^{(k)};\sigma\rangle=|{{q'}^{(1)}}\ot ...\ot {{q'}^{(k)}};\sigma'\rangle,\ \ \sigma \neq \sigma'.
\label{eq:diffperm}
\end{equation}
This occurs only when the representation of the symmetry group $G$ harbors some \textit{frozen states}, i.e., fully-polarized product states $\{\ket{f_{{\lambda,\alpha}}}^{\ot L}\}$ which belong to $d_\lambda$-dimensional irreps $\lambda$ (with $\alpha$ being a label that can also take continuous values) of the symmetry, such that $\lambda$ appears in $\mc H$ with multiplicity one.
Equivalently, these {$\{\ket{f_{{\lambda,\alpha}}}^{\ot L}\}$} are exact eigenstates of all $G$-symmetric operators, including the Brownian Hamiltonians $H(t)$, with eigenvalues that only depend on $\lambda$ and not $\alpha$, i.e.,
\begin{equation}
    H \ket{f_{{\lambda,\alpha}}}^{\ot L} = \varepsilon_{\lambda}(H)\ket{f_{{\lambda,\alpha}}}^{\ot L}\quad
    \forall H:[G,H]=0.
    \label{eq:frozenstatedefn}
\end{equation}
Hence we refer to the the frozen states with the same $\lambda$ as \textit{degenerate} frozen states.
Note that all irreps for Abelian symmetries are one-dimensional, i.e., $d_\lambda = 1$ for all $\lambda$, hence all frozen states are non-degenerate, hence there we can omit the label $\alpha$.
Given a symmetric Brownian model with frozen states, it is easy to see that the following replicated state composed by them is a ground state of $P^{(k)}$:
\begin{gather}
    P^{(k)} \sket{f_{\lambda,\{\alpha_a\}}}^{\ot L} = 0,\nonumber\\
    \sket{f_{\lambda,\{\alpha_a\}}} = \bigotimes_{m=1}^k \big( \sket{f_{\lambda, \alpha_{2m-1}}} \ot \sket{f_{\lambda, \alpha_{2m}}}^* \big)
    \label{eq:replicafrozenstatedefn}
\end{gather}
where the $2k$ copies are over replica space and the $L$ copies are over real space, and $\lambda$ is the same across replicas.
Since this replicated state is completely disentangled across all replicas, it can be interpreted as coming from states for the form of Eq.~(\ref{eq:qpermstate}) for any of the permutations, since for all $\sigma\in S_k$ we have
\begin{equation}
	\sket{f_{\lambda,\{\alpha_a\}}} =\big|\ot_{m = 1}^k \sketbra{f_{\lambda, \alpha_{2m-1}}}{f_{\lambda, \alpha_{2\sigma(m)}}};\sigma\big\rangle
\label{eq:fsigmasame}
\end{equation}
This shows the existence of solutions to Eq.~(\ref{eq:diffperm}), which in turn causes the different copies of $(M^{(1)}_G)^{\times k}$ to intersect.
We note that these general results are always valid under the assumption of having generators $\{h_{\alpha,ij}\}$ which form a symmetric $k$-design, which we can explicitly verify for the models under consideration using the framework in Ref.~\cite{the-algebra-paper}.
\subsubsection{Connection to Spontaneous Symmetry Breaking}
The fact that all ground states of $P^{(k)}$ are of the form Eq.~\eqref{eq:qpermstate} might be reminiscent of spontaneous symmetry breaking.
Indeed, these states break the $S_k\times S_k$ replica symmetry that permutes the odd and even copies separately while retaining a particular diagonal $S_k$ that permutes them jointly in accordance with the permutation $\sigma$.
In addition, they also break the $G^{\times 2k}$ symmetry to a smaller $G^{\times k}$ associated with the given permutation.
However, the manifold structure does not simply follow from this pattern of symmetry breaking alone, first since $M^{(k)}_G$ has generally a larger dimension than that of the product group $G^{\times k}$, and second because of the presence of frozen states, as we illustrate in the following section with examples of $U(1)$- and $SU(2)$-symmetric models.
\subsection{Examples of Commutant Manifolds}
\subsubsection{$U(1)$ symmetry}\label{sec:u1-gs-man}
Frozen states are a common feature in systems with continuous symmetries. For example, in $U(1)$-symmetric systems, there are often unique states of minimum or maximum charge, hence making them eigenstates of all symmetric operators.
For example, in a spin-$\frac{1}{2}$ model with $U(1)$ charge $Z_\mrm{tot}=\sum_i Z_i$, there are for example two non-degenerate ferromagnetic frozen states: fully-polarized {$\ket{\up}^{\ot L}$} and {$\ket{\dn}^{\ot L}$} that satisfy Eq.~(\ref{eq:frozenstatedefn}), and hence we have
\begin{equation}
    \ket{f_{\lambda = \up}} = \ket{\up},\;\;\;\ket{f_{\lambda = \dn}} = \ket{\dn},
\label{eq:U1frozen}
\end{equation}
where we have omitted the label $\alpha$ since $U(1)$ is Abelian.
By definition, these states cannot evolve under any $U(1)$-preserving dynamics.
As a consequence of Eq.~(\ref{eq:fsigmasame}), the existence of frozen states results in intersections between $(M^{(1)}_{U(1)})^{\times k}$, which is particularly simple to understand when $k = 2$.
For $U(1)$-conserving models, we have $M^{(1)}_{U(1)}\cong \mb S^2$, hence for $k=2$ we expect the commutant manifold to have two copies of $(\mb S^2\times\mb S^2)$, one for the $e$ (identity) permutation, and one for the $\eta$ (swap) permutation.
However, due to these frozen states, these two copies touch each other at two singular points corresponding to $\ket{\up}^{\ot 4}$ and $\ket{\dn}^{\ot 4}$, hence we get the full manifold to be
\begin{equation}
\begin{gathered}
	M_{U(1)}^{(2)}\cong(\mb S^2\times\mb S^2)_e\cup (\mb S^2\times\mb S^2)_\eta\\
	\text{where}\ \ (\mb S^2\times\mb S^2)_e\cap (\mb S^2\times\mb S^2)_\eta = \{\ket\dn^{\ot 4},\ket\up^{\ot 4}\},
\end{gathered}
\end{equation}
see Fig.~\ref{fig:gsman-all}b for a depiction.
For the generators of the $U(1)$ Brownian model listed in Eq.~\eqref{eq:u1-generators}, this can be shown to be the exact form of the ground state manifold for $k=2$.
The structure of the manifolds for $k \geq 3$ is discussed in Sec.~\ref{subsubsec:largerk}.
\subsubsection{$SU(2)$ symmetry}\label{sec:su2-gs-man}
At the $k=1$ level, not much changes when working with non-Abelian symmetries.
If we take the example [cf.~Eq.~\eqref{eq:su2-generators}] of an $SU(2)$-symmetric spin-$\frac{1}{2}$ many-body system, with conserved charges $\vec S_\mrm{tot}=\sum_i\vec S_i$, then any fully-polarized state on two copies of the Hilbert space is an element of the commutant, so that
\begin{equation}
    P^{(1)}\ket{O}^{\ot L}=0,\quad \forall O\in\mrm{End}(\mb C^2).\label{eq:allopssu2}
\end{equation}
This result is consistent with the fact that the effective Hamiltonian $P^{(1)}$ is an $SU(4)$-symmetric ferromagnetic model \cite{moudgalya2024, PhysRevB.108.054307}.
The commutant manifold in this case is therefore the smooth complex projective space
\begin{equation}
    M^{(1)}_{SU(2)} = \mb P\{\mrm{End}(\mb C^2)\} \cong \mb{CP}^{3},
\end{equation}
which has complex dimension $3$.
Non-Abelian symmetries however allow the possibility of degenerate frozen states in addition to non-degenerate ones.
For $SU(2)$, all fully-polarized product states of the form $\ket{\vec n}^{\ot L}$ for $\vec n\in\mb S^2$ [cf.~Eq.~\eqref{eq:blochsphere}] are degenerate frozen states, as they all belong to the same irrep.
Hence in the notation of Eq.~(\ref{eq:frozenstatedefn}), we have
\begin{equation}
    \ket{f_{\alpha = \vec{n}}} = \ket{\vec{n}},
\label{eq:SU2frozen}
\end{equation}
where we have omitted the label $\lambda$ since it is the same for all the frozen states.
The existence of this continuous family of frozen states is directly reflected in the geometry of the manifold.
If we consider the $k=2$ case, the commutant manifold $M^{(2)}_{SU(2)}$ is the union of two complex six-dimensional $\mb{CP}^{3}\times \mb{CP}^{3}$ smooth manifolds---one per permutation $\{e,\eta\}$ derived from $M^{(1)}_{SU(2)}$---and their intersection is given by all states of the form [cf. Eq.~(\ref{eq:fsigmasame})]
\begin{equation}\label{eq:su2-int-struct}
    \big |\ketbra{\vec n_1}{\vec n_2}\ot \ketbra{\vec n_3}{\vec n_4};e\big\rangle = \big |\ketbra{\vec n_1}{\vec n_4}\ot \ketbra{\vec n_3}{\vec n_2};\eta\big\rangle .
\end{equation}
The intersection of the two manifolds has therefore the structure of a continuous and connected manifold $(\mb S^2)^{\times 4}$, which has complex dimension $4$.
In summary, $M^{(2)}_{SU(2)}$ is composed of two six-dimensional complex manifolds intersecting in a four-dimensional complex manifold; a schematic representation is shown in Fig.~\ref{fig:gsman-all}c.

\subsubsection{Commutant manifolds for $k \geq 3$}\label{subsubsec:largerk}
As seen in Eq.~\eqref{eq:qpermstate}, the full manifold $M_G^{(k)}$ is the union of $k!$ identical manifolds with structure $(M_{G})^{\times k}$, potentially with some intersections.
For $k\geq 3$ the structure of the intersections becomes more complicated.
We have seen that states of the form Eq.~\eqref{eq:fsigmasame} are points where \textit{all} the $k!$ manifolds touch.
But for every pair of manifolds, their intersection depends on the Hamming distances between the permutations $\sigma,\sigma'$:
Eq.~\eqref{eq:diffperm} is satisfied only when $q^{(m)}=\ketbra{f_{\lambda,\alpha}}{f_{\lambda,\alpha'}}$ on the replicas $m$ where the two permutations differ, i.e. $\sigma(m)\neq \sigma'(m)$, while any operator $q^{(m)}\in M_{G}^{(1)}$ can be placed on the replicas where the two permutations coincide.
This can lead to a larger submanifolds of intersections between manifolds $(M_{G})^{\times k}$ corresponding to close permutations.
For example, if there is a unique frozen state, the manifolds $M_G^{\times k}$ corresponding to permutations $\sigma$ and $\sigma'$ intersect on two submanifolds, each of the form $M_{G}^{\times(k-d(\sigma, \sigma'))}$, where $d(\sigma, \sigma')$ is their Hamming distance.
These manifolds become more complicated when there are multiple frozen states, possibly forming a continuous manifold of frozen states such as in the $SU(2)$ case, and we will not discuss them in detail here.
\subsubsection{Free-fermion systems}\label{subsubsec:free-fermion-generic}
In a recent work~\cite{lastres2026geometryfreefermioncommutants} we characterized the commutant manifolds that appear in free-fermion systems, where $\{h_\alpha\}$ are fermionic bilinears with or without extra symmetries.
Unlike the generic case studied in Sec.~\ref{sec:fgsmanrep}, free fermion systems do not form symmetric $k$-designs for $k\geq 2$, hence their commutant manifolds are drastically different.
In particular, even free fermionic systems without any continuous symmetries have a continuous replica swap symmetry for $k \geq 2$ which extends the discrete $S_k\times S_k$ \cite{enriched-phases,swann2025,fava2024} appearing in generic systems.
This continuous replica symmetry for $k \geq 2$, also referred to as superoperator symmetry~\cite{lastres2026} or quadratic symmetry~\cite{zeier2011symmetry} for $k = 2$, results in the ground states of effective Hamiltonians $P^{(k)}$ for $k \geq 2$ to have a continuous manifold structure, even though there is no continuous symmetry in the system to begin with.
Indeed, the effective Hamiltonians $P^{(k)}$ in free-fermion systems can be shown to be ferromagnetic as well, and they are appropriate $SO(2k)$ or $SU(2k)$ generalizations of the standard Heisenberg model~\cite{swann2025, fava2023nonlinear, fava2024}.\footnote{When $k=2$, the effective Hamiltonian for noisy free fermions possesses an $SO(4)$ symmetry, and thanks to the decomposition $\mf{so}(4)\cong\mf{su}(2)\oplus \mf{su}(2)$, within a given sector this Hamiltonian is exactly the $SU(2)$ spin-$\frac{1}{2}$ Heisenberg model~\cite{swann2025}.}
This results in smooth and homogeneous manifolds $M_\mrm{GS}^{(k)}$ that take the form of (orthogonal) Grassmannian manifolds for systems (without) with a $U(1)$ symmetry~\cite{lastres2026geometryfreefermioncommutants}.
This continuous manifold structure implies that the dynamics of free-fermion systems, is broadly similar to the dynamics of interacting systems with continuous symmetries.
The main difference lies in the fact that the underlying manifold in free-fermions is smooth, whereas in interacting systems it generically contains singularities, and these can play an important role for certain observables, as we will discuss.
A toy model for free fermion replica systems, and its relation with real physical models, is discussed in Appendix~\ref{sec:delta1lim}.
\section{Geometry of Low-Energy Excitations and Time Evolution} % MARK: 3. TDVP
\label{sec:tdvp}
As shown in Eq.~\eqref{eq:superham}, the averaged dynamics of a replicated Brownian model reduces to imaginary-time evolution by an effective positive-semidefinite Hamiltonian $P^{(k)}$.
In the previous section we have also discussed the structure the ground states of $P^{(k)}$, which are ferromagnetic in the case of many simple symmetries we are interested in.
We will now show how to use this geometric structure of the ground state manifolds developed in the previous section to derive the late-time form of the imaginary-time dynamics of the effective Hamiltonian.
We are ultimately interested in computing amplitudes of the form of Eq.~(\ref{eq:squared-correlator}) or (\ref{eq:renyi-purity}), in general some observable of the form
\begin{equation}
    \sbra{\psi_{\mrm f}} e^{-\kappa P^{(k)} t} \sket{\psi_{\mrm i}}
\label{eq:Pkwanttocompute}
\end{equation}
where $\sket{\psi_{\mrm i}}$ and $\sket{\psi_{\mrm f}}$ are some replicated states on the $2k$ copy Hilbert space.
At late times, we expect the contributions to this propagator to only come from low-energy excitations on top of the ground states of $P^{(k)}$.
If the ground state manifold of $P^{(k)}$ is ferromagnetic and of the form of Eq.~(\ref{eq:gsman-def}), we expect that in the long wavelength or continuum limit, the low-energy excitations are captured by continuously varying configurations
\begin{equation}
    \ket{v(x)}=\bigotimes_i\ket{v(x_i)}, \qquad \ket{v(x_i)} \in M^{(k)}_{G},
\label{eq:varyingmanifold}
\end{equation}
where $v(x)$ is a field that smoothly varies with the position $x$ and parametrizes the configuration $\ket{v(x)}$.
Hence if $v(x)$ is independent of $x$, then $\ket{v(x)}$ is by definition a ground state of $P^{(k)}$, whereas if $v(x)$ varies slowly as a function of $x$, we expect that configuration to have small energy under $P^{(k)}$.
Such configurations can also be interpreted as approximate symmetries of the $k$-replica system~\cite{moudgalya2024}.
We then focus on the computation of the class of propagators of the form
\begin{equation}
    \bra{v_\mrm{f}(x)}e^{-\kappa P^{(k)}t}\ket{v_\mrm{i}(x)},
\label{eq:restprop}
\end{equation}
which is not as general as Eq.~(\ref{eq:Pkwanttocompute}) since it assumes a parametrization of the initial and final states as product states $\ket{\psi_\mrm{i}} = \ket{v_\mrm{i}(x)}$ and $\ket{\psi_\mrm{f}} = \ket{v_\mrm{f}(x)}$.
Remarkably, as we discuss in the upcoming sections, we find that such parametrizations exist for all the observables we are interested in here.
For example, consider the ferromagnetic Heisenberg model [Eq.~\eqref{eq:HeisModelHam}], whose ground state manifold is a sphere $\mb S^2$ parametrized by spin-coherent states [Eq.~\eqref{eq:HeisModelGS}] labeled by a unit vector $\vec{n}$.
In the continuum limit, it is known that its low-energy excitations (spin-waves) can be captured in terms of spatially varying fields $\vec n(x)\in\mb{S}^2$ (see Fig.~\ref{fig:heis}a)~\cite{fradkin2013field,Altland_Simons_2010}, and Eq.~(\ref{eq:varyingmanifold}) is the generalization of this to more general ferromagnetic manifolds.
Analogous to Eq.~(\ref{eq:restprop}), we can then ask how to compute propagators of the form
\begin{equation}
    \sbra{\vec{n}_\mrm{f}(x)} e^{-Ht}\sket{\vec{n}_\mrm{i}(x)}
\label{eq:Heispropagator}
\end{equation}
for some initial ($\vec{n}_\mrm{i}(x)$) and final ($\vec{n}_\mrm{f}(x)$) configurations of spin-coherent states, where $H$ is the Heisenberg Hamiltonian.

\subsection{Semiclassical Approaches and Quantum Fluctuations}\label{subsec:review}
We now give an overview of the general strategies to compute propagators of the form Eq.~\eqref{eq:restprop}.
One standard approach is through path-integral descriptions of the time evolution operator.
In Appendix~\ref{app:path-integral} we provide a detailed overview of the spin-coherent path integral formulation of the ferromagnetic Heisenberg model~\cite{kochetov95heis},  which enables the computation of propagators such as in Eq.~(\ref{eq:Heispropagator}).
Such a field-theoretical formulation of the imaginary time evolution can be naturally constructed whenever the ground state manifold of $H$ is smooth and homogeneous (i.e., a quotient of a Lie group by a closed subgroup $G/H$).
Such a manifold is said to parametrize a set of generalized coherent states $\ket v$ \cite{generalized-coher-st,perelomov1986generalized,Kochetov1995,Berezin:1978sn} that have properties analogous to the spin-coherent states.
Effective Hamiltonians $P^{(k)}$ obtained from free-fermion terms have this property~\cite{lastres2026geometryfreefermioncommutants,swann2025}, see Ref.~\cite{lastres2026geometryfreefermioncommutants} for a discussion of generalized coherent states in that context.
For such models, we expect a similar path integral formulation as in the Heisenberg model to hold, similar to techniques used in \cite{fava2023nonlinear,fava2024} (see Appendix~\ref{sec:delta1lim}).
Furthermore a simple scaling argument~\cite{swann2025} that we illustrate in Appendix~\ref{app:path-integral} for the Heisenberg model, proves that the evolution of propagators in such models is at long times $t$ is dominated by semiclassical trajectories corresponding to saddle-points of the action.
For more general ferromagnetic models that do not have such a smooth ground state manifolds,  the emergence of semiclassical behavior at late imaginary times can be justified via a more general physical argument, commonly found in the large-$S$ treatment of Heisenberg ferromagnets~\cite{fradkin2013field}.
At long times, the imaginary-time evolution operator $e^{-Ht}\ket \psi$ is expected to locally relax the initial state onto polarized ground states in order to minimize the energy.
From a real space renormalization group perspective, these locally ordered regions act as macroscopic degrees of freedom, thus suppressing quantum fluctuations and inherently justifying a semiclassical continuum description.
Regardless, in many situations quantum fluctuations that go beyond a simple semiclassical picture might be relevant for some aspects.
Conservation laws can for example constrain the evolved state onto a quantum superposition of locally ordered regions, and the presence of several equivalent semiclassical trajectories can similarly produce entanglement during the the time evolution.
These phenomena are all in principle captured by the path integral formalism.
However, for the systems under consideration here, formulating a path integral representation is a challenging task, due to the peculiar structure of the commutant manifolds $M^{(k)}_{G}$ studied in Sec.~\ref{sec:fgsmanrep}.
First, such manifolds do not admit a simple generalized coherent state parametrization~\cite{generalized-coher-st,perelomov1986generalized,Kochetov1995,Berezin:1978sn}.
In addition, for field configurations that traverse the contact singularities of the manifold, we expect to observe a behavior similar to that of fluctuating defects (see Appendix~\ref{sec:fluctuations} for a discussion of such defects in the $es\eta$ toy model of Sec.~\ref{sec:esn} and the $U(1)$-symmetric model of Sec.~\ref{sec:u1}), raising further theoretical challenges of addressing them in the path-integral formulation.
These issues can potentially be bypassed by working with the path integral on a larger homogenous manifold of which the ground state manifold of interest is a submanifold, and imposing the ground state via energetic restrictions in the Hamiltonian that corresponds to terms in the action, but we leave the exploration of such possibilities to future work.
The fact that ground state manifolds of ferromagnetic Hamiltonians are \textit{complex manifolds} (see Sec.~\ref{sec:fgsman}) does however provide a simple alternative approach to the study of semiclassical dynamics.
On such manifolds, it is possible to unambiguously formulate a set of semiclassical variational equations of motion for field configurations, following the Time-Dependent Variational Principle (TDVP)~\cite{tdvp20}.
Due to the structure of the ferromagnetic Hamiltonians, we find that these equations are mainly dependent on the \textit{geometry} of the variational manifold, which is directly related to the ground state manifold.
This structure makes TDVP readily applicable to the study of the imaginary-time evolution of an initial configuration of fields $\ket{v(x)}$.
We will present the TDVP in the forthcoming section, with a more detailed account of it and its relation with the path integral in Appendix~\ref{app:tdvp}.
In the case of the path integral, semiclassical trajectories are usually virtual intermediaries for the computation of propagators, and as such depend on \textit{both} the initial and final state, and hence cannot be interpreted as even an approximate physical evolutions of the initial state.
In contrast, the TDVP equations instead describe a unique \textit{physical} trajectory, which approximates the evolved quantum initial state; this property is helpful when one fixes the initial state and wants to make statements for a large class of final states, which is what we want to do in the computation of entanglement entropies.
While this single trajectory leads to a simple description of observables over time, this approach is also prone to variational error when computing propagators of the form Eq.~\eqref{eq:restprop}, which we however find to be small for generic overlaps at long times when comparing with numerical data.
The TDVP also does not provide a framework to study the subleading quantum fluctuations around the semiclassical trajecory.
We will therefore limit ourselves to a qualitative discussion of fluctuations, informed by numerical simulations and intuition provided by the structure and symmetries of the ground state manifolds under consideration.
Nevertheless, in some cases, we can map this imaginary-time evolution to a classical stochastic model, which provides a rigorous framework for studying fluctuations, and which in Appendix~\ref{sec:fluctuations} we demonstrate for the $es\eta$ toy model of Sec.~\ref{sec:esn} and the $U(1)$-symmetric model of Sec.~\ref{sec:u1}.

\begin{figure*}
	\centering
    \textbf{(a)}\includegraphics[width=0.37\linewidth,valign=t]{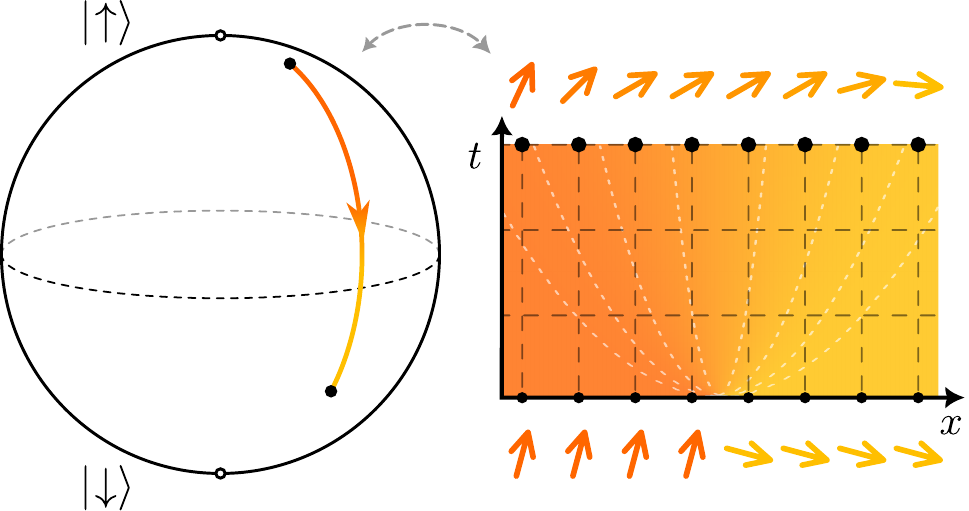}
    \textbf{(b)}\includegraphics[width=0.30\linewidth,valign=t]{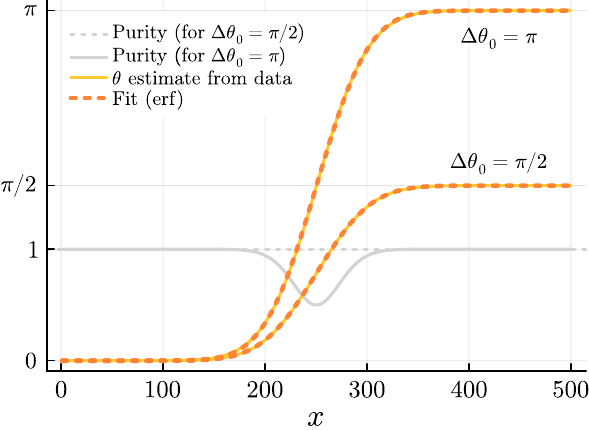}
    \textbf{(c)}\includegraphics[width=0.23\linewidth,valign=t]{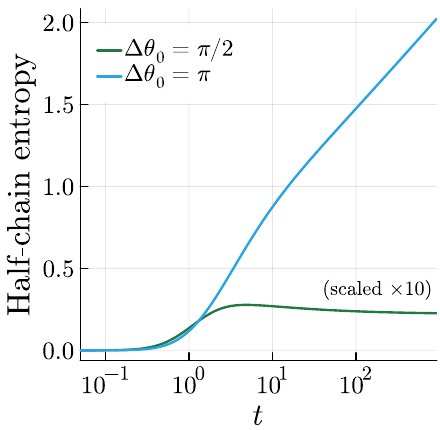}
	\caption{(a) Illustration of the ground state manifold and diffusive relaxation for the Heisenberg model imaginary-time dynamics. For a $1$d system, in the continuum limit a smooth field configuration can be seen as a curve on the Bloch sphere.
    (b) Spatial profile of the field $\theta(x,t)$ at late times extracted using Eq.~\eqref{eq:theta-est} from numerical data of the Heisenberg model, compared to a fit of the TDVP solution Eq.~\eqref{eq:sol-heis-th}. In light gray the purity of the one-site reduced density matrix for the $\Delta\theta_0=\pi$ evolved state, where the TDVP equations have multiple equivalent solutions.
    (c) Half-chain entanglement entropy from numerical data for two evolved domain walls in the Heisenberg model. The entropy in the $\Delta\theta_0=\pi/2$ state saturates to a small value and later slowly decreases (value scaled by $\times 10$ for visibility), while for the $\Delta\theta_0=\pi$ state it grows logarithmically over time consistently with the prediction in Sec.~\ref{sec:ex-dw}. Numerical data was obtained using MPS simulations for (b),(c) with $L=500$ and SVD cutoff $10^{-12}$ (that leads to a maximum bond dimension of $10$ and $35$ respectively for $\Delta\theta_0=\pi/2$ and $\pi$), where the imaginary time evolution is performed using an efficient implicit integration technique~\cite{Zima2026Jun}.
    }
	\label{fig:heis}
\end{figure*}

\subsection{Time-Dependent Variation Principle (TDVP)}\label{subsec:tdvp-main} % MARK: * tdvp
We now discuss in detail the TDVP~\cite{tdvp20} applied to the variational family of field configurations $\ket{v(x)}$ that at each point in space take value on the ground state manifold $M_\mrm{GS}$, as in Eq.~\eqref{eq:varyingmanifold}.
Since $M_{\mrm{GS}}$ is a complex projective manifold (apart from singularities), it possesses a quantum metric $g_{\alpha\beta}$ and a symplectic Berry curvature tensor $\omega_{\alpha\beta}$, which are defined in Eq.~\eqref{eq:tdvp-tensors-def}, where the indices $\alpha\beta\dots$ indicate intrinsic coordinates on the manifold.
For imaginary-time evolution, TDVP is equivalent to local energy relaxation via gradient descent along the metric $g_{\alpha\beta}$ of the variational manifold
\begin{equation}\label{eq:tdvp-continuum-main}
	g_{\alpha\beta} \dot v^\beta = -\frac{\delta H[v]}{\delta v^\alpha},
\end{equation}
where $\{v^\alpha\}$ are the components of the field $v$ that parametrizes the manifold, as in Eq.~\eqref{eq:varyingmanifold}.
We find it sometimes more convenient to parametrize the manifold with extrinsic coordinates, i.e, coordinates of a larger space in which $M_{GS}$ can be embedded. For example the Heisenberg ground state manifold $\mathbb{S}^2$ can be described with a pair of intrinsic spherical coordinates $(\theta, \varphi)$, or equivalently with a coordinate vector $\vec n\in\mb R^3$ with unit norm, as in Eq.~\eqref{eq:HeisModelGS}. In order to ensure that the evolution remains constrained to the variational manifold of interest, in extrinsic coordinates the TDVP equations require an explicit tangent space projection, which however does not change the underlying physics [see Eq.~(\ref{eq:tdvp-embed}) and Appendix \ref{app:tdvp} for specifics, along with examples of interest].
\subsubsection{Norm of the evolved state}
Since our variational family lies on a projective manifold, by definition the TDVP does not capture the information of the norm evolution of the state.
Hence within this space, the TDVP prediction for the evolved state identifies all rescalings of the same global state, i.e.,
\begin{equation}
    \ket{v(x,t)} \sim \lambda \ket{v(x,t)}
\label{eq:TDVPrescalings}
\end{equation}
for any complex number $\lambda$.
Nevertheless, the true quantum imaginary-time evolution from an initial state $\ket{v(x,0)}$ under Hamiltonian $H$ does have a well-defined norm:
\begin{equation}
    \norm{e^{-Ht}\ket{v(x,0)}}^2 \defn \bra{v(x,0)}e^{-2Ht}\ket{v(x,0)},
\label{eq:truenorm}
\end{equation}
and the knowledge of these norms is important in obtaining correct values for propagators such as Eq.~(\ref{eq:restprop}).
In order to capture the `norm' degree of freedom, we may parametrize the TDVP evolution as
\begin{equation}
    \ket{v(x,t)} = \mc N(t) \ket{v(x,t)}_\mrm{norm.}
\end{equation}
where $\ket{v(x,t)}_\mrm{norm.}$ is a representative of the state Eq.~\eqref{eq:TDVPrescalings} chosen such that its norm is one.
One canonical way to approximate the evolution of the norm is to apply the equations of motion Eq.~\eqref{eq:tdvp-continuum-main} directly to the additional degree of freedom $\mc N(t)$, which predicts that the final norm is obtained by integrating the energy of the state over time~\cite{tdvp20} [cf.~Eq.~\eqref{eq:en-dec-app}]:
\begin{equation}\label{eq:norm-evo-eq}
    |\mc N(t)|^2_\textsc{tdvp} \defeq |\mc N(0)|^2\exp(-2\int_0^t\dd \tau\, H[v(x,\tau)]),
\end{equation}
However, note that there is in general no self-consistent way to fix the norms for \textit{all} initial states under TDVP evolution.
For example, due to the variational error introduced in the TDVP evolution of a state, the predicted values for propagators of the form Eq.~\eqref{eq:restprop} will in general depend on whether one chooses to apply the TDVP evolution to $\ket{v_\mrm{i}(x)}$ or to $\bra{v_\mrm{f}(x)}$, thus violating the Hermiticity of the Hamiltonian.
For some applications it will be more natural to compute the norm of the state in alternative ways. In particular, if we know that under the exact quantum evolution, a field configuration decays exactly to another field configuration with zero energy
\begin{equation}
	\lim_{t\rt\infty} e^{-Ht} \ket{v(x,0)} \propto \ket{v_\infty(x)},
\end{equation}
we may as an alternative strategy approximate the true norm of Eq.~(\ref{eq:truenorm}) indirectly, by imposing that during the TDVP evolution, the overlap $\braket{v_\infty(x)}{v(x,t)}$ remains constant.
In practice, this means that given the projective evolution of the variational state $\ket{v(x,t)}$, we define the norm as
\begin{equation}\label{eq:norm-evo-eq-prime}
    |\mc N(t)|^2_{\textsc{tdvp}'} \defeq \frac{|\braket{v_\infty(x)}{v(x,0)}|^2}{|\braket{v_\infty(x)}{v(x,t)}_\mrm{norm.}|^2}.
\end{equation}
In this way, the estimate $|\mc N(t)|^2_{\textsc{tdvp}'}$ will be equal to the exact value of the norm at times $t=0$ and $t=\infty$.
Note that both formulas Eq.~\eqref{eq:norm-evo-eq} and ~\eqref{eq:norm-evo-eq-prime} become equivalent and exact if the full quantum evolution can be exactly captured by the variatonal manifold; however, they are not equivalent within any restricted variational ansatz, which is always the case for the systems we are studying here.
Nevertheless, we find the qualitative asymptotic behavior of the norm to be anyway independent of the choice of convention
\begin{equation}
    \norm{e^{-Ht}\ket{v(x,0)}} \sim |\mc N(t)|_{\textsc{tdvp}} \sim |\mc N(t)|_{\textsc{tdvp}'}
\end{equation}
and in the following we will use both for different applications, making explicit reference to the convention used.
\subsubsection{Equations of motion for ferromagnetic models}
As a useful example, let us now formulate the TDVP for the ferromagnetic Heisenberg model, where the variational manifold is given by the spin-coherent state configurations of the form $\ket{\vec{n}(x)}$.
The continuum Hamiltonian can be written in the vector coordinates $\vec n\in\mb S^{2}$ or the $(\theta,\varphi)$ spherical coordinates as [cf.~Eq.~\eqref{eq:heis-ham-cont-lim}]:
\begin{equation}\label{eq:cont-heis-mtxt}
    H[\vec{n}] \!=\! \frac{J}{2}\!\int\!\dd^d x \, |\bm\nabla \vec{n}|^2 =\! \frac{J}{2}\!\int\!\dd^d x\,[|\bm\nabla \theta|^2 + \sin^2\theta |\bm\nabla\varphi|^2].
\end{equation}
Moreover, the metric for this manifold $g_{\alpha\beta}$ is proportional to the standard Euclidean metric on the sphere $\mathbb{S}^2$, which, after applying Eq.~(\ref{eq:tdvp-continuum-main}) ultimately leads to equations of motion of the form [cf.~Eq.~(\ref{eq:tdvp_heis_final})]
\begin{equation}
	\partial_t \vec{n} = 2 J \left( \bm \nabla^2 \vec{n} + |\bm \nabla \vec{n}|^2 \vec{n} \right),\label{eq:diff-eq-heis}
\end{equation}
which are equivalent to the dissipative part of the Landau-Lifshitz-Gilbert equations with an effective magnetic field $\vec H_\mrm{eff}\sim 2J\bm\nabla^2\vec n$ on every site \cite{aharoni2000introduction}.
This equation of motion is closely related to those obtained from the path integral methods [see Eq.~\eqref{eq:eoms-z} and the discussion in Appendix~\ref{app:tdvp-vs-pi}].
More generally, we expect ferromagnetic Hamiltonians with a smooth continuous ground state manifold to be dominated by a leading term with second order spatial derivatives [cf.~Eq.~\eqref{eq:exp-generic}], analogous to Eq.~(\ref{eq:cont-heis-mtxt}) for the Heisenberg model, and consistent with the expected quadratic dispersion relation for frustration-free Hamiltonians~\cite{Masaoka2024Nov}.
For the effective models that we will study in the following sections, they all take the specific form
\begin{equation}\label{eq:deriv-exp-H}
	H[ v]=J\int_{M_\mrm{GS}}\dd^d x\ g_{\alpha\beta}\bm\nabla v^\alpha\cdot\bm\nabla v^\beta+\dots,
\end{equation}
where $g_{\alpha\beta}$ is the metric of the ground state manifold written in its intrinsic coordinates; see discussion around Eq.~(\ref{eq:explicit-rho-expansion}) for a derivation under fairly general assumptions.
For example, the Heisenberg Hamiltonian of Eq.~(\ref{eq:cont-heis-mtxt}) is also of this form, evident in its expression in spherical coordinates intrinsic to the spin-coherent state manifold.
The TDVP equations of Eq.~(\ref{eq:tdvp-continuum-main}) for such a Hamiltonian simply read [see Eq.~\eqref{eq:geom-heat-flow}]
\begin{equation}\label{eq:harm-map-heat-flow}
	\partial_t v^\gamma=2J\left(\bm\nabla^2v^\gamma+\Gamma^\gamma_{\alpha\beta}\bm\nabla v^\alpha\cdot\bm\nabla v^\beta\right),
\end{equation}
i.e. the heat equation along the metric of the variational manifold where
\begin{equation}
    \Gamma^\gamma_{\alpha\beta}=\frac{1}{2}G^{\gamma\delta}\big(\partial_\alpha g_{\delta\beta}+\partial_\beta g_{\delta\alpha}-\partial_\delta g_{\alpha\beta}\big)
\end{equation}
are the Christoffel symbols for the metric $g_{\alpha\beta}$, with $G^{\alpha\beta}$ being the inverse of the metric.
For ground state manifolds with singularities, these equations are still expected to hold away from the singularities, but the singularity itself needs to be treated separately with the precise conditions depending on its nature, as we will discuss in the upcoming sections.
Note that for systems of finite size, there are natural Neumann boundary conditions at the spatial boundary for fields $v(x)$:
\begin{equation}\label{eq:neumann-conditions}
    \bm n\cdot \bm\nabla v(x,t)\big|_{x\,\in\,\mrm{bd}} = 0,
\end{equation}
where $\bm n$ is the unit vector normal to the boundary. They can be seen as conditions that minimize the energy density along the free open boundary, and a derivation from TDVP can be found in Appendix~\ref{app:bound-cond}.
\subsubsection{Stationary solutions}
The stationary solutions to the TDVP equations of Eq.~\eqref{eq:harm-map-heat-flow}, i.e., those that satisfy $\partial_t v^\gamma =  0$, capture the effect of low-energy eigenstates of quantum Hamiltonians with the given ground state manifold.
For example, in 1d, the stationary solutions are exactly the configurations where $v(x)$ follows a geodesic on the ground state manifold, since Eq.~(\ref{eq:harm-map-heat-flow}) then reduces to the well-known equations for a geodesic.
A class of trivial geodesics are just points on the ground state manifold $v(x)\equiv\mrm{const.}$, which are the exact ferromagnetic ground states of the model.
In addition, the other geodesics variationally correspond to low-energy eigenstates of the Hamiltonian, whose norm would decay exponentially under imaginary-time evolution at an inverse-rate proportional to the energy, which can be made arbitrarily small in an infinite system, signifying gapless excitations of the Hamiltonian.
Of course, these are not the true eigenstates of the microscopic quantum Hamiltonian, and the discrepancy can be quantified through the \textit{leakage} of the trajectory outside the variational manifold $\{\ket{v(x)}\}$, which is generally always non-zero for non-trivial field configurations.
This leakage is derived in Appendix~\ref{app:leakage} for the ferromagnetic Heisenberg model, showing that it goes to zero as $|\bm \nabla\vec n|^4$ for slowly-varying solutions.
\subsubsection{Multiple equivalent solutions}\label{subsubsec:multiplesoln}
Often, particularly for discontinuous initial field configurations, there can be multiple trajectories of field configurations that satisfy the TDVP equations of motion.
While among these solutions there is often a unique one that possess the lowest energy, and is thus the unambiguous leading contribution, in some cases multiple equivalent lowest-energy solutions might also exist.
If that is the case, we posit that the true evolution is best approximated by an equal superposition of all such solutions.
Unlike the case where a unique solution exists, which represents a semiclassical state with no entanglement, the superposition of multiple solutions does result in an entangled state.
While this is strictly speaking an extension of standard TDVP, we will still refer to this as the TDVP solution.
We will discuss many such examples below.
\subsection{Example: Domain wall evolution in the Heisenberg model}\label{sec:ex-dw}
Relevant to the computation of entanglement in the rest of the paper is the case where the initial state for the TDVP evolution is a domain wall configuration between two distinct ground states [cf.~Eq.~\eqref{eq:definition-dw-id-swap}].
Here we illustrate the imaginary-time TDVP dynamics of such an initial state in the ferromagnetic Heisenberg model.
In spherical coordinates Eq.~\eqref{eq:blochsphere} let us consider the following initial state on an infinite 1d system:
\begin{equation}
    \!\vec n_{\textsc{dw}}(x,t=0)=\begin{cases}
        (\theta=\theta_0-\frac{\Delta\theta_0}{2},\ \varphi=0)\quad x<0,\\
        (\theta=\theta_0+\frac{\Delta\theta_0}{2},\ \varphi=0)\quad x>0.
    \end{cases}\!\!
\label{eq:heis_init_conditions}
\end{equation}
For a domain wall between any two points of the sphere $\mb S^2$ we can always choose the coordinates such that $\varphi=0$ everywhere.
This leads to a simplification of the TDVP equations of motion Eq.~\eqref{eq:diff-eq-heis}, which become:
\begin{equation}\label{eq:eom-theta}
    \partial_t \theta = 2J\,\partial_x^2\theta,\quad \varphi(x,t)=0.
\end{equation}
For infinite system size, the solution to this differential equation with the initial conditions of Eq.~(\ref{eq:heis_init_conditions}) can be solved using straightforward Fourier analysis, and is expressed in terms of the error function
\begin{equation}
    \theta(x,t)=\theta_0+\frac{\Delta\theta_0}{2}\,\mrm{erf}\left(\frac{x}{\sqrt{8Jt}}\right),\;\;\varphi(x,t) = 0,
    \label{eq:sol-heis-th}
\end{equation}
and corresponds to a diffusive melting of the domain wall, as pictured in Fig.~\ref{fig:heis}a.
In Fig.~\ref{fig:heis}b, we compare this solution to matrix product state simulations of the imaginary-time evolution, where, while the true evolved state is entangled, we can estimate $\theta_i$ using the reduced density matrix $\rho_i$ at site $i$ as
\begin{equation}\label{eq:theta-est}
    \theta_i\approx 2\arctan\sqrt{\frac{c_{\dn\dn}}{c_{\up\up}}},\quad \rho_i=\begin{pmatrix}
        c_{\up\up}&c_{\up\dn}\\
        c_{\dn\up}&c_{\dn\dn}\\
    \end{pmatrix},
\end{equation}
which demonstrates remarkable agreement with the TDVP prediction.
Moreover, we also find numerically in Fig.~\ref{fig:heis}c that the entanglement of the evolved domain wall is extremely small, consistent with the semiclassical TDVP approach, which remains within the product state ansatz.
This problem has also been studied with semiclassical path-integral techniques~\cite{swann2025} that we discuss in Appendix~\ref{app:path-integral} and using the Bethe ansatz solution in the case of the $\ket{\up}$-$\ket{\dn}$ domain wall~\cite{Stephan2017Oct}.

Note that since the value of $\theta$ at $x=0$ in Eq.~\eqref{eq:sol-heis-th} is always equal to $\theta_0$ throughout the evolution, this trajectory also solves Eq.~(\ref{eq:eom-theta}) with the Dirichlet boundary condition
\begin{equation}
    \theta(x=0,t)=\theta_0.
\end{equation}
For solutions where $\varphi(x,t)=\mrm{const.}$, the energy functional of Eq.~\eqref{eq:cont-heis-mtxt} can be rewritten as
\begin{equation}\label{eq:heis-theta}
    H[\theta]=\frac{J}{2}\int\dd^d x\, |\bm\nabla\theta|^2.
\end{equation}
The energy of the solution of Eq.~(\ref{eq:sol-heis-th}) can be then seen to decay algebraically over time, which leads to:
\begin{equation}
    E\sim (\Delta \theta_0)^2\sqrt{J/t}\;\;\;\implies\;\; \norm{\,\ket{\vec n_{\textsc{dw}}(x,t)}\,} \sim e^{-\sqrt{Jt}},
\label{eq:heisenbergenergy}
\end{equation}
where we have used the norm formula of Eq.~\eqref{eq:norm-evo-eq}.
Note that the $(\Delta \theta_0)^2$ can be seen as a geometric consequence of the energy functional Eq.~\eqref{eq:heis-theta} being quadratic in gradients of the field, and we will use this intuition in examples later in this work.
This form also leads to overlaps of the form $e^{-\sqrt{Jt}}$ with typical translation-invariant product initial states, which is directly relevant for the diffusive entanglement growth in free Majorana Brownian models, as studied in \cite{swann2025} using path integrals.
For $0<\Delta\theta_0<\pi$ we expect this solution to correctly describe the imaginary time evolution for the domain wall, since this is the unique solution to the TDVP equations that minimizes the energy $H[\theta]$.
However, setting $\Delta\theta_0=\pi$ (and $\theta_0=\frac{\pi}{2}$), which corresponds to having
\begin{equation}
    \ket{\vec n_\textsc{dw}(x,t=0)}=\ket{\cdots\dn\dn\up\up\cdots},
\label{eq:heisinitsuperposition}
\end{equation}
leads to an ambiguity.
Since $\varphi$ is not well-defined for this initial state and can be any angle $\varphi_0\in[0,2\pi)$, there is a continuum of equivalent TDVP solutions where $\theta(x,t)$ satisfies Eq.~\eqref{eq:sol-heis-th} and $\varphi(x,t) = \varphi_0$.
Following our prescription for multiple equivalent solutions in Sec.~\ref{subsubsec:multiplesoln}, we expect the evolved state to be an equal superposition of the form:
\begin{equation}\label{eq:heis-dw-cont-sup}
    \ket{\vec n_\textsc{dw}(x,t)}\approx\int\frac{\dd\varphi_0}{2\pi}\ket{\theta(x,t),\varphi(x,t) = \varphi_0}.
\end{equation}
As shown in Appendix~\ref{app:entr}, such a solution would lead to a logarithmic growth in the bipartite entanglement entropy at $x=0$, and this is confirmed by our simulations in Fig.~\ref{fig:heis}c.
This is also captured in Fig.~\ref{fig:heis}b, where the reduced density matrix $\rho_i$ is shown to be impure around $x=0$ for $\Delta\theta_0=\pi$.
\section{Dynamics of the $es\eta$ Toy Model} % MARK: 4. ESN
\label{sec:esn}
Given the TDVP understanding of the ferromagnetic Heisenberg model, we now move on to effective Hamiltonians $P^{(2)}$ for various continuous symmetries.
However, as a simpler intermediate step, we will first study what we call the \textit{$es\eta$ model}, given by the Hamiltonian we refer to as $P^{(es\eta)}$.
This is a simplified toy model that has a ground state manifold with a singularity and captures the key features of effective Hamiltonians $P^{(2)}$ we are interested in.
It is defined on a many-body system with three local degrees of freedom:
\begin{equation}
	\ket e,\ \ket s,\ \ket\eta.
\end{equation}
The states $\ket e$ and $\ket\eta$ are meant to represent states in permutations $e,\eta\in S_2$, as in Eq.~\eqref{eq:permstate}, while $\ket s$ represents a state that can be interpreted as being in both permutations.
$\ket s$ is therefore analogous to a replicated frozen state $(\sket{f}\ot\sket{f}^*)^{\ot 2}$ of the symmetry group [cf.~Eq.~\eqref{eq:replicafrozenstatedefn}], which are also part of both permutations, as discussed in Sec.~\ref{sec:fgsmanrep}.
Of course, for actual physical symmetries such as $U(1)$, the on-site Hilbert space of these effective models involves more local degrees of freedom, that we will discuss in detail in Sec.~\ref{sec:u1} and later.
But the local degrees of freedom there too can either be associated to one of the two permutations, which are analogs of the states $\ket{e}$ and $\ket{\eta}$, or can be part of both permutations, which are analogs of $\ket{s}$.
Here we will show how the essential physics there is already captured by this simple three-state model.
Given these local degrees of freedom, we define the two-site ferromagnetic interaction for the $es\eta$ model as:
\begin{gather}
	P^{(es\eta)}_{ij}=P^{(es)}_{ij}+P^{(\eta s)}_{ij}+P^{(e\eta)}_{ij},\label{eq:esn-ham}\\
	P^{(\sigma s)}_{ij}=\ketbra{\sigma s}{\sigma s}+\ketbra{s\sigma}{s\sigma}-
	\ketbra{\sigma s}{s\sigma}-\ketbra{s\sigma}{\sigma s}\!,\\
	P^{(e\eta)}_{ij}=\Delta(\ketbra{\eta e}{\eta e}+\ketbra{e\eta}{e\eta})+
	\ketbra{\eta e}{e\eta}+\ketbra{e\eta}{\eta e}\!,\!
\end{gather}
where $\sigma\in\{e,\eta\}$ and $\Delta > 1$.
These terms are a ferromagnetic Heisenberg interactions between the states $e$-$s$ and $\eta$-$s$, and a ferromagnetic XXZ-type interactions between the states $e$-$\eta$, which preserve the particle numbers of all three species.
The form of this Hamiltonian is derived from the fact that it is exactly a subsector of the effective Hamiltonian for the $U(1)$-conserving replica model discussed in Sec.~\ref{sec:u1}, but here we will simply treat this as a toy model and extract the relevant physics.
The $\Delta\rt 1$ limit in 1d is a toy model for an effective Hamiltonian $P^{(2)}$ of a Brownian system with free-fermion generators, and is discussed in Appendix~\ref{sec:delta1lim}.

\begin{figure*}
	\centering
    \textbf{(a)}\includegraphics[width=0.26\linewidth,valign=t]{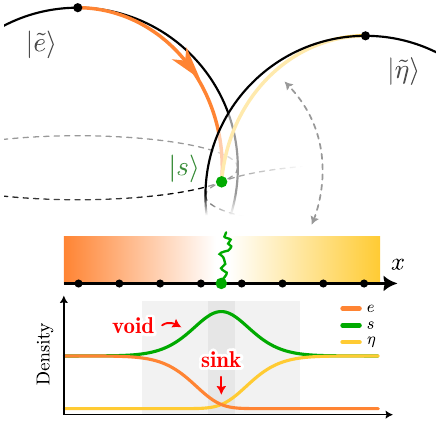}
    \textbf{(b)}\includegraphics[width=0.32\linewidth,valign=t]{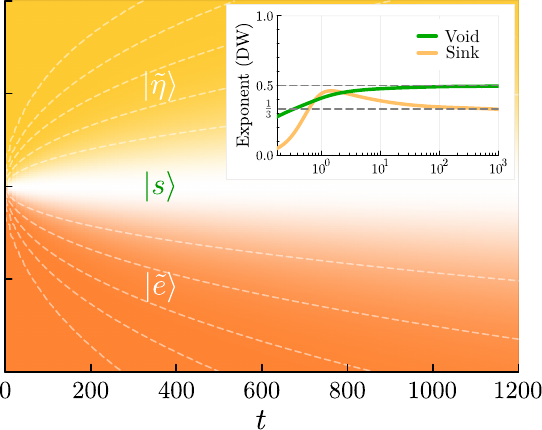}
	\textbf{(c)}\includegraphics[width=0.32\linewidth,valign=t]{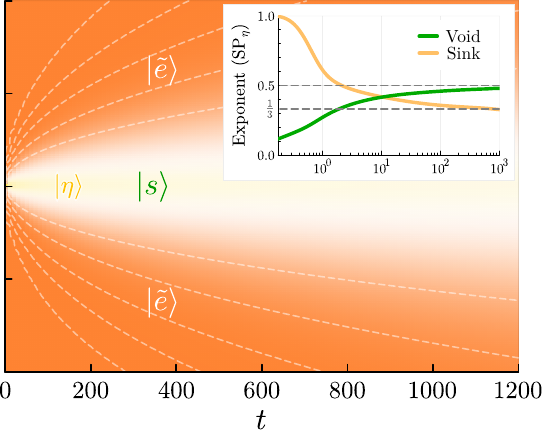}

	\caption{(a) Illustration of a field configuration for the $es\eta$ model, in particular associated to the evolution of the domain wall state $\ket{\textsc{dw}(t)}$. To each position in space we associate a point on the ground state manifold $M_\mrm{GS}$ of Fig.~\ref{fig:gsman-all}a, and we treat the position of the $s$-point as a defect, which in the semiclassical limit remains stationary. We call \textit{void} the region where the density of the $\ket{s}$ state dominates, and \textit{sink} the region where density of $\ket{e}$ and $\ket{\eta}$ states overlap.
    (b) Density of $\ket{e}$ and $\ket{\eta}$ states in the evolution of $\ket{\textsc{dw}(t)}$, showing the formation of the void the $\ket{s}$ density dominates. Inset, power law exponent for the growth of the void (diffusive) and sink (subdiffusive), the latter estimated as in Eq.~\eqref{eq:sink-region-dw}.
    (c) Density of $\ket{e}$ states  in the evolution of $\ket{\textsc{sp}_\eta(t)}$, showing the formation of the void. There is a single $\eta$-particle which wanders within the void. Inset, power law exponent for the growth of the void (diffusive) and sink (subdiffusive), the latter estimated by the extent of the wavefunction for the $\eta$-particle. Numerical data was obtained using MPS simulations with $L=500,501$ and bond dimension $100$, where the imaginary time evolution performed using an efficient implicit integration technique~\cite{Zima2026Jun}.}
	\label{fig:esn}
\end{figure*}

\subsection{Ground State Manifold and Low-energy Excitations}
The ferromagnetic nature of the $es\eta$ model and the geometry of its ground state manifold have been discussed in Sec.~\ref{sec:esn-man}.
Due to the gapped interaction between $\ket e$ and $\ket\eta$ states, the ground states of this model are exactly the ferromagnetic ground states of $P^{(es)}$ and of $P^{(\eta s)}$, which can be parametrized by the ground state manifold $M_\mrm{GS}$ composed of two Bloch spheres of the form of Eq.~(\ref{eq:esn-gs-param}) that touch at a unique point $\ket{s}$, as seen in Fig.~\ref{fig:gsman-all}a.
This manifold closely resembles the $k=2$ commutant manifolds in systems with continuous symmetries, as described in Sec.~\ref{sec:fgsmanrep}: the union of two smooth manifolds associated to each of the two permutations $e$ and $\eta$ (two spheres $\mb S^2$ in this toy model), which intersect on a shared submanifold associated to frozen states (a single point in this toy model).
Given this ground state manifold we can formulate the low-energy physics of the $es\eta$ toy Hamiltonian of Eq.~(\ref{eq:esn-ham}) in the continuum in terms of a slow field which takes values in the target manifold $M_\mrm{GS}$ (see Fig.~\ref{fig:esn}a), as argued in Sec.~\ref{sec:tdvp}.
In terms of the physical degrees of freedom, this is equivalent to only incorporating the effect of gapless spinwave-like excitations between $s$-$e$ and $s$-$\eta$, while neglecting the effect of the gapped $e$-$\eta$ interactions, which, to leading order, are not expected to  play a role in the low-energy physics of the model (the parameter $\Delta$ is expected to flow to strong coupling in a renormalization group sense).
Concretely, we use the following labeling for states on the manifold, such that any point $\ket{v}\in M_{\mrm{GS}}$ can be represented by
\begin{equation}\label{eq:esn-field-coord}
    v = (\vec n,\sigma)\;\;\text{where}\;\; \sigma\in\{e,\eta\},\ \  \vec n\in (\mb S^2)_\sigma,
\end{equation}
where $\sigma$ is a ``branch'' label, and $\vec n$ is a smooth coordinate on the corresponding Bloch sphere.
Following Eq.~\eqref{eq:esn-gs-param}, the $s$-point will have coordinates $\vec n=(0,0,1)$ (or $\theta=0$) on both spheres, while the $e/\eta$ points will correspond to its antipode $\vec n=(0,0,-1)$ (or $\theta=\pi$).
Hence the low-energy excitations are described by spatially varying configurations $v(x) = (\vec{n}(x), \sigma(x))$, where $\sigma(x)$ can only take two discrete values $e$ or $\eta$.
\subsection{TDVP Equations and Treatment of Singularities}\label{subsec:TDVPsingularityprescript}
With these coordinates, when $\sigma(x)$ is fixed to be $\sigma \in \{e, \eta\}$, the field $\vec n(x)$ on that branch would have the diffusive TDVP equations of motion of Eq.~(\ref{eq:diff-eq-heis}) due to the Heisenberg interactions $P^{(\sigma s)}$.
However, the non-standard feature here is the possible discrete jumps of $\sigma(x)$ from $e$ to $\eta$ and vice versa.
For simplicity, let us first analyze it in 1d.
For smoothly varying configurations on the manifold, this jump is only possible when $\vec{n}(x) \cong s$, hence the spatial coordinates $\{x_i\}$ of the crossing points will act as Dirichlet boundaries for the $\vec n$ fields of each branch, i.e., we will demand
\begin{equation}
    \theta(x_i) = 0\;\;\;\text{if}\;\;\;\sigma(x_i^-) \neq \sigma(x_i^+),
\label{eq:crossingdirichlet}
\end{equation}
where $x_i^\pm$ denotes points in the immediate neighborhood of $x_i$, where the branch changes discontinuously.
The positions $x_i$ of such crossings can also change in time. To leading order, they evolve deterministically according to the condition that the \textit{pressure} imposed by the field configurations on both sides must equalize on both sides
\begin{equation}\label{eq:equality-pressures}
    |\partial_x \vec n(x_i^-)|^2 = |\partial_x \vec n(x_i^+)|^2.
\end{equation}
Intuitively, if the fields on either side of the crossing points $x_i(t)$ are imbalanced, then $x_i(t)$ will move, thus compressing one side and expanding the other, while satisfying the equalization of pressures at all times.
A derivation of this condition from the general form fo the Hamiltonian, and the higher-dimensional generalization from TDVP, can be found in Appendix~\ref{app:bound-cond}.
In the simplest cases that we will study, the position of the crossing will be fixed by spatially symmetric initial conditions which automatically enforce Eq.~\eqref{eq:equality-pressures}.
Nevertheless, the motion of the crossing point is important to capture behaviors of the system that lead to phenomena such as entanglement saturation at late-times.
Beyond this semiclassical behavior, we find that there are fluctuations of $\{x_i\}$ that are not captured by TDVP, which only lead to subleading contributions to the quantities of interest, which is something that we also verify numerically, as we discuss in the following sections.
In Appendix~\ref{sec:fluctuations} we also discuss a potential framework to understand the nature of such fluctuations by mapping the imaginary-time evolution under $P^{(es\eta)}$ to a classical stochastic model.
\subsection{Initial States for Time Evolution}\label{subsec:esn-init-states}
We mainly study the evolution of three different initial states under imaginary-time dynamics under the $es\eta$ model in $1$d:
\begin{equation}\label{eq:esn-initial}
\begin{gathered}
	\ket{\textsc{dw}}=\ket{\tilde e\,...\,\tilde e\,\tilde e\,\tilde \eta\,\tilde \eta\,...\,\tilde \eta},\\[.4em]
	\ket{\textsc{sp}_e}=\ket{\tilde e\,...\,\tilde e\, e \,\tilde e\,...\,\tilde e},\quad \ket{\textsc{sp}_\eta}=\ket{\tilde e\,...\,\tilde e\,\eta \,\tilde e\,...\,\tilde e},\\[.4em]
	\text{where}\ \ \ket{\tilde e}=\ket{s}+\ket{e}\ \  \text{and}\ \  \ket{\tilde \eta}=\ket{s}+\ket{\eta}.
\end{gathered}
\end{equation}
As we discuss below, the domain wall $\ket{\textsc{dw}}$ and single particle $\ket{\textsc{sp}_{\sigma}}$ states are related to the two kinds of two-replica observables introduced in Sec.~\ref{sec:observables-general}: the domain wall dynamics parallels the evolution of the purity defined in Eq.~\eqref{eq:renyi-purity}, while the single particle dynamics corresponds to the correlators of Eq.~\eqref{eq:squared-correlator}, for local observables.
We will also discuss generalizations to higher dimensions.
\subsubsection{Domain wall state}
The reason we associate the evolution of this domain wall $\ket{\textsc{dw}}$ in the toy model to entanglement growth in interacting symmetric Brownian models, is because of Eq.~\eqref{eq:definition-dw-id-swap}.
The $k=2$ domain wall $\sket{A\!:\!\bar A}$ there for an infinite one-dimensional system with
\begin{equation}\label{eq:bipartition}
    A=\{x>0\}
\end{equation}
is composed of $\sket{\1^{\ot 2},e}$ and $\sket{\1^{\ot 2},\eta}$ states on the left and right halves of the system.
Both have non-zero overlap with any replicated frozen state [defined in Eq.~(\ref{eq:frozenstatedefn})], since
\begin{equation}
\begin{gathered}
	\langle{\1^{\ot 2},e}\,|(\sket{f}\ot\sket{f}^*)^{\ot 2} = \langle{\1^{\ot 2},\eta}\,|(\sket{f}\ot\sket{f}^*)^{\ot 2}=\\
	=|\braket{f}|^2=1
\end{gathered}
\label{eq:DWinitoverlapeqn}
\end{equation}
This makes $\ket{\tilde e}$ and $\ket{\tilde \eta}$ analogous to $\sket{\1^{\ot 2},e}$ and $\sket{\1^{\ot 2},\eta}$, since they both have overlap equal to one with $\ket s$, which is the analog of the frozen state, and they are points on $(\mb S^2)_e$ and $(\mb S^2)_\eta$ respectively (see Fig.~\ref{fig:gsman-all}a).
As in Eq.~\eqref{eq:renyi-purity}, we will be interested on the overlap of the imaginary time evolved state
\begin{equation}
    \ket{\textsc{dw}(t)} = \exp(-\kappa P^{(es\eta)}t)\ket{\textsc{dw}}.
\label{eq:DWevolution}
\end{equation}
with various simple initial states that will represent different initial states analogous to $\sket{\rho^{\ot 2}}$.
While this will strongly depend on the choice of initial states, the overall decay in the norm of $\ket{\textsc{dw}(t)}$ will indicate the generic rate of decay of such overlaps at long times.
\subsubsection{Single particle states}

The single particle states $\ket{\textsc{sp}_\sigma}$ in Eq.~(\ref{eq:esn-initial}) are meant to represent two different types of replicated operators.
The state $\ket{\textsc{sp}_e}$ models the evolution of replicated \textit{hydrodynamic} operators in the replica theory, i.e., operators that overlap with the conserved charges of the system (assuming a continuous symmetry group $G$).
Following the definition in Sec.~\ref{subsubsec:hydrononhydro} [see Eq.~(\ref{eq:hydrodefn})], a local hydrodynamic operator $O_i$ has non-zero overlap with some states in the $k=1$ commutant manifold $M_G^{(1)}$.
It follows that the replicated hydrodynamic operator on a single site $O \otimes O\+$ [see Eq.~(\ref{eq:replicatedhydroop})], will have non-zero overlap with the $e$ branch of the $k=2$ commutant manifold $M_G^{(2)}$, i.e.,
\begin{equation}
	\ket{O\otimes O\+;e}\not\perp (M_G^{(1)}\times M_G^{(1)})_e.
\end{equation}
For a non-trivial operator $O$ which satisfies $\sket{O\otimes O\+;e} \not\propto \sket{\mathds{1}^{\ot 2}; e}$, the initial state of the replicated hydrodynamic operator $O_i$ shown in Eq.~(\ref{eq:replicatedhydroop}) is analogous to the $\ket{\textsc{sp}_e}$ state in Eq.~(\ref{eq:esn-initial}), where the $\tilde{e}$ state represents $\sket{\mathds{1}^{\ot 2}; e}$ and $e$ represents $\sket{O\otimes O\+;e}$.
Conversely, non-hydrodynamic operators are those that are completely orthogonal to states in $M^{(1)}_G$, and hence their replicated versions are orthogonal to the $e$ branch of $M_G^{(2)}$.
Hence the single-site non-hydrodynamic operator $\sket{O \otimes O\+;e}$ is analogous to $\ket{\eta}$ here, and so the replicated non-hydrodynamic operator of Eq.~(\ref{eq:replicatedhydroop}) is analogous to the state $\ket{\textsc{sp}_\eta}$ of Eq.~(\ref{eq:esn-initial}).
Then we will be interested in computing the imaginary-time evolved single-particle states, which are given by
\begin{equation}
    \begin{gathered}
    \ket{\textsc{sp}_e(t)} = \exp(-\kappa P^{(es\eta)}t)\ket{\textsc{sp}_e}, \\
    \ket{\textsc{sp}_\eta(t)} = \exp(-\kappa P^{(es\eta)}t)\ket{\textsc{sp}_\eta}.
\label{eq:SPevolution}
\end{gathered}
\end{equation}
In particular, for autocorrelation functions similar to Eq.~(\ref{eq:squared-autocorrelator}), we are interested in computing the norms since
\begin{equation}
    \frac{1}{2^{L/2}}\braket{\textsc{sp}_\sigma(0)}{\textsc{sp}_\sigma(t)} \propto \norm{\,\sket{\textsc{sp}_\sigma(t/2)}\,}^2.
\label{eq:SPnorm}
\end{equation} 
The overall decay in the norm of $\ket{\textsc{sp}_\sigma(t)}$ will indicate the expected rate of decay of such overlaps at long times.
\subsection{Domain Wall Dynamics}\label{subsec:esn-dw} % MARK: * dw ent
We first illustrate the TDVP continuum dynamics of domain walls in one dimension [cf.~Eq.~\eqref{eq:esn-initial}], which appproximates Eq.~(\ref{eq:DWevolution}).
We can picture the initial state in an infinite system as being described by a discontinuous field, using the coordinates of Eq.~\eqref{eq:esn-field-coord},
\begin{equation}\label{eq:initial-cond-dw-esn}
    v(x,0)=\begin{cases}
       (\vec n_0,e) & x<0\\
       (\vec n_0,\eta) & x>0,
    \end{cases}\quad \vec n_0=(1,0,0),
\end{equation}
where $\vec n_0$ in polar coordinates is $({\theta_0=\frac{\pi}{2}},{\varphi_0=0})$.
Given the general geometric considerations discussed in Sec.~\ref{sec:tdvp}, imaginary time evolution will smoothen this discontinuity, and at hydrodynamic timescales, the leading field configuration $v(x,t)$ will smoothly connect the $\tilde e$-point with the $\tilde \eta$-point passing through the $s$-point, as shown in Fig.~\ref{fig:gsman-all}a.
\subsubsection{Explicit TDVP solution and void formation}\label{subsubsec:DWTDVP}
The TDVP equations of motion for $v(x,t)$ as long as $\vec{n}$ is away from the $s$-point is exactly the same as for the ferromagnetic Heisenberg model Eq.~\eqref{eq:diff-eq-heis}.
Due to the symmetry of the initial condition, the angle $\varphi(x,t)$ remains zero throughout the time-evolution, and the position of the $s$-point will remain fixed at $x=0$ due to the satisfaction of Eq.~(\ref{eq:equality-pressures}), thus providing a fixed Dirichlet boundary condition for the equations of motion, where $v(0,t)\cong s$ or equivalently, $\vec{n}(0, t) = (0,0,1)$.
Using the parametrization of Eq.~\eqref{eq:esn-gs-param} we can therefore simply write [cf.~\eqref{eq:eom-theta} and \eqref{eq:crossingdirichlet}]:
\begin{equation}\label{eq:diffeq-esn-dw}
       \partial_t\theta(x,t) = \kappa\, \partial_x^2 \theta,\;\;\;
       \theta(0,t) = 0.
\end{equation}
The Dirichlet boundary conditions can be imposed using the method of images, where we parametrize $\theta(x,t)$ as 
\begin{gather}
\theta(x,t) = |\omega(x,t)|,\;\;
\partial_t \omega(x,t) = \kappa\, \partial_x^2 \omega \nonumber\\
\omega(0, t) = 0,\;\;\omega(x, 0) = \begin{cases}
    +\frac{\pi}{2}\;\;x > 0\\
    -\frac{\pi}{2}\;\;x < 0
\end{cases}.
\label{eq:methodofimages}
\end{gather}
The exact solution for $\omega(x,t)$ with the required conditions can then be directly obtained using Fourier analysis, which ultimately yields:
\begin{equation}\label{eq:esn-sol-dw}
    \theta(x,t)=\frac{\pi}{2}\left|\mrm{erf}\left(\frac{x}{\sqrt{4\kappa t}}\right)\right|\!,\;\; \sigma(x,t)=\begin{cases}
       e & x<0,\\
       \eta & x>0,
    \end{cases}
\end{equation}
analogous to Eq.~\eqref{eq:sol-heis-th} for the Heisenberg model.
This explicit solution can be interpreted as the formation of a \textit{void} of width $w\sim\sqrt{\kappa t}$, a region where the local degrees of freedom are roughly aligned along the $\ket s$ state.\footnote{The term `void' appears in the recent works Refs.~\cite{McCulloch_2026, mcculloch2026longlivedlocalquantumcoherences} in reference to a similar phenomenon observed in replica systems with a conserved $U(1)$ charge which we discuss in Sec.~\ref{sec:u1}. }
This void formation is a clear dynamical signature of the contact singularity in the ground state manifold: any continuous spatial configuration that needs to switch branch due to the initial conditions will necessarily ``pass through'' the singular point $s$, thus creating a diffusively large spatial region dominated by the $\ket s$ state.
The domain wall does not remain sharp, but rather melts, due to the presence of a connected low-energy path between the ground states configurations $\ket{\tilde e}$ and $\ket{\tilde\eta}$ on either side of the domain wall.
\subsubsection{Sink growth and fluctuations beyond TDVP}\label{subsubsec:DWannihilation}
Heuristically, we can think of the formation of the void by considering the evolved wavefunction as an ensemble of various computational basis configurations of $e$, $s$, and $\eta$, which we refer to as the \textit{components} of this wavefunction.
Physically, the formation of this region is driven by the diffusive transport of $e$- and $\eta$-particles on an $s$-background in the $es\eta$ model, whose interaction is simply governed by the ferromagnetic Heisenberg model.
Once they meet in the middle they effectively annihilate under the imaginary-time evolution operator $e^{-\kappa P^{(es\eta)}\delta t}$ due to the gapped $e$-$\eta$ interaction, which exponentially suppresses any component of the wavefunction where $\ket e$ and $\ket \eta$ states are neighbors.
We will refer to the region where this effective annihilation process takes place as the \textit{sink}.
According to the TDVP ansatz discussed above, this sink is point-like and stationary, positioned at $x=0$.

\begin{figure}
    \centering
    \includegraphics[width=\linewidth]{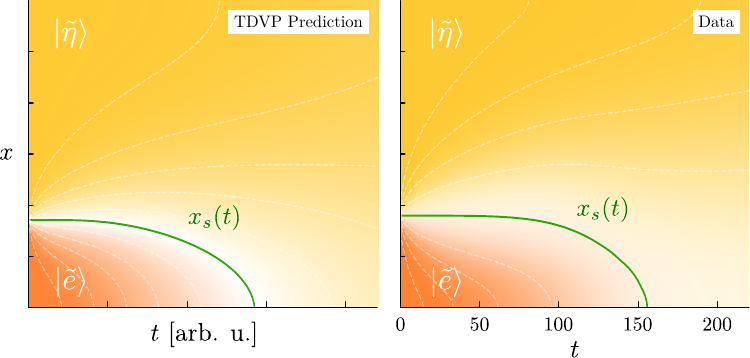}
    \caption{Comparison between solution to the TDVP equations of motion (left) and numerical simulation data (right) for a domain wall close to the spatial boundary of the system, showing its annihilation, as in Sec.~\ref{subsubsec:finite-size-dw}. Time in the TDVP continuum prediction is in arbitrary units and has been approximately rescaled to match the data. Numerical data obtained using MPS simulation with $L=60$, with other parameters similar to Fig.~\ref{fig:esn}.}
    \label{fig:esn-unb}
\end{figure}

However, this sink is not exactly stationary, since---as argued in Sec.~\ref{subsec:TDVPsingularityprescript}---the effective position where the `annihilation' occurs can fluctuate.
Physically this means that the $e$ and $\eta$ states can coexist in the same position, although on different components of the wavefunction.
If we treat the position $x_s(t)$ of the $s$-point as an interface between the $e$ and $\eta$ regions (as shown in Fig.~\ref{fig:esn}a), then its spatial fluctuations precisely delineates the sink region.
From numerical simulations, we expect these quantum fluctuations around the $s$-point to provide the leading correction to the TDVP trajectory of Eq.~\eqref{eq:esn-sol-dw}.
The size of the sink region can be estimated as the width of the region where the product of the $e$ and $\eta$ densities is not negligible
\begin{equation}\label{eq:sink-region-dw}
    \langle\rho_e(x,t)\rangle \cdot \langle\rho_\eta(x,t)\rangle \gtrsim 0.
\end{equation}
Then, the data shown in Fig.~\ref{fig:esn}b indicates that this region grows subdiffusively as $\sim t^{1/3}$ at long times.
Nevertheless, this is small compared to the $\sim t^{1/2}$ width of the void region driven by diffusion, and hence justifies the leading TDVP approximation for studying the properties of the void at late-times.
We discuss how to better understand these fluctuations in Appendix~\ref{sec:fluctuations}.
\subsubsection{Domain wall annihilation and entanglement saturation}\label{subsubsec:finite-size-dw}
For a system of finite size $L$, where $x \in [-L/2, L/2]$, the above picture and the TDVP solution Eq.~\eqref{eq:esn-sol-dw} will only hold at intermediate timescales, before the domain wall senses the finiteness of the system.
In the context of a concrete Brownian model with a continuous symmetry, at this timescale the system would enter the regime where purity approaches its saturation value.
In order to capture the asymptotic behavior of the domain wall initialized at $x_s(0)=0$, we should in addition introduce Neumann conditions $\partial_x \vec n(x) = 0$ [see Eq.~\eqref{eq:neumann-conditions}], i.e. $\partial_x \theta(x) = 0$, at the spatial boundaries of the system.
In such a case, we expect the system to be in the hydrodynamic regime for times up to $t\sim\mc O(L^2/\kappa)$, where the system size is effectively not felt by the melting domain wall.
At later times, we should instead expect the wavefunction $\ket{\textsc{dw}(t)}$ to only contain a few hydrodynamic excitations, as the imaginary-time evolution asymptotically converges to the ground state $\ket{s\cdots s}$ in a finite system.
This physics is reflected in the TDVP dynamics, where around this time the growing void region reaches the boundary of the system, and the solution is approximately described by the lowest Fourier mode (i.e., spin-wave) compatible with the Neumann boundary conditions
\begin{equation}\label{eq:esn-asympt-sol}
	\theta(x,t)\approx \frac{4}{\pi}\left|\sin(\frac{\pi x}{L})\right|\, e^{-\frac{\kappa \pi^2}{L^2}t},\quad x\in[-L/2,L/2].
\end{equation}
Thus at late times, the half-chain purity approaches exponentially fast its saturation value, dictated by the asymptotic state $\ket{\textsc{dw}(\infty)}=\ket{s\cdots s}$.
If however the initial domain wall is not centered, but is instead positioned at $x_s(0)=-L/2+\ell$ for $\ell<L/2$, the additional condition Eq.~\eqref{eq:equality-pressures} will also play a role, leading to a ``pressure-induced'' drift in the position of $x_s(t)$.
At short timescales $t \ll O(\ell^2/\kappa)$, the evolution of the domain wall and the void formation is as in the infinite size system.
At timescales $t\sim\mc O(\ell^2/\kappa)$, the void region hits the edges of the system, which creates an imbalance between the `$e$' ($x<x_s$) and `$\eta$' ($x>x_s$) density on both sides of the domain wall.
If $x_s(t)$ remained stationary, this would lead to a violation of Eq.~(\ref{eq:equality-pressures}), so $x_s(t)$ is forced to drift towards $x=-L/2$.
In this way, $x_s(t)$ will eventually hit the boundary, which results in a field configuration that lies completely in the $\eta$ branch of the ground state manifold.
Therefore, at later times the fields evolve purely under the Heisenberg dynamics of $P^{(\eta s)}$, and eventually settles to a ground state containing only $\ket{\eta}$ and $\ket{s}$ states.
This picture follows explicitly from the solution to the TDVP equations of motion, which can be obtained through the method of images similar to Eq.~(\ref{eq:methodofimages}), where we write
\begin{equation}\label{eq:image-method1}
    \theta(x,t) = |\omega(x,t)|,\;\;\; \sigma(x,t)=\begin{cases}
       e & \omega(x,t)<0,\\
       \eta & \omega(x,t)>0,
    \end{cases}
\end{equation}
where $\omega(x,t)$ solves the conditions
\begin{equation}\label{eq:image-method2}
    \begin{cases}
       \partial_t \omega=\kappa\, \partial_x^2 \omega,\\
       \partial_x \omega(\pm \frac{L}{2},t)=0,
    \end{cases}
    \omega(x,0)=\begin{cases}
       -\frac{\pi}{2} & x<x_s(0),\\
       +\frac{\pi}{2} & x>x_s(0).
    \end{cases}
\end{equation}
The numerical solution to these equations is shown and compared to simulation data in Fig.~\ref{fig:esn-unb}.
A similar phenomenon also occurs when the subsystem $A$ in Eq.~\eqref{eq:definition-dw-id-swap} is a finite interval of size $\ell$, corresponding here to a finite region composed of $\ket{\tilde e}$ states, surrounded by a larger region of $\ket{\tilde \eta}$ states.
In this case, the two domain walls with positions $x_{s,L}(t)$ and $x_{s,R}(t)$ generate voids that intersect, which then leads to them moving closer due to the pressure difference.
They then collide and annihilate, leading to the finite interval to be eventually fully depleted of $e$ particles, and leaving thus behind a state composed exclusively of $\ket{\eta}$ and $\ket{s}$ states.
The evolution then converges to a ground state composed of just those two states, and result in the saturation of any overlaps involved.
\subsubsection{Norm and overlap evolutions}
Since the observable of interest in the evolution of the domain wall state $\ket{\textsc{dw}(t)}$ is its overlap with fixed initial states over time, we first need to fix the norm of the evolved state.
Since we know that asymptotically $\ket{\textsc{dw}(t)}$ will approach the ground state $\ket{s \cdots s}$, we choose to follow the normalization convention of Eq.~\eqref{eq:norm-evo-eq-prime}; therefore we enforce the exact condition
\begin{equation}
    \braket{\textsc{dw}(t)}{s \cdots s} = \bra{\textsc{dw}} e^{-\kappa P^{(es\eta)}t}\ket{s \cdots s} =  1.
\label{eq:DWnormfix}
\end{equation}
In the continuum limit, since TDVP evolves $\ket{\textsc{dw}(t)}$ within the space of product states of the form of Eq.~(\ref{eq:esn-gs-param}) with a spatial profile specified by $\theta(x)$ and $\sigma(x)$ of Eq.~(\ref{eq:esn-sol-dw}), imposing the normalization of Eq.~(\ref{eq:DWnormfix}) is equivalent to rescaling the local states on each site such that the coefficient of $\ket{s}$ is $1$, i.e.,
\begin{equation}
    \ket{\theta,\sigma} \defn \ket{s} + \tan\frac{\theta}{2}\ket{\sigma},\;\;\;\sigma \in \{e, \eta\}.
\label{eq:rescaled_esn}
\end{equation}
Hence the prediction for the evolved domain wall in the continuum is given by
\begin{equation}
    \ket{\textsc{dw}(t)} \approx \bigotimes_{x}\ket{\theta(x),\sigma(x)}.
\label{eq:DWTDVPevolution}
\end{equation}
This leads to the norm of the state being
\begin{equation}
    \norm{\,\ket{\textsc{dw}(t)}\,}^2\approx \prod_x \left(1+\tan\frac{\theta(x)}{2}\right).
\end{equation}
Note that since $\tan\frac{\theta(x)}{2}$ is everywhere equal to one in the initial state Eq.~\eqref{eq:initial-cond-dw-esn}, while it is smaller than one within the void, the growth of the void directly dictates the decay of the norm, resulting in
\begin{equation}
	\norm{\,\ket{\textsc{dw}(t)}\,} \sim e^{-\sqrt{\kappa t}}.
\end{equation}
\begin{figure}
    \centering
    \includegraphics[width=0.9\linewidth]{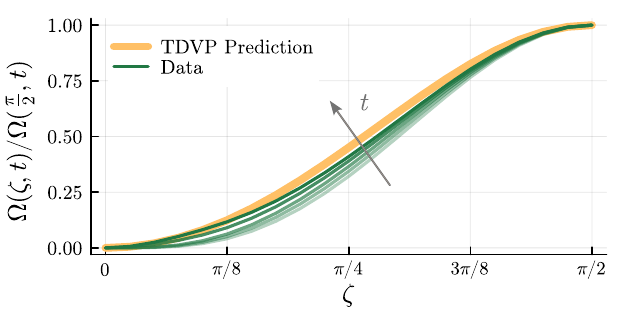}
    \caption{Comparison with numerical data of the TDVP prediction for the asymptotic value of the of overlaps between $\ket{\textsc{dw}(t)}$ [obtained using Eqs.~(\ref{eq:DWevolution}) and (\ref{eq:DWTDVPevolution}) for the numerical data and TDVP respectively] and $\ket{\zeta}^{\ot L}$. 
    Data is the same as in Fig.~\ref{fig:esn}; times shown between $t=12$ and $t=244$.}
    \label{fig:esn-data}
\end{figure}
We can now compare the TDVP prediction of certain overlaps with numerical data from simulations.
In Appendix~\ref{app:conti} we compute the predictions for the overlap of the evolved domain wall state $\ket{\textsc{dw}(t)}$ and a family of product states $\{\ket{\zeta}^{\ot L},\,\zeta\in[0,\frac{\pi}{2}]\}$ where
\begin{equation}
\ket{\zeta}\defn\cos^2\zeta\ket{s}+\sin^2\zeta\frac{\ket{e}+\ket{s}+\ket{\eta}}{2}.
\end{equation}
Since they are computed in the continuum limit, and since overlaps are multiplicative in the system size, the physically meaningful quantities to compare with predictions are ratios of the logarithms of the overlaps.
\begin{equation}
    \Omega(\zeta,t)\defeq -\log(\braket{\textsc{dw}(t)}{\zeta}^{\ot L}).
\label{eq:Omegabehavior}
\end{equation}
For large system sizes $L$, at intermediate times $ t\ll O(L^2/\kappa)$, these are predicted to scale as $\sim\sqrt{t}$ with a prefactor that only depends on $\zeta$.
We also test this prefactor numerically in Fig.~\ref{fig:esn-data}, finding good qualitative agreement in the expected intermediate timescales.
These types of overlaps model the entanglement generated according to Eq.~\eqref{eq:renyi-purity} for simple product initial states, with $\ket{\zeta}$ having a structure analogous to a family of states that we will study in the $U(1)$ model of Sec.~\ref{sec:u1}.
\subsubsection{Higher dimensions}\label{subsubsec:higherd}
So far we have discussed the semiclassical evolution of domain walls in 1d.
In $d$ dimensions, domain walls correspond to $(d-1)$-dimensional defects.
The point $x_s(t)$ which indicates in 1d the position where the field crosses the $s$-point on the ground state manifold, then should be replaced by a fluctuating hypersurface, along which $e$- and $\eta$-particles effectively annihilate as explained in Sec.~\ref{subsubsec:DWannihilation}.
Within the TDVP ansatz, voids will again form and expand diffusively around these interfaces between `$e$' and `$\eta$' domains.
Since under the normalization condition Eq.~\eqref{eq:rescaled_esn} the norm of the evolved domain wall state is bound to the size of the dynamically generated void, the intermediate-time dynamics of log-overlaps such as Eq.~(\ref{eq:Omegabehavior}) is expected to scale as $\sim \sqrt{t}$ with a pre-factor that is proportional to the size of the domain wall (e.g., its length in $d = 2$ and area in $d = 3$).
However, an important qualitative difference from $1$d  comes from the fact that in higher dimensions, domain walls generally possess a non-trivial curved shape.
Such curvature immediately creates a density imbalance between the two sides of the defect, which locally shrinks to satisfy the equal-pressure condition [cf.~Eq.~(\ref{eq:higherdimspressure})]
\begin{equation}\label{eq:press-eq-esn-d}
    |\bm n\cdot\bm\nabla \vec n(x,t)|^2\,\Big|_{x\in\mrm{bd}^+} = |\bm n\cdot\bm\nabla \vec n(x,t)|^2\,\Big|_{x\in\mrm{bd}^-},
\end{equation}
where $\bm n$ is the unit vector normal to the hypersurface (see Appendix~\ref{app:bound-cond}).
Unless the domain wall is at least locally straight, there is therefore no separation of timescales between the void formation and the drift of the defect.
If the subsystem $A$ is of finite size, the defects will eventually annihilate, and the overlaps of interest will saturate at diffusive timescales.
The condition Eq.~\eqref{eq:press-eq-esn-d}, together with the bulk equations of motion Eq.~\eqref{eq:diff-eq-heis} can be solved by Fourier analysis using the method of images similar to the 1d case of Eq.~\eqref{eq:image-method1}. 
\subsection{Single Particle Dynamics} % MARK: * part
\label{subsec:sp}
We now move on to the single particle states $\ket{\textsc{sp}_e}$ and $\ket{\textsc{sp}_\eta}$ from Eq.~\eqref{eq:esn-initial}, where we are interested in characterizing the imaginary time evolution Eq.~(\ref{eq:SPevolution}). 
Due to the global number conservation for all three species in this model, $\ket{\textsc{sp}_e(t)}$ will always stay within the subspace spanned by $\ket e$ and $\ket s$ states, and $\ket{\textsc{sp}_\eta(t)}$ will always possess a single $\eta$ particle.
We can picture the initial state $\ket{\textsc{sp}_\sigma}$ here in an infinite system as being described by a simple background field, with a \textit{point-like defect} at position $x_\sigma(t=0)=0$.
In this setting, only the background field can be described semiclassically through TDVP, and not the single-particle motion of the defect, although as we discuss below, an exact solution is possible for $\sigma = e$.
Using the coordinates of Eq.~\eqref{eq:esn-field-coord}, for this state we have the background field
\begin{equation}\label{eq:initial-cond-sp-esn}
    v(x,0)= (\vec n_0,e)\;\;\text{for}\;\; x \neq 0,
\end{equation}
where $\vec n_0$ in polar coordinates is $({\theta_0=\frac{\pi}{2}},{\varphi_0=0})$.
This initial condition also generalizes straightforwardly to higher dimensions.
\subsubsection{$e$ particle}
Since the Hamiltonian acting in the sector only containing $e$ and $s$ is just the Heisenberg model, the dynamics of the state $\ket{\textsc{sp}_e}$ are equivalent to the dynamics of the state
\begin{equation}
	\ket{\rt\dots\rt\,\,\up\,\,\rt\dots\rt}
\end{equation}
under the ferromagnetic Heisenberg model Eq.~\eqref{eq:HeisModelHam}.
This evolution can be solved exactly within the spin-wave sector of the Hilbert space, and results in the `up' spin diffusing on the background as a single free particle under imaginary time evolution, similar to the case of hydrodynamic operators on a single replica~\cite{moudgalya2024}.
Explicitly, in the continuum limit at late times, we get:
\begin{equation}
    \ket{\textsc{sp}_e(t)}\sim \int\dd{x_e}\psi(x_e,t) \ket{...\rt\,\,\up_{x_e}\rt...}\label{eq:e-wavefunct}
\end{equation}
where $\psi(x_e,t)\sim e^{-x_e^2/\kappa t}$ is the wavefunction of the particle.
Due to this form the energy of this state for an infinite system decays as $\sim 1/t$ over time, which leads to an algebraic decay of the norm according to Eq.~\eqref{eq:norm-evo-eq} (which is exact in this case of an exact solution to the quantum evolution)
\subsubsection{$\eta$ particle}
A more peculiar behavior is instead found for $\ket{\textsc{sp}_\eta}$.
As discussed in the case of domain walls, the only way to spatially connect regions where $\sigma(x)=e$ and regions where $\sigma(x)=\eta$ while staying on the ground state manifold, is by having $\vec{n}(x)$ pass through the $s$-point.
In order to minimize the energy---starting from a single $\eta$ particle in a $\tilde{e}$ background---the $\eta$ particle must therefore be surrounded by a void region, analogous to the one described in Sec.~\ref{subsec:esn-dw}; in this void the $\eta$ particle can then move freely.
Hence in the continuum, the position of the $\eta$ particle will enforce a Dirichlet boundary condition $v(x_\eta,t)\cong s$ for the background field Eq.~\eqref{eq:initial-cond-sp-esn} [see Eq.~(\ref{eq:crossingdirichlet})], with $x_\eta$ acting as a defect for the field.
As discussed in Sec.~\ref{subsec:TDVPsingularityprescript}, since the TDVP equations of motion Eq.~(\ref{eq:tdvp-continuum-main}) do not describe large superpositions of quantum states, we posit that the defect $x_\eta(t)$ remains stationary to leading order, and numerically justify that its fluctuations are subleading.
If we use the parametrization in spherical coordinates of Eq.~\eqref{eq:esn-field-coord} for the background field, the equations of motion then become completely identical to the ones studied in the case of the domain wall initial state in Eq.~(\ref{eq:diffeq-esn-dw}), with the initial position of the $\eta$ particle playing the role of the initial domain wall.
The solution too has a similar form, and leads then to the creation of a diffusively large void of size $w\sim\sqrt{\kappa t}$, as discussed in Sec.~\ref{subsubsec:DWTDVP} (see Fig.~\ref{fig:esn}c).
Explicitly, the solution to the TDVP equations of motion is [cf.~Eq.~\eqref{eq:esn-sol-dw}]
\begin{equation}\label{eq:esn-sol-sp}
    \theta(x,t)=\frac{\pi}{2}\left|\mrm{erf}\left(\frac{x}{\sqrt{4\kappa t}}\right)\right|\!,\quad \sigma(x, t) = e
\end{equation}
The energy of this state can be estimated to scale as $E\sim \sqrt{\kappa/t}$, and hence according to Eq.~\eqref{eq:norm-evo-eq} its norm scales as
\begin{equation}\label{eq:norm-dec-1d}
	\norm{\,\ket{\textsc{sp}_\eta(t)}\,}^2\sim e^{-\sqrt{\kappa t}}.\quad (d=1)
\end{equation}
Recall that this squared norm is the analog of the non-hydrodynamic correlator Eq.~\eqref{eq:squared-autocorrelator} [see Eq.~(\ref{eq:SPnorm})].
Remarkably, not only is this leading TDVP behavior the same as in the case with the domain walls, but---according to the data in Fig.~\ref{fig:esn}c---so are the spatial fluctuations of the defect $x_\eta(t)$, which correspond to the sink region (discussed for the domain wall in Sec.~\ref{subsubsec:DWannihilation}).
Here too, the $\eta$ particle can be seen as a sink for diffusing $e$ particles, due to exponential suppression $e^{-P^{(es\eta)}\delta t}$ of the components of the wavefunctions where the two species come into contact.
Simulations in Fig.~\ref{fig:esn}c show that the width of this sink zone, estimated here through the standard deviation of the position $\langle x_\eta^2(t)\rangle^{1/2}$, grows subdiffusively as $\sim t^{1/3}$, similar to the sink region in the domain wall case.
This corroborates the idea that in the continuum limit, any appropriate field theory that captures effects beyond TDVP would describe both objects---the $\eta$ particles and the domain walls---as fluctuating defects, as argued more rigorously in Appendix~\ref{sec:fluctuations}.
\subsubsection{Dimension $d>1$}
Note that the correspondence of the $\eta$ particle dynamics with domain wall dynamics can only be valid in $d=1$, where particles and domain walls are both point-like objects.
For $d>1$, the TDVP equations of motion for the evolution of a single-particle state $\ket{\textsc{sp}_\eta}$ are given by Eq.~\eqref{eq:diff-eq-heis}, with the Dirichlet condition $\vec{n}(\bm 0,t)=(0,0,1)$ given by the $\eta$ particle at a fixed position $x_\eta(t)=\bm 0$.
Here we claim that the behavior of Eq.~(\ref{eq:norm-dec-1d}) no longer holds in $d > 1$.
Since the initial state is everywhere $\vec{n}(\bm x)=(1,0,0)$, we can rewrite the TDVP equations for this initial state only in terms of the polar angle $\theta(\bm x,t)$ as:
\begin{equation}
       \partial_t\theta(\bm x,t) = \kappa \bm\nabla^2 \theta,\;\;\;
       \theta(\bm x, t) = 0\;\;\text{if}\;\; |\bm x| = b,
\end{equation}
where, for reasons that will be clear, we have regularized the Dirichlet condition to hold within a radius $b$, which acts as a UV cutoff for the TDVP equations and is at least of the order of the lattice spacing.
In addition, for the initial state of the form of Eq.~(\ref{eq:initial-cond-sp-esn}) on an infinite system, we should also impose the boundary condition at infinity as
\begin{equation}
    \lim_{|\bm x|\rightarrow \infty} \theta(\bm x, t) = \frac{\pi}{2}\;\;\text{for any finite $t$}.
\end{equation}
We see that in dimensions $d\geq 3$, these equations of motion admit non-trivial stationary solutions, i.e., with $\partial_t\theta_\mrm{st}(\bm x) = 0$.
These correspond to $\theta_\mrm{st}(\bm x)$ being a bounded harmonic function which satisfies the boundary conditions at the origin and at infinity, which here reads:
\begin{equation}
	\theta_\mrm{st}(\bm x)=\frac{\pi}{2}-\frac{\pi}{2}\left(\frac{b}{|\bm x|}\right)^{d-2}.
\label{eq:statprofile}
\end{equation}
According to the bare $P^{(es\eta)}$ Hamiltonian in the continuum [cf.~Eq.~\eqref{eq:heis-ham-cont-lim}], this stationary configuration has energy 
\begin{equation}
    E_\mrm{st}\propto \kappa \left(\frac{b}{a}\right)^{d-2},
\label{eq:Estval}
\end{equation}
where $a$ is the lattice spacing, and we expect the ratio $b/a$ to be finite in the continuum limit.
The initial state $\ket{\textsc{sp}_\eta}$ can asymptotically approach this stationary state in an infinite system.
With this asymptotic profile of Eq.~(\ref{eq:statprofile}), we can conclude that the void region dominated by the $\ket{s}$ state ($\theta = 0$), saturates to a finite size of order $\mathcal{O}(b)$ for $d\geq 3$, unlike with $d = 1$, where its size grows indefinitely with time.
Using Eq.~(\ref{eq:norm-evo-eq}), this finite energy asymptotic stationary state results in the exponential decay of the norm:
\begin{equation}\label{eq:norm-dec-3d}
	\norm{\,\ket{\textsc{sp}_\eta(t)}\,}^2\sim e^{-\kappa t}.\quad (d\geq 3)
\end{equation}
The size of the radius $b$ might be increased in order to take into account the fluctuations of the position of the $\eta$-particle, but note that since a larger $b$ implies a larger energy $E_\mrm{st}$ in Eq.~(\ref{eq:Estval}) for the stationary solution, we should expect $x_\eta(t)$ to remain confined to a region of finite size during the evolution, consistent with the TDVP ansatz.
When $d=2$, similar to $d=1$, there is no bounded stationary solution to the TDVP equations that satisfies the boundary conditions, so the energy of the evolving configuration will decay to zero at long times.
This also means that the void region will grow indefinitely, similar to the $d = 1$ case.
We can estimate the time-dependent solution to be almost stationary within a radius given by standard diffusion
\begin{equation}
	\theta(\bm x, t)\approx \frac{\pi}{2}\frac{\log(|\bm x|/b)}{\log(\sqrt{4\kappa t}/b)} ,\quad |\bm x|\ll \sqrt{4\kappa t},
\end{equation}
while remaining at its initial value $\theta(\bm x, t)=\frac{\pi}{2}$ outside of this radius.
This quasi-stationary solution is simply the unique harmonic function that satisfies $\theta(|\bm x|=b)=0$ and $\theta(|\bm x|=\sqrt{4\kappa t})=\frac{\pi}{2}$.
It provides estimates for energy and norm decay:
\begin{equation}\label{eq:norm-dec-2d}
	E(t)\sim\frac{\kappa}{\log t},\quad \norm{\,\ket{\textsc{sp}_\eta(t)}\,}^2\sim e^{-\frac{\kappa t}{\log t}},\quad (d = 2)
\end{equation}
where we use Eq.~\eqref{eq:norm-evo-eq} for estimating the norm, and approximate $\int\frac{\dd t}{\log t}\sim \frac{t}{\log t}$.
These results for the asymptotic decay of the norm are consistent with the survival probabilities of particles in a related stochastic model~\cite{Blythe_2003} that should also capture fluctuations beyond TDVP, as discussed in Appendix~\ref{sec:fluctuations}.

\begin{figure*}
	\centering
    \textbf{(a)}\includegraphics[width=0.32\linewidth,valign=t]{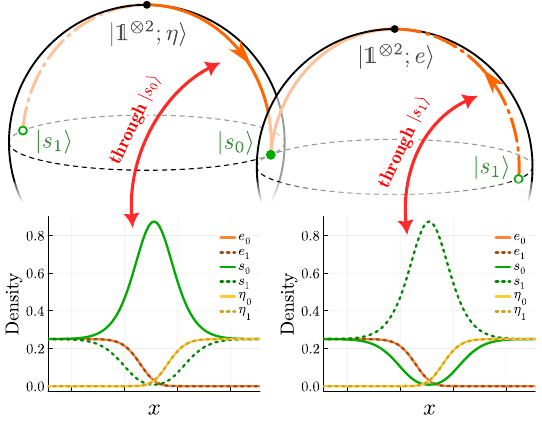}
    \textbf{(b)}\includegraphics[width=0.32\linewidth,valign=t]{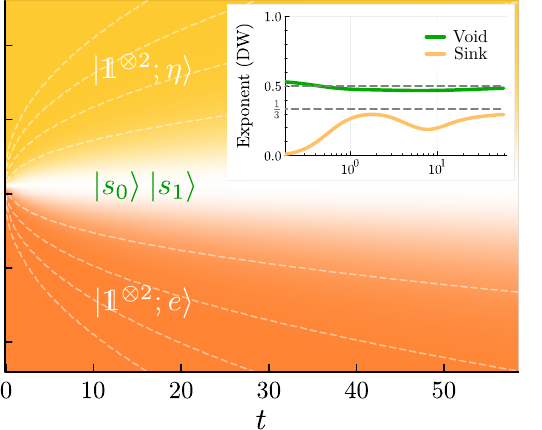}
    \textbf{(c)}\includegraphics[width=0.26\linewidth,valign=t]{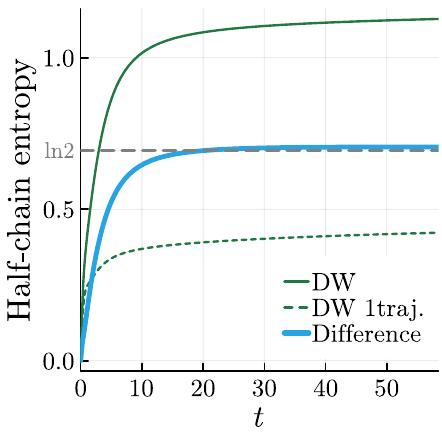}
	\caption{(a) Illustration of the two symmetry-related semiclassical trajectories that contribute to the formation of the void in the domain wall state. The plots show the spatial density of states for each trajectory, given by numerical simulations. (b) Density of $\{\ket{e_p}\}$ and $\{\ket{\eta_p}\}$ states in the evolution of $\sket{A\!:\!\bar A\,(t)}$ where $A$ is a half-line, showing the formation of the void where the $\{\ket{s_p}\}$ density dominates. Inset, power law exponent for the growth of the void (diffusive) and sink (subdiffusive). (c) Plot of the half-chain entropy for the evolved domain wall state: `DW' is the actual evolution of the domain wall state, while `DW 1traj.' is an initial domain wall state biased to select the trajectory through $\ket{s_0}$, as explained in the main text. Numerical data was obtained via MPS simulations with $L=180,181$ and bond dimension $80$, where imaginary time evolution was performed using an efficient implicit integration technique~\cite{Zima2026Jun}.}
	\label{fig:u1}
\end{figure*}
\section{Models with $U(1)$ Symmetry} % MARK: 5. U(1)
\label{sec:u1}
We now turn to the dynamics of noisy $U(1)$ symmetric models, and focus mostly on 4-copy ($k=2$) replica models that are relevant for the computation of averages of purities and squares of two-point correlators.
Random unitary charge-conserving many-body dynamics has been studied in the literature in the context of
charge statistics~\cite{McCulloch_2023,Turkeshi_mpemba},
evolution of correlators~\cite{Khemani2018Sep,Rakovszky2018Sep,Agarwal2023May,McCulloch_2026,mcculloch2026longlivedlocalquantumcoherences}, and entanglement growth~\cite{Rakovszky_2019,Zhou2020Jul,Han2023Dec}.
Many of these works, however, rely on the now standard practice of first introducing large charge-neutral ancillary degrees of freedom to simplify the statistical model under consideration, and then removing the effects of those degrees of freedom in the final answer for observables.
Using the geometric understanding of ferromagnetic models developed in the previous sections, we are able to avoid such a construction and directly study the replica dynamics of the system without any ancillary degrees of freedom.
Our results mostly rely on the geometry of the underlying commutant manifold, which, as discussed in Sec.~\ref{sec:fgsmanrep} is completely characterized by the $U(1)$ symmetry.
Hence they are expected to be independent of the precise details of the random Brownian models being used.
Moreover, this underlying geometry closely resembles the ground state geometry of the $es\eta$ toy model studied in Sec.~\ref{sec:esn}, hence many of the results obtain there can be straightforwardly generalized to understand the dynamics of $U(1)$ symmetric systems.
We compare the predictions for our approach to previous literature in Sec.~\ref{subsubsec:breakdownofmembrane} and~\ref{subsubsec:semicl-pur-u1}.
Below we mostly focus on models with a qubit local Hilbert space $\mc H_\mrm{loc}=\mb C^2$ with a basis $\{\ket\downarrow,\ket\uparrow\}$, where the generators of Brownian evolutions [cf.~Eq.~(\ref{eq:brownhamdef})] commute with the total charge
\begin{equation}
	Z_\mrm{tot}\defeq\sum_i Z_i.\label{eq:ztot-def-u1}
\end{equation}
We discuss systems with higher spin in Sec.~\ref{subsec:higher-spin}.
\subsection{Nearest-neighbor qubit model}
As a simple illustratory example, we choose the following generators:
\begin{equation}
	\{Z_iI_j,\,I_iZ_j,\,\lambda Z_iZ_j,\,X_iX_j+Y_iY_j,\,X_iY_j-Y_iX_j\},\label{eq:u1-generators}
\end{equation}
where $\{I,X,Y,Z\}$ are the Pauli matrices and $\lambda>0$; however, as we will discuss, our results are valid much more generally.
\subsubsection{Replica Hilbert space reduction}
Since all generators in Eq.~\eqref{eq:u1-generators} either commute or anti-commute with $Z_i$ for every site $i$, the effective Hamiltonian $P^{(k)}$ preserves the following local parity operator on every site:
\begin{equation}
	\mc Z_i \defeq \prod_{a=1}^{2k} Z_i^{[a]},
\label{eq:k2localsym}
\end{equation}
as explained in Eq.~\eqref{eq:acc-symms} [see Eq.~(\ref{eq:opconvention}) for the replica labeling convention].
Since the relevant boundary states for all observables of interest introduced in Sec.~\ref{sec:observables-general}---namely, the domain wall Eq.~\eqref{eq:definition-dw-id-swap} and the replicated operators Eq.~\eqref{eq:squared-correlator} for $A,B$ Pauli matrices---satisfy $\mc Z_i\equiv +1$ on every site $i$, we will restrict ourselves to analyzing this sector.
Hence for $k=2$ the local Hilbert space dimension of the relevant replicated degrees of freedom is $8$, and an orthonormal basis for this space is:
\begin{equation}\label{eq:letterbasis}
	\begin{split}
		\ket{s_0}&=\ket{\downarrow\downarrow\downarrow\downarrow},\\
		\ket{e_0}&=\ket{\downarrow\downarrow\uparrow\uparrow},\\
		\ket{\eta_0}&=\ket{\downarrow\uparrow\uparrow\downarrow},\\
		\ket{\xi_0}&=\ket{\downarrow\uparrow\downarrow\uparrow},
	\end{split}
	\qquad
	\begin{split}
		\ket{s_1}&=\ket{\uparrow\uparrow\uparrow\uparrow},\\
		\ket{e_1}&=\ket{\uparrow\uparrow\downarrow\downarrow},\\
		\ket{\eta_1}&=\ket{\uparrow\downarrow\downarrow\uparrow},\\
		\ket{\xi_1}&=\ket{\uparrow\downarrow\uparrow\downarrow},
	\end{split}
\end{equation}
where the four spins in each state represent the states in 4 copies of the local Hilbert space required for $k = 2$ replicas, as in Eq.~\eqref{eq:1s-basis}.
The letters denote the appearance of the states in the computational basis expansion of $\sket{\mathds{1}^{\ot 2}; e}$ and $\sket{\mathds{1}^{\ot 2}; \eta}$ as defined in Eq.~(\ref{eq:permstate}).
In particular, $\ket{\sigma_0}$ and $\ket{\sigma_1}$ for $\sigma \in \{e, \eta\}$ appear only in $\sket{\mathds{1}^{\ot 2}; \sigma}$, the fully polarized $\ket{s_0}$ and $\ket{s_1}$ appear in both, while $\ket{\xi_0}$ and $\ket{\xi_1}$ appear in neither:
\begin{equation}\label{eq:states-entanglement}
	\begin{aligned}
		\sket{\1^{\ot 2};e} &= \ket{s_0}+\ket{s_1}+\ket{e_0}+\ket{e_1}, \\
		\sket{\1^{\ot 2};\eta} &= \ket{s_0}+\ket{s_1}+\ket{\eta_0}+\ket{\eta_1}.
	\end{aligned}
\end{equation}
Furthermore it is easy to verify that the states $\ket{\xi_p}$ are non-dynamical under the effective Hamiltonian $P^{(2)}$ obtained from the generators of Eq.~(\ref{eq:u1-generators}) and possess a positive energy due to the one-site generators $\{Z_iI_j,I_iZ_j\}$ in Eq.~\eqref{eq:u1-generators}, hence they do not contribute to its ground states.
To study the long-time dynamics of the model, which is determined by the low-energy physics of $P^{(2)}$, we can therefore further restrict the local Hilbert space to the remaining 6 states.
\subsubsection{Effective Hamiltonian and its ground states}
Within the 6-dimensional local Hilbert space spanned by the states $\{\ket{s_p},\ket{e_p},\ket{\eta_p}\}_{p=0,1}$ in Eq.~(\ref{eq:letterbasis}), the effective Hamiltonian $P^{(2)}$ of Eq.~\eqref{eq:superham-local} can be computed to be equal to the $es\eta$ Hamiltonian of Eq.~\eqref{eq:esn-ham} when restricted to any triplet of $e$-, $s$-, and $\eta$-states:
\begin{equation}\label{eq:esn-u1-term}
	P^{(2)}_{ij}\Big|_{\mrm{span}\{\ket{e_p},\ket{s_q},\ket{\eta_r}\}} = P^{(es\eta)}_{ij},\quad(\Delta=1+\lambda^2)
\end{equation}
which in part motivates the analysis of the previous section.
Furthermore, the effective Hamiltonian possesses an additional interaction term whose action reads
\begin{equation}\label{eq:u1-ham-idstates}\begin{split}
	P^{(2)}_{ij}\ket{\alpha_p\alpha_q}=(1-\delta_{pq}) \bigg( 2\ket{\alpha_p\alpha_q}+\qquad\qquad\qquad\\
	+\,\!\!\!\!\!\sum_{\beta\in\{e,s,\eta\}\setminus\{\alpha\}}(-1)^{\delta_{\alpha s}+\delta_{\beta s}}\Big(\ket{\beta_0\beta_1}+\ket{\beta_1\beta_0}\Big) \bigg),
\end{split}
\end{equation}
for $\alpha\in\{e,s,\eta\}$ and $p,q \in \{0,  1\}$.
For convenience, we have re-scaled the definition of $P^{(2)}$ by an overall prefactor, which can be absorbed in $\kappa$ [see Eq.~\eqref{eq:superham}].
For $\lambda > 0$, this model has a ferromagnetic ground state manifold composed of fully polarized states as defined in Eq.~(\ref{eq:def-ferro}).
Indeed the ground state manifold is exactly\footnote{This is rigorously proven in Ref.~\cite{the-algebra-paper}.} the $k = 2$ commutant manifold $M_{U(1)}^{(2)}$ for systems with the $U(1)$ symmetry of Eq.~\eqref{eq:ztot-def-u1}, i.e., the one identified in Sec.~\ref{sec:u1-gs-man}.
It is composed of two $(\mb S^2\times\mb S^2)_{e/\eta}$ manifolds, which intersect at the points associated to $\ket{s_0}$ and $\ket{s_1}$, which can also be viewed as arising from the frozen states as discussed in Sec.~\ref{subsubsec:frozenstates}.
Similar to the $es\eta$ model, $e$- and $\eta$-states cannot coexist within a fully-polarized ground state.
\subsubsection{TDVP equations and replica decoupling}\label{subsubsec:TDVPreplicadecoup}
Before discussing the dynamics of the $P^{(2)}$ Hamiltonian, let us first recall the  $k=1$ case.
Here the commutant manifold is a sphere and can be parametrized as in Eq.~\eqref{eq:U1manifold}, or more conveniently for us as:
\begin{equation}
	M^{(1)}_{U(1)} = \Big\{\ket{\vec n}=\cos\frac{\theta}{2} \ket{\Pi_\up} + e^{i\varphi}\sin\frac{\theta}{2} \ket{\Pi_\dn}\Big\},
\end{equation}
where $\Pi_\up=\ketbra{\up}{\up}$ and $\Pi_\dn=\ketbra{\dn}{\dn}$, and where the spherical angles are bundled into a unit vector $\vec n\in\mb S^2$ as in Eq.~\eqref{eq:blochsphere}. 
The action of the effective Hamiltonian $P^{(1)}$ on the space spanned by these states is the same as that of the $SU(2)$ Heisenberg ferromagnet Eq.~\eqref{eq:HeisModelHam}~\cite{moudgalya2024}, and we have already studied its semiclassical dynamics within the TDVP ansatz in Sec.~\ref{sec:tdvp}.
Considering now the $k=2$ case, we can parametrize the points on $(\mb S^2\times\mb S^2)_\sigma$ (where the permutation $\sigma\in\{e,\eta\}$ identifies the branch of the commutant manifold) with a pair of unit vectors $(\vec n_1,\vec n_2)$:
\begin{equation}
    \label{eq:explicit-coh-st-u1}
	\begin{gathered}
		\ket{\vec n_1,\vec n_2;\sigma} = \cos\frac{\theta_1}{2}\cos\frac{\theta_2}{2}\ket{s_1}+\\
		\!+\,e^{i\varphi_1}\sin\frac{\theta_1}{2}\cos\frac{\theta_2}{2}\ket{\sigma_0}+
		e^{i\varphi_2}\cos\frac{\theta_1}{2}\sin\frac{\theta_2}{2}\ket{\sigma_1}+\\
		+\,e^{i(\varphi_1+\varphi_2)}\sin\frac{\theta_1}{2}\sin\frac{\theta_2}{2}\ket{s_0}.
	\end{gathered}
\end{equation}
As explained in Appendix~\ref{app:replica-decoupling}, when focusing on the dynamics of states of the same permutation $\sigma$, the different replicas dynamically \textit{decouple} at low-energies under the $P^{(2)}$ dynamics.
This means that whenever the two fields $\vec n_1(x,t)$ and $\vec n_2(x,t)$ that parametrize one branch of the manifold are away from the $s_0$ and $s_1$ points characterized below in Eq.~(\ref{eq:s0s1point}), the energy functional for $P^{(2)}$ factorizes in the continuum limit as
\begin{equation}
    H[\vec{n}_1, \vec{n}_2] = \frac{\kappa}{4} \int\ \mathrm{d}^d x\ [|\bm\nabla \vec{n}_1|^2 + |\bm\nabla \vec{n}_2|^2],
\label{eq:U1decoupledfunctional}
\end{equation}
Hence $\vec{n}_1$ and $\vec{n}_2$ evolve independently according to the TDVP equations for $k=1$ replica, which are precisely those for the Heisenberg model given in Eq.~\eqref{eq:diff-eq-heis}.
At these singular points, however, the inter-replica interactions can become important, and can be captured by defects that impose boundary conditions on the fields, as we discuss below.
Since we characterize the low-energy physics using smoothly varying states on the manifold, the value of the permutation $\sigma(x,t)$ can only switch when the state Eq.~(\ref{eq:explicit-coh-st-u1}) approaches the points (see Fig.~\ref{fig:u1}a)
\begin{equation}
s_0\cong\big(\theta_1=0,\theta_2=0\big),\;\;s_1\cong\big(\theta_1=\pi,\theta_2=\pi\big).
\label{eq:s0s1point}
\end{equation}
As in the $es\eta$ case, in the TDVP equations these points act as Dirichlet boundary conditions for the fields $\vec n_1$ and $\vec n_2$ [cf.~Eq.~(\ref{eq:crossingdirichlet})].
As in the $es\eta$ model, the quantum dynamics would allow these crossing points $\{x_{s_p,i}(t)\}$ to spatially fluctuate, providing subleading corrections that are not captured within the TDVP framework.
Furthermore, a boundary condition analogous to Eq.~\eqref{eq:equality-pressures} would allow the position $x_{s_p,i}(t)$ of these points to drift for unbalanced initial conditions (see Sec.~\ref{subsubsec:finite-size-dw} for comparison, and Appendix~\ref{app:bound-cond})
\begin{equation}
	\begin{gathered}
		|\partial_x \vec n_1(x_{s_p,i}^-)|^2+|\partial_x \vec n_2(x_{s_p,i}^-)|^2 =\qquad\qquad\\
		\qquad\qquad= |\partial_x \vec n_1(x_{s_p,i}^+)|^2+|\partial_x \vec n_2(x_{s_p,i}^+)|^2.
	\end{gathered}
\end{equation}
For the initial states studied here, which are symmetric under $\vec n_1\leftrightarrow\vec n_2$, we have that $|\partial_x \vec n_1|=|\partial_x \vec n_2|$ and this condition factorizes for the two fields, which independently satisfy Eq.~\eqref{eq:equality-pressures}.
In all, the dynamics of the fields on each branch is---to leading order---mostly very simple.
On each branch, we obtain two copies of Heisenberg dynamics, one for each replica.
The interaction between the different branches, and as a consequence, the coupling between the different replicas is---again to leading order---only encoded by a localized set of ``crossing points'', which can simply be treated as aborbing impurities in this context.
\subsubsection{Generality of the model}
The fact that the ground state manifold of $P^{(2)}$ is exactly $M_{U(1)}^{(2)}$ means that Brownian evolutions generated by operators in Eq.~(\ref{eq:u1-generators}) form a symmetric $2$-design, as expected for generic $U(1)$-symmetric circuits.
Hence, the addition of any other set of symmetric generators to the set of Eq.~(\ref{eq:u1-generators}) will not change these ferromagnetic ground states.
Since the low-energy dynamics of the model mostly follows from the structure of this ground state manifold rather than the precise form of the effective Hamiltonian, we expect our results to hold for generic Brownian evolutions that form symmetric $2$-designs, as long as they are chosen isotropically across the system (i.e., with the same gates and strengths for any set of local sites).
The high degree of symmetry of the effective model for the particular set of generators of Eq.~(\ref{eq:u1-generators}) nevertheless provides significant advantages for numerical simulations.
Note that removing certain gates from the set Eq.~(\ref{eq:u1-generators}) can alter its $2$-design property and hence increase the number of ground states of the effective Hamiltonian.
For example, in the regime where $\lambda=0$, the ground state structure changes and is no longer evidently ferromagnetic, and has a structure similar to the $\Delta\rt 1$ limit of the $es\eta$ case that is discussed in Appendix~\ref{sec:delta1lim}.
In 1d, the gates can then be mapped onto quadratic fermionic gates via a Jordan-Wigner transformation, and the effective model can again be brought to an evidently ferromagnetic form again by appropriate transformations as discussed in Appendix~\ref{sec:delta1lim}.
The ground state manifold there is then the homogeneous manifold $U(4)/(U(2)\times U(2))$, which is known as the Grassmannian manifold $\text{Gr}(2, 4)$~ \cite{lastres2026geometryfreefermioncommutants}.
The simple smooth structure of that manifold makes dynamics on them very similar to that of the Heisenberg model discussed in Sec.~\ref{sec:tdvp}.
We provide a brief discussion of this limit in Appendix~\ref{sec:delta1lim}, and focus on the interacting case in the main text.
\subsection{Dynamics of Entanglement} % MARK: * dw ent
We consider the initial state defined in Eq.~\eqref{eq:definition-dw-id-swap}, where in an infinite 1d chain $A=\{x>0\}$, so that the configuration is a single domain wall.
In terms of the basis Eq.~\eqref{eq:letterbasis}, the states on either side of the domain wall are the ones of Eq.~\eqref{eq:states-entanglement}.
These two points are on the commutant manifold $M^{(2)}_{U(1)}$ and are highlighted in Fig.~\ref{fig:gsman-all}b.
We are then interested in the imaginary time evolution of this initial state under the Hamiltonian $P^{(2)}$ in order to study the evolution of purity as in Eq.~\eqref{eq:renyi-purity}.
\subsubsection{Semiclassical dynamics and void formation}\label{subsubsec:semiclassicaldynamics}
Before moving onto the quantitative TDVP analysis of this dynamics, we first provide a qualitative summary.
From the analysis of the domain wall evolutions in the Heisenberg model in Sec.~\ref{sec:ex-dw} and in the $es\eta$ model in Sec.~\ref{subsec:esn-dw}, we expect the domain wall here to eventually relax to a configuration that is composed of smoothly varying states along the shortest path between the two points $\sket{\1^{\ot 2};e}$ and $\sket{\1^{\ot 2};\eta}$ along the commutant manifold.
Here there are two equally shortest paths between them, one passing through $\ket{s_0}$ and the other passing through $\ket{s_1}$, which are the two points of contact between the $e$ and $\eta$ branches of the manifold (see Fig.~\ref{fig:gsman-all}b).
These two paths are related by the symmetry of the Hamiltonian $\ket{\alpha_0}\leftrightarrow\ket{\alpha_1}$ for $\alpha\in\{e,s,\eta\}$ that is evident in Eq.~(\ref{eq:u1-ham-idstates}).
They correspond to two equivalent TDVP trajectories,
\begin{equation}\label{eq:two-trajectories}
    \sket{A\!:\!\bar A_0(t)}\;\;\text{and}\;\;\sket{A\!:\!\bar A_1(t)}
\end{equation}
which at $x=0$ respectively satisfy the conditions
\begin{equation}
    \!(\vec n_1(0,t),\vec n_2(0,t))_{0}\cong s_0,\,\,\,
    (\vec n_1(0,t),\vec n_2(0,t))_{1}\cong s_1.\!
\end{equation}
Along the lines of the domain-wall evolution physics of the $es\eta$ model, they develop void regions dominated by states polarized in the $\ket{s_0}$ and $\ket{s_1}$ directions respectively (shown in Fig.~\ref{fig:u1}a, see Sec.~\ref{sec:esn} for the definition of a \textit{void}).
The evolved domain wall $\sket{A\!:\!\bar A(t)}$ here is hence expected to be described at the semiclassical level by an equal superposition of smoothly varying states along both trajectories,
\begin{equation}
    \sket{A\!:\!\bar A(t)}\approx \sket{A\!:\!\bar A_0(t)} + \sket{A\!:\!\bar A_1(t)}
\label{eq:DWcat}
\end{equation}
resulting in a `cat'-like state in the middle of the void.
With this picture in mind, we can discuss the TDVP evolution more quantitatively.
Using the parametrization of Eq.~\eqref{eq:explicit-coh-st-u1}, we can write the initial conditions associated to the domain wall $\ket{A\!:\!\bar A}$ as
\begin{equation}%\label{eq:initial-cond-dw-esn}
    \vec n_m(x,0)=(1,0,0),\quad \sigma(x,0)=\begin{cases}
       e & x<0,\\
       \eta & x>0,
    \end{cases}
\label{eq:sigmavariation}
\end{equation}
for $m\in\{1,2\}$.
Along the evolution, we will enforce the Dirichlet boundary conditions $\vec n_m(0,t)=(0,0,-1)$ for the trajectory passing through $\ket{s_0}$ and $\vec n_m(0,t)=(0,0,+ 1)$ for the one passing through $\ket{s_1}$.
Knowing that the fields $\vec n_1(x,t)$ and $\vec n_2(x,t)$ evolve independently under Eq.~\eqref{eq:diff-eq-heis}, we can use the symmetry of the initial condition to see that only the angles $\theta_1,\theta_2$ will evolve in time according to Eq.~(\ref{eq:eom-theta}), with $\varphi_1,\varphi_2=0$.
The equations of motion for the evolution of this domain wall are then entirely equivalent to those for the evolution of the $\ket{\textsc{dw}}$ state in the $es\eta$ model [cf.~Eq.~\eqref{eq:diffeq-esn-dw}], with different Dirichlet boundary conditions at $x=0$ for the two trajectories of Eq.~\eqref{eq:two-trajectories}:
\begin{equation}
    \label{eq:dw-diffeq-both-u1}
	\begin{cases}
		\partial_t\theta_m(x,t)=\kappa\,\partial^2_x\theta_m,\\
		\theta_m(x,0)=\frac{\pi}{2},
	\end{cases}
    \theta_m(0,t)=\begin{cases}
    0&\text{for }\sket{A\!:\!\bar A_0},\\
    \pi&\text{for }\sket{A\!:\!\bar A_1}.
    \end{cases}
\end{equation}
where $m\in\{1,2\}$, and the permutation $\sigma(x,t)$ remains equal to its original value Eq.~(\ref{eq:sigmavariation}) throughout the evolution. 
The solution to these equations is directly related to the solution $\theta_{es\eta}(x,t)$ of Eq.~\eqref{eq:esn-sol-dw}
\begin{equation}\label{eq:solution-u1-dw}
    \theta_m(x,t)=\begin{cases}
    \theta_{es\eta}(x,t)&\text{for }\sket{A\!:\!\bar A_0},\\
    \pi- \theta_{es\eta}(x,t)&\text{for }\sket{A\!:\!\bar A_1}.
    \end{cases}
\end{equation}
and it provides predictions for the structure of the evolved state $\sket{A\!:\!\bar A(t)}$.
As we will discuss in Sec.~\ref{subsubsec:semicl-pur-u1}, this can then be used to compute predictions for the averaged purity under $U(1)$-symmetric Brownian evolution by evaluating the inner product with the replicated initial state $\sket{\rho^{\ot 2}}$ as in Eq.~\eqref{eq:renyi-purity}.
This will result in the diffusive growth of entanglement as $\sim\sqrt{t}$ for generic initial product states.
We also expect the physics in higher dimensions to be similar to those in higher dimensional $es\eta$ models discussed in Sec.~\ref{subsubsec:higherd}, with single TDVP trajectories replaced by a superposition of equivalent TDVP trajectories.
\subsubsection{Numerical simulation of the domain wall evolution}
This prediction of an equal superposition of two semiclassical trajectories can be checked numerically by computing the half-chain entanglement of the $\sket{A\!:\!\bar A(t)}$ domain wall.
In particular, in Fig.~\ref{fig:u1}c we compare it with the half-chain entanglement of a modified initial state, biased to pass through only one of the two trajectories:
\begin{equation}
	\sket{A\!:\!\bar A_p(0)} = \cdots \ot \sket{\1^{\ot 2};e}\ot \ket{s_p} \ot \sket{\1^{\ot 2};\eta} \ot \cdots
\end{equation}
where $p\in\{0,1\}$.
In the continuum TDVP prediction, $\sket{A\!:\!\bar A_p(t)}$ is expected to not have any entanglement, whereas $\sket{A\!:\!\bar A(t)}$ is expected to have a $\log(2)$ entanglement by virtue of being an equal superposition of two unentangled states Eq.~\eqref{eq:DWcat}.
While this is no longer exactly true in numerical simulations on a lattice, the difference of the two half-chain entropies nevertheless converges close to $\log(2)$, which is evidence in favor of the double-trajectory picture.
The evolution of this modified initial state is also the one shown in the plots of Fig.~\ref{fig:u1}a, which clearly show the formation of a void dominated by $\ket{s_p}$.
In Fig.~\ref{fig:u1}b we more explicitly show the evolution of the domain wall. Compared to the $es\eta$ model, simulations in this case are more challenging to perform due to the larger local Hilbert space dimension and complexity of the evolution.
Despite not reaching the same long timescales, we can see that the exponent for the growth of the void region is around $\frac{1}{2}$, as expected from diffusion, while the size of the sink region associated with quantum fluctuations [defined in Eq.~\eqref{eq:sink-region-dw}\footnote{While this definition was for the $es\eta$ model, it can also be applied in this case since in the evolution of $\sket{A\!:\!\bar A}$, since the density profile of $\ket{e_0}$ and $\ket{\eta_0}$ states is identical to that of $\ket{e_1}$ and $\ket{\eta_1}$ states respectively (see Fig.~\ref{fig:u1}a).}] grows subdiffusively, and appears to saturate to the same $\frac{1}{3}$ exponent as in the $es\eta$ model, although the data here is not as clean.
However, due to the similarities shared between this model and the $es\eta$, especially the point-like nature of the contact between the two permutation branches of the ground state manifold, we should expect fluctuations of the crossing points to also behave similarly.
\begin{figure}
    \centering
    \includegraphics[width=0.9\linewidth]{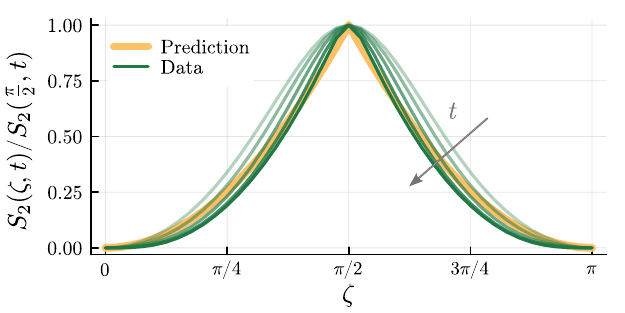}
    \caption{Comparison with numerical data of the TDVP prediction for the asymptotic value of the of annealed average of the second Rényi entropy between $\ket{A\!:\!\bar A(t)}$ [obtained using Eq.~(\ref{eq:purity-prediction})] and a fully-polarized initial state with $\langle Z_x\rangle = \cos\zeta$.
    Data is the same as in Fig.~\ref{fig:u1}; times shown between $t=0.4$ and $t=18.6$ (ideally, longer and more precise simulations would be needed to observe the asymptotic behavior). }
    \label{fig:u1-pred-data}
\end{figure}
\subsubsection{Breakdown of the entanglement membrane picture}\label{subsubsec:breakdownofmembrane}
Before we move on to precise predictions for the diffusive growth of Rényi entropies in $U(1)$ symmetric systems, we briefly comment on some fundamental underlying differences of the behavior of entanglement domain walls compared to systems without symmetries.
In particular, domain walls here melt by developing a large \textit{void} region, rather than propagating sharply, as in effective models for systems without symmetry \cite{PhysRevX.7.031016,Zhou_2019}; this can be attributed to the fact that domain walls there are gapped quasiparticle excitations~\cite{vardhan_entanglement_2024}, while here they decay to lower energy configurations.
This qualitative difference leads to a breakdown of the entanglement membrane picture for Rényi entanglement entropy.
Previous proposals of a modification of the membrane picture for charge-conserving systems~\cite{Rakovszky_2019,Turkeshi_mpemba}, where a sharp domain wall performs a random walk while emitting hydrodynamic modes that propagate diffusively in the bulk, also cannot be valid in the regime we study, since they effectively treat domain walls as stable quasiparticles, which is not the case here.
The discrepancy arises because this earlier picture is derived from circuit models where each charge-carrying qubit is paired up with neutral ancillary qudits, so that the local degrees of freedom are $\mc H_\mrm{loc}=\mb C^2\ot\mb C^q$.
This setup is commonly found in the literature and is perturbatively solvable in the limit $q\rt\infty$~\cite{Barratt2021Nov,Agrawal2021Jul,McCulloch_2023,Khemani2018Sep,Rakovszky2018Sep, Rakovszky_2019}.
When $q > 1$, \textit{no frozen states exist}, leading to a disconnected commutant manifold
\begin{equation}
	M_{U(1)}^{(2)} \cong (\mb S^2\times\mb S^2)_e \sqcup (\mb S^2\times\mb S^2)_\eta.
\end{equation}
Since there is no gapless excitation connecting the two branches of the commutant manifold, this structure is therefore consistent with the existence of sharp domain wall excitations, and therefore with a membrane picture.
One can geometrically think of the emission of hydrodynamic modes by the domain wall as the coupling of the domain wall excitation with the dynamics of the smooth field within each of the two disconnected components.
Ultimately, while entanglement grows ballistically $\sim t$ in these large-$q$ models due to the finite energy cost of maintaining a domain wall, the perturbative results are then be extrapolated to $q\rt 1$, thus obtaining the expected leading diffusive behavior.
Such extrapolations, also yield precise predictions for the evolution of purities in the $q \rt 1$ limit from various initial product states~\cite{Rakovszky_2019}, and remarkably they have a similar form to the predictions we obtain, as we will discuss in the following.
\subsubsection{Semiclassical predictions for entanglement growth}\label{subsubsec:semicl-pur-u1}
In order to compute concrete predictions for entanglement growth [cf.~\eqref{eq:renyi-purity}] from the TDVP solution to the domain wall evolution Eq.~\eqref{eq:solution-u1-dw}, we must first fix the norms appropriately.
In a system of finite size the domain wall state ${\sket{A\!:\!\bar A(t)}}$ decays at long times towards the superposition
\begin{equation}
    \sket{A\!:\!\bar A(\infty)} = \ket{s_0\cdots s_0}+\ket{s_1\cdots s_1},
\end{equation}
where the two components are coming from the saturation of ${\sket{A\!:\!\bar A_0(t)}}$ and ${\sket{A\!:\!\bar A_1(t)}}$ respectively.
Hence we will normalize the variationally-evolved states ${\sket{A\!:\!\bar A_0(t)}}$ and ${\sket{A\!:\!\bar A_1(t)}}$ according to Eq.~\eqref{eq:norm-evo-eq-prime}, thus obtaining the following parametrization for the state ${\sket{A\!:\!\bar A(t)}}$ predicted by the TDVP solution Eq.~\eqref{eq:solution-u1-dw}:
\begin{multline}\label{eq:u1-dw-parametrization}
    \sket{A\!:\!\bar A(t)}\approx \bigotimes_x \ket{\theta_{es\eta}(x,t),\sigma(x,t);0} + \\[-0.5em]
    +\bigotimes_x \ket{\theta_{es\eta}(x,t),\sigma(x,t);1},
\end{multline}
where $\theta_{es\eta}(x,t)$ is as in Eq.~\eqref{eq:esn-sol-dw}, and [cf.~Eq.~\eqref{eq:explicit-coh-st-u1}]
\begin{equation}
    \ket{\theta,\sigma;p} \defeq \ket{s_p}+\tan^2\frac{\theta}{2}\ket{s_{1-p}}+\tan\frac{\theta}{2}(\ket{\sigma_0}+\ket{\sigma_1}),
\end{equation}
where $p\in\{0,1\}$.
This normalization also ensures that the purity of frozen initial states
\begin{equation}
    \rho = (\ketbra{\up}{\up})^{\ot L}\;\;\text{and}\;\; \rho = (\ketbra{\dn}{\dn})^{\ot L},
\end{equation}
which do not evolve under the $U(1)$ symmetric gates, stays equal to $1$ at all times.

With this parametrization, in Appendix~\ref{app:conti} we perform computations of purity for simple unentangled initial states.
This results in Eq.~\eqref{eq:u1-purity-sol} for the decay of purity, which for states where the local mean charge $\langle Z_x\rangle$ is constant takes the form
\begin{equation}\label{eq:purity-prediction}
    \sbra{A\!:\!\bar A(t)}\rho\rangle^{\ot 2} \approx \prod_x\left(\frac{1+ r(x,t) |\langle Z_x\rangle|}
    {1+r(x,t)}\right)^2
\end{equation}
where
\begin{equation}
    r(x,t)=\tan\left[\frac{\pi}{4}\left(1-\left|\mrm{erf}\left(\frac{x}{\sqrt{4\kappa t}}\right)\right|\right)\right].
\end{equation}
This formula receives leading contributions from the region where $|x|\lesssim \sqrt{\kappa t}$---weighted by $r(x,t)$---resulting in a $\sim e^{-\sqrt{\kappa t}}$ decay of purity, and therefore a $\sim\sqrt{\kappa t}$ growth in the second Rényi entropy.
In Fig.~\ref{fig:u1-pred-data} we compare the prediction of Eq.~\eqref{eq:purity-prediction} with numerical data, where as in Eq.~\eqref{eq:Omegabehavior} we define the log-overlaps
\begin{equation}
    S_2(\zeta,t)=-\log(\sbra{A\!:\!\bar A(t)}\rho\rangle^{\ot 2})
\end{equation}
for a fully-polarized initial state with $\langle Z_x\rangle=\cos\zeta$, which here is simply the annealed average of the second Rényi entropy.
Note that due to the high computational cost of reaching thermalization timescales in this effective model, and due to the high sensitivity of the overlaps of Eq.~\eqref{eq:renyi-purity} to truncation error in MPS-based methods, we only show results for short times.

\begin{figure}
    \centering
    \includegraphics[width=0.62\linewidth]{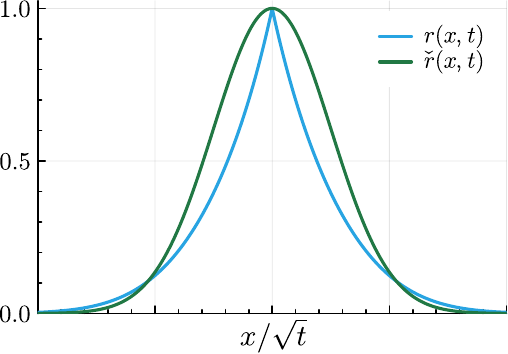}
    \caption{Comparison of the weight functions appearing in the prediction for half-chain growth of purity Eq.~\eqref{eq:purity-prediction}: in this work ($r(x,t)$) and in Ref.~\cite{Rakovszky_2019} ($\check r(x,t)$).}
    \label{fig:comparison}
\end{figure}

This formula (and also the full result in the appendix) has a similar form to what has been proposed in Ref.~\cite{Rakovszky_2019} for this same quantity.
The main difference is in the weight function $r(x,t)$, which was instead found to be a Gaussian of the form
\begin{equation}
    \check r(x,t)=\exp(-x^2/2\kappa t).
\end{equation}
The two curves are plotted in Fig.~\ref{fig:comparison}.
This discrepancy can be understood from the difference in the two approaches taken in their derivation, explained in the previous section Sec.~\ref{subsubsec:breakdownofmembrane}.
In our case we find that the domain wall melts diffusively with a profile given approximately by Eqs.~\eqref{eq:solution-u1-dw} and~\eqref{eq:esn-sol-dw}.
By ignoring the subleading fluctuations of the sink region (which are expected to go as $t^{\frac{1}{3}}$ due to physics similar to the discussion in Sec.~\ref{subsubsec:DWannihilation}), the profile displays a sharp cusp in the middle, which is then reflected in the weight function $r(x,t)$.
Since the sink fluctuations are subdiffusive, we expect that the cusp should remain in the scaling variables $x/\sqrt{t}$ at late times.
In contrast, from the modified membrane picture which emerges when charged degrees of freedom are coupled to neutral ancillas (i.e., $\mc H_\mrm{loc}=\mb C^2\ot\mb C^q$), which is the approach used by Ref.~\cite{Rakovszky_2019}, the fluctuations of the domain wall are overestimated as being diffusive, leading to a smooth Gaussian profile in the scaling variables, without any cusps.
In the setting of Fig.~\ref{fig:u1-pred-data} the predictions from $r(x,t)$ or $\check r(x,t)$ would be almost indistinguishable, since translationally invariant states are not very sensitive to the choice of weight function. More refined simulations are therefore needed to discriminate the two.
A more detailed discussion of the predictions is found in Appendix~\ref{app:conti}.
\subsubsection{Rényi entropies for $k\geq 3$}
Given these results for the purity, one might wonder if analogous results hold for the $k$-th purities or the $k$-th Rényi entropies for $k \geq 3$.
This requires understanding the geometry of higher $k$-commutants. 
As we discussed in Sec.~\ref{subsubsec:largerk}, the complexity of the commutant manifold greatly increases as the number $k$ of replicas grows: it consists of $k!$ identical smooth components, and the structure of their intersections depends on the structure of the symmetric group $S_k$.
From Eq.~\eqref{eq:fsigmasame} we do however know that in the case of the model under consideration, the branches of the commutant manifold will all intersect on the two replicated frozen states, i.e.,
\begin{equation}
	\bigcap_{\sigma\in S_k}\big((\mb S^2)^{\times k}\big)_\sigma = \big\{\ket{\dn}^{\ot 2k},\,\ket{\up}^{\ot 2k}\big\},\label{eq:only-two-contact-points}
\end{equation}
which in the $k=2$ case are the $\ket{s_0}$ and $\ket{s_1}$ states.
As seen in Eq.~(\ref{eq:definition-dw-id-swap}), $e$ and $\eta$ are the only permutations that are relevant for the $k$-th purities, and it is easy to see that those branches of the commutant manifold only intersect at two points, which are exactly those of Eq.~(\ref{eq:only-two-contact-points}).
Hence the shortest paths connecting $\sket{\1^{\ot k};e}$ to $\sket{\1^{\ot k};\eta}$ on the commutant manifold, will be the two paths passing through the contact points Eq.~\eqref{eq:only-two-contact-points}.
While these were the only available paths for $k = 2$, there are many other paths for $k \geq 3$, which are longer and pass through other intermediate branches of the manifold (which are associated to permutations $\sigma\notin\{e,\eta\}$).
To each path we can associate a semiclassical solution to the equations of motion by following similar arguments as in Sec.~\ref{subsubsec:TDVPreplicadecoup} to derive the TDVP equation, and those in Sec.~\ref{subsubsec:semiclassicaldynamics} to solve them.
The energy of each trajectory will decay over time as $\sqrt{\kappa/t}$, with a prefactor that we expect to be be roughly proportional to the squared length of the path on the manifold, since the energy functional Eq.~\eqref{eq:deriv-exp-H} is quadratic in gradients of the field.
We can therefore conjecture that the leading contribution at long times will be given by the two shortest trajectories that go directly from $e$ to $\eta$ through the contact points of Eq.~\eqref{eq:only-two-contact-points}, resulting in a picture very similar to the one presented here for $k=2$, which results in void formation and a scaling of the $k$-th purity as $e^{-\sqrt{\kappa t}}$, ultimately leading to diffusive  $\sim \sqrt{\kappa t}$ growth of $k$-th Rényi entanglement entropies.
This diffusive scaling of all $k$-th Rényi entropies raises a puzzle on the behavior of von Neumann entropies, which is obtained as the $k \rightarrow 1$ limit of these Rényi entropies. 
While one might naively expect it to also grow diffusively, it has been observed to instead grow ballistically as $\sim t$~\cite{Rakovszky_2019}.
However, taking this replica limit is rather subtle, and would require a deeper understanding of all the factors of $k$ that appear in the expression of the $k$-th purity, including prefactors that potentially also depend on fluctuations that are not captured by TDVP.
We leave a detailed study of these effects to future work.

\subsubsection{Extensions to other entropic quantities}\label{subsubsec:U1extensions}
The dynamics of states similar to $\sket{A\!:\!\bar A}$ can be used to study many other entropic quantities.
For example, if in the definition Eq.~\eqref{eq:renyi-purity} one dresses the $\eta$ permutation with a charged operator as follows
\begin{equation}\label{eq:symmetry-resolved}
    \bigotimes_{i\in \bar A,\,j\in A}|\1^{\ot 2};e\rangle_{i}|e^{i\theta_1Z}\ot e^{i\theta_2Z};\eta\rangle_{j},
\end{equation}
then the evolution of this state can be used to compute the (averaged) dynamics of the ``charged moments'' which are used in the study of symmetry-resolved entanglement~\cite{Goldstein2018May,Fraenkel2020Mar,Bonsignori2019Oct,Murciano2020Aug,Murciano2022Aug} and entanglement asymmetry~\cite{Ares2023Apr,Ares2025Aug,Fossati2024May,Lastres2025Jan,Murciano2024Jan}, under $U(1)$-conserving Brownian dynamics.
The semiclassical evolution of this modified domain-wall state is completely analogous to that of $\sket{A\!:\!\bar A}$, and can actually be obtained from it by performing a simple change of initial conditions, which also amounts to a simple coordinate on the $\eta$ branch of the $k = 2$ commutant manifold.

\subsection{Dynamics of Correlators} % MARK: * nonhydro
\label{sec:U1correlators}
We now turn to the study of squared autocorrelation functions [see Eq.~\eqref{eq:squared-correlator}]
\begin{equation}\label{eq:u1-autocorr-ref}
	\begin{gathered}
	\langle O_i(t)O_i(0)\rangle= \frac{1}{2^L}\tr(O_i(t)O_i(0))\\
	\overline{|\tr( O_i(t)O_i(0))|^2} = \sbra{O_i^{\ot 2}} e^{-\kappa P^{(2)}t}\sket{O_i^{\ot 2}}\\
	\sket{O_i^{\ot 2}} = \dots \ot \sket{\1^{\ot 2};e} \ot \sket{O^{\ot 2}}_i \ot \sket{\1^{\ot 2};e} \ot \dots
	\end{gathered}
\end{equation}
where for simplicity we choose $O$ to be Hermitian.
In this way, we can compute the average decay in the magnitude of the autocorrelator even when its sign fluctuates over time.
We will be interested in the case where $O$ is either a hydrodynamic or non-hydrodynamic operator, where the only purely-hydrodynamic traceless single-site operator is $Z$, which is an element of the $k=1$ commutant manifold Eq.~\eqref{eq:U1manifold}.
\subsubsection{Hydrodynamic Correlators}\label{subsubsec:U1hydrocorrelators}
Using the basis Eq.~(\ref{eq:letterbasis}), we can express the squared autocorrelation function of $Z_i$ through the evolution of [cf.~Eq.~\eqref{eq:u1-autocorr-ref}]
\begin{equation}
	\sket{Z^{\ot 2}} = \ket{s_0}+\ket{s_1}-\ket{e_0}-\ket{e_1}.
\end{equation}
Under the assumption of replica decoupling, which holds to leading order at times $t\geq \mc O(1)$, the two $Z$ operators would evolve freely in their respective replica, under the ferromagnetic Heisenberg dynamics of $P^{(1)}$, under which the operator string $\ket{\cdots\1 Z\1\cdots}$ is analogous to $\ket{\cdots\up\dn\up\cdots}$. 
This results in two independent Gaussian profiles for the two evolved operators [similar to Eq.~\eqref{eq:e-wavefunct}], and Eq.~\eqref{eq:u1-autocorr-ref} evaluates to the return probability of both operators to the origin at time $t$.
In an infinite system, for each operator this is $\sim\frac{1}{\sqrt{\kappa t}}$ so that in total we have
\begin{equation}
	\overline{|\tr( Z_i(t)Z_i(0))|^2}\sim \frac{1}{\kappa t}.\quad(d = 1)\label{eq:one-hydro-correlator}
\end{equation}
Overall, replica decoupling implies the formal equivalence of the averaged squared correlator and the square of the averaged correlator, which means that to leading order:
\begin{equation}
	\overline{|\tr( Z_i(t)Z_i(0))|^2}\approx |\overline{\tr( Z_i(t)Z_i(0))}|^2.
\end{equation}
Subleading corrections to this are given by the Hamiltonian term Eq.~\eqref{eq:u1-ham-idstates}, which introduces a two-particle $Z$-$Z$ interaction across the two replicas.
\subsubsection{Non-Hydrodynamic Correlators}\label{subsubsec:U1nonhydro}
For the autocorrelation function of the local non-hydrodynamic operator of $X_i$, in the basis of Eq.~(\ref{eq:letterbasis}) we have [cf.~Eq.~\eqref{eq:u1-autocorr-ref}]
\begin{equation}
	\sket{X^{\ot 2}} = \ket{\eta_0}+\ket{\eta_1}+\ket{\xi_0}+\ket{\xi_1}.
\end{equation}
For this initial state, the non-dynamical states $\ket{\xi_p}$ will be exponentially suppressed over time due to their gap [see discussion below Eq.~(\ref{eq:states-entanglement})], only leaving the $\ket{\eta_p}$ states on a background of `identity' states, which are instead points on the $e$ branch of the commutant manifold.
The dynamics is then that of a single $\eta$-particle in the effective model, and as in the case of the $es\eta$ toy model discussed in Sec.~\ref{subsec:sp}, we expect it to share great similarity with domain-wall dynamics in 1d.
The replicated $X_i$ operator will therefore act as a fluctuating Dirichlet boundary condition for the background field, which in the semiclassical limit will be approximated as being immobile.
Diffusive transport will then drive the formation of a void completely analogous to the one of the domain wall case, with voids here having a similar `cat'-like structure as shown in Eq.~(\ref{eq:DWcat}), due to the existence of two equivalent semiclassical solutions.
Explicitly, we have
\begin{equation}
	\sket{X_i^{\ot 2}(t)}\approx \sket{(X_i^{\ot 2})_0(t)} + \sket{(X_i^{\ot 2})_1(t)}
\end{equation}
where $\sket{(X_i^{\ot 2})_p(t)}$ for $p\in\{0,1\}$ are TDVP trajectories that with the parametrization of Eq.~\eqref{eq:explicit-coh-st-u1} are described by $\varphi_1(x,t)=\varphi_2(x,t)=0$ and
\begin{equation}
    \theta_m(x,t)=\begin{cases}
    \theta_{es\eta}(x,t)&\text{for }p=0,\\
    \pi- \theta_{es\eta}(x,t)&\text{for }p=1.
    \end{cases}
\end{equation}
where $\theta_{es\eta}(x,t)$ is the solution from Eq.~\eqref{eq:esn-sol-sp}.
A similar analysis of the energies of these solutions, shows that it scales as $E \sim \sqrt{\kappa/t}$, as in the $es\eta$ model, which ultimately leads to a decay of the correlator in 1d as
\begin{equation}
	\overline{|\tr( X_i(t)X_i(0))|^2}\sim e^{-\sqrt{\kappa t}},\quad(d = 1)\label{eq:one-nonhydro-correlator}
\end{equation}
similar to Eq.~(\ref{eq:norm-dec-1d}).
We also expect the higher-dimensional physics to be similar to the single-particle dynamics of the $es\eta$ model, leading the non-hydrodynamic correlators to take the forms of Eqs.~(\ref{eq:norm-dec-3d}) and (\ref{eq:norm-dec-2d}) in $d \geq 3$ and $d = 2$ respectively.
Subleading contributions due to quantum fluctuations of the non-hydrodynamic operator can in this setting be captured through a classical stochastic model which approximates the effective imaginary-time quantum evolution.
An exact stochastic mapping of the full $U(1)$ model cannot be performed using the same methods, due to a sign problem in the term Eq.~\eqref{eq:u1-ham-idstates} of the effective Hamiltonian, in particular since some of its off-diagonal elements in the computational basis are positive.
However, an approximate mapping that works for the study of non-hydrodynamic correlators was introduced in Ref.~\cite{mcculloch2026longlivedlocalquantumcoherences}, where a $\frac{1}{3}$ exponent for the subdiffusion of the operator within the void was predicted in 1d, similar to what we observed in Fig.~\ref{fig:esn}c.
We review this model in Appendix~\ref{sec:fluctuations}, and we argue that it is not suitable for the study of general observables such as entanglement, since it ignores relevant terms in the effective Hamiltonian, which leads to a change of the ground state manifold.
However, we note that it might be possible to obtain a sign-problem-free effective Hamiltonian that has the same ground state manifold, which would then lead to stochastic models that also capture the behavior of other observables.
We leave the exploration of such possibilities to future work.
\subsection{Connection to void arguments}
To our knowledge, the behavior of non-hydrodynamic correlators as Eq.~\eqref{eq:one-nonhydro-correlator} was first described in Ref.~\cite{McCulloch_2026}, where this scaling is argued for through the use of a \textit{void argument}, which were in previous works also used in the context of diffusive Rényi entropy scaling~\cite{Huang_2020,Rakovszky_2019,Rakovszky2021May}.
These types of arguments provide rigorous lower bounds on the asymptotic decay of observables by expressing them as a weighted sum over expectation values of the same observables in individual states in the Hilbert space, and using the analytically tractable contribution from carefully chosen states to establish lower bounds.
The states typically used in the computation are those that possess some extended region of a many-body state with a fully-polarized frozen product state such as $\ket{\dots\dn\dn\dots}$, called \textit{voids}.
Since the void regions themselves do not evolve under charge-conserving dynamics, their structure simplifies analytic tractability, with the only processes being diffusive transport of any charge on top of the void, and the diffusive melting of the edges of the void regions.
These can be accounted for rigorously using standard hydrodynamic computations, which can then be used to describe the evolution of these modified states, which in turn provide bounds on the value of the original observables of interest.
These voids are therefore in some sense an artificial construction used in a mathematical argument, and are distinct from the dynamically generated void regions in our effective models $P^{(k)}$, which natural form from the boundary states used for the study of entanglement and appropriate correlators.
The recent work Ref.~\cite{mcculloch2026longlivedlocalquantumcoherences} also explicitly solved the dynamics of physical voids in a similar effective model within the framework of macroscopic fluctuation theory, in order to describe the evolution of non-hydrodynamic correlators such as Eq.~\eqref{eq:one-nonhydro-correlator}, and going beyond the semiclassical results provided here.

\subsection{$U(1)$-Symmetric Models with Higher Spin} % MARK: * more spin
\label{subsec:higher-spin}
So far we have discussed the dynamics of a specific model, with Brownian generators Eq.~\eqref{eq:u1-generators}.
Based on the formalism of ferromagnetic models, we argued that the continuum effective Hamiltonians describing other noisy systems should not change, as long as the commutant remains the same.
However, the commutant does change if the local representation of the $U(1)$ group is chosen differently, and here we discuss the behavior of these quantities in such settings.
The spin-$\frac{1}{2}$ case studied so far corresponds to a local representation with charges $\mu_{\up/\dn}=\pm 1$ on every site. More generally, one might consider $r$-dimensional qudits labelled by $\{\ket{n},\;\;1\leq n\leq r\}$ with on-site charges $\{\mu_r\}$ such that
\begin{equation}
    \mu_1\leq ...\leq \mu_r\quad (\mu_n\in\mb Z).
\end{equation}
In this case, the existence of frozen states relies on having non-degenerate states of lowest and highest charge:
\begin{equation}\label{eq:condition-for-frozen-u1}
    \begin{aligned}
        \mu_1<\mu_2 \ \ &\implies& \ket{1}\text{ is frozen,}\\
        \mu_{r-1}<\mu_r \ \ &\implies& \ket{r}\text{ is frozen.}
    \end{aligned}
\end{equation}
The case discussed in Sec.~\ref{subsubsec:breakdownofmembrane} of a spin-$\frac{1}{2}$ site paired with a charge-neutral $\mb C^q$ qudit corresponds to
\begin{equation}
    \mu_1=\dots =\mu_q = -1,\quad \mu_{q+1}=\dots=\mu_{2q} = +1.
\end{equation}
where no frozen states are present due to the degeneracy of the highest and lowest charge states, resulting in a disconnected commutant manifold for $k\geq 2$ (see Sec.~\ref{subsubsec:frozenstates} for the connection of commutant manifolds and frozen states).
If only one of the conditions of Eq.~\eqref{eq:condition-for-frozen-u1} is satisfied, then only one frozen state exists, and the $k=2$ commutant manifolds will possess a single contact point between the $e$ and $\eta$ branches, similar to the $es\eta$ model in Fig.~\ref{fig:gsman-all}a.
This for example happens on a spin-$1$ model with global charge $\sum_i{(S^z_i)^2}$, corresponding to local charges $\mu_0=0$ and $\mu_+=\mu_- =1$, and where $\ket{0}$ is thus the only frozen state.
A standard spin-$s$ many body system with the global charge operator $\sum_i{S^z_i}$ where $\mu_n=0, \dots ,2s-1$, possesses exactly two frozen states.
Beyond frozen states, the structure of the commutant manifolds for generic number of replicas $k$ only depends on the $k=1$ commutant manifold for the given representation, which we generically expect to be a sphere,
\begin{equation}
    M_{U(1)}^{(1)}=\mb P\{\exp(zN),\, z\in\mb C\}\stackrel{?}{\cong} \mb S^2.
\end{equation}
where $N$ is the local charge density operator.
In Ref.~\cite{the-algebra-paper} we show for example that this holds for any standard spin-$s$ many-body system.
While we do not exclude that particular choices of generators for higher-spin systems might result in different dynamics of the field, we generally expect these to be diffusive [cf.~Eq.~\eqref{eq:exp-generic}] as long as they form $k$-designs, i.e., when their $k$-commutant matches what is expected from just their symmetry operators.
Since all these standard higher spin models possess two frozen states, we expect dynamics of entanglement and of non-hydrodynamic correlators completely analogous to those studied here in the spin-$\frac{1}{2}$ case.
This prediction is in agreement with the ``frozen region condition'' of Ref.~\cite{Rakovszky2021May} regarding entanglement. 
We also find that all single-site diagonal operators are hydrodynamic in higher-spin models according to the definition of Eq.~\eqref{eq:hydrodefn}, which---as highlighted in Ref.~\cite{Rakovszky2021May}---might explain the discrepancy with the predictions of Ref.~\cite{Znidaric2020Jun} for the rate of growth of Rényi entropies.
Interestingly, it is also possible to construct less generic situations where the $k = 1$ commutant manifold is a sphere with cuspidal singularities.
For example if we take a spin-$\frac{3}{2}$ site and remove the state of charge $2$ by some energetic means such that the local Hilbert space is:
\begin{equation}
    \mc H_\mrm{loc}=\mrm{span}\{\ket{0},\ket{1},\ket{3}\}.
\end{equation}
Then the $k=1$ commutant manifold of the total spin $U(1)$ symmetry is described by the projective variety specified by~\cite{the-algebra-paper}
\begin{equation}
    x_3 x_0^2=x_1^3\;\;\;\text{in}\;\;\;\mb P\{x_0\!\ketbra{0}{0}+x_1\!\ketbra{1}{1}+x_3\!\ketbra{3}{3}\}\cong \mb{CP}^2,
\end{equation}
which describes a sphere with a singularity at the point associated with $\ketbra{3}{3}$.
It is possible to construct a Brownian model with local interactions of range three that corresponds to this manifold, and it could be interesting to explore whether this affects the dynamics of observables in a significant way.

\section{Models with $SU(2)$ and other Non-Abelian Symmetries} \label{sec:su2} % MARK: 5. SU(2)
We finally discuss the example of 4-copy models ($k=2$) for Brownian circuits with a continuous non-Abelian symmetry, mainly focusing on the case of $SU(2)$ for concreteness.
While several works study the saturation of observables in this class of random circuits~\cite{hulse2021qudit,Hulse2024Jul,Li2024Dec,Li2025Oct,Wu2026Apr}, to our knowledge the quantum dynamics of these models remains widely unexplored from a theoretical standpoint, with only the $k=1$ effective Hamiltonian being fully understood~\cite{PhysRevB.108.054307,moudgalya2024}.
\subsection{Nearest-neighbor $SU(2)$ spin-$\frac{1}{2}$ model}
We consider a spin-$\frac{1}{2}$ many-body system $\mc H_\mrm{loc}=\mb C^2$ with basis $\{\ket\downarrow,\ket\uparrow\}$.
On every bond $\langle ij\rangle$ we will consider the Brownian generator
\begin{equation}
	\vec S_i\cdot \vec S_j = \frac{1}{4}(X_iX_j+Y_iY_j+Z_iZ_j)\label{eq:su2-generators}
\end{equation}
where $\{X,Y,Z\}$ are the Pauli matrices; this is the only non-trivial $SU(2)$-symmetric two-site interaction term.
The resulting effective Hamiltonian $P^{(2)}$ is ferromagnetic (proven in Ref.~\cite{the-algebra-paper}) and does not have any additional accidental symmetries of the kind described in Sec.~\ref{subsec:accidentalsym}.
Its ground state manifold (derived in Sec.~\ref{sec:su2-gs-man} and pictured in Fig.~\ref{fig:gsman-all}c) has the form
\begin{equation}\label{eq:gsman-su2-intersection}
    \begin{gathered}
        M_{SU(2)}^{(2)} \cong (\mb{CP}^3\times \mb{CP}^3)_e\cup (\mb{CP}^3\times \mb{CP}^3)_\eta\\
        \text{where}\quad (\mb{CP}^3\times \mb{CP}^3)_e\cap (\mb{CP}^3\times \mb{CP}^3)_\eta\cong (\mb S^2)^{\times 4}.
    \end{gathered}
\end{equation}
The states parametrized by it span the whole local Hilbert space $\mc H_\mrm{loc}^{\ot 4}\cong \mb C^{16}$.
As a result, defining $P^{(2)}$ requires a many-body system with an on-site Hilbert space dimension of 16, thus introducing significant theoretical and computational difficulties.
Hence we restrict our analysis to results from TDVP, which are indeed very natural generalizations of results obtained in Sec.~\ref{sec:esn} and \ref{sec:u1}, and which are in turn much more easily verifiable numerically.
\subsubsection{TDVP equations}
The formalism developed in Sec.~\ref{sec:tdvp} does however allow us to analyze the expected low-energy behavior of $P^{(2)}$ with relative ease.
We start from the continuum description of the $k=1$ model.
By direct computation, the effective Hamiltonian $P^{(1)}$ can be shown to be an $SU(4)$ Heisenberg ferromagnet~\cite{moudgalya2024, PhysRevB.108.054307}
\begin{equation}
    P^{(1)} = \frac{1}{2} \sum_{\langle ij\rangle} (\Pi_{\Lambda^2\mb C^4})_{ij}
\end{equation}
where $\Pi_{\Lambda^2\mb C^4}$ is the projector onto the set of antisymmetric states on two $\mb C^2\ot\mb C^2\cong \mb C^4$ sites.
As explained in Appendix~\ref{app:cont-exp}, the emergent $SU(4)$ symmetry of the Hamiltonian forces its continuum expansion to take the form of Eq.~\eqref{eq:deriv-exp-H}, with $v\in\mb{CP}^3 \cong M_{SU(2)}^{(1)}$ and $g_{\alpha\beta}$ being proportional to the Fubini-Study metric on $\mb{CP}^3$~\cite{Bengtsson2017Aug}.
The TDVP equations of motion then take exactly the form of Eq.~\eqref{eq:harm-map-heat-flow}.
For our purposes it will be convenient to over-parametrize $\mb{CP}^3$ using $4\times 4$ complex matrices:
\begin{equation}
    \!\mb{CP}^3\cong\{V\!\in\mrm{End}(\mb C^4): V=V\+=V^2,\,\tr(V)=1\}.\!
\end{equation}
The matrices $V$ that satisfy these conditions are exactly the pure density matrices on $\mb C^4$
\begin{equation}\label{eq:Vdef}
    V=\ketbra{v}{v},\;\;\text{where}\;\; \ket{v}\in\mb C^4\;\;\text{and}\;\;\braket{v}=1.
\end{equation}
This is in part analogous to the parametrization that we used in Eq.~\eqref{eq:blochsphere} for $\mb{CP}^1\cong \mb S^2$, the ground state manifold of the spin-$\frac{1}{2}$ Heisenberg model, where vectors $\vec n$ of norm one parametrize spin-coherent states $\ket{\vec n}$, or equivalently pure density matrices $\rho =\ketbra{\vec n}{\vec n} =\frac{1}{2}(\1+\vec n\cdot\vec\sigma)$.
Here, we find it more convenient to work with the density matrices $V$ directly, rather than with vectors that parametrize them as a 15-dimensional vector over a fixed basis of traceless $4 \times 4$ Hermitian matrices.
With this parametrization we find for $k=1$ the following continuum Hamiltonian and TDVP equations (cf.~Appendix~\ref{app:cont-exp})
\begin{gather}
    H[V] = \frac{\kappa}{2} \int\dd^d x\,\tr(\bm\nabla V\cdot \bm\nabla V)\\
    \partial_t V = \kappa\,\Pi_V[\bm \nabla^2 V] = \kappa\,[V,[V,\bm \nabla^2 V]]\label{eq:eoms-su2}
\end{gather}
where $\Pi_V[\,\cdot\,]=[V,[V,\,\cdot\,]]$ is the projector onto matrices tangent to $\mb{CP}^3$ at the point $V$, as shown in Eq.~(\ref{eq:tangentproj}).
For $P^{(2)}$, the continuum Hamiltonian and TDVP equations will also take a similar form.
We can parametrize a point on the $k=2$ ground state manifold with a pair of matrices $V_1,V_2\in \mb{CP}^3$ and a permutation $\sigma\in\{e,\eta\}$ which indicates the branch of the manifold as in Eq.~\eqref{eq:gsman-su2-intersection}, and we can provide coordinates to field configurations as
\begin{equation}
    \ket{v(x)} \cong (V_1(x), V_2(x), \sigma(x)).
\label{eq:vxfieldconfigs}
\end{equation}
Similar the $U(1)$ case, $P^{(2)}$ exhibits replica decoupling at low energies (discussed in Appendix \ref{app:replica-decoupling}), which means that its energy functional, when applied to product state configurations composed of smooth variations on the ground state manifold, factorizes as
\begin{equation}
    H[V_1,V_2] = \frac{\kappa}{2} \int\dd^d x\,[\tr(\bm\nabla V_1\cdot \bm\nabla V_1)+\tr(\bm\nabla V_2\cdot \bm\nabla V_2)],\label{eq:ham-su2}
\end{equation}
as in Eq.~(\ref{eq:U1decoupledfunctional}) for the $U(1)$ case.
This leads to the decoupling of equations of motion for $V_1$ and $V_2$, which have the same form as the $k=1$ case of Eq.~\eqref{eq:eoms-su2}.
This decoupling however only holds away from the singularities of the manifold, where inter-replica interactions become important.
\subsubsection{Treatment of singularities and boundary conditions for branch switching}
As described more generally in Sec.~\ref{sec:tdvp}, low-energy field configurations $v(x)$ of the form of Eq.~(\ref{eq:vxfieldconfigs}) are expected to vary smoothly on the manifold.
This means that $\sigma(x)$ can only change at points in space where the field lies on the intersection of the two branches of the manifold Eq.~\eqref{eq:gsman-su2-intersection}, i.e., [cf.~Eq.~\eqref{eq:su2-int-struct}]
\begin{equation}\label{eq:intersection-points}
    \ket{\vec n_1}\ot \ket{\vec n_2}^*\ot \ket{\vec n_3}\ot \ket{\vec n_4}^*\in (\mb S^2)^{\times 4}.
\end{equation}
In the matrix coordinates of Eq.~\eqref{eq:Vdef}, this intersection manifold is identified by the condition
\begin{equation}\label{eq:intersection-points-coord}
    V_1 = \rho_1\ot_K\rho_2^T,\quad V_2 = \rho_3\ot_K\rho_4^T,
\end{equation}
where $\rho_a=\ketbra{\vec n_a}{\vec n_a}$ and $\otimes_K$ is the Kronecker product.
These two equations are enforced as boundary conditions for field configurations $V_1(x_i)$ and $V_2(x_i)$ at all such points $\{x_i\}$ where $\sigma(x_i)$ switches value discretely between $e$ and $\eta$, analogous to Eq.~(\ref{eq:crossingdirichlet}).
We expect all such crossing points $\{x_i\}$ where $\sigma(x)$ is discontinuous to act as defects in ways that are for the most part analogous to what we have seen in Sec.~\ref{sec:esn} and \ref{sec:u1}.
However, a crucial difference in this case is that the singularities [cf.~Eq.~\eqref{eq:gsman-su2-intersection}] are not a discrete set of points, but rather 4-dimensional manifolds themselves.
Hence the defects themselves here possess continuous internal degrees of freedom (labelled by $\{\vec n_a\}$) which would allow for low-energy quantum fluctuations of the defects that have no counterpart in the previously studied cases.
For example, in $d \geq 2$ dimensions, the value of the fields $(V_1,V_2)$ on any $(d-1)$-dimensional interface is not necessarily constrained to a uniform value $(\rho_1\ot_K\rho_2^T,\rho_3\ot_K\rho_4^T)$, but rather the value of the matrices $\{\rho_a\}_{1 \leq a \leq 4}$ can instead vary spatially along the defect.
While we do not have any proposal to treat such fluctuations, we will argue from symmetry that in one dimension, the semiclassical TDVP trajectory can be obtained in a manner directly analogous to the $U(1)$ case, by imposing simple Dirichlet boundary conditions for initial states of interest here.
We leave a broader understanding of the evolution of these defects for future investigations.
\subsection{Dynamics of Entanglement} % MARK: * dw ent
The setup for the computation of purities remains essentially identical as in the $U(1)$ case, with a domain wall configuration ${\sket{A\!:\!\bar A(t)}}$ in 1d between the states $\sket{\1^{\ot 2};e}$ and $\sket{\1^{\ot 2};\eta}$ of Eq.~\eqref{eq:states-entanglement}.
Recall that in the $U(1)$ case, the presence of two contact points ($\ket{s_0}$ and $\ket{s_1}$) between the two branches $e$ and $\eta$ of the commutant manifold led to a semiclassical evolution of the domain wall dominated by a superposition of two trajectories, each interpolating the domain wall passing through either $\ket{s_0}$ or $\ket{s_1}$, shown in Eq.~(\ref{eq:u1-dw-parametrization}).
The continuous intersection manifold that occurs in the $SU(2)$ case, therefore suggests that the domain wall evolves into an equal superposition of potentially \textit{infinitely many} semiclassical trajectories, parametrized by points on $(\mb S^2)^{\times 4}$ through which they pass, provided they are all energetically equivalent.
With this dynamical information, we can now formulate a semiclassical ansatz for the evolution of the domain wall.
As in Sec.~\ref{subsubsec:semiclassicaldynamics}, we again argue that the evolved field $v(x,t)$ for the domain wall initial state in 1d should relax to a configuration that follows a curve of shortest length between the points $\sket{\1^{\ot 2};e}$ and $\sket{\1^{\ot 2};\eta}$ on the commutant manifold, which in an infinite system act as boundary conditions at $x=\pm\infty$ (depicted in Fig.~\ref{fig:gsman-all}c).
This is expected because the energy functional Eq.~\eqref{eq:ham-su2}, by virtue of being quadratic in the gradients, explicitly penalizes length of the curves $V_1(x)$ and $V_2(x)$ on $\mb{CP}^3$.
Then, it is natural to require that dominant semi-classical trajectories must pass only through the points in Eq.~\eqref{eq:intersection-points} that are \textit{closest} to both $\sket{\1^{\ot 2};e}$ and $\sket{\1^{\ot 2};\eta}$.
\subsubsection{Determination of the shortest paths}
According to the Fubini-Study metric, the total distance traveled between two points labelled by matrices $V = \ketbra{v}{v}$ and $W = \ketbra{w}{w}$ on $\mb{CP}^3$ is~\cite{Bengtsson2017Aug}
\begin{equation}
    \dd(V,W) = \arccos(\sqrt{\tr(VW)}) = \arccos(|\braket{v}{w}|).
\label{eq:CP3distance}
\end{equation}
The points $\sket{\1^{\ot 2};e}$ and $\sket{\1^{\ot 2};\eta}$ correspond on their respective branch of the commutant manifold to the coordinates [cf. Eq.~(\ref{eq:vxfieldconfigs})]
\begin{equation}
    \sket{\1^{\ot 2}; \sigma}\cong (V_\1,V_\1, \sigma),\quad V_\1 = \frac{1}{2}\begin{pmatrix}
        1&0&0&1\\
        0&0&0&0\\
        0&0&0&0\\
        1&0&0&1
    \end{pmatrix}.
\end{equation}
Thus we can directly use Eq.~(\ref{eq:CP3distance}) compute that the length of a curve between the two points $\sket{\1^{\ot 2};e}$ and $\sket{\1^{\ot 2};\eta}$ that passes through the intersection point of Eq.~\eqref{eq:intersection-points} and~\eqref{eq:intersection-points-coord} is given by
\begin{equation}\label{eq:distance-function-su2}
    \begin{gathered}
        \dd(\{\vec n_a\}) = \sqrt{\dd_{\vec n_1,\vec n_2}^2+\dd_{\vec n_1,\vec n_4}^2+\dd_{\vec n_2,\vec n_3}^2+\dd_{\vec n_3,\vec n_4}^2},\\
        \dd_{\vec n_a,\vec n_b} = \arccos(\sqrt{\frac{1+\vec n_a\cdot\vec n_b}{4}}),
    \end{gathered}
\end{equation}
which is minimized when $\vec n_1=\vec n_2=\vec n_3=\vec n_4$.
This shows that there is a continuum of equivalent TDVP semi-classical solutions, which are labelled by the parameter $\vec{n} \in \mathbb{S}^2$.
\subsubsection{Semiclassical solution}
Following the prescription for multiple equivalent solutions in Sec.~\ref{subsubsec:multiplesoln}, we therefore formulate the semiclassical solution as a continuous superposition of trajectories:
\begin{equation}\label{eq:superposition-sphere}
    \sket{A\!:\!\bar A(t)} \approx \int_{\mb S^2}\dd\mu(\vec n)\, \sket{A\!:\!\bar A_{\,\vec n}(t)}
\end{equation}
where $\sket{A\!:\!\bar A_{\,\vec n}(t)}$ solves the equations of motion with the Dirichlet boundary condition that the field configuration $(V_1(x), V_2(x))$ equals $\ket{\vec n}\ot\ket{\vec n}^*\ot \ket{\vec n}\ot\ket{\vec n}^*$ at $x=0$ at all times.
Explicitly, we can parametrize each of these shortest paths by a single parameter $\omega\in[-\frac{\pi}{2},+\frac{\pi}{2}]$ as
\begin{gather}
    \ket{\omega;\vec n} = \begin{cases}
        \sket{\Omega_{\ket{\vec n}}(\omega)^{\ot 2};e},\quad &-\frac{\pi}{2}\leq\omega\leq 0,\\
        \sket{\Omega_{\ket{\vec n}}(\omega)^{\ot 2};\eta},\quad &\phantom{-}0\leq\omega\leq\frac{\pi}{2},
    \end{cases}\ \ \\
    \Omega_{\ket{\vec n}}(\omega) = \cos\frac{\omega}{2}\ketbra{\vec n}{\vec n} + \sin\frac{\omega}{2}\, (\1 - \ketbra{\vec n}{\vec n}),
\end{gather}
which is analogous to the parametrization implicit in Eqs.~(\ref{eq:esn-sol-dw}) and (\ref{eq:dw-diffeq-both-u1}), where there are one and two shortest paths respectively.
We can write this in the matrix coordinates Eq.~\eqref{eq:Vdef} on $\mb{CP}^3$ using
\begin{equation}\label{eq:trajectory-su2-rotation}
    \Omega_{\ket{\vec n}}(\omega)\cong U_{\vec n}\ot U_{\vec n}^*\begin{pmatrix}
        \sin^2\frac{\omega}{2}&0&0&\frac{\sin\omega}{2}\\
        0&0&0&0\\
        0&0&0&0\\
        \frac{\sin\omega}{2}&0&0&\cos^2\frac{\omega}{2}
    \end{pmatrix}U_{\vec n}\+\ot U_{\vec n}^T
\end{equation}
where $U_{\vec n}$ is an $SU(2)$ symmetry operator such that $U_{\vec n}\ket{\up}=\ket{\vec n}$, and the inner matrix is the coordinate expression for $\Omega_{\ket{\up}}$.
In this way, the equations of motion, which are Eq.~\eqref{eq:eoms-su2} for $V_1$ and $V_2$, and the domain wall initial conditions, reduce to equations for the single parameter $\omega$, which after explicit computation yields
\begin{equation}
    \sket{A\!:\!\bar A_{\,\vec n}(t)} = \ket{\omega(x,t);\vec n},\quad
    \begin{cases}
       \partial_t\omega(x,t) = \kappa\, \partial_x^2 \omega,\\
       \omega(x,0) = \mrm{sgn}(x)\frac{\pi}{2},\\
       \omega(0,t) = 0.
    \end{cases}\!\!\!\!\!\!
\end{equation}
These are exactly the same equations as Eq.~\eqref{eq:diffeq-esn-dw}, and as in the previous sections, the solution of these semiclassical equations of motion are the same as the ones for an initial domain wall in the Heisenberg model Eq.~\eqref{eq:sol-heis-th}
\begin{equation}
    \omega(x,t)=\frac{\pi}{2}\,\mrm{erf}\left(\frac{x}{\sqrt{4\kappa t}}\right).
    \label{eq:sol-heis-th-2}
\end{equation}
Each of these trajectories hence possess a diffusively large void regions dominated by the replicated frozen state $\ket{\vec n}\ot \ket{\vec n}^*\ot \ket{\vec n}\ot \ket{\vec n}^*$.
Similar arguments to those used in Sec.~\ref{subsec:esn-dw} show that under the normalization convention Eq.~\eqref{eq:norm-evo-eq-prime}, void formation leads to the norm of each trajectory ${\sket{A\!:\!\bar A_{\,\vec n}(t)}}$ decaying as $\sim e^{-\sqrt{\kappa t}}$, and the Rényi entropy is expected to grow as $\sim\sqrt{\kappa t}$ for generic initial product states.
\subsubsection{Structure of the void}\label{subsubsec:voidstructure}
The superposition of a continuum of paths of Eq.~\eqref{eq:superposition-sphere} also has a signature in the entanglement of the imaginary-time evolved state ${\sket{A\!:\!\bar A(t)}}$ itself, even in the semiclassical approximation.
In Appendix~\ref{app:entr}, we show that this continuous superposition generates a logarithmic growth of entanglement over time within the void:
\begin{equation}
    S_{\rm void}\big(\sket{A\!:\!\bar A(t)}\big)\sim \log(\kappa t).
\end{equation}
This is similar to what we found in Sec.~\ref{sec:ex-dw} for the domain wall between polarized ground states composed of the $\ket{\up}$ and $\ket{\dn}$ states in the spin-$\frac{1}{2}$ Heisenberg model, where a continuous family of equivalent TDVP solutions exists.
This logarithmic growth in the entanglement of the solution contrasts it for example with the $U(1)$ case of Sec.~\ref{sec:u1}, where the semiclassical imaginary-time evolution of the domain wall is dominated by a finite number of trajectories, whose entanglement saturates at a finite value.
\subsubsection{Extension to other entropic quantities}
The study of related entropic quantities, such as symmetry-resolved entanglement~\cite{Goldstein2018May,Fraenkel2020Mar,Bonsignori2019Oct,Murciano2020Aug,Murciano2022Aug} and entanglement asymmetry~\cite{Ares2023Apr,Ares2025Aug,Fossati2024May,Lastres2025Jan,Murciano2024Jan} would follow similarly from the study of symmetry-dressed domain walls Eq.~\eqref{eq:symmetry-resolved}, whose evolution would be similar to that of ${\sket{A\!:\!\bar A}}$, as discussed in Sec.~\ref{subsubsec:U1extensions} in the context of $U(1)$.
However, since the frozen states $\ket{\vec n}$ belong to a non-trivial irrep of the $SU(2)$ symmetry, there is a non-trivial manifold of such states, and the distance function Eq.~\eqref{eq:distance-function-su2} will change form, and the manifold of shortest-path trajectories will change, as we now discuss.
Consider for example a domain wall between $\ket{\1^{\ot 2};e}$ and $\ket{U_1\ot U_2;\eta}$ for $U_1,U_2\in SU(2)$, which are relevant for the computation of ``charged moments''~\cite{Goldstein2018May,Ares2023Apr}.
The distance function Eq.~\eqref{eq:distance-function-su2} between $\sket{\1^{\ot 2};e}$ and a state in $(\mb S^2)^{\times 4}$ of the form Eq.~\eqref{eq:intersection-points} is minimized when $\vec n_1=\vec n_2$ and $\vec n_3=\vec n_4$, and by symmetry we can show that the distance between $\ket{U_1\ot U_2;\eta}$ and this same set is instead minimized when $\vec n_1=R_1\vec n_4$ and $\vec n_3=R_2\vec n_2$, where $R_1$ and $R_2$ are the rotation matrices in $\mb R^3$ associated to $U_1$ and $U_2$ respectively.
When in Eq.~\eqref{eq:distance-function-su2} we had $R_1=R_2=\1$, we found that both conditions could be simultaneously satisfied by choosing all $\vec n_a$ vectors to be equal, leading to a manifold $\mb S^2$ of degenerate solutions.
However, in this more general case, this condition is satisfied when $\vec n_1=\vec n_2=R_1\vec n_3=R_1\vec n_4$ is a vector stabilized by the composite rotation $R_1R_2$, a condition that allows exactly two solutions unless $R_1=R_2^T$ (i.e., $U_1=\pm U_2\+$), in which case a full sphere $\mb S^2$ of minimum-distance points again exists within the intersection manifold.
This would predict a different pattern of void formation depending on the domain wall involved, and it would be interesting in the future to explore physical consequences of this prediction.

\subsection{Dynamics of Correlators} % MARK: * nonhydro
In the spin-$\frac{1}{2}$ $SU(2)$-symmetric model, the dynamics of correlation functions display a peculiar behavior.
Since the $k=1$ commutant manifold contains all single-site operators [cf.~Eq.~\eqref{eq:allopssu2}], \textit{all single-site operators are hydrodynamic} (defined in Sec.~\ref{subsubsec:hydrononhydro}).
This means that for an operator $O\in\mrm{End}(\mb C^2)$ simple correlation functions will map to free single-particle problems in the $SU(4)$ ferromagnet~\cite{moudgalya2024, PhysRevB.108.054307} [cf.~Eq.\eqref{eq:simple-correlator}], and decay algebraically,
\begin{equation}
    \overline{\langle O(0)\+ O(t)\rangle}\sim t^{-\frac{1}{2}}.
\end{equation}
Squared correlators [Eq.~\eqref{eq:squared-correlator}], will similarly map to a single $\sket{O\ot O\+;e}\in(\mb{CP}^3\times \mb{CP}^3)_e$ state in a background of $\sket{\1^{\ot 2};e}$ states, leading to first order to
\begin{equation}
    \overline{|\langle O(0)\+ O(t)\rangle|^2}\sim t^{-1},
\end{equation}
similar to the $U(1)$ case discussed in Sec.~\ref{subsubsec:U1hydrocorrelators}.
More broadly, we can show that here \textit{no} single-particle excitation in the replicated model can develop a void, unlike in the $U(1)$ case, where non-hydrodynamic operators are such excitations that lead to void formation.
This is because states of the form $\{\ket{A\ot B;e}\}$ completely span the single-site Hilbert space\footnote{Even just the states of Eq.~\eqref{eq:intersection-points} that lie in the intersection of the two manifolds are easily seen to be sufficient.} $\mc H_\mrm{loc}^{\ot 4}$.
Therefore, if we consider for example the initial state
\begin{equation}
    \dots\ot\sket{\1^{\ot 2};e}\ot\ket{\psi}\ot\sket{\1^{\ot 2};e}\ot\dots
\end{equation}
where $\ket{\psi}$ is any state in $\mc H_\mrm{loc}^{\ot 4}$, it is possible to decompose it as a sum of `hydrodynamic' states $\ket{A\ot B;e}$, which evolve almost-freely in the `identity' background.
Note however that it is nevertheless possible to construct entangled multiple-site operators that are non-hydrodynamic.
For example, any two-site operator of the form $O_{\langle ij\rangle}=A_iB_j-B_iA_j$ will be orthogonal to all ferromagnetic ground states of $P^{(1)}$ by the virtue of being spatially antisymmetric. 
The replicated operator $(O_{\langle ij\rangle})^{\ot 2}$ is however spatially symmetric, and so it can in principle have a non-zero overlap with ferromagnetic ground states of the $\eta$ type on $k=2$ replicas [cf.~Eq.~\eqref{eq:gsman-su2-intersection}].
Such an operator should naively behave like a point-like defect on top of a background of states in the $e$ permutation in the continuum, and we expect its dynamics to also be dominated by void formation as in Sec.~\ref{subsubsec:U1nonhydro}, where the voids here can instead be highly entangled, as in Sec.~\ref{subsubsec:voidstructure}.
\subsection{Other Models with Non-Abelian Symmetries}
So far we have discussed the specific case of a spin-$\frac{1}{2}$ chain, and here we show that many of its features generalize directly to other models with non-Abelian symmetries, i.e., other representations of $SU(2)$ or of other non-Abelian continuous symmetry groups.
\subsubsection{Continuum of frozen states and TDVP solutions}
The main key distinction that we find with respect to the Abelian $U(1)$ case of Sec.~\ref{sec:u1}---namely the continuous family of equivalent semiclassical trajectories Eq.~\eqref{eq:superposition-sphere}---can be shown to be a general feature characteristic of systems with continuous non-Abelian symmetries.
Recall that in Sec.~\ref{sec:fgsman} we showed that frozen states Eq.~\eqref{eq:frozenstatedefn} dictate the geometry of the intersection between the $e$ and $\eta$ branches of the manifold for $k=2$ replicas through Eq.~\eqref{eq:replicafrozenstatedefn}.
We also noted that since Abelian symmetries only possess one-dimensional irreps, their frozen states can only be non-degenerate, and hence only result in a discrete set of contact points between the two branches.
The general structure for non-Abelian symmetries is as follows.
The local Hilbert space decomposes as
\begin{equation}
    \mc H_\mrm{loc}=\bigoplus_{\lambda} \mc M_\lambda \ot \mc V_\lambda,
\end{equation}
where $V_\lambda$ are the irreps of the symmetry group $G$ and $\mc M_\lambda$ the multiplicity spaces (such that the multiplicity of each representation is $\mu_\lambda=\mrm{dim}(\mc M_\lambda)$).
Notice then that no vector $\ket{v}\in \mc M_\lambda\ot \mc V_\lambda$ can be frozen if $\mu_\lambda>1$, since symmetric operators on a single site can act non-trivially on the multiplicity degrees of freedom $\mc M_\lambda$.
Then frozen states $\ket{f_{\lambda,\alpha}}$ are those vectors in multiplicity-one irreps $\mc V_\lambda$ such that $\ket{f_{\lambda,\alpha}}^{\ot L}$ lies within a multiplicity-one irrep on $L$ sites.
These are exactly the highest-weight vectors in the representation~\cite{hall2015lie}, or equivalently, the generalized coherent states associated to the irrep~\cite{generalized-coher-st}.
For an irrep $\mc V_\lambda$, the projectivized set of all its generalized coherent states is a homogeneous manifold $G/H_\lambda$, where $H_\lambda$ is the stabilizer group for the coherent states.
For example, in the $SU(2)$ spin-$\frac{1}{2}$ example studied above, all states in the single-site representation are of highest-weight, and they form the manifold of spin-coherent states $SU(2)/U(1)\cong \mb S^2$.
We refer readers to \cite{lastres2026geometryfreefermioncommutants} for details discussions on this for $SO(2k)$ and $SU(2k)$ symmetry groups.
The intersection manifold for the $k=2$ commutants for such continuous symmetry groups $G$ then generally takes the form [cf. Eq.~\eqref{eq:replicafrozenstatedefn}]
\begin{equation}
    (M_G^{(2)})_e\cap(M_G^{(2)})_\eta \cong \bigcup_{\lambda:\mu_\lambda=1} (G/H_\lambda)^{\times 4}.
\end{equation}
Note that this formula is also trivially valid in the Abelian case, where $(G/H_\lambda)^{\times 4}$ is always a single point.
Intersection points where $\mrm{dim}(\mc V_\lambda)>1$ spontaneously break the replicated $G$ symmetry on the intersection manifold, since $G$ acts non-trivially on each of the frozen states.
Therefore, if some $G$-symmetric domain wall configuration between the $e$ and $\eta$ branch of the manifold evolves semiclassically through some replicated degenerate frozen state $\ket{f_{\lambda,\{\alpha_a\}}}$ [defined in Eq.~\eqref{eq:replicafrozenstatedefn}], then a continuous family of equivalent trajectories can be obtained simply by acting with $G$ on the whole trajectory, similar to our formulation of the entanglement domain wall evolution in Eq.~\eqref{eq:trajectory-su2-rotation} for the $SU(2)$ case.
Hence the existence of a continuum of equivalent TDVP solutions is a generic feature of non-Abelian continuous symmetries.
\subsubsection{Absence of single-site non-hydrodynamic operators} % MARK: * more
We have seen that in the specific case of a spin-$\frac{1}{2}$ chain, all single-site operators are elements of the $k=1$ commutant manifold (and thus hydrodynamic).
This property holds more broadly for all $SU(q)$-symmetric models where each site is a $q$-dimensional qudit in the fundamental representation.\footnote{Although note that for such models, 3-site interactions are needed for a circuit to form a 2-design~\cite{hulse2021qudit,Hulse2024Jul}} However, the absence of non-hydrodynamic local operators is a more general feature of all models where the single-site representation is irreducible, for example general $SU(2)$-symmetric spin-$s$ chains. As we show below, even though the single-site operators are not directly elements of $k = 1$ commutant manifold, they are nevertheless linear combinations of them.
As discussed in the $SU(2)$ spin-$\frac{1}{2}$ case above, this does not prevent operators on multiple sites from being non-hydrodynamic.
This result is a simple consequence of Schur's lemma and the Double Commutant Theorem (DCT)~\cite{landsman1998lecture}, using which we can prove the following Lemma:
\begin{lemma}
    Let $G$ be a continuous connected compact Lie group which acts irreducibly on a finite-dimensional local Hilbert space $\mc H_\mrm{loc}$.
    Then all single-site operators are hydrodynamic, as defined in Eq.~\eqref{eq:hydrodefn}.
\end{lemma}
\begin{proof}
    We start by showing that any local operator $O\in\mrm{End}(\mc H_\mrm{loc})$ has is a linear combination of some states in the commutant manifold $M_G^{(1)}$.
    By Schur's lemma, if $G$ acts irreducibly on $\mc H_\mrm{loc}$, then the only $G$-symmetric operator on a single site is the identity operator.
    The set of all single-site operators that commute with the identity---i.e., the \textit{double commutant} of the unitary on-site representation operators $\{U_g\}$---is exactly $\mrm{End}(\mc H_\mrm{loc})$.
    It follows from the DCT that the double commutant of $\{U_g\}$ is equal to its linear span, proving that all single-site operators are linear combinations of group elements.
    To conclude, we must then show that condition Eq.~\eqref{eq:hydrodefn} holds for states of the form $\ket{\cdots\1U_g\1\cdots}$.
    Since any $\ket{U_{g'}}^{\ot L}$ is an element of the $k=1$ commutant, to show this we need to find $g'\in G$ such that
    \begin{equation}\label{eq:lemma-condition}
        0\neq \braket{\1}{U_{g'}} = \chi(g'),\quad 0\neq \braket{U_{g}}{U_{g'}} = \chi(g^{-1}g'),
    \end{equation}
    where $\chi(g)=\tr(U_{g})$ is the character of the single site representation.
    Since the character is a non-zero real-analytic function~\cite{hall2015lie} on a connected set, then the sets where $\chi(g')=0$ and where $\chi(g^{-1}g')=0$ have both measure zero, which means that condition Eq.~\eqref{eq:lemma-condition} is satisfied by almost all $g'\in G$.
    This concludes the proof that all single-site operators are hydrodynamic in any system with such non-Abelian symmetries.
\end{proof}
\section{Discussion and Outlooks}\label{sec:outlooks} % MARK: 7. OUTLOOK
In this work, we have introduced a general geometric framework for studying the dynamics of observables in noisy quantum many-body systems with continuous symmetries.
This is enabled by the fact that noise-averaged observables can be expressed as matrix elements of imaginary-time evolution under effective Hamiltonians acting on multiple replicas of the system, so that their late-time dynamics is determined by the low-energy physics of these Hamiltonians.
To access this physics, we first study the ground states of these effective Hamiltonians, which are tightly related to the symmetries of $k$ replicas of the system, an algebraic object known as its $k$-commutant.
In the systems we study, we find that these $k$-commutants additionally have the structure of a continuous \textit{ferromagnet}, leading to a geometric characterization in terms of \textit{commutant manifolds}, distinct from their algebraic characterization.
For generic interacting systems with continuous symmetries, more specifically those that form symmetric $k$-designs~\cite{hearth2025unitary, liu2024unitary, mitsuhashi2025unitary}, these commutant manifolds can be obtained directly from the group manifold of the symmetry group $G$, allowing us to characterize them explicitly.
Interestingly, in contrast to the $k$-commutant manifolds in free-fermion systems characterized in our previous work~\cite{lastres2026geometryfreefermioncommutants}, those in interacting systems with continuous symmetries possess \textit{singularities} for $k \geq 2$, making them complex projective varieties in the mathematical sense.
In upcoming work~\cite{the-algebra-paper}, we provide the rigorous mathematical details needed to precisely characterize such manifolds starting from a general set of Brownian generators composing the noisy system.
Our emphasis on the geometry of these commutant manifolds is further justified by the fact that the low-energy excitations of the effective Hamiltonians can be expressed as smooth variations over these ferromagnetic manifolds.
We make this concrete using the time-dependent variational principle (TDVP)~\cite{tdvp20} to capture imaginary-time evolution under the effective Hamiltonians, leading to simple semiclassical equations governed by the quantum geometry of these manifolds.
The structure of the manifolds of interest here also requires us to extend TDVP appropriately, by providing prescriptions for singularities and including superpositions of equivalent semiclassical solutions.
We verified the validity of this framework in simple models such as the Heisenberg model and a toy $es\eta$ model that we introduced, where we carefully calibrated our results against simulation data.
The physics of these simple models naturally generalizes to effective Hamiltonians with $U(1)$ and $SU(2)$ symmetries, allowing us to gain a geometric understanding of the structure of replica circuit models without relying on analytical tricks used in the literature, such as introducing large charge-neutral degrees of freedom to artificially decouple the replicas.
More importantly, this perspective directly demonstrates the importance of void states during the dynamics of these observables, which were heavily used in earlier works to bound entanglement growth and non-hydrodynamic correlators~\cite{Huang_2020, Rakovszky_2019, McCulloch_2026}.
This geometric point of view opens several exciting directions for future research.
Most straightforwardly, the framework could be applied to other observables with continuous symmetries, such as Out-of-Time-Ordered Correlators (OTOCs)~\cite{Rakovszky2018Sep, Khemani2018Sep}, as well as numerous information-theoretic quantities in the presence of continuous symmetries, e.g., symmetry-resolved entropies~\cite{Goldstein2018May,Fraenkel2020Mar,Bonsignori2019Oct,Murciano2020Aug,Murciano2022Aug}, entanglement asymmetry~\cite{Ares2023Apr,Ares2025Aug,Fossati2024May,Lastres2025Jan,Murciano2024Jan}, coherences~\cite{Aditya2026Apr}, and negativity~\cite{li2023hilbert, li2025highly}.
Another direction is to explore inhomogeneous effective Hamiltonians arising from spatially varying strengths of randomness, or symmetry-breaking impurities, for which much of the hydrodynamics associated with continuous symmetries is known to survive on certain intermediate timescales~\cite{li2025dynamics, wang2025exponential, han2024exponentially}.
Finally, our geometric characterization has so far been restricted to systems with simple on-site continuous symmetries, and extending it to other kinds of symmetries presents several further directions.
For example, additional discrete symmetries in the problem should change the structure of the commutant manifolds and potentially lead to novel effects for certain observables.
Another is to characterize the geometry of commutants for symmetries of a different kind, such as multipole symmetries, where transport is similar to systems with $U(1)$ but with different exponents~\cite{feldmeier2020anomalous, moudgalya2021spectral}.
We no longer expect the commutant manifolds to be ferromagnetic in such cases, but possibly rather simple generalizations thereof, whose geometry would be interesting to characterize.
These manifolds should be even more exotic in constrained systems exhibiting Hilbert space fragmentation~\cite{sala2020fragmentation, khemani2020localization, moudgalya2019thermalization, moudgalya2021review, yang2022distinction, yang2020hilbert}, where the algebraic structure can already be rather complicated~\cite{read2007enlarged, rakovszky2020statistical, moudgalya2022}.

Returning to entanglement, it would be useful to further characterize the behavior of entropy densities, their flow in real space, and their relation to conserved-charge densities, potentially leading to a hydrodynamic or a field-theoretic~\cite{sensarma2023} understanding of entanglement analogous to the membrane picture in the absence of symmetries.
Our analysis suggests that the membrane picture~\cite{PhysRevX.7.031016, zhou_nahum, jonay_huse, vardhan_entanglement_2024}, or its modifications~\cite{Rakovszky_2019,Turkeshi_mpemba}, do not quite hold in the absence of additional charge-neutral degrees of freedom.
However, a membrane picture might survive in other sense, e.g., through the motion of the defect characterizing the ``position'' of the domain wall, as seen in Fig.~\ref{fig:esn-unb}, and this possibility deserves further exploration.
More technically, while we have mostly worked with the simplest case of $k=2$ replicas, the framework should be explored more thoroughly at higher replica numbers, which might enable the replica limit and attempt to compute observables such as von Neumann entanglement entropies.
These entropies have been conjectured to grow ballistically as $\sim t$ even in the presence of continuous symmetries such as $U(1)$~\cite{Rakovszky_2019}, distinguishing them from Rényi entropies, and therefore might support a membrane picture.
More generally, while TDVP clearly provides a first approximation to the dynamics, it is important to develop a more complete framework capable of capturing higher order effect such as fluctuations of singularities, which would require a path-integral picture in which corrections to the semiclassical trajectories can be incorporated systematically.
Although it is not obvious how to construct such a picture in the absence of coherent-state parametrizations of the ground-state manifold, one promising approach is to start from the free-fermion limit, where path integrals can be formulated using coherent states~\cite{swann2025, fava2023nonlinear, fava2024}, and treat interactions systematically, similarly to Ref.~\cite{swann2026continuummechanicsentanglementnoisy}.
Another possibility is to use mappings to classical stochastic models, which have been successful for certain observables~\cite{mcculloch2026longlivedlocalquantumcoherences}, which might be generalizable to arbitrary observables using insights from our work. 
Finally, it would be interesting to determine whether this geometric understanding carries over beyond unitary dynamics.
The evolution of observables in the presence of symmetric measurements can potentially be understood as imaginary-time evolution of effective Hamiltonians that are perturbations of the ones presented here~\cite{Agrawal2021Jul, Barratt2021Nov, zerba2025dipole}.
Even though extreme care should be exercised with issues such as replica limits and fluctuations, useful insights can be obtained even without such limits, and exploring how the geometric picture changes under such perturbations may provide novel perspectives into phenomena such as charge-sharpening transitions~\cite{Agrawal2021Jul, Barratt2021Nov, zerba2025dipole, nahum2025bayesian, gopalakrishnan2026monitored}.
Finally, given that many of these effective Hamiltonians are known to have Lindbladian forms with Hermitian jump operators~\cite{ogunnaike2023unifying} opens up the natural possibility of using such geometric approaches to study the relaxation of Lindbladians with strong continuous symmetries~\cite{hauser2026strongtoweak}.
We leave these questions for future work.
\section*{Acknowledgements} % MARK: *ackn
We are particularly grateful to Lesik Motrunich for enlightening conversations and encouragement.
We also acknowledge useful discussions with Sarang Gopalakrishnan, Alex Jacoby, Curt von Keyserlingk, Ewan McCulloch, Adam Nahum, Frank Pollmann, Tibor Rakovszky, and Shreya Vardhan.
S.M. thanks Lesik Motrunich for collaboration on \cite{moudgalya2024} and Shreya Vardhan for collaboration on \cite{vardhan_entanglement_2024}.
We acknowledge support from the Munich Quantum Valley, which is supported by the Bavarian state government with funds from the Hightech Agenda Bayern Plus, and the Munich Center for Quantum Science and Technology (MCQST), supported by the Deutsche Forschungsgemeinschaft (DFG, German Research Foundation) under Germany’s Excellence Strategy–EXC–2111–390814868.
\newpage
\bibliography{refs}

@article{read2007enlarged,
  author = {N. Read and H. Saleur},
  title = {{Enlarged symmetry algebras of spin chains, loop models, and S-matrices}},
  journal = {Nuclear Physics B},
  year = {2007},
  volume = {777},
  number = {3},
  pages = {263--315},
  doi = {10.1016/j.nuclphysb.2007.03.007},
  url = {http://www.sciencedirect.com/science/article/pii/S0550321307001745}
}

@article{chan2018solution,
  title = {Solution of a Minimal Model for Many-Body Quantum Chaos},
  author = {Chan, Amos and De Luca, Andrea and Chalker, J. T.},
  journal = {Phys. Rev. X},
  volume = {8},
  issue = {4},
  pages = {041019},
  numpages = {17},
  year = {2018},
  month = {Nov},
  publisher = {American Physical Society},
  doi = {10.1103/PhysRevX.8.041019},
  url = {https://link.aps.org/doi/10.1103/PhysRevX.8.041019}
}

@article{moudgalya2024,
  title = {{Symmetries as Ground States of Local Superoperators: Hydrodynamic Implications}},
  author = {Moudgalya, Sanjay and Motrunich, Olexei I.},
  journal = {PRX Quantum},
  volume = {5},
  issue = {4},
  pages = {040330},
  numpages = {41},
  year = {2024},
  month = {Nov},
  publisher = {American Physical Society},
  doi = {10.1103/PRXQuantum.5.040330},
  url = {https://link.aps.org/doi/10.1103/PRXQuantum.5.040330}
}

@article{khemani2020localization,
  author = {Khemani, Vedika and Hermele, Michael and Nandkishore, Rahul},
  title = {{Localization from Hilbert space shattering: From theory to physical realizations}},
  journal = {Phys. Rev. B},
  year = {2020},
  volume = {101},
  pages = {174204},
  doi = {10.1103/PhysRevB.101.174204},
  url = {https://link.aps.org/doi/10.1103/PhysRevB.101.174204}
}

@inbook{moudgalya2019thermalization,
  author = {Sanjay   Moudgalya  and  Abhinav   Prem  and  Rahul   Nandkishore  and  Nicolas   Regnault  and  B. Andrei   Bernevig },
  title = {{Thermalization and Its Absence within Krylov Subspaces of a Constrained Hamiltonian}},
  booktitle = {Memorial Volume for Shoucheng Zhang},
  pages = {147--209},
  doi = {10.1142/9789811231711_0009},
  url = {https://www.worldscientific.com/doi/abs/10.1142/9789811231711_0009},
  chapter = {7}
}

@article{Znidaric2020Jun,
	author = {{\ifmmode\check{Z}\else\v{Z}\fi}nidari{\ifmmode\check{c}\else\v{c}\fi}, Marko},
	title = {{Entanglement growth in diffusive systems}},
	journal = {Communications Physics},
	volume = {3},
	number = {100},
	pages = {100},
	year = {2020},
	month = jun,
	issn = {2399-3650},
	publisher = {Nature Publishing Group},
	doi = {10.1038/s42005-020-0366-7}
}

@article{d2016quantum,
  author = {{D'Alessio}, Luca and {Kafri}, Yariv and {Polkovnikov}, Anatoli and {Rigol}, Marcos},
  title = {{From quantum chaos and eigenstate thermalization to statistical mechanics and thermodynamics}},
  journal = {Advances in Physics},
  year = {2016},
  volume = {65},
  number = {3},
  pages = {239--362},
  doi = {10.1080/00018732.2016.1198134}
}

@article{polkovnikov2011colloquium,
  author = {Polkovnikov, Anatoli and Sengupta, Krishnendu and Silva, Alessandro and Vengalattore, Mukund},
  title = {{Colloquium: Nonequilibrium dynamics of closed interacting quantum systems}},
  journal = {Reviews of Modern Physics},
  year = {2011},
  volume = {83},
  pages = {863--883},
  doi = {10.1103/RevModPhys.83.863},
  url = {https://link.aps.org/doi/10.1103/RevModPhys.83.863}
}

@article{rigol2008thermalization,
  author = {{Rigol}, Marcos and {Dunjko}, Vanja and {Olshanii}, Maxim},
  title = {{Thermalization and its mechanism for generic isolated quantum systems}},
  journal = {Nature},
  year = {2008},
  volume = {452},
  number = {7189},
  pages = {854--858},
  doi = {10.1038/nature06838}
}

@article{zanardi2001virtual,
  author = {Zanardi, Paolo},
  title = {{Virtual Quantum Subsystems}},
  journal = {Phys. Rev. Lett.},
  year = {2001},
  volume = {87},
  pages = {077901},
  doi = {10.1103/PhysRevLett.87.077901},
  url = {https://link.aps.org/doi/10.1103/PhysRevLett.87.077901}
}

@article{bartlett2007reference,
  author = {Bartlett, Stephen D. and Rudolph, Terry and Spekkens, Robert W.},
  title = {{Reference frames, superselection rules, and quantum information}},
  journal = {Rev. Mod. Phys.},
  year = {2007},
  volume = {79},
  pages = {555--609},
  doi = {10.1103/RevModPhys.79.555},
  url = {https://link.aps.org/doi/10.1103/RevModPhys.79.555}
}

@article{Huang_2020,
   title={Dynamics of Rényi entanglement entropy in diffusive qudit systems},
   volume={1},
   ISSN={2633-1357},
   url={http://dx.doi.org/10.1088/2633-1357/abd1e2},
   DOI={10.1088/2633-1357/abd1e2},
   number={3},
   journal={IOP SciNotes},
   publisher={IOP Publishing},
   author={Huang, Yichen},
   year={2020},
   month=Dec, pages={035205} }

@article{hulse2021qudit,
  author = {Austin Hulse and Hanqing Liu and Iman Marvian},
  title = {{Qudit circuits with SU(d) symmetry: Locality imposes additional conservation laws}},
  journal = {arXiv preprint arXiv:2105.12877},
  year = {2021},
  eprint = {2105.12877},
  archiveprefix = {arXiv},
  primaryclass = {quant-ph},
  howpublished = {arXiv preprint arXiv:2105.12877}
}

@article{Li2024Dec,
	author = {Li, Zimu and Zheng, Han and Liu, Junyu and Jiang, Liang and Liu, Zi-Wen},
	title = {{Designs from Local Random Quantum Circuits with $\mathrm{SU}(d)$ Symmetry}},
	journal = {PRX Quantum},
	volume = {5},
	number = {4},
	pages = {040349},
	year = {2024},
	month = dec,
	publisher = {American Physical Society},
	doi = {10.1103/PRXQuantum.5.040349}
}

@article{Li2025Oct,
	author = {Li, Zimu and Zheng, Han and Wang, Yunfei and Jiang, Liang and Liu, Zi-Wen and Liu, Junyu},
	title = {{SU(d)-symmetric random unitaries: quantum scrambling, error correction, and machine learning}},
	journal = {npj Quantum Information},
	volume = {11},
	number = {158},
	pages = {158},
	year = {2025},
	month = oct,
	issn = {2056-6387},
	publisher = {Nature Publishing Group},
	doi = {10.1038/s41534-025-01045-6}
}

@article{Wu2026Apr,
	author = {Wu, Yuhan and Rodriguez-Nieva, Joaquin F.},
	title = {{Quantum state randomization constrained by non-Abelian symmetries}},
	journal = {ArXiv e-prints},
	year = {2026},
	month = apr,
	eprint = {2604.05043},
	doi = {10.48550/arXiv.2604.05043}
}

@article{Hulse2024Jul,
	author = {Hulse, Austin and Liu, Hanqing and Marvian, Iman},
	title = {{A framework for semi-universality: Semi-universality of 3-qudit SU(d)-invariant gates}},
	journal = {ArXiv e-prints},
	year = {2024},
	month = jul,
	eprint = {2407.21249},
	doi = {10.48550/arXiv.2407.21249}
}

@article{sala2020fragmentation,
  author = {Sala, Pablo and Rakovszky, Tibor and Verresen, Ruben and Knap, Michael and Pollmann, Frank},
  title = {{Ergodicity Breaking Arising from Hilbert Space Fragmentation in Dipole-Conserving Hamiltonians}},
  journal = {Phys. Rev. X},
  year = {2020},
  volume = {10},
  pages = {011047},
  doi = {10.1103/PhysRevX.10.011047},
  url = {https://link.aps.org/doi/10.1103/PhysRevX.10.011047}
}

@article{rakovszky2020statistical,
  author = {Rakovszky, Tibor and Sala, Pablo and Verresen, Ruben and Knap, Michael and Pollmann, Frank},
  title = {{Statistical localization: From strong fragmentation to strong edge modes}},
  journal = {Phys. Rev. B},
  year = {2020},
  volume = {101},
  pages = {125126},
  doi = {10.1103/PhysRevB.101.125126},
  url = {https://link.aps.org/doi/10.1103/PhysRevB.101.125126}
}

@article{landsman1998lecture,
  author = {Landsman, Nicolas P},
  title = {{Lecture notes on C*-algebras, Hilbert C*-modules, and quantum mechanics}},
  journal = {arXiv preprint math-ph/9807030},
  year = {1998}
}

@article{Zima2026Jun,
	author = {Zima, John P. and Stoudenmire, E. Miles and White, Steven R. and Parcollet, Olivier and Kaye, Jason},
	title = {{Fast Tensor Network Imaginary Time Evolution by Implicit Stepping on Logarithmic Grids}},
	journal = {ArXiv e-prints},
	year = {2026},
	month = jun,
	eprint = {2606.02930},
	doi = {10.48550/arXiv.2606.02930}
}

@article{Zhou2019May,
	author = {Zhou, Tianci and Chen, Xiao},
	title = {{Operator dynamics in a Brownian quantum circuit}},
	journal = {Physical Review E},
	volume = {99},
	number = {5},
	pages = {052212},
	year = {2019},
	month = may,
	publisher = {American Physical Society},
	doi = {10.1103/PhysRevE.99.052212}
}

@article{Rakovszky2021May,
	author = {Rakovszky, Tibor and Pollmann, Frank and von Keyserlingk, Curt},
	title = {{Entanglement growth in diffusive systems with large spin}},
	journal = {Communications Physics},
	volume = {4},
	number = {91},
	pages = {91},
	year = {2021},
	month = may,
	issn = {2399-3650},
	publisher = {Nature Publishing Group},
	doi = {10.1038/s42005-021-00594-4}
}

@article{moudgalya2022,
  author = {Moudgalya, Sanjay and Motrunich, Olexei I.},
  title = {{Hilbert Space Fragmentation and Commutant Algebras}},
  journal = {Phys. Rev. X},
  year = {2022},
  volume = {12},
  pages = {011050},
  doi = {10.1103/PhysRevX.12.011050},
  url = {https://link.aps.org/doi/10.1103/PhysRevX.12.011050}
}

@article{Agrawal2021Jul,
	author = {Agrawal, Utkarsh and Zabalo, Aidan and Chen, Kun and Wilson, Justin H. and Potter, Andrew C. and Pixley, J. H. and Gopalakrishnan, Sarang and Vasseur, Romain},
	title = {{Entanglement and charge-sharpening transitions in U(1) symmetric monitored quantum circuits}},
	journal = {ArXiv e-prints},
	year = {2021},
	month = jul,
	eprint = {2107.10279},
	doi = {10.1103/PhysRevX.12.041002}
}

@article{Barratt2021Nov,
	author = {Barratt, Fergus and Agrawal, Utkarsh and Gopalakrishnan, Sarang and Huse, David A. and Vasseur, Romain and Potter, Andrew C.},
	title = {{Field theory of charge sharpening in symmetric monitored quantum circuits}},
	journal = {ArXiv e-prints},
	year = {2021},
	month = nov,
	eprint = {2111.09336},
	doi = {10.1103/PhysRevLett.129.120604}
}

@article{keyserlingk2018hydro,
  title = {Operator Hydrodynamics, OTOCs, and Entanglement Growth in Systems without Conservation Laws},
  author = {von Keyserlingk, C. W. and Rakovszky, Tibor and Pollmann, Frank and Sondhi, S. L.},
  journal = {Phys. Rev. X},
  volume = {8},
  issue = {2},
  pages = {021013},
  numpages = {19},
  year = {2018},
  month = {Apr},
  publisher = {American Physical Society},
  doi = {10.1103/PhysRevX.8.021013},
  url = {https://link.aps.org/doi/10.1103/PhysRevX.8.021013}
}

@article{Fraenkel2020Mar,
	author = {Fraenkel, Shachar and Goldstein, Moshe},
	title = {{Symmetry resolved entanglement: exact results in 1D and beyond}},
	journal = {Journal of Statistical Mechanics: Theory and Experiment},
	volume = {2020},
	number = {3},
	pages = {033106},
	year = {2020},
	month = mar,
	issn = {1742-5468},
	publisher = {IOP Publishing and SISSA},
	doi = {10.1088/1742-5468/ab7753}
}

@article{Bonsignori2019Oct,
	author = {Bonsignori, Riccarda and Ruggiero, Paola and Calabrese, Pasquale},
	title = {{Symmetry resolved entanglement in free fermionic systems}},
	journal = {Journal of Physics A: Mathematical and Theoretical},
	volume = {52},
	number = {47},
	pages = {475302},
	year = {2019},
	month = oct,
	issn = {1751-8121},
	publisher = {IOP Publishing},
	doi = {10.1088/1751-8121/ab4b77}
}

@article{Murciano2024Jan,
	author = {Murciano, Sara and Ares, Filiberto and Klich, Israel and Calabrese, Pasquale},
	title = {{Entanglement asymmetry and quantum Mpemba effect in the XY spin chain}},
	journal = {Journal of Statistical Mechanics: Theory and Experiment},
	volume = {2024},
	number = {1},
	pages = {013103},
	year = {2024},
	month = jan,
	issn = {1742-5468},
	publisher = {IOP Publishing},
	doi = {10.1088/1742-5468/ad17b4}
}

@article{Lastres2025Jan,
	author = {Lastres, Marco and Murciano, Sara and Ares, Filiberto and Calabrese, Pasquale},
	title = {{Entanglement asymmetry in the critical XXZ spin chain}},
	journal = {Journal of Statistical Mechanics: Theory and Experiment},
	volume = {2025},
	number = {1},
	pages = {013107},
	year = {2025},
	month = jan,
	issn = {1742-5468},
	publisher = {IOP Publishing},
	doi = {10.1088/1742-5468/ada497}
}

@article{Fossati2024May,
	author = {Fossati, Michele and Ares, Filiberto and Dubail, J{\ifmmode\acute{e}\else\'{e}\fi}r{\ifmmode\hat{o}\else\^{o}\fi}me and Calabrese, Pasquale},
	title = {{Entanglement asymmetry in CFT and its relation to non-topological defects}},
	journal = {Journal of High Energy Physics},
	volume = {2024},
	number = {5},
	pages = {59},
	year = {2024},
	month = may,
	issn = {1029-8479},
	publisher = {Springer Berlin Heidelberg},
	doi = {10.1007/JHEP05(2024)059}
}

@article{Ares2025Aug,
	author = {Ares, Filiberto and Murciano, Sara and Calabrese, Pasquale and Piroli, Lorenzo},
	title = {{Entanglement asymmetry dynamics in random quantum circuits}},
	journal = {Physical Review Research},
	volume = {7},
	number = {3},
	pages = {033135},
	year = {2025},
	month = aug,
	publisher = {American Physical Society},
	doi = {10.1103/m3np-p5xj}
}

@article{Ares2023Apr,
	author = {Ares, Filiberto and Murciano, Sara and Calabrese, Pasquale},
	title = {{Entanglement asymmetry as a probe of symmetry breaking}},
	journal = {Nature Communications},
	volume = {14},
	number = {2036},
	pages = {2036},
	year = {2023},
	month = apr,
	issn = {2041-1723},
	publisher = {Nature Publishing Group},
	doi = {10.1038/s41467-023-37747-8}
}

@article{Murciano2020Aug,
	author = {Murciano, Sara and Di Giulio, Giuseppe and Calabrese, Pasquale},
	title = {{Entanglement and symmetry resolution in two dimensional free quantum field theories}},
	journal = {Journal of High Energy Physics},
	volume = {2020},
	number = {8},
	pages = {73},
	year = {2020},
	month = aug,
	issn = {1029-8479},
	publisher = {Springer Berlin Heidelberg},
	doi = {10.1007/JHEP08(2020)073}
}

@article{Murciano2022Aug,
	author = {Murciano, Sara and Calabrese, Pasquale and Piroli, Lorenzo},
	title = {{Symmetry-resolved Page curves}},
	journal = {Physical Review D},
	volume = {106},
	number = {4},
	pages = {046015},
	year = {2022},
	month = aug,
	publisher = {American Physical Society},
	doi = {10.1103/PhysRevD.106.046015}
}

@article{Goldstein2018May,
	author = {Goldstein, Moshe and Sela, Eran},
	title = {{Symmetry-Resolved Entanglement in Many-Body Systems}},
	journal = {Physical Review Letters},
	volume = {120},
	number = {20},
	pages = {200602},
	year = {2018},
	month = may,
	publisher = {American Physical Society},
	doi = {10.1103/PhysRevLett.120.200602}
}

@article{fava2023nonlinear,
  title = {Nonlinear Sigma Models for Monitored Dynamics of Free Fermions},
  author = {Fava, Michele and Piroli, Lorenzo and Swann, Tobias and Bernard, Denis and Nahum, Adam},
  journal = {Phys. Rev. X},
  volume = {13},
  issue = {4},
  pages = {041045},
  numpages = {33},
  year = {2023},
  month = {Dec},
  publisher = {American Physical Society},
  doi = {10.1103/PhysRevX.13.041045},
  url = {https://link.aps.org/doi/10.1103/PhysRevX.13.041045}
}

@ARTICLE{liu2024unitary,
       author = {{Liu}, Hanqing and {Hulse}, Austin and {Marvian}, Iman},
        title = "{Unitary Designs from Random Symmetric Quantum Circuits}",
      journal = {arXiv e-prints},
         year = 2024,
        month = aug,
archivePrefix = {arXiv},
       eprint = {2408.14463},
 primaryClass = {quant-ph},
       adsurl = {https://ui.adsabs.harvard.edu/abs/2024arXiv240814463L}
}

@article{mitsuhashi2025unitary,
  title = {Unitary Designs of Symmetric Local Random Circuits},
  author = {Mitsuhashi, Yosuke and Suzuki, Ryotaro and Soejima, Tomohiro and Yoshioka, Nobuyuki},
  journal = {Phys. Rev. Lett.},
  volume = {134},
  issue = {18},
  pages = {180404},
  numpages = {8},
  year = {2025},
  month = {May},
  publisher = {American Physical Society},
  doi = {10.1103/PhysRevLett.134.180404},
  url = {https://link.aps.org/doi/10.1103/PhysRevLett.134.180404}
}

@article{sunderhauf2019quantum,
  author = {S{\"u}nderhauf, Christoph
and Piroli, Lorenzo
and Qi, Xiao-Liang
and Schuch, Norbert
and Cirac, J. Ignacio},
  title = {{Quantum chaos in the Brownian SYK model with large finite N : OTOCs and tripartite information}},
  journal = {Journal of High Energy Physics},
  year = {2019},
  volume = {2019},
  number = {11},
  pages = {38},
  doi = {10.1007/JHEP11(2019)038}
}

@article{bauer2017stochastic,
  author = {Michel Bauer and Denis Bernard and Tony Jin},
  title = {{Stochastic dissipative quantum spin chains (I) : Quantum fluctuating  discrete hydrodynamics}},
  journal = {SciPost Phys.},
  year = {2017},
  volume = {3},
  pages = {033},
  doi = {10.21468/SciPostPhys.3.5.033},
  url = {https://scipost.org/10.21468/SciPostPhys.3.5.033}
}

@article{xu2019locality,
  author = {Xu, Shenglong and Swingle, Brian},
  title = {{Locality, Quantum Fluctuations, and Scrambling}},
  journal = {Phys. Rev. X},
  year = {2019},
  volume = {9},
  pages = {031048},
  doi = {10.1103/PhysRevX.9.031048},
  url = {https://link.aps.org/doi/10.1103/PhysRevX.9.031048}
}

@article{lashkari2013towards,
  author = {Lashkari, Nima
and Stanford, Douglas
and Hastings, Matthew
and Osborne, Tobias
and Hayden, Patrick},
  title = {{Towards the fast scrambling conjecture}},
  journal = {Journal of High Energy Physics},
  year = {2013},
  volume = {2013},
  number = {4},
  pages = {22},
  doi = {10.1007/JHEP04(2013)022}
}

@article{Gross2007Designs,
  author = {Gross, David and Audenaert, Koenraad and Eisert, Jens},
  title = {Evenly distributed unitaries: On the structure of unitary designs},
  journal = {Journal of Mathematical Physics},
  volume = {48},
  pages = {052104},
  year = {2007}
}

@article{HarrowLow2009,
  author = {Harrow, Aram W. and Low, Richard A.},
  title = {Random quantum circuits are approximate 2-designs},
  journal = {Communications in Mathematical Physics},
  volume = {291},
  pages = {257--302},
  year = {2009}
}

@article{Dankert2009,
  author = {Dankert, Christoph and Cleve, Richard and Emerson, Joseph and Livine, Emmanuel},
  title = {Exact and approximate unitary 2-designs},
  journal = {Phys. Rev. A},
  volume = {80},
  pages = {012304},
  year = {2009}
}

@article{Nahum2018,
  author = {Nahum, Adam and Ruhman, Jonathan and Vijay, Sagar and Haah, Jeongwan},
  title = {Operator spreading in random unitary circuits},
  journal = {Physical Review X},
  volume = {8},
  pages = {021014},
  year = {2018}
}

@ARTICLE{hunterjones2019unitary,
       author = {{Hunter-Jones}, Nicholas},
        title = "{Unitary designs from statistical mechanics in random quantum circuits}",
      journal = {arXiv e-prints},
         year = 2019,
        month = may,
archivePrefix = {arXiv},
       eprint = {1905.12053},
 primaryClass = {quant-ph},
       adsurl = {https://ui.adsabs.harvard.edu/abs/2019arXiv190512053H}
}

@article{Han2023Dec,
	author = {Han, Yiqiu and Chen, Xiao},
	title = {{Entanglement dynamics in U(1) symmetric hybrid quantum automaton circuits}},
	journal = {Quantum},
	volume = {7},
	pages = {1200},
	year = {2023},
	month = dec,
	publisher = {Verein zur F{\ifmmode\ddot{o}\else\"{o}\fi}rderung des Open Access Publizierens in den Quantenwissenschaften},
	eprint = {2305.18141v2},
	doi = {10.22331/q-2023-12-06-1200}
}

@article{Bertini2026May,
	author = {Bertini, Bruno},
	title = {{Non-equilibrium quantum many-body physics with quantum circuits}},
	journal = {SciPost Physics Lecture Notes},
	pages = {124},
	year = {2026},
	month = may,
	issn = {2590-1990},
	doi = {10.21468/SciPostPhysLectNotes.124}
}

@article{fisher2023random,
  author = "Fisher, Matthew P.A. and Khemani, Vedika and Nahum, Adam and Vijay, Sagar",
  title = "Random Quantum Circuits",
  journal = "Annual Review of Condensed Matter Physics",
  year = "2023",
  volume = "14",
  number = "Volume 14, 2023",
  pages = {335--379},
  doi = {10.1146/annurev-conmatphys-031720-030658},
  url = "https://www.annualreviews.org/content/journals/10.1146/annurev-conmatphys-031720-030658",
  type = "Journal Article"
}

@article{li2025dynamics,
  author = {Li, Yahui and Sala, Pablo and Pollmann, Frank and Moudgalya, Sanjay and Motrunich, Olexei},
  title = {{Dynamics in the presence of local symmetry-breaking impurities}},
  journal = {Phys. Rev. B},
  year = {2025},
  volume = {112},
  pages = {155108},
  doi = {10.1103/p1cm-9z8n},
  url = {https://link.aps.org/doi/10.1103/p1cm-9z8n}
}

@article{ogunnaike2023unifying,
  author = {Ogunnaike, Olumakinde and Feldmeier, Johannes and Lee, Jong Yeon},
  title = {{Unifying Emergent Hydrodynamics and Lindbladian Low-Energy Spectra across Symmetries, Constraints, and Long-Range Interactions}},
  journal = {Phys. Rev. Lett.},
  year = {2023},
  volume = {131},
  pages = {220403},
  doi = {10.1103/PhysRevLett.131.220403},
  url = {https://link.aps.org/doi/10.1103/PhysRevLett.131.220403}
}

@article{vardhan_entanglement_2024,
  title = {{Entanglement dynamics from universal low-lying modes}},
  author = {Vardhan, Shreya and Moudgalya, Sanjay},
  journal = {Phys. Rev. B},
  volume = {113},
  issue = {1},
  pages = {014308},
  numpages = {40},
  year = {2026},
  month = {Jan},
  publisher = {American Physical Society},
  doi = {10.1103/prp6-y5hl},
  url = {https://link.aps.org/doi/10.1103/prp6-y5hl}
}

@article{moudgalya2022from,
  author = {Sanjay Moudgalya and Olexei I. Motrunich},
  title = {{From symmetries to commutant algebras in standard Hamiltonians}},
  journal = {Annals of Physics},
  year = {2023},
  volume = {455},
  pages = {169384},
  doi = {10.1016/j.aop.2023.169384},
  url = {https://www.sciencedirect.com/science/article/pii/S0003491623001707}
}

@article{moudgalya2022exhaustive,
  author = {{Moudgalya}, Sanjay and {Motrunich}, Olexei I.},
  title = {{Exhaustive Characterization of Quantum Many-Body Scars using Commutant Algebras}},
  year = {2024},
  journal = {Physical Review X},
  volume = {14},
  pages = {041069},
  doi = {10.1103/PhysRevX.14.041069},
  url = {https://doi.org/10.1103/PhysRevX.14.041069}
}

@article{moudgalya2021review,
  author = {Sanjay Moudgalya and B Andrei Bernevig and Nicolas Regnault},
  title = {{Quantum many-body scars and Hilbert space fragmentation: a review of exact results}},
  journal = {Reports on Progress in Physics},
  year = {2022},
  volume = {85},
  number = {8},
  pages = {086501},
  doi = {10.1088/1361-6633/ac73a0}
}

@article{mori2018thermalization,
  author = {Takashi Mori and Tatsuhiko N Ikeda and Eriko Kaminishi and Masahito Ueda},
  title = {{Thermalization and prethermalization in isolated quantum systems: a theoretical overview}},
  journal = {Journal of Physics B: Atomic, Molecular and Optical Physics},
  year = {2018},
  volume = {51},
  number = {11},
  pages = {112001},
  doi = {10.1088/1361-6455/aabcdf}
}

@book{aharoni2000introduction,
  title={Introduction to the Theory of Ferromagnetism},
  author={Aharoni, A.},
  isbn={9780198508083},
  lccn={00126806},
  series={International series of monographs on physics},
  year={2000},
  publisher={Oxford University Press}
}

@article{Kochetov1995,
	author = {Kochetov, E. A.},
	title = {{Path integral over the generalized coherent states}},
	journal = {Journal of Mathematical Physics},
	year = {1995},
	doi = {10.1063/1.531078}
}

@book{Kipnis1999,
	author = {Kipnis, Claude and Landim, Claudio},
	title = {{Scaling Limits of Interacting Particle Systems}},
	journal = {SpringerLink},
	year = {1999},
	isbn = {978-3-662-03752-2},
	publisher = {Springer},
	address = {Berlin, Germany},
	url = {https://link.springer.com/book/10.1007/978-3-662-03752-2}
}

@article{Khemani2018Sep,
	author = {Khemani, Vedika and Vishwanath, Ashvin and Huse, David A.},
	title = {{Operator Spreading and the Emergence of Dissipative Hydrodynamics under Unitary Evolution with Conservation Laws}},
	journal = {Physical Review X},
	volume = {8},
	number = {3},
	pages = {031057},
	year = {2018},
	month = sep,
	publisher = {American Physical Society},
	doi = {10.1103/PhysRevX.8.031057}
}

@article{zeier2011symmetry,
  author = {Zeier,Robert  and Schulte-Herbrüggen,Thomas},
  title = {{Symmetry principles in quantum systems theory}},
  journal = {Journal of Mathematical Physics},
  year = {2011},
  volume = {52},
  number = {11},
  pages = {113510},
  doi = {10.1063/1.3657939}
}

@article{Sala2024Oct,
	author = {Sala, Pablo and Gopalakrishnan, Sarang and Oshikawa, Masaki and You, Yizhi},
	title = {{Spontaneous strong symmetry breaking in open systems: Purification perspective}},
	journal = {Physical Review B},
	volume = {110},
	number = {15},
	pages = {155150},
	year = {2024},
	month = oct,
	publisher = {American Physical Society},
	doi = {10.1103/PhysRevB.110.155150}
}

@article{Lessa2025Mar,
	author = {Lessa, Leonardo A. and Ma, Ruochen and Zhang, Jian-Hao and Bi, Zhen and Cheng, Meng and Wang, Chong},
	title = {{Strong-to-Weak Spontaneous Symmetry Breaking in Mixed Quantum States}},
	journal = {PRX Quantum},
	volume = {6},
	number = {1},
	pages = {010344},
	year = {2025},
	month = mar,
	publisher = {American Physical Society},
	doi = {10.1103/PRXQuantum.6.010344}
}

@article{Chen2025Feb,
	author = {Chen, Langxuan and Sun, Ning and Zhang, Pengfei},
	title = {{Strong-to-weak symmetry breaking and entanglement transitions}},
	journal = {Physical Review B},
	volume = {111},
	number = {6},
	pages = {L060304},
	year = {2025},
	month = feb,
	publisher = {American Physical Society},
	doi = {10.1103/PhysRevB.111.L060304}
}

@article{Huang2025Mar,
	author = {Huang, Xiaoyang and Qi, Marvin and Zhang, Jian-Hao and Lucas, Andrew},
	title = {{Hydrodynamics as the effective field theory of strong-to-weak spontaneous symmetry breaking}},
	journal = {Physical Review B},
	volume = {111},
	number = {12},
	pages = {125147},
	year = {2025},
	month = mar,
	publisher = {American Physical Society},
	doi = {10.1103/PhysRevB.111.125147}
}

@article{moudgalya2023numerical,
  author = {Moudgalya, Sanjay and Motrunich, Olexei I.},
  title = {{Numerical methods for detecting symmetries and commutant algebras}},
  journal = {Phys. Rev. B},
  year = {2023},
  volume = {107},
  pages = {224312},
  doi = {10.1103/PhysRevB.107.224312},
  url = {https://link.aps.org/doi/10.1103/PhysRevB.107.224312}
}

@book{Altland_Simons_2010, place={Cambridge}, edition={2}, title={Condensed Matter Field Theory}, publisher={Cambridge University Press}, author={Altland, Alexander and Simons, Ben D.}, year={2010}}

@article{lastres2026,
  title = {{Nonuniversality from conserved superoperators in unitary circuits}},
  author = {Lastres, Marco and Pollmann, Frank and Moudgalya, Sanjay},
  journal = {Phys. Rev. B},
  volume = {113},
  issue = {1},
  pages = {014310},
  numpages = {33},
  year = {2026},
  month = {Jan},
  publisher = {American Physical Society},
  doi = {10.1103/8jfm-l4ml},
  url = {https://link.aps.org/doi/10.1103/8jfm-l4ml}
}

@book{fradkin2013field,
  author = {Fradkin, Eduardo},
  title = {{Field theories of condensed matter physics}},
  year = {2013}
}

@article{Rakovszky2018Sep,
	author = {Rakovszky, Tibor and Pollmann, Frank and von Keyserlingk, C. W.},
	title = {{Diffusive Hydrodynamics of Out-of-Time-Ordered Correlators with Charge Conservation}},
	journal = {Physical Review X},
	volume = {8},
	number = {3},
	pages = {031058},
	year = {2018},
	month = sep,
	publisher = {American Physical Society},
	doi = {10.1103/PhysRevX.8.031058}
}

@ARTICLE{jonay_huse,
       author = {{Jonay}, Cheryne and {Huse}, David A. and {Nahum}, Adam},
        title = "{Coarse-grained dynamics of operator and state entanglement}",
      journal = {arXiv e-prints},
         year = 2018,
        month = feb,
archivePrefix = {arXiv},
       eprint = {1803.00089},
 primaryClass = {cond-mat.stat-mech},
       adsurl = {https://ui.adsabs.harvard.edu/abs/2018arXiv180300089J}
}

@ARTICLE{mark_membrane,
       author = {{Mezei}, M{\'a}rk},
        title = "{Membrane theory of entanglement dynamics from holography}",
      journal = {\prd},
         year = 2018,
        month = nov,
       volume = {98},
       number = {10},
          eid = {106025},
        pages = {106025},
          doi = {10.1103/PhysRevD.98.106025},
archivePrefix = {arXiv},
       eprint = {1803.10244},
 primaryClass = {hep-th},
       adsurl = {https://ui.adsabs.harvard.edu/abs/2018PhRvD..98j6025M}
}

@ARTICLE{zhou_nahum,
       author = {{Zhou}, Tianci and {Nahum}, Adam},
        title = "{Entanglement Membrane in Chaotic Many-Body Systems}",
      journal = {Physical Review X},
         year = 2020,
        month = jul,
       volume = {10},
       number = {3},
          eid = {031066},
        pages = {031066},
          doi = {10.1103/PhysRevX.10.031066},
archivePrefix = {arXiv},
       eprint = {1912.12311},
 primaryClass = {cond-mat.str-el},
       adsurl = {https://ui.adsabs.harvard.edu/abs/2020PhRvX..10c1066Z}
}

@article{Aditya2026Apr,
	author = {Aditya, Sreemayee and Tirrito, Emanuele and Sierant, Piotr and Turkeshi, Xhek},
	title = {{Coherence dynamics in quantum many-body systems with conservation laws}},
	journal = {ArXiv e-prints},
	year = {2026},
	month = apr,
	eprint = {2604.23192},
	doi = {10.48550/arXiv.2604.23192}
}

@misc{poetri2026,
	title={Theory of the Matchgate Commutant}, 
	author={Piotr Sierant and Xhek Turkeshi and Poetri Sonya Tarabunga},
	year={2026},
	eprint={2603.12392},
	archivePrefix={arXiv},
	primaryClass={quant-ph},
	url={https://arxiv.org/abs/2603.12392}
}

@misc{larocca2026,
      title={The commutant of fermionic Gaussian unitaries}, 
      author={Paolo Braccia and N. L. Diaz and Martin Larocca and M. Cerezo and Diego García-Martín},
      year={2026},
      eprint={2603.19210},
      archivePrefix={arXiv},
      primaryClass={quant-ph},
      url={https://arxiv.org/abs/2603.19210}, 
}

@article{enriched-phases,
	title = {Symmetry enriched phases of quantum circuits},
	journal = {Annals of Physics},
	volume = {435},
	pages = {168618},
	year = {2021},
	note = {Special issue on Philip W. Anderson},
	issn = {0003-4916},
	doi = {https://doi.org/10.1016/j.aop.2021.168618},
	url = {https://www.sciencedirect.com/science/article/pii/S0003491621002244},
	author = {Yimu Bao and Soonwon Choi and Ehud Altman}
}

@article{fava2024,
	title={Monitored fermions with conserved ${U}(1)$ charge},
	volume={6},
	ISSN={2643-1564},
	url={http://dx.doi.org/10.1103/PhysRevResearch.6.043246},
	DOI={10.1103/physrevresearch.6.043246},
	number={4},
	journal={Physical Review Research},
	publisher={American Physical Society (APS)},
	author={Fava, Michele and Piroli, Lorenzo and Bernard, Denis and Nahum, Adam},
	year={2024},
	month=dec
}

@article{swann2025,
	title = {Spacetime picture for entanglement generation in noisy fermion chains},
	author = {Swann, Tobias and Bernard, Denis and Nahum, Adam},
	journal = {Phys. Rev. B},
	volume = {112},
	issue = {6},
	pages = {064301},
	numpages = {23},
	year = {2025},
	month = {Aug},
	publisher = {American Physical Society},
	doi = {10.1103/PhysRevB.112.064301},
	url = {https://link.aps.org/doi/10.1103/PhysRevB.112.064301}
}

@mastersthesis{kerin2003homogeneous,
  title={Homogeneous metrics on spheres},
  author={Kerin, Martin and Wraith, David},
  year={2003},
  school={National University of Ireland, Maynooth}
}

@article{Stephan2017Oct,
	author = {St{\ifmmode\acute{e}\else\'{e}\fi}phan, Jean-Marie},
	title = {{Return probability after a quench from a domain wall initial state in the spin-1/2 XXZ chain}},
	journal = {Journal of Statistical Mechanics: Theory and Experiment},
	volume = {2017},
	number = {10},
	pages = {103108},
	year = {2017},
	month = oct,
	issn = {1742-5468},
	publisher = {IOP Publishing and SISSA},
	doi = {10.1088/1742-5468/aa8c19}
}

@book{hall2015lie,
	title={Lie Groups, Lie Algebras, and Representations: An Elementary Introduction},
	author={Hall, B.},
	isbn={9783319134673},
	lccn={2015935277},
	series={Graduate Texts in Mathematics},
	year={2015},
	publisher={Springer International Publishing}
}

@book{perelomov1986generalized,
	year = {1986},
	author = {Perelomov, A. M.},
	booktitle = {Generalized coherent states and their applications},
	isbn = {0387159126},
	lccn = {85030306},
	publisher = {Springer-Verlag},
	series = {Texts and monographs in physics},
	title = {Generalized coherent states and their applications},
}

@article{hearth2025unitary,
  title = {Unitary $k$-Designs from Random Number-Conserving Quantum Circuits},
  author = {Hearth, Sumner N. and Flynn, Michael O. and Chandran, Anushya and Laumann, Chris R.},
  journal = {Phys. Rev. X},
  volume = {15},
  issue = {2},
  pages = {021022},
  numpages = {23},
  year = {2025},
  month = {Apr},
  publisher = {American Physical Society},
  doi = {10.1103/PhysRevX.15.021022},
  url = {https://link.aps.org/doi/10.1103/PhysRevX.15.021022}
}

@article{generalized-coher-st,
	author = "Perelomov, A. M.",
	title = "{Coherent states for arbitrary lie groups}",
	doi = "10.1007/BF01645091",
	journal = "Commun. Math. Phys.",
	volume = "26",
	pages = "222--236",
	year = "1972"
}

@article{Berezin:1978sn,
    author = "Berezin, F. A.",
    title = "{Models of {Gross-Neveu} Type as Quantization of Classical Mechanics With Nonlinear Phase Space}",
    reportNumber = "ITF-78-119E",
    doi = "10.1007/BF01220849",
    journal = "Commun. Math. Phys.",
    volume = "63",
    pages = "131--153",
    year = "1978"
}

@article{Zhou2020Jul,
	author = {Zhou, Tianci and Ludwig, Andreas W. W.},
	title = {{Diffusive scaling of R{\ifmmode\backslash\else\textbackslash\fi}'enyi entanglement entropy}},
	journal = {Physical Review Research},
	volume = {2},
	number = {3},
	pages = {033020},
	year = {2020},
	month = jul,
	publisher = {American Physical Society},
	doi = {10.1103/PhysRevResearch.2.033020}
}

@article{nahum2022realtime,
  title = {Real-time correlators in chaotic quantum many-body systems},
  author = {Nahum, Adam and Roy, Sthitadhi and Vijay, Sagar and Zhou, Tianci},
  journal = {Phys. Rev. B},
  volume = {106},
  issue = {22},
  pages = {224310},
  numpages = {26},
  year = {2022},
  month = {Dec},
  publisher = {American Physical Society}}

@article{liu2018lectures,
	author = {Liu, Hong and Glorioso, Paolo},
	title = {{Lectures on non-equilibrium effective field theories and fluctuating hydrodynamics}},
	journal = {PoS},
	volume = {TASI2017},
	pages = {008},
	year = {2018},
	doi = {10.22323/1.305.0008}
}

@article{li2023hilbert,
  title = {Hilbert space fragmentation in open quantum systems},
  author = {Li, Yahui and Sala, Pablo and Pollmann, Frank},
  journal = {Phys. Rev. Res.},
  volume = {5},
  issue = {4},
  pages = {043239},
  numpages = {17},
  year = {2023},
  month = {Dec},
  publisher = {American Physical Society},
  doi = {10.1103/PhysRevResearch.5.043239},
  url = {https://link.aps.org/doi/10.1103/PhysRevResearch.5.043239}
}

@article{yang2022distinction,
  title = {Distinction between transport and R\'enyi entropy growth in kinetically constrained models},
  author = {Yang, Zhi-Cheng},
  journal = {Phys. Rev. B},
  volume = {106},
  issue = {22},
  pages = {L220303},
  numpages = {5},
  year = {2022},
  month = {Dec},
  publisher = {American Physical Society},
  doi = {10.1103/PhysRevB.106.L220303},
  url = {https://link.aps.org/doi/10.1103/PhysRevB.106.L220303}
}

@article{moudgalya2021spectral,
  title = {Spectral statistics in constrained many-body quantum chaotic systems},
  author = {Moudgalya, Sanjay and Prem, Abhinav and Huse, David A. and Chan, Amos},
  journal = {Phys. Rev. Res.},
  volume = {3},
  issue = {2},
  pages = {023176},
  numpages = {27},
  year = {2021},
  month = {Jun},
  publisher = {American Physical Society},
  doi = {10.1103/PhysRevResearch.3.023176},
  url = {https://link.aps.org/doi/10.1103/PhysRevResearch.3.023176}
}

@article{gopalakrishnan2026monitored,
  title = {Monitored Fluctuating Hydrodynamics},
  author = {Gopalakrishnan, Sarang and McCulloch, Ewan and Vasseur, Romain},
  journal = {Phys. Rev. X},
  volume = {16},
  issue = {1},
  pages = {011024},
  numpages = {17},
  year = {2026},
  month = {Feb},
  publisher = {American Physical Society},
  doi = {10.1103/295c-lj1w},
  url = {https://link.aps.org/doi/10.1103/295c-lj1w}
}

@article{nahum2025bayesian,
  title = {Bayesian critical points in classical lattice models},
  author = {Nahum, Adam and Jacobsen, Jesper Lykke},
  journal = {Phys. Rev. B},
  volume = {112},
  issue = {23},
  pages = {235113},
  numpages = {57},
  year = {2025},
  month = {Dec},
  publisher = {American Physical Society},
  doi = {10.1103/7dpt-d4s5},
  url = {https://link.aps.org/doi/10.1103/7dpt-d4s5}
}

@ARTICLE{hauser2026strongtoweak,
       author = {{Hauser}, Jacob and {Su}, Kaixiang and {Ha}, Hyunsoo and {Lloyd}, Jerome and {Kiely}, Thomas G. and {Vasseur}, Romain and {Gopalakrishnan}, Sarang and {Xu}, Cenke and {Fisher}, Matthew P.~A.},
        title = "{Strong-to-Weak Symmetry Breaking in Open Quantum Systems: From Discrete Particles to Continuum Hydrodynamics}",
      journal = {arXiv e-prints},
         year = 2026,
        month = feb,
archivePrefix = {arXiv},
       eprint = {2602.16045},
 primaryClass = {quant-ph},
       adsurl = {https://ui.adsabs.harvard.edu/abs/2026arXiv260216045H}
}

@ARTICLE{zerba2025dipole,
       author = {{Zerba}, Caterina and {Gopalakrishnan}, Sarang and {Knap}, Michael},
        title = "{Strong-to-weak symmetry breaking in monitored dipole conserving quantum circuits}",
      journal = {arXiv e-prints},
         year = 2025,
        month = dec,
archivePrefix = {arXiv},
       eprint = {2512.14830},
 primaryClass = {quant-ph},
       adsurl = {https://ui.adsabs.harvard.edu/abs/2025arXiv251214830Z}
}

@article{sensarma2023,
  title = {Building entanglement entropy out of correlation functions for interacting fermions},
  author = {Moitra, Saranyo and Sensarma, Rajdeep},
  journal = {Phys. Rev. B},
  volume = {108},
  issue = {17},
  pages = {174309},
  numpages = {20},
  year = {2023},
  month = {Nov},
  publisher = {American Physical Society},
  doi = {10.1103/PhysRevB.108.174309},
  url = {https://link.aps.org/doi/10.1103/PhysRevB.108.174309}
}

@article{feldmeier2020anomalous,
  title = {Anomalous Diffusion in Dipole- and Higher-Moment-Conserving Systems},
  author = {Feldmeier, Johannes and Sala, Pablo and De Tomasi, Giuseppe and Pollmann, Frank and Knap, Michael},
  journal = {Phys. Rev. Lett.},
  volume = {125},
  issue = {24},
  pages = {245303},
  numpages = {6},
  year = {2020},
  month = {Dec},
  publisher = {American Physical Society},
  doi = {10.1103/PhysRevLett.125.245303},
  url = {https://link.aps.org/doi/10.1103/PhysRevLett.125.245303}
}

@ARTICLE{han2024exponentially,
       author = {{Han}, Yiqiu and {Chen}, Xiao and {Lake}, Ethan},
        title = "{Exponentially slow thermalization and the robustness of Hilbert space fragmentation}",
      journal = {arXiv e-prints},
         year = 2024,
        month = jan,
archivePrefix = {arXiv},
       eprint = {2401.11294},
 primaryClass = {quant-ph},
       adsurl = {https://ui.adsabs.harvard.edu/abs/2024arXiv240111294H}
}

@ARTICLE{wang2025exponential,
       author = {{Wang}, Cheng and {Balasubramanian}, Shankar and {Han}, Yiqiu and {Lake}, Ethan and {Chen}, Xiao and {Yang}, Zhi-Cheng},
        title = "{Exponentially slow thermalization in 1D fragmented dynamics}",
      journal = {arXiv e-prints},
         year = 2025,
        month = jan,
archivePrefix = {arXiv},
       eprint = {2501.13930},
 primaryClass = {quant-ph},
       adsurl = {https://ui.adsabs.harvard.edu/abs/2025arXiv250113930W}
}

@article{yang2020hilbert,
  title = {Hilbert-Space Fragmentation from Strict Confinement},
  author = {Yang, Zhi-Cheng and Liu, Fangli and Gorshkov, Alexey V. and Iadecola, Thomas},
  journal = {Phys. Rev. Lett.},
  volume = {124},
  issue = {20},
  pages = {207602},
  numpages = {6},
  year = {2020},
  month = {May},
  publisher = {American Physical Society},
  doi = {10.1103/PhysRevLett.124.207602},
  url = {https://link.aps.org/doi/10.1103/PhysRevLett.124.207602}
}

@article{li2025highly,
  title = {Highly Entangled Stationary States from Strong Symmetries},
  author = {Li, Yahui and Pollmann, Frank and Read, Nicholas and Sala, Pablo},
  journal = {Phys. Rev. X},
  volume = {15},
  issue = {1},
  pages = {011068},
  numpages = {37},
  year = {2025},
  month = {Mar},
  publisher = {American Physical Society},
  doi = {10.1103/PhysRevX.15.011068},
  url = {https://link.aps.org/doi/10.1103/PhysRevX.15.011068}
}

@ARTICLE{2026arXiv260602391D,
       author = {{Delacretaz}, Luca V.},
        title = "{Boulder Lectures on Thermal Dynamics and Hydrodynamic EFTs}",
      journal = {arXiv e-prints},
         year = 2026,
        month = jun,
archivePrefix = {arXiv},
       eprint = {2606.02391},
 primaryClass = {hep-th},
       adsurl = {https://ui.adsabs.harvard.edu/abs/2026arXiv260602391D}
}

@ARTICLE{crossley,
       author = {{Crossley}, Michael and {Glorioso}, Paolo and {Liu}, Hong},
        title = "{Effective field theory of dissipative fluids}",
      journal = {Journal of High Energy Physics},
         year = 2017,
        month = sep,
       volume = {2017},
       number = {9},
          eid = {95},
        pages = {95},
          doi = {10.1007/JHEP09(2017)095},
       adsurl = {https://ui.adsabs.harvard.edu/abs/2017JHEP...09..095C}
}

@article{Zhou_2019,
   title={Emergent statistical mechanics of entanglement in random unitary circuits},
   volume={99},
   ISSN={2469-9969},
   url={http://dx.doi.org/10.1103/PhysRevB.99.174205},
   DOI={10.1103/physrevb.99.174205},
   number={17},
   journal={Physical Review B},
   publisher={American Physical Society (APS)},
   author={Zhou, Tianci and Nahum, Adam},
   year={2019},
   month={may} }

@article{10.1063/1.524610,
    author = {Thomas, Lawrence E.},
    title = {Quantum Heisenberg ferromagnets and stochastic exclusion processes},
    journal = {Journal of Mathematical Physics},
    volume = {21},
    number = {7},
    pages = {1921-1924},
    year = {1980},
    month = {07},
    issn = {0022-2488},
    doi = {10.1063/1.524610},
    url = {https://doi.org/10.1063/1.524610},
}

@article{Bravyi2015,
	author = {Bravyi, Sergey},
	title = {{Monte Carlo simulation of stoquastic Hamiltonians}},
	journal = {Quant.Inf.Comput},
	volume = {15},
	number = {13-14},
	pages = {1122--1140},
	year = {2015},
	doi = {10.26421/QIC15.13-14-3}
}

@article{Blythe_2003,
   title={Survival probability of a diffusing particle in the presence of Poisson-distributed mobile traps},
   volume={67},
   ISSN={1095-3787},
   url={http://dx.doi.org/10.1103/PhysRevE.67.041101},
   DOI={10.1103/physreve.67.041101},
   number={4},
   journal={Physical Review E},
   publisher={American Physical Society (APS)},
   author={Blythe, R. A. and Bray, A. J.},
   year={2003},
   month=Apr }

@article{Anton_2004,
   title={Spatial fluctuations of a surviving particle in the trapping reaction},
   volume={38},
   ISSN={1361-6447},
   url={http://dx.doi.org/10.1088/0305-4470/38/1/009},
   DOI={10.1088/0305-4470/38/1/009},
   number={1},
   journal={Journal of Physics A: Mathematical and General},
   publisher={IOP Publishing},
   author={Anton, L and Blythe, R A and Bray, A J},
   year={2004},
   month=Dec, pages={133–144} }

@article{PhysRevX.7.031016,
  title = {Quantum Entanglement Growth under Random Unitary Dynamics},
  author = {Nahum, Adam and Ruhman, Jonathan and Vijay, Sagar and Haah, Jeongwan},
  journal = {Phys. Rev. X},
  volume = {7},
  issue = {3},
  pages = {031016},
  numpages = {30},
  year = {2017},
  month = {Jul},
  publisher = {American Physical Society},
  doi = {10.1103/PhysRevX.7.031016},
  url = {https://link.aps.org/doi/10.1103/PhysRevX.7.031016}
}

@article{Collins_2010,
   title={Random Quantum Channels I: Graphical Calculus and the Bell State Phenomenon},
   volume={297},
   ISSN={1432-0916},
   url={http://dx.doi.org/10.1007/s00220-010-1012-0},
   DOI={10.1007/s00220-010-1012-0},
   number={2},
   journal={Communications in Mathematical Physics},
   publisher={Springer Science and Business Media LLC},
   author={Collins, Benoît and Nechita, Ion},
   year={2010},
   month=feb, pages={345–370} }

@article{Agarwal2023May,
	author = {Agarwal, Lakshya and Sahu, Subhayan and Xu, Shenglong},
	title = {{Charge transport, information scrambling and quantum operator-coherence in a many-body system with U(1) symmetry}},
	journal = {Journal of High Energy Physics},
	volume = {2023},
	number = {5},
	pages = {37},
	year = {2023},
	month = may,
	issn = {1029-8479},
	publisher = {Springer Berlin Heidelberg},
	doi = {10.1007/JHEP05(2023)037}
}

@article{Agarwal_2022,
   title={Emergent symmetry in Brownian SYK models and charge dependent scrambling},
   volume={2022},
   ISSN={1029-8479},
   url={http://dx.doi.org/10.1007/JHEP02(2022)045},
   DOI={10.1007/jhep02(2022)045},
   number={2},
   journal={Journal of High Energy Physics},
   publisher={Springer Science and Business Media LLC},
   author={Agarwal, Lakshya and Xu, Shenglong},
   year={2022},
   month=feb }

@Article{10.21468/SciPostPhys.15.4.175,
	title={{Exact entanglement in the driven quantum symmetric simple exclusion process}},
	author={Denis Bernard and Ludwig Hruza},
	journal={SciPost Phys.},
	volume={15},
	pages={175},
	year={2023},
	publisher={SciPost},
	doi={10.21468/SciPostPhys.15.4.175},
	url={https://scipost.org/10.21468/SciPostPhys.15.4.175},
}

@article{Bernard_2021,
   title={Solution to the Quantum Symmetric Simple Exclusion Process: The Continuous Case},
   volume={384},
   ISSN={1432-0916},
   url={http://dx.doi.org/10.1007/s00220-021-04087-x},
   DOI={10.1007/s00220-021-04087-x},
   number={2},
   journal={Communications in Mathematical Physics},
   publisher={Springer Science and Business Media LLC},
   author={Bernard, Denis and Jin, Tony},
   year={2021},
   month=apr, pages={1141–1185} }

@misc{swann2026continuummechanicsentanglementnoisy,
      title={Continuum mechanics of entanglement in noisy interacting fermion chains}, 
      author={Tobias Swann and Adam Nahum},
      year={2026},
      eprint={2601.21134},
      archivePrefix={arXiv},
      primaryClass={cond-mat.stat-mech},
      url={https://arxiv.org/abs/2601.21134}, 
}

@misc{lastres2026geometryfreefermioncommutants,
      title={Geometry of Free Fermion Commutants}, 
      author={Marco Lastres and Sanjay Moudgalya},
      year={2026},
      eprint={2604.05031},
      archivePrefix={arXiv},
      primaryClass={quant-ph},
      url={https://arxiv.org/abs/2604.05031}, 
}

@article{kochetov95,
  title = {SU(2) coherent-state path integral for the Heisenberg ferromagnet},
  author = {Kochetov, E. A.},
  journal = {Phys. Rev. B},
  volume = {52},
  issue = {6},
  pages = {4402--4408},
  numpages = {0},
  year = {1995},
  month = {Aug},
  publisher = {American Physical Society},
  doi = {10.1103/PhysRevB.52.4402},
  url = {https://link.aps.org/doi/10.1103/PhysRevB.52.4402}
}

@article{kochetov95heis,
    author = {Kochetov, E. A.},
    title = {SU(2) coherent‐state path integral},
    journal = {Journal of Mathematical Physics},
    volume = {36},
    number = {9},
    pages = {4667-4679},
    year = {1995},
    month = {09},
    issn = {0022-2488},
    doi = {10.1063/1.530913},
    url = {https://doi.org/10.1063/1.530913},
}

@article{klauder79,
  title = {Path integrals and stationary-phase approximations},
  author = {Klauder, John R.},
  journal = {Phys. Rev. D},
  volume = {19},
  issue = {8},
  pages = {2349--2356},
  numpages = {0},
  year = {1979},
  month = {Apr},
  publisher = {American Physical Society},
  doi = {10.1103/PhysRevD.19.2349},
  url = {https://link.aps.org/doi/10.1103/PhysRevD.19.2349}
}

@book{Bengtsson2017Aug,
	author = {Bengtsson, Ingemar and {\ifmmode\dot{Z}\else\.{Z}\fi}yczkowski, Karol},
	title = {{Geometry of Quantum States: An Introduction to Quantum Entanglement}},
	year = {2017},
	month = aug,
	isbn = {978-1-10702625-4},
	publisher = {Cambridge University Press},
	address = {Cambridge, England, UK}
}

@article{tdvp20,
	title={{Geometry of variational methods: dynamics of closed quantum systems}},
	author={Lucas Hackl and Tommaso Guaita and Tao Shi and Jutho Haegeman and Eugene Demler and J. Ignacio Cirac},
	journal={SciPost Phys.},
	volume={9},
	pages={048},
	year={2020},
	publisher={SciPost},
	doi={10.21468/SciPostPhys.9.4.048},
	url={https://scipost.org/10.21468/SciPostPhys.9.4.048},
}

@article{Turkeshi_mpemba,
  title = {Quantum Mpemba Effect in Random Circuits},
  author = {Turkeshi, Xhek and Calabrese, Pasquale and De Luca, Andrea},
  journal = {Phys. Rev. Lett.},
  volume = {135},
  issue = {4},
  pages = {040403},
  numpages = {10},
  year = {2025},
  month = {Jul},
  publisher = {American Physical Society},
  doi = {10.1103/5d6p-8d1b},
  url = {https://link.aps.org/doi/10.1103/5d6p-8d1b}
}

@article{Rakovszky_2019,
   title={Sub-ballistic Growth of Rényi Entropies due to Diffusion},
   volume={122},
   ISSN={1079-7114},
   url={http://dx.doi.org/10.1103/PhysRevLett.122.250602},
   DOI={10.1103/physrevlett.122.250602},
   number={25},
   journal={Physical Review Letters},
   publisher={American Physical Society (APS)},
   author={Rakovszky, Tibor and Pollmann, Frank and von Keyserlingk, C. W.},
   year={2019},
   month=June }

@article{McCulloch_2023,
	title={Full Counting Statistics of Charge in Chaotic Many-Body Quantum Systems},
	volume={131},
	ISSN={1079-7114},
	url={http://dx.doi.org/10.1103/PhysRevLett.131.210402},
	DOI={10.1103/physrevlett.131.210402},
	number={21},
	journal={Physical Review Letters},
	publisher={American Physical Society (APS)},
	author={McCulloch, Ewan and De Nardis, Jacopo and Gopalakrishnan, Sarang and Vasseur, Romain},
	year={2023},
	month=Nov
}

@article{PhysRevLett.55.537,
  title = {Nonlinear excitations on a quantum ferromagnetic chain},
  author = {Balakrishnan, Radha and Bishop, A. R.},
  journal = {Phys. Rev. Lett.},
  volume = {55},
  issue = {5},
  pages = {537--540},
  numpages = {0},
  year = {1985},
  month = {Jul},
  publisher = {American Physical Society},
  doi = {10.1103/PhysRevLett.55.537},
  url = {https://link.aps.org/doi/10.1103/PhysRevLett.55.537}
}

@misc{mcculloch2026longlivedlocalquantumcoherences,
	title={Long-lived local quantum coherences from hydrodynamic large deviations}, 
	author={Ewan McCulloch and J. Alexander Jacoby and Sarang Gopalakrishnan},
	year={2026},
	eprint={2604.27074},
	archivePrefix={arXiv},
	primaryClass={quant-ph},
	url={https://arxiv.org/abs/2604.27074}, 
}

@article{McCulloch_2026,
   title={Subexponential Decay of Local Correlations from Diffusion-Limited Dephasing},
   volume={136},
   ISSN={1079-7114},
   url={http://dx.doi.org/10.1103/393g-z21y},
   DOI={10.1103/393g-z21y},
   number={19},
   journal={Physical Review Letters},
   publisher={American Physical Society (APS)},
   author={McCulloch, Ewan and Jacoby, J. Alexander and von Keyserlingk, Curt and Gopalakrishnan, Sarang},
   year={2026},
   month=May }

@book{klauder1985coherent,
  title={Coherent States : Applications In Physics And Mathematical Physics},
  author={Klauder, J.R. and Skagerstam, B.},
  isbn={9789814590884},
  url={https://books.google.de/books?id=qXsGCwAAQBAJ},
  year={1985},
  publisher={World Scientific Publishing Company}
}

@article{JohnSchliemann_1998,
doi = {10.1088/0953-8984/10/5/016},
url = {https://doi.org/10.1088/0953-8984/10/5/016},
year = {1998},
month = {feb},
publisher = {},
volume = {10},
number = {5},
pages = {1091},
author = {John Schliemann and Franz G Mertens},
title = {Semiclassical description of Heisenberg models via spin-coherent states},
journal = {Journal of Physics: Condensed Matter}
}

@article{JMRadcliffe_1971,
doi = {10.1088/0305-4470/4/3/009},
url = {https://doi.org/10.1088/0305-4470/4/3/009},
year = {1971},
month = {may},
publisher = {},
volume = {4},
number = {3},
pages = {313},
author = {J M Radcliffe},
title = {Some properties of coherent spin states},
journal = {Journal of Physics A: General Physics}
}

@article{Masaoka2024Nov,
	author = {Masaoka, Rintaro and Soejima, Tomohiro and Watanabe, Haruki},
	title = {{Quadratic dispersion relations in gapless frustration-free systems}},
	journal = {Physical Review B},
	volume = {110},
	number = {19},
	pages = {195140},
	year = {2024},
	month = nov,
	publisher = {American Physical Society},
	doi = {10.1103/PhysRevB.110.195140}
}

@article{PhysRevB.108.054307,
  title = {Critical phase and spin sharpening in SU(2)-symmetric monitored quantum circuits},
  author = {Majidy, Shayan and Agrawal, Utkarsh and Gopalakrishnan, Sarang and Potter, Andrew C. and Vasseur, Romain and Halpern, Nicole Yunger},
  journal = {Phys. Rev. B},
  volume = {108},
  issue = {5},
  pages = {054307},
  numpages = {13},
  year = {2023},
  month = {Aug},
  publisher = {American Physical Society},
  doi = {10.1103/PhysRevB.108.054307},
  url = {https://link.aps.org/doi/10.1103/PhysRevB.108.054307}
}

@incollection{SCHUTZ20011,
    title = {1 - Exactly Solvable Models for Many-Body Systems Far from Equilibrium},
    editor = {C. Domb and J.L. Lebowitz},
    series = {Phase Transitions and Critical Phenomena},
    publisher = {Academic Press},
    volume = {19},
    pages = {1-251},
    year = {2001},
    issn = {1062-7901},
    doi = {https://doi.org/10.1016/S1062-7901(01)80015-X},
    url = {https://www.sciencedirect.com/science/article/pii/S106279010180015X},
    author = {G.M. Schütz}
}

@article{the-algebra-paper,
  author = {Lastres, Marco and Moudgalya, Sanjay},
  journal = {},
  title = {{in preparation}},
  year = {2026}
}

\onecolumngrid
\appendix %%%%%%%%%%%%%%%%%%%%%%%%%%%%%%%%%%%%%%%%%%%%%%%%%%%%%%
\newpage

\section{Review of the Spin Coherent Path Integral} % MARK: APP: PATH INT
\label{app:path-integral}
In this appendix, we provide a self-contained review of the $SU(2)$ spin coherent state path integral~\cite{kochetov95,kochetov95heis}, including its application to the study of the ferromagnetic Heisenberg model Eq.~\eqref{eq:HeisModelHam}.
We focus on the subtle role of boundary terms, showing these features explicitly through the example of a single spin in a longitudinal magnetic field.
While this is not directly used for any of the results in the main text, it is instructive to review this powerful approach, which has been successfully applied to study certain Brownian models close to free-fermion limits in \cite{swann2025, swann2026continuummechanicsentanglementnoisy}, and compare these results to those from the time dependent variational principle, which we do in Appendix~\ref{app:tdvp-vs-pi}.
This comparison might also be useful for formulating path integrals for more general commutant manifolds discussed in the main text, where our results have been limited to those obtained using TDVP.
\subsection{Coherent State Representation and the Path Integral Action}
A spin-coherent state $\ket{\vec{n}}$ \cite{JMRadcliffe_1971} for a spin-$\frac{1}{2}$ degree of freedom is parameterized by a three-dimensional unit vector $\vec{n}$ defined as
\begin{equation}
	\ket{\vec{n}} = \cos\frac{\theta}{2}\ket\up + e^{i\varphi}\sin\frac{\theta}{2}\ket\dn,\;\;\;  \vec{n} \defn (\sin\theta\cos\varphi, \sin\theta\sin\varphi, \cos\theta)^T\label{eq:spin-coh-states}
\end{equation}
which can be thought of as a point on the Bloch sphere $\mb{S}^2$.
By design, the expectation value of the Pauli vector recovers the unit vector: $\bra{\vec{n}}\vec{\sigma}\ket{\vec{n}} = \vec{n}$.
Alternatively, these states can be parametrized by complex coordinates obtained by mapping the Bloch sphere onto the complex plane $\mb{C}$ via stereographic projection:
\begin{equation}
	\ket{z} = \frac{z\ket\up + \ket\dn}{\sqrt{1 + \bar{z}z}}=\ket{\vec n},\;\;\; z \defn e^{-i\varphi}\cot\frac{\theta}{2}.
\label{eq:coherentstatecomplex}
\end{equation}
Here $\bar{z}$ represent the complex conjugate of the variable $z$, but it will later be treated as an independent variable due to technicalities in the path integral formalism.
This parametrization is the more appropriate one when discussing the spin-coherent path integral.
The resolution of identity on the single-site Hilbert space is given by:
\begin{equation}\label{eq:res-id-s2-formula}
	\1 = \frac{2}{\pi}\int\frac{\dd ^2z}{{(1 + \bar{z}z)^2}} \ketbra{z}{z} = \int \frac{\mrm{d}\Omega}{2\pi} \ketbra{\vec{n}}{\vec{n}}.
\end{equation}
To construct the path integral for a Hamiltonian $H$ evolving in imaginary time $\tau \in [0, t]$ for the transition amplitude between the initial and final coherent states $\ket{z_\mrm{i}}$ and $\ket{z_\mrm{f}}$, we first Trotterize the evolution operator $e^{-Ht}$ into $N$ time steps of width $\delta \tau = t/N$.
Inserting the resolution of identity at each step, and following standard algebraic simplifications, this can ultimately be written as \cite{kochetov95}
\begin{equation}
	Z(z_\mrm{i}, z_\mrm{f}; t) = \bra{z_\mrm{f}}e^{-Ht}\ket{z_\mrm{i}} = \int_{z(0)=z_\mrm{i}}^{\bar z(t)=z_\mrm{f}^*} \mrm{D}z\,\mrm{D}\bar{z}\, \exp\left( -S[z, \bar{z}] \right),
\end{equation}
where we obtain a Berry phase contribution, Hamiltonian contribution, and a boundary contribution as
\begin{equation}
	S[z, \bar{z}] = S_{\mrm{Berry}}[z, \bar{z}] + S_H[z, \bar{z}] + S_{\mrm{bd}}[z(t), \bar{z}(0)].
	\label{eq:S_total}
\end{equation}
For a generic spin-$\frac{1}{2}$ system, we always get
\begin{gather}
	S_{\mrm{Berry}}[z, \bar{z}] = -\frac{1}{2} \int_0^t \dd \tau \, \frac{\dot{\bar{z}}z - \bar{z}\dot{z}}{1 + \bar{z}z},\quad
	S_H[z, \bar{z}] = \int_0^t\dd\tau\bra{z}\!H\!\ket{z},\quad
	S_{\mrm{bd}}[z, \bar{z}] = -\frac{1}{2} \log \left[ \frac{(1 + z_\mrm{f}^* z(t))(1 + \bar{z}(0)z_\mrm{i})}{(1 + z_\mrm{f}^* z_\mrm{f})(1 + z_\mrm{i}^* z_\mrm{i})} \right].\label{eq:action_terms}
\end{gather}
For a quantum many-body system composed of spin-$\frac{1}{2}$ sites, the resolution in terms of spin-coherent states can be applied site by site, and we can define the continuum fields $z(x)$ and $\bar z(x)$, thus obtaining in the continuum limit:
\begin{gather}
	S_{\mrm{Berry}}[z, \bar{z}] = -\frac{1}{2} \int\dd^d x\int_0^t \dd \tau \, \frac{\dot{\bar{z}}z - \bar{z}\dot{z}}{1 + \bar{z}z},\quad
	S_{\mrm{bd}}[z(t), \bar{z}(0)] = -\frac{1}{2} \int\dd^d x\, \log \left[ \frac{(1 + z_\mrm{f}^* z(t))(1 + \bar{z}(0)z_\mrm{i})}{(1 + z_\mrm{f}^* z_\mrm{f})(1 + z_\mrm{i}^* z_\mrm{i})} \right],
\end{gather}
where we leave implicit the $x$ dependence of all the fields.
The Hamiltonian in a many-body system is also a sum of local terms, and we can take the continuum limit of $S_H[z, \bar z]$ to also write it as an integral over space.

\subsection{Subtleties on the Boundary Conditions of the Path Integral}
The prescription for the boundary conditions of the fields $z,\bar z$ in Eq.~\eqref{eq:S_total} deserves some clarifications.
In the usual derivation of the path integral action, one would expect that the boundary conditions for all fields are determined at both temporal boundaries, thus fixing ($z(0)$, $\bar{z}(0)$, $z(t)$, $\bar{z}(t)$).
However, since $S_\mrm{Berry}$ contains first-order time derivatives, semiclassical saddle-point trajectories will satisfy first-order differential equations in time, which, for times $\tau \notin \{0, t\}$ are of the form
\begin{equation}
	\delta S[z, \bar z] = 0 \;\;\;\implies\;\;\;\dot z= -(1+\bar z z)^2\frac{\delta H(z,\bar z)}{\delta \bar z},\quad \dot {\bar z}= (1+\bar z z)^2\frac{\delta H(z,\bar z)}{ \delta z},\;\;\;H(z,\bar z)\defn\bra z\! H\!\ket z.\label{eq:eqsofmotionzzbar}
\end{equation}
Fixing all the fields at both temporal boundaries would then result in an overdetermined system of differential equations.
This paradox was solved by Ref.~\cite{klauder79} as follows.
If one keeps $\delta t$ finite in the Trotterization procedure, the Berry phase term in the action will possess an additional term with second-order time derivatives, and a small coefficient of order $\epsilon = \mc O(\delta t)$.
The classical equations of motion then become second-order differential equations in time, allowing us to fix both temporal boundaries for all fields.
Doing so however reveals that the fields $z(\tau)$ and $\bar z(\tau)$ evolve according to Eqs.~(\ref{eq:eqsofmotionzzbar}) for most of the trajectory and perform rapid jumps at times $\tau\sim t-\epsilon$ and $\tau\sim \epsilon$ respectively, which in the $\delta t \rightarrow 0$ limit results in singular behaviors at $\tau = 0$ and $\tau = t$ for the $\bar z$ and $z$ fields respectively.
Hence two of the four boundary conditions can be discarded in the continuum limit, leaving us with the remaining two boundary conditions
\begin{equation}\label{eq:temp-bound-cond}
	z(\tau = 0) = z_\mrm{i} \quad \text{and} \quad \bar{z}(\tau = t) = z_\mrm{f}^*,
\end{equation}
which can then be imposed on Eq.~(\ref{eq:eqsofmotionzzbar}) to solve them consistently.
This then leads to the picture where $z$ and $\bar z$ are independent fields that ``originate'' from different temporal boundaries and propagate forward and backward in time respectively  (note that the signs of their derivatives in Eq.~(\ref{eq:eqsofmotionzzbar}) are opposite), while potentially interacting at other times away from the temporal boundaries.
This can then be solved, generically through numerical integration, as done for the Heisenberg model for specific initial and final conditions in Ref.~\cite{swann2025}, to obtain the solutions for the field configurations.
Note that these trajectories are in most cases unphysical, since there is no constraint for $\bar z(\tau)$ to equal $z(\tau)^*$, which in the language of $\vec{n}$ would mean that the angles $\theta$, $\varphi$ take unphysical complex values.
They should simply be seen as virtual solutions that correspond to the stationary value of the action.
It is also interesting to notice that all this is consistent with also setting $\delta S = 0$ for $\tau \in\{0, t\}$ since the variation of boundary action $S_\mrm{bd}$ perfectly cancels the term originating from integration by parts in the variation of $S_{\mrm{Berry}}$, if one assumes that $\delta z(0)=\delta \bar z(t)=0$:
\begin{gather}\label{eq:dberry-app}
	\delta S_{\mrm{Berry}} = \int \dd^d x \int_0^t \dd \tau \, \frac{\dot{z}\delta\bar{z} - \dot{\bar{z}}\delta z}{(1 + \bar{z}z)^2} - \frac{1}{2} \int \dd^d x \left[ \frac{z(\tau)\delta\bar{z}(\tau) - \bar{z}(\tau)\delta z(\tau)}{1 + \bar{z}z} \right]_0^t,\\
	\delta S_{\mrm{bd}} = -\frac{1}{2}\int \dd^d x \left[ \frac{z_f^* \delta z(t)}{1 + z_f^* z(t)} + \frac{z_i \delta\bar{z}(0)}{1 + \bar{z}(0) z_i} \right].
\end{gather}
\subsection{Example: A single spin under $H = J \sigma^z$}
To illustrate the crucial role of the boundary conditions and the boundary action $S_\mrm{bd}$, we  consider a single spin governed by the Hamiltonian $H = J \sigma^z$.
In terms of coherent states, the Hamiltonian term and the variation of energy action read
\begin{equation}
	\bra z \!H \!\ket z = -J \frac{1 - \bar{z}z}{1 + \bar{z}z},\;\;\;\delta S_H = \int_0^t \dd \tau \, \frac{2J \left(z\delta\bar{z} + \bar{z}\delta z\right)}{(1 + \bar{z}z)^2}.
\end{equation}
Combining this with $\delta S_{\mrm{Berry}}$ of Eq.~(\ref{eq:dberry-app}) with $d = 0$, we get the  Euler-Lagrange equations decouple into simple linear ODEs:
\begin{equation}\label{eq:simple-z-zbar}
	\dot{z} = -2 J z \quad \text{and} \quad \dot{\bar{z}} = +2 J \bar{z}\;\;\;\implies\;\;\;z(\tau) = z_i e^{-2J\tau} \quad \text{and} \quad \bar{z}(\tau) = z_f^* e^{-2J(t - \tau)},
\end{equation}
where in the second step we have integrated the equations by imposing the boundary conditions $z(0) = z_i$ and $\bar{z}(t) = z_f^*$:
Noting that the product of the fields is uniformly constant throughout the entire bulk: $\bar{z}(\tau)z(\tau) = z_f^* z_i e^{-2J t} \equiv \gamma e^{-2J t}$, where $\gamma = z_f^* z_i$, we can evaluate each component of the total action on this classical saddle-point trajectory analytically:
\begin{gather}
	S_{\mrm{Berry}} = -\int_0^t \dd \tau \, \frac{2J \bar{z}z}{1 + \bar{z}z} = -2J t \frac{\gamma e^{-2J t}}{1 + \gamma e^{-2J t}}, \;\;\;S_H = -\int_0^t \dd \tau \, J \frac{1 - \bar{z}z}{1 + \bar{z}z} = -J t \frac{1 - \gamma e^{-2J t}}{1 + \gamma e^{-2J t}}, \\
	S_{\mrm{bd}} = -\log\left( \frac{1 + \gamma e^{-2J t}}{\sqrt{1+z_\mrm{i}^*z_\mrm{i}}\sqrt{1+z_\mrm{f}^*z_\mrm{f}}} \right).
\end{gather}
Summing these parts, the total semi-classical action evaluates to:
\begin{equation}
	S_{\mrm{total}} = -J t - \log\left( \frac{1 + \gamma e^{-2J t}}{\sqrt{1+z_\mrm{i}^*z_\mrm{i}}\sqrt{1+z_\mrm{f}^*z_\mrm{f}}} \right)
    \;\;\implies\;\;
    e^{-S_{\rm total}} = \frac{e^{Jt} + \gamma e^{-J t}}{\sqrt{1+z_\mrm{i}^*z_\mrm{i}}\sqrt{1+z_\mrm{f}^*z_\mrm{f}}} = \bra{z_\mrm{f}} e^{-Ht} \ket{z_\mrm{i}}.
\end{equation}
This demonstrates that under imaginary-time evolution, the fields $z(\tau)$ and $\bar{z}(\tau)$ independently grow and decay exponentially away from their respective boundaries to satisfy the path integral.
The boundary term $S_{\mrm{bd}}$ provides a crucial normalization factor, necessary for the path integral to correctly reproduce the transition amplitude $\bra{z_\mrm{f}} e^{-J \sigma^z t} \ket{z_\mrm{i}}$.
\subsection{Path Integral Formulation of the Ferromagnetic Heisenberg Model} % MARK: * p-i
For the Heisenberg model of Eq.~\eqref{eq:HeisModelHam} in the continuum limit, the energy functional $H[z, \bar z]$ is~\cite{kochetov95heis}
\begin{equation}\label{eq:heis-cont-ham-z}
	H[z,\bar z] = 2J\int\dd^d x\, \frac{\bm\nabla \bar z\cdot\bm\nabla z}{(1+\bar z z)^2}.
\end{equation}
This functional form is equivalent to the expression Eq.~(\ref{eq:cont-heis-mtxt}) [derived later in Eq.~\eqref{eq:heis-ham-cont-lim}], under the mapping of Eq.~\eqref{eq:coherentstatecomplex}.
Given that this Hamiltonian term has a second derivative in space, and that the Berry phase term in Eq.~(\ref{eq:action_terms}) has a first derivative in time, in an infinite system we can always rescale to dimensionless coordinates $(x,\tau)\mapsto (\tilde x = x/\sqrt{Jt},\,\tilde\tau = \tau/t)$, such that a large prefactor appears in front of the action $S\mapsto \sqrt{Jt}\cdot \tilde S$~\cite{swann2025}.
Note that this scaling does not depend on the $SU(2)$ symmetry of the current model, and it generalizes to homogeneous ferromagnets with larger symmetry groups.
This justifies performing a saddle-point approximation at large $t$ to obtain a semiclassical trajectory at long times, and this approach has been employed in previous works \cite{swann2025, swann2026continuummechanicsentanglementnoisy} in the context of free-fermion or weakly interacting noisy dynamics, where the effective Hamiltonian can be mapped onto the ferromagnetic Heisenberg model (within a certain sector).
The classical equations of motion are
\begin{gather}\label{eq:eoms-z}
    \dot z = +2J\left(\bm\nabla^2 z - \frac{2\bar z(\bm\nabla z)^2}{1+\bar z z}\right),\qquad    \dot {\bar z} = -2J\left(\bm\nabla^2 \bar z - \frac{2 z(\bm\nabla \bar z)^2}{1+\bar z z}\right).
\end{gather}
As before, these can be interpreted as evolving $z(x,\tau)$ and $\bar z(x,\tau)$ respectively forwards and backwards in time, and since $\dot z\neq (\dot{\bar z})^*$, these equations of motion always evolve the field into unphysical configurations.
Note that these equations of motion couple the $z$ and $\bar z$ field, unlike the case of a single spin whose equations of motion are Eq.~\eqref{eq:simple-z-zbar}.
The boundary value problem for semiclassical trajectories in the Heisenberg model is in general analytically untractable, since these $z$ and $\bar z$ fields originate on different boundaries and interact in the bulk. One generally has to resort to numerics to obtain a self-consistent solution for the configuration of the fields in the bulk.
\subsection{Generalization to Matrix Coherent States and Coset Manifolds}\label{subsec:matrixcoherent}
The construction detailed above for $SU(2)$ spin-$\frac{1}{2}$ coherent states can be extended to broader classes of quantum systems where the manifold of states is as a homogeneous coset space $G/H$ \cite{generalized-coher-st,klauder1985coherent} for a Lie group $G$ and a closed subgroup $H$.
For instance, in the case of $U(1)$-conserving free fermion systems, the ground state space for the replica system is smoothly described by the complex Grassmannian manifold $\mrm{Gr}(M, N) = U(N)/[U(M) \times U(N-M)]$ \cite{lastres2026geometryfreefermioncommutants}.
In such cases, one can define generalized coherent states similar to the ones above, where the complex fields $z$ are replaced by matrix-valued fields $Z$, and the denominators $(1+\bar{z}z)$ generalize to determinant expressions of the form $\det(\1+ Z^\dagger Z)$~\cite{perelomov1986generalized,Kochetov1995,Berezin:1978sn}.
These coset manifolds are smooth homogeneous spaces, and do not contain any singularities.
However, to our knowledge, analogous constructions are not available for manifolds of the kind we are interested in, which are complex projective varieties with singularities.

\section{Review of the Time Dependent Variational Principle} % MARK: APP: TDVP
\label{app:tdvp}
In this Appendix, we formulate and discuss the Time-Dependent Variational Principle (TDVP) for continuum Hamiltonians acting on the space of variational product space parametrized as $v(x)\in M\subseteq\mb P(\mc H)$.
Following the geometric framework formalized in Ref.~\cite{tdvp20}, we outline how quantum dynamics can be projected onto a variational subspace, construct the continuum expression for the Heisenberg model Hamiltonian, and contrast the resulting physical state trajectories with the virtual semiclassical trajectories of the path integral.
\subsection{Geometric Formulation on Kähler Manifolds}
Let $M\subseteq \mb P(\mc H)$ be a variational manifold of quantum states $\ket{\psi(\theta)}$ parameterized by a set of real coordinates $\theta^\alpha$.
These states are identified up to normalization and phase factors, so that $\ket{\psi(\theta)}\sim \lambda\ket{\psi(\theta)}$ for $\lambda \in \mathbb{C}$.
When $M$ is a Kähler manifold of projective states [e.g., when it is holomorphic up to a normalization factor, such as the coherent state manifold of Eq.~(\ref{eq:coherentstatecomplex})], it naturally inherits three compatible geometric structures from the full Hilbert space: a Riemannian metric $g_{\alpha\beta}(\theta)$, a symplectic two-form $\omega_{\alpha\beta}(\theta)$, and a complex structure ${J^\alpha}_\beta(\theta)$.
These tensors are tied together by the fundamental relations:
\begin{equation}
	{J^\alpha}_\gamma {J^\gamma}_\beta = -\delta^\alpha_\beta, \quad \text{and} \quad \omega_{\alpha\beta} = - g_{\alpha\gamma} {J^\gamma}_\beta.
\end{equation}
The metric and symplectic forms are explicitly given by the real and imaginary components of the overlap of quantum vectors $\{\ket{V_\alpha}\}$ tangent to $M$~\cite{tdvp20}:
\begin{equation}\label{eq:tdvp-tensors-def}
	g_{\alpha\beta}(\theta) = 2\frac{\Re \braket{V_\alpha}{V_\beta}}{\braket{\psi}}, \quad
	\omega_{\alpha\beta}(\theta) = 2 \frac{\Im  \braket{V_\alpha}{V_\beta}}{\braket{\psi}},\;\;\;\ket{V_\alpha} \defn \left(1 - \frac{\ketbra{\psi}{\psi}}{\braket{\psi}{\psi}}\right) \frac{\partial}{\partial \theta^\alpha}\ket{\psi(\theta)},
\end{equation}
where the projection from $\partial_{\theta^\alpha}\ket{\psi(\theta)}$ to $\ket{V_\alpha}$ is only needed to make sure that $\ket{V_\alpha}$ is tangent to $\mb P(\mc H)$ (i.e., orthogonal to $\ket{\psi(\theta)}$) independently of the chosen representative $\ket{\psi(\theta)}\in \mc H$.
Notice that an alternate natural parametrization of projective states comes from pure density matrices, defined as
\begin{equation}
\rho=\frac{\ketbra{\psi(\theta)}{\psi(\theta)}}{\braket{\psi(\theta)}},
\end{equation}
since they are invariant under rescalings $\ket{\psi(\theta)}\mapsto {\lambda}\ket{\psi(\theta)}$ for $\lambda \in \mathbb{C}$.
When parametrized by such normalized pure density matrices, the metric in Eq.~\eqref{eq:tdvp-tensors-def} is exactly the restriction of the matrix Frobenius metric to the variational manifold of states
\begin{equation}\label{eq:matrix-metric}
    g_{\alpha\beta}(\theta)=\tr(\frac{\partial\rho}{\partial\theta^\alpha}  \frac{\partial\rho}{\partial\theta^\beta}).
\end{equation}
We will mostly work with this density matrix representation in the examples illustrated in the following.
Note that in the following we will sometime abbreviate the notation $\partial_\alpha\defn\partial_{\theta^\alpha}$.
The projection of the full Schrödinger equation onto the tangent space $T_\psi M$ yields deterministic differential equations for the coordinates $\theta^\alpha(t)$.
Physically, for real-time evolution, the state evolves along a symplectic vector field that follows the energy contours on the variational manifold, whereas for (normalized) imaginary-time evolution the state performs a steepest gradient descent perpendicular to the energy contours on the manifold, thus minimizing its energy.
The respective equations of motion can be compactly expressed as
\begin{equation}
	\omega_{\alpha\beta} \dot{\theta}^\beta = - \frac{\partial E}{\partial \theta^\alpha}\;\;\text{(real-time)},\;\;\;
	g_{\alpha\beta} \dot{\theta}^\beta = - \frac{\partial E}{\partial \theta^\alpha}\;\;\text{(imaginary-time)},\;\;\;
	E(\theta) \defn \frac{\bra{\psi(\theta)}\!H\!\ket{\psi(\theta)}}{\braket{\psi(\theta)}{\psi(\theta)}}.
	\label{eq:tdvp}
\end{equation}
It is then also straightforward using the properties of $g_{\alpha\beta}$ and $\omega_{\alpha\beta}$ to show that the time-derivative of the energy satisfies~\cite{tdvp20}
\begin{equation}
    \dot E(\theta) = 0\;\;\text{(real-time)},\;\;\;\dot E(\theta) \leq 0\;\;\text{(imaginary-time)},
\end{equation}
consistent with expectations.
These equations capture the evolution of the projective states, which do not contain the information of the norm.
The decay of the norm under imaginary time evolution can be captured separately by using the relation
\begin{equation}\label{eq:en-dec-app}
	\dv{t}\braket{\psi(t)} = \bra{\psi}\dv{t}e^{-2Ht} \ket{\psi} = -2 \bra{\psi(t)} H \ket{\psi (t)} \;\;\implies\;\; \dv{t} \log\braket{\psi(\theta)} = -2E(\theta).
\end{equation}
However, this is not the only way to compute the decay of the norm, and we discuss alternative approaches in Sec.~\ref{subsec:tdvp-main}.

By denoting the inverse of the symplectic form $\Omega^{\alpha\beta}$ and the inverse metric $G^{\alpha\beta}$, the connection between the real- and imaginary-time trajectories can be seen to be a tangent-space Wick rotation mediated by the complex structure ${J^\alpha}_\beta$ on the manifold:
\begin{equation}
	\dot{\theta}^\alpha_{\mrm{imag}} = - G^{\alpha\beta} \frac{\partial E}{\partial \theta^\beta} = - \left( -{J^\alpha}_\gamma \Omega^{\gamma\beta} \right) \frac{\partial E}{\partial \theta^\beta} =  -{J^\alpha}_\gamma \dot{\theta}^\alpha_{\mrm{real}}.
\end{equation}
This shows that imaginary-time TDVP is simply a rotation of the real-time evolution within the tangent space using the intrinsic complex structure of the variational manifold, again consistent with the propagation of the state vector along and perpendicular to the energy contours in the different cases.
While the TDVP formalism is commonly presented~\cite{tdvp20} in terms of intrinsic coordinates $\alpha\beta\cdots$ of the variational manifold $M$ (such as spherical coordinates $\theta, \varphi$ or projective coordinates $z$ for $M = \mathbb{S}^2$), for explicit applications it is sometimes more convenient to parametrize points on $M$ in terms of extrinsic coordinates $AB\cdots$ of a larger embedding space (such as $x,y,z$ coordinates of $\mathbb{R}^3$ where $M = \mathbb{S}^2$ is embedded in).
If the embedding space is endowed with a metric $g_{AB}$ and/or a symplectic form $\omega_{AB}$ whose restrictions to the variational manifold exactly give the metrics $g_{\alpha\beta}$ and/or symplectic form $\omega_{\alpha\beta}$ of Eq.~\eqref{eq:tdvp-tensors-def}, then the equations of motion Eq.~\eqref{eq:tdvp} will take the similar form
\begin{equation}
	\omega_{AB} \dot{\theta}^B = - {\Pi^C\!}_A\frac{\partial E}{\partial \theta^C}\;\;\text{(real-time)},\;\;\;
	g_{AB} \dot{\theta}^B = - {\Pi^C\!}_A\frac{\partial E}{\partial \theta^C}\;\;\text{(imaginary-time)},
	\label{eq:tdvp-embed}
\end{equation}
where ${\Pi^C\!}_A$ is the projector onto the tangent space of $M$, which constrains the evolution to this variational set.
Everywhere below, we distinguish intrinsic and extrinsic coordinates by using Greek lower case $\alpha\beta \cdots$ for the former, and Roman upper case $A B \cdots$ for the latter.
\subsection{Application to spatially varying fields}
While Eq.~\eqref{eq:tdvp} is the general expression for TDVP on a variational manifold of finite dimension, we are interested in applying this framework to the case of a variational family of spatially varying fields $v(x)$ which take values on a manifold $M$.
It is convenient to discretize the spatial direction and consider product states on a many-body system where states on each site are parametrized by $M$.
Then the whole family will be parametrized by the product manifold $M^{\times L}$, and the variational parameters will be $\{v_i^\alpha\}$.
On this product manifold, the TDVP symplectic form and metric of Eq.~\eqref{eq:tdvp-tensors-def} simply factorize along the spatial direction:
\begin{equation}
	\omega_{(\alpha,i)(\beta,j)} = \omega_{\alpha\beta}\delta_{ij},\quad g_{(\alpha,i)(\beta,j)} = g_{\alpha\beta}\delta_{ij},
\end{equation}
where $\alpha$ and $\beta$ label the directions on the manifold and $i$ and $j$ the site indices.
In such a case, the equations of motion of Eq.~\eqref{eq:tdvp} factorize into separate equations for $v_i$ on each site, with only the energy $E$ that couples different sites.
In the continuum limit, we can also replace the gradient of energy with a functional derivative as
\begin{equation}\label{eq:functional-der-a}
    \frac{\partial E}{\partial v^\alpha_i}
    \rt
    a^d\frac{\delta E}{\delta v^\alpha(x)},
\end{equation}
to obtain the equations of motion
\begin{equation}
\omega_{\alpha\beta} \dot{v}^\beta = - \frac{\delta E}{\delta v^\alpha}\;\;\text{(real-time)},\;\;\;
	g_{\alpha\beta} \dot{v}^\beta = - \frac{\delta E}{\delta v^\alpha}\;\;\text{(imaginary-time)},
\end{equation}
where $E$ now is the global energy functional that depends on the field configuration $v(x)$, and the lattice-spacing dependence is absorbed in the space and time variables or the constants in the Hamiltonian, as we will show in the example of the Heisenberg model.
\subsection{Semiclassical Trajectories in TDVP vs Path Integral}\label{app:tdvp-vs-pi}
We remark that the real time TDVP equation in Eq.~\eqref{eq:tdvp} can be derived as the Euler-Lagrange equation of the Lagrangian:
\begin{equation}
	L(\theta,\partial_t{\theta};t) = \mrm{Re}\left[\frac{\bra{\psi(\theta)} i\partial_t-H\ket{\psi(\theta)}}{\braket{\psi(\theta)}}\right]
\end{equation}
If we assume that the coordinates $\theta^\alpha$ are holomorphic,\footnote{This means that they can be paired up to form complex variables $z^{\rho}=\theta^{2\rho}+i\theta^{2\rho+1}$ which are the actual holomorphic coordinates, which satisfy ${J^\rho}_\sigma= i$} then the imaginary time TDVP equations in Eq.~\eqref{eq:tdvp} can be derived from the Wick-rotated Lagrangian (where one replaces $t\mapsto -i t$).
Note that in both cases, these Lagrangians match exactly the ones associated to the spin-coherent path integral discussed in the previous section [Eq.~\eqref{eq:action_terms}], and so they will result in the same equations of motion for the holomorphic coordinates.
Indeed, as shown in the Heisenberg ferromagnet example in Appendix~\ref{app:path-integral}, the equations of motion of Eq.~\eqref{eq:eoms-z} for the holomorphic coordinate $z(x,\tau)\in\mb{CP}^1$ are the same as the TDVP ones, while those for the anti-holomorphic one are not.
However, the trajectories in TDVP should not be confused with the semiclassical trajectories used in the saddle-point approximation of the path integral:
\begin{itemize}
	\item TDVP solves an \textit{initial-value problem} describing the trajectory of a state within a variational manifold. Since at every time the trajectory describes a physical state, there is no need to treat anti-holomorphic parts such as $\bar z$ as independent variables.
	\item The saddle-point trajectory for the path integral solves a \textit{boundary-value problem} to the find transition amplitude between two states, and does not directly correspond to a physical trajectory. The solutions may need to leave the manifold of physical states in order to accommodate the boundary conditions. 
    This is particularly the case in imaginary time, where, as discussed in Appendix~\ref{app:path-integral} the equation of motion for $\dot {\bar z}$ is not equal to $(\dot z)^*$, immediately leading to a loss of physicality for the state. 
\end{itemize}
\subsection{Continuum Evolution for the Heisenberg Model}
We now specialize this framework to the continuum limit of the spin-$\frac{1}{2}$ ferromagnetic Heisenberg model [Eq.~\eqref{eq:HeisModelHam}], similar to Ref.~\cite{JohnSchliemann_1998}.
Our variational manifold consists of a continuous tensor product of independent single-site coherent states $\ket{\vec n_i}$ as defined in Eq.~\eqref{eq:spin-coh-states}. The local density matrix representation of this state is:
\begin{equation}
	\rho(\vec{n}_i) = \ketbra{\vec n_i}{\vec n_i}=\frac{1}{2} \left( \1 + \vec{n}_i \cdot \vec{\sigma} \right).
\end{equation}
This is an overparametrization of the sphere $\mb S^2$ defined by $|\vec n|=1$, and vectors $\delta \vec n$ tangent to this manifold must satisfy $\vec n\cdot\delta\vec n=0$.
In these coordinates the metric tensor $g_{\alpha\beta}$ [Eq.~\eqref{eq:tdvp-tensors-def}] on the tangent plane of the sphere $\mb{S}^2$ is the restriction of the Frobenius metric for the matrices $\rho(\vec{n})$ [cf.~Eq.~\eqref{eq:matrix-metric}], which is proportional to the Euclidean metric for $\vec n\in\mb R^3$:
\begin{equation}
	g_{\alpha\beta} = (g_{AB})\big|_\mrm{tangent}\qquad g_{AB} = \tr \left( \frac{\partial \rho}{\partial n^A} \frac{\partial \rho}{\partial n^B} \right) = \frac{1}{2} \delta_{AB}.
\label{eq:metricheisenberg}
\end{equation}

To derive the equations of motion, we evaluate the energy functional for the ferromagnetic Heisenberg model with small lattice spacing $a$. The microscopic interaction between adjacent sites reads $h_{ij} = 2J ( \1 - \Sigma_{ij} )$ where $\Sigma_{ij}$ is the swap operator. Its expectation value is:
\begin{equation}\label{eq:meanh}
	\langle h_{ij} \rangle = 2J \tr [ ( \1 - \Sigma_{ij} ) \rho ]
	= 2J (1 - \tr[\rho(\vec n_i)\rho(\vec n_j)])
	= J \left( 1 - \vec{n}_i \cdot \vec{n}_{j} \right).
\end{equation}
For smooth configurations, we expand the field in space by taking
\begin{equation}
    \vec{n}_j = \vec{n}_i +  {a^\mu\partial_\mu} \vec{n}_i + \frac{1}{2} {a^\mu a^\nu \partial_\mu\partial_\nu} \vec{n}_i+\mc O(a^3),    
\end{equation}
where repeated indices are summed over, $\langle ij\rangle$ are adjacent sites and $\bm a = \bm x_j-\bm x_i$ (with components $a^\mu$ with $|a^\mu| \sim a$) is the spatial vector which separates the two sites (in this derivation we will use bold symbols and indices $\mu\nu...$ to indicate spatial vectors).
Substituting this expansion into the expectation value we get:
\begin{equation}
	1 - \vec{n}_i \cdot \vec{n}_j = - a^{\mu}\vec{n}_i \cdot \partial_\mu \vec{n}_i - \frac{a^\mu a^\nu}{2} \vec{n}_i \cdot \partial_\mu \partial_\nu \vec{n}_i +\mc O(a^3),
\end{equation}
Differentiating the identity $\vec{n} \cdot \vec{n} = 1$, one finds that $\vec{n} \cdot \partial_\mu \vec{n} = 0$ and $\vec{n} \cdot \partial_\mu \partial_\nu \vec{n} = -\partial_\mu \vec{n} \cdot \partial_\nu \vec{n}$.
Assuming for simplicity a hypercubic lattice with lattice spacing $a$ in every direction, the leading-order continuum expression for energy simplifies to:
\begin{equation}
	E[\vec{n}] = \frac{J}{2}\sum_{i}{a^\mu a^\nu \partial_\mu \vec{n}_i \cdot \partial_\nu \vec{n}_i} = \frac{a^{2} J}{2} \int \dd^d x \, |\bm\nabla \vec{n}|^2,\label{eq:heis-ham-cont-lim}
\end{equation}
where in the second step we have defined fields $\vec{n}(\bm x)$ from the bare lattice fields $\vec{n}_i$ as $\vec{n}(\bm x_i) \defn a^{-\frac{d}{2}}\vec{n}_i$ which enables the conversion of the sum into an integral.
For regular lattices with different coordination number, only the numerical prefactor of the Hamiltonian changes.
To apply the TDVP equations of motion, we first take the unconstrained functional derivative of the energy with respect to the coordinates, obtaining $\frac{\delta E}{\delta \vec{n}} = -a^{2} J \bm \nabla^2 \vec{n}$.
To enforce the $|\vec{n}|^2 = 1$ constraint, this gradient must be projected onto the tangent space of the sphere using the projector $\Pi_{AB}(\vec{n}) = \delta_{AB} - n_A n_B$ [cf.~Eq.~\eqref{eq:tdvp-embed}]:
\begin{equation}\label{eq:heis-var-proj}
	\left. \frac{\delta E}{\delta n^A} \right|_{\mrm{tangent}} = -a^{2} J \left( \delta_{AB} - n_A n_B \right) \bm \nabla^2 n^B = -a^{2} J \left( \bm \nabla^2 n_A + |\bm \nabla \vec{n}|^2 n_A \right).
\end{equation}
Substituting this into Eq.~(\ref{eq:tdvp-embed}), with $g_{AB} = \frac{1}{2}\delta_{AB}$, we arrive at the diffusion equation with $\vec n$ constrained to the surface of the unit sphere:
\begin{equation}
	\partial_t \vec{n} = 2 a^2J \left( \bm \nabla^2 \vec{n} + |\bm \nabla \vec{n}|^2 \vec{n} \right).
	\label{eq:tdvp_heis_final}
\end{equation}
where $|\cdot|$ indicates the Euclidean norm. We can effectively set $a=1$ by replacing $t\mapsto t_\mrm{ph}=a^2t$, a rescaling that highlights the diffusive nature of the dynamics.
This can alternatively be done by setting $J\mapsto J_\mrm{ph} = Ja^2$, which is kept constant as $a \rightarrow 0$. 
\subsection{Quantifying Error via Manifold Leakage}
\label{app:leakage}
One way to assess the validity of the semi-classical continuum limit, is to measure the ``leakage'' rate $\epsilon^2$ out of the variational manifold.
Let $\Pi_{\psi}$ be the projector onto the tangent space of the manifold at the point $\ket\psi$; then the imaginary-time TDVP equations of Eq.~\eqref{eq:tdvp} are equivalent to \cite{tdvp20}:
\begin{equation}
	\partial_t \ket\psi \big|_{\mrm{TDVP}} = -\Pi_\psi H\ket \psi.
\end{equation}
This evolution approximates the exact normalized imaginary-time evolution:
\begin{equation}
	\partial_t \ket\psi = -(H-E_\psi)\ket \psi\quad\text{where}\quad E_\psi=\bra\psi H\ket\psi.
\end{equation}
The manifold leakage is then defined as the squared norm of the
\begin{equation}
	\epsilon^2_\psi = \Big\lVert \partial_t \ket\psi- \partial_t \ket\psi \big|_{\mrm{TDVP}} \Big\rVert^2 = \bra\psi\Big((H-E_\psi)^2-H\Pi_\psi H\Big)\ket\psi = -\frac{1}{2} \Big(\partial_t E_\psi-\partial_t E_\psi \big|_{\mrm{TDVP}}\Big)
\end{equation}
where we used the fact that $\Pi_\psi\ket\psi=0$. Also note that in the expression $\partial_t E_\psi = -2\bra\psi(H-E_\psi)^2\ket\psi=-2\mrm{Var}_\psi[E]$.
We can derive these two quantities explicitly for the Heisenberg model in the continuum limit.
For a product state $\ket\psi = \bigotimes_i \ket{\vec{n}_i}$, the energy variance decomposes as:
\begin{equation}
	\mrm{Var}_\psi[E] = \frac{1}{2}\sum_i \sum_{jj'\in\partial i}\Big(\langle h_{ij} h_{ij'}\rangle-\langle h_{ij}\rangle \langle h_{ij'}\rangle\Big),
\end{equation}
where $\partial i$ indicates the neighbors of $i$. When $j=j'$, the expression in the sum is the variance of $h_{ij}$, which can be computed using Eq.~\eqref{eq:meanh} and the fact that $h_{ij}^2=4Jh_{ij}$.
When $j\neq j'$ we can instead compute:
\begin{equation}
	\langle \{ h_{ij} , h_{ij'}\}\rangle-2\langle h_{ij}\rangle \langle h_{ij'}\rangle=2J^2(\vec n_j\cdot\vec n_{j'}-(\vec n_i\cdot\vec n_j)(\vec n_i\cdot\vec n_{j'})).
\end{equation}
In total we find:
\begin{equation}
	\mrm{Var}_\psi[E]\approx a^{4} J^2\int\dd^d x \,\Big(|\bm\nabla^2\vec n|^2 - \frac{3}{4}|\bm\nabla\vec n|^4 \Big).
\end{equation}
The time derivative of energy along the TDVP evolution can instead be computed directly from Eq.~\eqref{eq:tdvp_heis_final} to be:
\begin{equation}
	\frac{1}{2}\partial_t E_\psi \big|_{\mrm{TDVP}} = -a^{4} J^2 \int\dd^d x \,\Big(|\bm\nabla^2\vec n|^2 - |\bm\nabla\vec n|^4 \Big),
\end{equation}
so that finally:
\begin{equation}
	\epsilon^2_\psi = \frac{a^{4} J^2}{4}\int\dd^d x\, |\bm\nabla\vec n|^4.
\end{equation}
While it might appear that $\epsilon_\psi$ goes to zero in the continuum limit, this should be integrated over timescales of the order of $t\sim t_\mrm{ph}/a^2$, so that this error is in general not negligible.
Alternately, if $J_\mrm{ph} = J a^2$ is finite in the continuum limit, this leakage is generically non-zero.
However, since the error goes to zero as $|\bm\nabla\vec n|^4$, and since long-time imaginary dynamic suppresses gradients, we expect TDVP predicitons to become more accurate with time, once the higher-energy modes have decayed.
\subsection{Continuum Expansion of Effective Hamiltonians} % MARK: * continuum
In this section, we provide a general geometric derivation for the leading order of the continuum limit of effective Hamiltonians with continuous symmetries. We first treat the ferromagnetic case associated to standard on-site symmetries as discussed in Sec.~\ref{sec:fgsmanrep}.
\subsubsection{General Structure of Continuous Ferromagnetic Hamiltonians}\label{app:cont-exp}
Consider a ferromagnetic model [Eq.~\eqref{eq:def-ferro}] with a continuous ground state manifold $M_\mrm{GS}$.
In Sec.~\ref{sec:tdvp} we argued that under imaginary time evolution in the continuum limit, its dynamics can generally be captured by semiclassical evolution of a field $v(x)$ which takes value in the target manifold $M_\mrm{GS}$.
Note that since all physical quantities are invariant under the local phase transformation $\ket{v(x)}\mapsto e^{i\varphi(x)}\ket{v(x)}$, the terms in the expansion of the Hamiltonian can only depend on $\rho(x)=\ketbra{v(x)}{v(x)}$ and its derivatives.
The expansion might also include additional constant structure tensors, beyond the ones derived directly from the local density matrix (e.g., terms of the form $\tr(W_1\partial_x^n\rho W_2\partial_x^{m}\rho \cdots)$ where $W_1,W_2, \cdots$ are constant operators).
In terms of intrinsic coordinates on $M_{GS}$, a general derivative expansion for the Hamiltonian is
\begin{equation}
	H[v] = \int\dd^d x\, \Big( V( v(x)) + u^\mu_\alpha( v(x)) \partial_\mu v^\alpha + K_{\alpha\beta}^{\mu\nu}( v(x)) \partial_\mu  v^\alpha  \partial_\nu  v^\beta + \dots \Big),
\label{eq:Hexpansion}
\end{equation}
where $\mu\nu \cdots$ are spatial indices.
The assumption that the model is ferromagnetic means that $v(x)$ is a ground state if and only if $v(x)= v_0$ for some constant $v_0\in M_\mrm{GS}$.
This condition forces $V(v)=0$ for all $v\in M_\mrm{GS}$. It also forces $u^\mu_\alpha( v)=0$, since otherwise it would be possible to locally perturb a constant configuration $v(x)_\alpha\approx ( v_0)_\alpha-u^\mu_\alpha( v_0)x_\mu$ to reduce the energy density below zero.
Since for effective replica models we assume the form $P^{(k)}=\sum_{\langle ij\rangle}P_{ij}^{(k)}$ where all $P_{ij}^{(k)}$ are equal on a regular lattice, we will also assume the continuum limit to be spatially isotropic, so that $K_{\alpha\beta}^{\mu\nu}( v)=\delta^{\mu\nu} K_{\alpha\beta}(v)$.
Furthermore this structure, together with the fact that states of the form $\ket{v}_i\ket{v+\delta v}_j$ cannot locally be ground states, implies that the second order of the Hamiltonian expansion is positive and non-zero:
\begin{equation}\label{eq:sketch-expansion}
	0 < \bra{v}\bra{v+\delta v} P_{ij}^{(k)} \ket{v}\ket{v+\delta v} = \mc O(|\delta v|^2)
\end{equation}
so that $K_{\alpha\beta}(v)$ is a positive Riemannian metric on $M_\mrm{GS}$.
Therefore,
\begin{equation}
	H[ v] = \int\dd^d x\, \Big( K_{\alpha\beta}( v(x)) \bm\nabla  v^\alpha  \cdot\bm\nabla  v^\beta + \dots \Big),\label{eq:exp-generic}
\end{equation}
where we have represented spatial derivatives compactly using the gradient vector $\bm\nabla$ to avoid confusion with derivatives w.r.t. intrinsic coordinates of the ground state manifold.
From this we can apply Eq.~\eqref{eq:tdvp-continuum-main} to obtain the TDVP equations of motion. The variation of the energy is
\begin{multline}\label{eq:grosso-passaggio}
    \delta H = \int \dd^d x \left[ \frac{\partial K_{\alpha\beta}}{\partial v^\rho} \delta v^\rho \bm\nabla v^\alpha \cdot \bm\nabla v^\beta + 2 K_{\rho\beta} \bm\nabla (\delta v^\rho) \cdot \bm\nabla v^\beta \right]=\\
    =\int \dd^d x\,\delta v^\rho \left[ \frac{\partial K_{\alpha\beta}}{\partial v^\rho} \bm\nabla v^\alpha \cdot \bm\nabla v^\beta - 2 \left( K_{\rho\beta} \bm\nabla^2 v^\beta + \frac{\partial K_{\rho\beta}}{\partial v^\alpha} \bm\nabla v^\alpha \cdot \bm\nabla v^\beta \right) \right],
\end{multline}
and symmetrizing $2\frac{\partial K_{\rho\beta}}{\partial v^\alpha} \bm\nabla v^\alpha \cdot \bm\nabla v^\beta = \big(\frac{\partial K_{\rho\beta}}{\partial v^\alpha}+\frac{\partial K_{\rho\alpha}}{\partial v^\beta}\big)\bm\nabla v^\alpha \cdot \bm\nabla v^\beta$, we find:
\begin{equation}
    \partial_t v^\gamma = 2 G^{\gamma\rho} \Big(K_{\rho\delta} \bm\nabla^2 v^\delta + C_{\rho\alpha\beta} \bm\nabla v^\alpha \cdot \bm\nabla v^\beta\Big),\quad\text{where}\quad C_{\rho\alpha\beta} = \frac{1}{2} \left( \frac{\partial K_{\rho\beta}}{\partial v^\alpha} + \frac{\partial K_{\rho\alpha}}{\partial v^\beta} - \frac{\partial K_{\alpha\beta}}{\partial v^\rho} \right),
\label{eq:TDVPeqnsfinal}
\end{equation}
and $G^{\gamma\rho}$ is the inverse of the manifold metric $g_{\gamma\rho}$.
In some cases, the metric $K_{\alpha\beta}(v)$ can be uniquely identified.
If for example $M_\mrm{GS}\cong G/H$ is a homogeneous manifold, and the unbroken symmetries $H$ of the Hamiltonian (i.e., the symmetry operators that preserve a given ground state $v(x)= v_0$) act irreducibly on the tangent space of $M_\mrm{GS}$, then Schur's lemma implies that all $G$-invariant metrics are equal up to rescaling~\cite{kerin2003homogeneous}.
In our case, we can explicitly construct this unique metric using the derivative expansion in terms of $\rho(x)=\ketbra{v(x)}{v(x)}$; indeed to leading order we can write the explicitly $G$-invariant second order term:
\begin{equation}\label{eq:explicit-rho-expansion}
	H[ v] \propto \int\dd^d x\, \Big( \tr(\bm\nabla\rho\cdot \bm\nabla\rho) + \dots \Big)
	= \int\dd^d x\, \Big( \tr(\partial_\alpha\rho\,\partial_\beta\rho) \bm\nabla v^\alpha\cdot \bm\nabla v^\beta + \dots \Big),
\end{equation}
such that $K_{\alpha\beta}( v)\propto\tr(\partial_\alpha\rho\,\partial_\beta\rho)=g_{\alpha\beta}$, and $g_{\alpha\beta}$ is the TDVP metric of Eq.~\eqref{eq:matrix-metric} inherited by $M_\mrm{GS}$ from the Hilbert space~\cite{Bengtsson2017Aug}, and coincides with Eq.~\eqref{eq:tdvp-tensors-def}.
This is the correct form of the Hamiltonian in the case of the $U(1)$- and $SU(2)$-symmetric effective models studied in the main text: for $k=1$ both models can explicitly be seen to be isotropic ferromagnets on their ground state manifolds (the $SU(2)$ and $SU(4)$ Heisenberg model respectively~\cite{moudgalya2024,PhysRevB.108.054307}), while for $k\geq 2$ the structure of $K_{\alpha\beta}$ follows from the general replica decoupling arguments of the next section.
If $K_{\alpha\beta}( v)=Jg_{\alpha\beta}$ as in Eq.~\eqref{eq:heis-ham-cont-lim}, the imaginary-time TDVP equations for Eq.~\eqref{eq:TDVPeqnsfinal} are thus equivalent
\begin{equation}
	\partial_t  v^\gamma=2J\left(\bm\nabla^2 v^\gamma+\Gamma^\gamma_{\alpha\beta}\bm\nabla v^\alpha\cdot\bm\nabla v^\beta\right),
    \label{eq:geom-heat-flow}
\end{equation}
where $\Gamma^\gamma_{\alpha\beta}$ are the Christoffel symbols for the metric $g_{\alpha\beta}$.
This is exactly the generalized heat equation, for fields $v(x)$ constrained to the surface of a Riemannian manifold $M_\mrm{GS}$.
In $1$d, its stationary solutions are those where $v^\gamma(x)$ are geodesics on the ground state manifold as a function of $x$, since by setting $\partial_t v^\gamma = 0$, Eq.~(\ref{eq:geom-heat-flow}) then reduces to the geodesic equation.
In the main text we however avoid for the most part the intrinsic coordinate representation of Eq.~\eqref{eq:geom-heat-flow}, and rather use an overparametrization in terms of the $\rho$ matrix itself.
In particular, for the $SU(q)$ Heisenberg model on $\mb C^q$, we parametrize the ground state manifold $\mb{CP}^{q-1}$ using a $q\times q$ complex pure density matrix $\rho=\ketbra{v}{v}$, which must therefore satisfy
\begin{equation}\label{eq:conditions-pure-rho}
    \rho=\rho\+=\rho^2,\quad \tr(\rho)=1.
\end{equation}
Note that in the natural matrix coordinates $\rho^A$ where the index $A$ enumerates all real and imaginary components of matrix elements of $\rho$, the Frobenius metric Eq.~\eqref{eq:matrix-metric}---which is equal to the TDVP metric Eq.~\eqref{eq:tdvp-tensors-def} when restricted to the pure state manifold---is simply the Euclidean metric
\begin{equation}
    \tr[(\delta\rho)(\delta\rho)] = \delta_{AB}(\delta\rho)^A(\delta\rho)^B,
\end{equation}
assuming $\rho=\rho\+$.
From the explicit $\rho$ expansion of Eq.~\eqref{eq:explicit-rho-expansion} we can write
\begin{equation}
    H[\rho] = J \int\dd^d x\,\tr(\bm\nabla \rho\cdot \bm\nabla \rho)
\end{equation}
which is equivalent to Eq.~\eqref{eq:heis-ham-cont-lim} for the $q=2$ parametrization of pure density matrices $\rho=\frac{\1+\vec n\cdot\vec\sigma}{2}$ (setting lattice spacing $a=1$). 
To derive the equations of motion we have to compute the unconstrained variation of the energy $\frac{\delta H}{\delta \rho}=-2J\bm\nabla^2\rho$ and then project it onto the tangent space of the manifold $\mb{CP}^{q-1}$, similar to Eq.~\eqref{eq:heis-var-proj}.
Conditions on the tangent vectors $\{Y = \delta \rho\}$ to the set of $\{\rho\}$ that satisfy Eq.~\eqref{eq:conditions-pure-rho} are derived by simply taking variations of the equation $\rho = \rho^2$, which gives $Y=\rho Y+Y\rho$.\footnote{The condition $\tr(Y)=0$ follows from $Y=\rho Y+Y\rho$ and $\rho=\rho^2$.}
The projector onto that tangent space takes the form
\begin{equation}
    \Pi_\rho[X] = [\rho,[\rho,X]] = \rho X + X \rho - 2\rho X \rho,
\label{eq:tangentproj}
\end{equation}
assuming $X=X\+$, which is an expression where we can verify that $\Pi_\rho\+=\Pi_\rho$, $\Pi_\rho^2 = \Pi_\rho$, and that $Y=\Pi_\rho[Y]$ for any $Y$ tangent to $\rho$.
Using this we find the TDVP equations of motion [cf. Eq.~\eqref{eq:tdvp-embed}]
\begin{equation}
    \partial_t \rho = -\Pi_\rho\!\left[\frac{\delta H}{\delta \rho}\right] = 2J \,\Pi_\rho[\bm \nabla^2 \rho] = 2J \,[\rho,[\rho,\bm \nabla^2 \rho]].
\end{equation}
\subsubsection{Equations of motion at boundaries and interfaces}\label{app:bound-cond}
Having computed the equations of motion for the field in the bulk, we now discuss their behaviour at the boundary.
To do this, we apply the imaginary-time TDVP principle Eq.~\eqref{eq:tdvp} to the energy density Eq.~\eqref{eq:explicit-rho-expansion} in a 1d system with boundary.
Consider the Hamiltonian of Eq.~(\ref{eq:explicit-rho-expansion}) with left and right boundaries (writing in one dimension for clarity, but the generalization to higher dimension is straightforward)
\begin{equation}
    H[v; x_L, x_R] \defn J\int_{x_L}^{x_R}\dd x\, \Big( g_{\alpha\beta}\, \partial_x  v^\alpha\, \partial_x  v^\beta\Big).
\label{eq:Hboundaries}
\end{equation}
Its variation then reads [cf.~Eq.~\eqref{eq:grosso-passaggio} and the substitutions in Eqs.~(\ref{eq:TDVPeqnsfinal}) and (\ref{eq:geom-heat-flow})]
\begin{align}
    \delta H[v; x_L, x_R] &= -2J\int_{x_L}^{x_R}\dd x\, g_{\gamma\rho}\left(\partial_x^2 v^\gamma+\Gamma^\gamma_{\alpha\beta}\partial_x v^\alpha\partial_x v^\beta\right)\delta v^\rho +  \left.(2g_{\alpha\beta} \partial_x v^\alpha) \delta v^\beta\right|_{x_R} - \left.(2g_{\alpha\beta} \partial_x v^\alpha) \delta v^\beta\right|_{x_L} + \nonumber \\
    &+\left.(g_{\alpha\beta}\partial_x v^\alpha\partial_x v^\beta)\right|_{x_R}\delta x_R - \left.(g_{\alpha\beta}\partial_x v^\alpha\partial_x v^\beta)\right|_{x_L}\delta x_L.
\label{eq:app-en-var-bd}
\end{align}
Overall, in order to have a well-defined set of equations of motion, we will have to require the boundary term to vanish, which is the principle we use to derive conditions for boundaries and interfaces.
For a quantum system with a (stationary) open boundary $x_b$, we should consider the Hamiltonian $H[v; -\infty, x_b]$ and set $\delta x_b = 0$.
Using Eq.~(\ref{eq:app-en-var-bd}) and the standard boundary conditions at $-\infty$, we find that the field must satisfy Neumann boundary conditions
\begin{equation}
    (2g_{\alpha\beta} \partial_x v^\alpha)\big|_{x_b} = 0 \;\;\implies\;\; \partial_x v^\alpha\big|_{x_b} = 0.
\end{equation}
We can treat similarly the semiclassical evolution of ``unbalanced'' domain walls.
Here we consider $x_b(t)$ to act as the interface between two ``fluids'' (i.e., field configurations) $v_L$ and $v_R$, which occupy respectively the left and right sides of the system, which corresponds to a Hamiltonian of the form
\begin{equation}
    H_{\rm interface} = H[v_L; -\infty, x_b(t)] + H[v_R; x_b(t), \infty].
\end{equation}
We also impose that both fields will have a Dirichlet boundary condition at $x_b(t)$
\begin{equation}
    v_L(x_b(t))=v_R(x_b(t))=v_*.
\end{equation}
This is a general feature of systems such as in Sec.~\ref{subsubsec:finite-size-dw}, where $x_s(t)$ is the interface between the regions where the field $\sigma(x)$ is in the $e$ and $\eta$ configurations, and the fields $\vec{n}$ on the two branches of the manifold are the two fluids on the left and the right of the system.
By differentiating the Dirichlet constraint we find that
\begin{equation}
    \delta v^\alpha_L = - \partial_x v^\alpha_L \delta x_b\qquad
    \delta v^\alpha_R = - \partial_x v^\alpha_R \delta x_b.
\end{equation}
This condition follows from the fact that $x_b(t)$ is not an independent variable, but it is rather a function of $v_L(x,t)$ and $v_R(x,t)$, since it tracks the position where they equal $v_*$.
Using Eq.~(\ref{eq:app-en-var-bd}) and imposing the vanishing of the boundary terms that appear in the variation of the total Hamiltonian we get
\begin{equation}
    (g_{\alpha\beta} \partial_x v^\alpha_L \partial_x v^\beta_L)\big|_{x_b^-} = (g_{\alpha\beta} \partial_x v^\alpha_R \partial_x v^\beta_R)\big|_{x_b^+},
\label{eq:pressureequality1d}
\end{equation}
which implicitly provides the equation of motion for $x_b(t)$, together with the bulk equations Eq.~\eqref{eq:geom-heat-flow} on both sides of the interface.
These boundary conditions take the same form when using an extrinsic coordinate system.
In higher dimensions, the spatial boundary is a hypersurface of dimension $d-1$, and computing the variation Eq.~\eqref{eq:app-en-var-bd} reduces to a one-dimensional problem along the direction perpendicular to the surface at every point (due to the Dirichlet boundary condition, the gradient of the field is orthogonal to the surface).
In some cases, such as the $SU(2)$-symmetric model of Sec.~\ref{sec:su2}, the Dirichlet boundary condition is replaced with $v_L(x_b(t))=v_R(x_b(t))\in M_*$ where $M_*\subseteq M$. 
In this case the gradients of the field at the boundary need not be orthogonal to it; however, their parallel components will be equal $\bm\nabla_\Vert v_L\big|_\mrm{bd}=\bm\nabla_\Vert v_R\big|_\mrm{bd}$, which ensures that the boundary terms in $\delta H$ only depend on the component of the gradient normal to the surface.
If $\bm n$ is the unit vector normal to the boundary, these conditions the naturally generalize to
\begin{equation}
    \text{Open boundary:}\;\;\; \bm n\cdot\bm\nabla v^\alpha =0,\qquad \text{Interface:}\;\;\; |\bm n\cdot\bm\nabla v_L|^2 = |\bm n\cdot\bm\nabla v_R|^2.
\label{eq:higherdimspressure}
\end{equation}
Physically, both conditions require the boundary be at a stationary equilibrium of forces at all times: in the former the pressure of the field on the boundary is set to zero, in the latter the pressure of fields on both sides must be equal.
\subsubsection{Replica decoupling}\label{app:replica-decoupling}
In Sec.~\ref{sec:fgsmanrep} we have seen that for generic Brownian circuits with a continuous on-site symmetry $G$, their effective replica Hamiltonians $P^{(k)}$ possess a ground state manifold that locally has the product structure $(M_{G}^{(1)})^{\times k}$.
Assuming the leading continuum expansion for ferromagnetic models of Eq.~\eqref{eq:exp-generic}, we can easily argue that the metric $K_{\alpha\beta}^{(k)}$ that arises in the gradient expansion of $P^{(k)}$ is simply the direct sum $\bigoplus_{m=1}^{k} K_{\alpha\beta}^{(1)}$ of the metric derived from $P^{(1)}$ on each factor $M_{G}^{(1)}$.
For notational simplicity we show the argument for $k=2$, but generalizations to $k>2$ follow along the same lines.
From the definitions of Eq.~\eqref{eq:superham}, we can write:
\begin{equation}\label{eq:k2-peff-exp}
	P^{(2)}_{ij} = P^{(1)}_{ij} \ot \1^{\ot 2} + \1^{\ot 2} \ot P^{(1)}_{ij}
	 +2 \sum_{\alpha} \ad{h_{\alpha,ij}}^{(1)}\ot\ad{h_{\alpha,ij}}^{(1)},
\end{equation}
where the tensor product is over the replicas.
Let us write each point $v\in M^{(1)}_G\times M^{(1)}_G$ as a pair $v=(w^{[1]},w^{[2]})$. If we expand now the field as in Eq.~\eqref{eq:sketch-expansion}, we can decompose the tangent vector as the sum of the components tangent to each factor in the product $M_{G}^{(1)}\times M_{G}^{(1)}$:
\begin{equation}
    \ket{v+\delta v}=(\sket{w^{[1]}}+\sket{\delta w^{[1]}})\otimes(\sket{w^{[2]}}+\sket{\delta w^{[2]}}).
\end{equation}
Using the fact that $\ad{h_{\alpha,ij}}^{(1)}\ket{w}_i\ket{w}_j=\,_i\!\bra{w}_j\!\bra{w}\ad{h_{\alpha,ij}}^{(1)}=0$ for each factor in Eq.~\eqref{eq:k2-peff-exp}, we find that
\begin{equation}
\begin{aligned}
    \bra{v}\bra{v+\delta v} P_{ij}^{(2)} \ket{v}\ket{v+\delta v} =& \\
    =\sbra{w^{[1]}}\sbra{w^{[1]}+\delta w^{[1]}} &P^{(1)}_{ij} \sket{w^{[1]}}\sket{w^{[1]}+\delta w^{[1]}}+
    \sbra{w^{[2]}}\sbra{w^{[2]}+\delta w^{[2]}} P^{(1)}_{ij} \sket{w^{[2]}}\sket{w^{[2]}+\delta w^{[2]}}
    +O(|\delta v|^3),
\end{aligned}
\end{equation}
where in the continuum limit for smooth field configurations $|\delta v|=\mc O(a)$. 
This implies that the energy of the field $v(x,t)\in M_{G}^{(1)}\times M_{G}^{(1)}\subseteq M_G^{(2)}$ in the product manifold under the effective Hamiltonian $P^{(2)}$, completely decouples to second order to the sum of the energies of the fields $w^{[1]}\in M_G^{(1)}$ and $w^{[2]}\in M_G^{(1)}$ under the effective Hamiltonian $P^{(1)}$, and this decoupling is exact in the continuum limit where $a \rightarrow 0$.
Since the metric of $M^{(2)}_G$ also similarly decomposes into a direct sum of the metrics on each $M^{(1)}_G$ factor, this means that the TDVP equations Eq.~\eqref{eq:tdvp} decouple across all replicas, away from the contact points that may appear in these types of manifolds (as discussed in Sec.~\ref{sec:fgsmanrep}), where interactions between different permutation branches of the manifold can play a role.
In particular, when $K_{\alpha\beta}^{(1)}$ [cf.~Eq.~\eqref{eq:exp-generic}] on $M_{G}^{(1)}$ is the metric in Eq.~\eqref{eq:tdvp-tensors-def}, then the TDVP equation is again the heat equation Eq.~\eqref{eq:geom-heat-flow} on the product manifold.

\section{Mapping Effective Hamiltonians to Classical Stochastic Models} % MARK: * fluct
\label{sec:fluctuations}
We now briefly discuss a framework that could be used to extend the results on the dynamics of effective Hamiltonians beyond the semiclassical limit, and study the quantum fluctuations around the leading trajectories.
This proceeds by mapping the imaginary-time evolution of a quantum Hamiltonian to a classical stochastic process, which we review below for the ferromagnetic Heisenberg model and apply to the $es\eta$ toy model of Sec.~\ref{sec:esn} and the $U(1)$-symmetric model of Sec.~\ref{sec:u1}.
\subsection{Ferromagnetic Heisenberg model}
The mapping for the Heisenberg model is the well-known mapping between imaginary-time ferromagnetic Heisenberg evolution, and the Symmetric Simple Exclusion Process (SSEP) \cite{10.1063/1.524610, SCHUTZ20011}.
Consider a state on a spin-$\frac{1}{2}$ many-body quantum system, and its imaginary-time quantum evolution under the Heisenberg Hamiltonian $H$:
\begin{equation}
	\ket{\psi} = \sum_{\{\varsigma_i\}} \psi_{\{\varsigma_i\}}\ket{\varsigma_1...\varsigma_L},\quad\ket{\psi(t)} = e^{-Ht}\ket{\psi},
\end{equation}
where $\psi_{\{\varsigma_i\}}$ is real and non-negative for every spin pattern $\{\varsigma_i\}$.
Then we can interpret the state as a probability distribution over spin configurations by normalizing:
\begin{equation}
	p(\{\varsigma_i\}) = \frac{\psi_{\{\varsigma_i\}}}{\norm{\ket\psi}_1},\quad \norm{\ket\psi}_1\defeq \sum_{\{\varsigma_i\}} \psi_{\{\varsigma_i\}} = \braket{\boldsymbol{1}}{\psi},
\end{equation}
where we have defined the $\ket{\boldsymbol{1}}$ state to be a uniform superposition of all spin configurations, which is proportional to the fully-polarized $\ket{\dots\rt\rt\dots}$.
The Heisenberg Hamiltonian [Eq.~\eqref{eq:HeisModelHam}] then has the property that under imaginary-time evolution, the weights $\psi_{\{\varsigma_i\}}$ remain positive, and $\norm{\ket{\psi(t)}}_1$ is preserved, since $\ket{\boldsymbol{1}}$ is a ground state of $H$.
The imaginary-time quantum evolution of the state $\ket\psi$ then corresponds to a continuous-time Markovian stochastic process, which in this case is the SSEP.
In the SSEP, we interpret $\ket\dn$ states as empty sites, and $\ket\up$ states as occupied sites, and time evolution consists in an equal probability of particles to hop to any adjacent site, constrained to a maximum occupancy of one particle per site.
Note that $\norm{\ket{\psi(t)}}_2=\sqrt{\braket{\psi(t)}}$ will generally decrease under such a Markovian process, consistent with imaginary-time evolution under a positive-semidefinite Hamiltonian.

\subsection{$es\eta$ Toy Model}
In the case of the $es\eta$ model, we can construct a similar process for the Hamiltonian $P^{(es\eta)}$ of Eq.~(\ref{eq:esn-ham}).
The $e$-$s$ and $\eta$-$s$ interactions in Eq.~\eqref{eq:esn-ham} are identical Heisenberg couplings, and so they can be turned into SSEP generators for configurations of $e$- and $\eta$-particles by interpreting $\ket s$ as the empty state
\begin{equation}
	\varsigma_i \in\{\emptyset,e,\eta\},
\end{equation}
an intuition that we have used in Sec.~\ref{sec:esn}.
Since the $e$-$\eta$ interaction strength $\Delta$ flows to strong coupling at long times, we can replace it with a diagonal potential term $P^{(e\eta)}_{ij}\approx \Delta V_{ij}=\Delta (\ketbra{e\eta}{e\eta}+\ketbra{\eta e}{\eta e})$.
Unlike the hopping terms, this term in the Hamiltonian does not preserve the total probability under time evolution, but rather acts as a probability sink for any configuration where $e$- and $\eta$-particles are neighbors.
We can use the Feynmann-Kac formula \cite{Kipnis1999} to translate this interaction into a penalty for each many-body trajectory within the stochastic process
\begin{equation}
	\!\!\sbra{\{\varsigma_i^{(\mrm f)}\}} e^{-\kappa P^{(es\eta)}t}\sket{\{\varsigma_i^{(\mrm i)}\}} \!=\! \Big\langle e^{-\Delta \int_0^t V(\{\varsigma_i(t)\})\dd\tau}\Big\rangle_\mrm{SSEP}\!\!\!\!
\label{eq:feynmankacweight}
\end{equation}
where the mean is taken over all SSEP trajectories $\{\varsigma_i(\tau)\}$ which satisfy the temporal boundary conditions at times $\tau=0$ and $\tau=t$.
Here the SSEP dynamics accounts for the terms $P^{(es)}$ and $P^{(\eta s)}$, whereas the potential term being averaged accounts for the diagonal $P^{(e\eta)}$ interactions.
In the RG limit where $\Delta\rt\infty$, the state $\ket{\psi(t)}$ represents the outcome of
\begin{equation}\label{eq:words}
		\textit{independent SSEP evolution of $e$- and $\eta$-particles conditioned to non-crossing of $e$- and $\eta$-particles,}
\end{equation}
since in that limit, any configuration where the particles $e$ and $\eta$ are on neighboring sites is infinitely penalized (i.e., ignored) in the weighting of Eq.~(\ref{eq:feynmankacweight}).
The decay of the norm $\norm{\ket{\psi(t)}}_1/\norm{\ket{\psi(0)}}_1$ represents the \textit{survival probability} of the process, i.e. the probability that trajectories satisfy the non-crossing condition.
The formation of the void observed in the $es\eta$ model is here explained by this post-selection process, whereby trajectories where $e$- and $\eta$-particles interact are eliminated from the ensemble; the growth of the sink region is instead driven by density fluctuations within the void.
Broadly speaking, this model falls within the class of \textit{reaction-diffusion} processes~\cite{SCHUTZ20011}.
For example, the initial state $\ket{\textsc{sp}_\eta}$ of Eq.~\eqref{eq:esn-initial}, is the state of a single $\eta$-particle within a background of Poisson distributed $e$-particles.
The evolution of this state under the dynamics described here, to our knowledge has been studied for the first time in Ref.~\cite{Blythe_2003}.
There, the decay rate of the survival probability of the process is rigorously computed to leading order, resulting in scalings identical to what we propose for the decay of the norm [cf.~Eqs.~(\ref{eq:norm-dec-1d},\ref{eq:norm-dec-3d},\ref{eq:norm-dec-2d})].
In Ref.~\cite{Anton_2004}, an exponent of $\frac{1}{4}$ is proposed for subdiffusion of the $\eta$-paticle in 1d, although our numerical data (see Fig.~\ref{fig:esn}c) and recent theoretical results for a closely related model~\cite{mcculloch2026longlivedlocalquantumcoherences} (discussed in the next section), point to a power-law exponent of $\frac{1}{3}$.
\subsection{$U(1)$ effective Hamiltonian}
Using the same methods, it is sometimes possible to approximately map the spin-$\frac{1}{2}$ $U(1)$-symmetric model of Sec.~\ref{sec:u1} to classical stochastic dynamics that have a similar form to those that we derived for the $es\eta$ model.
An exact mapping is not however possible in general since the Hamiltonian term Eq.~\eqref{eq:u1-ham-idstates} possesses a positive off-diagonal matrix element between $\ket{e_pe_q}\leftrightarrow\ket{\eta_r\eta_s}$, which results in a sign problem.
This problem can be circumvented by simply ignoring these transition rules, thus effectively setting these matrix elements to zero.
If the initial state is dominated by a fully-polarized background in a single permutation, as in the case of the correlators in Sec.~\ref{sec:U1correlators}, where a replicated local operator $\ket{O^{\ot 2}}$ is placed in a $\sket{\1^{\otimes 2};e}$ background, a stronger approximation can be performed, which further simplifies the dynamics.
In this case, the creation of pairs $\ket{\eta_r\eta_s}$ in a background of $e$ states, which results in energetic configurations, suppressed by the gapped $e$-$\eta$ interaction.
In Eq.~\eqref{eq:u1-ham-idstates} we can therefore only consider the off-diagonal transition between states $\ket{s_ps_q}\leftrightarrow\ket{e_re_s}$, and ignore the rest.
In this way, if we re-label the states (using the notation of Ref.~\cite{mcculloch2026longlivedlocalquantumcoherences}):
\begin{equation}
    \ket{s_0}=\ket{\textbf{d}}\ot\ket{\textbf{d}},\quad\ket{s_1}=\ket{\textbf{u}}\ot\ket{\textbf{u}},\quad\ket{e_0}=\ket{\textbf{d}}\ot\ket{\textbf{u}},\quad\ket{e_1}=\ket{\textbf{u}}\ot\ket{\textbf{d}},
\end{equation}
the imaginary-time quantum dynamics can be mapped to independent SSEP evolution of the left and right ``$\textbf{d}/\textbf{u}$'' labels.
Furthermore, in this model, $\ket{\eta_0}$ and $\ket{\eta_1}$ particles are conserved, and perform SSEP evolution on a $\ket{s_0}$ or $\ket{s_1}$ background.
Due to the gapped $e$-$\eta$ interaction, in the continuum limit we obtain through the Feynmann-Kac formula a non-crossing rule between $\ket{e_p}$ and $\ket{\eta_q}$ states, analogous to the one of Eq.~\eqref{eq:words}.
This model was derived and studied in Ref.~\cite{mcculloch2026longlivedlocalquantumcoherences} by similar arguments for the study of non-hydrodynamic correlators.
We must emphasize that this procedure of ignoring terms in Eq.~(\ref{eq:u1-ham-idstates}) is justified only in the setting of single-particle initial states which is mostly composed of $e$ states, and not for domain-wall initial states where both the $e$ and $\eta$ states are important.
Indeed, the structure of the $\eta$ branch of the ground state manifold changes upon ignoring the $\ket{s_ps_q}\leftrightarrow\ket{\eta_r\eta_s}$ term, and we hence we should expect that term to be important to the physics of general initial states (even though its effect might drop out for certain special initial states).
If instead of removing all $\ket{\eta_p\eta_q}$ terms from the Hamiltonian term Eq.~\eqref{eq:u1-ham-idstates}, we instead only neglect the $\ket{s_ps_q}\leftrightarrow\ket{e_re_s}$ transitions, this does not change the ground state manifold of the effective Hamiltonian, since $e$-$\eta$ interactions are already gapped due to the term Eq.~\eqref{eq:esn-u1-term}.
Thus we expect the resulting model to correctly capture its leading behavior in a general setting. 
Having resolved the sign problem, this Hamiltonian can also be mapped to a classical stochastic model with similar methods as before~\cite{Bravyi2015}, although with more complex interactions.
We leave the exploration of such a model to future work.
\section{The $\Delta=1$ Limit of the $es\eta$ model} % MARK: * delta=1
\label{sec:delta1lim}
For $\Delta =1$ the $e$-$\eta$ interaction of the $es\eta$ model of Eq.~\eqref{eq:esn-ham} becomes gapless, and so the ground-state structure of the model changes.
Generally, this leads to a non-ferromagnetic $P^{(es\eta)}$ Hamiltonian, due to the `plus' sign of the exchange interaction $P^{(e\eta)}$ in Eq.~\eqref{eq:esn-ham}.
However in 1d, it is possible to perform a Jordan-Wigner transformation that flips the sign of this interaction while preserving the rest:
\begin{equation}
    \ket{e}_i\mapsto (-1)^{\prod_{j<i}\ketbra{\eta}{\eta}_j}\ket{e}_i.
\end{equation}
At this point, the $es\eta$ model then becomes a toy model for effective Hamiltonians of \textit{free-fermion Brownian models with continuous symmetries} at $k=2$ replicas.
After this transformation, the $P^{(es\eta)}$ Hamiltonian, which is in general symmetric under exchange $\ket e\leftrightarrow \ket \eta$, acquires an extended continuous $U(1)$ exchange symmetry
\begin{equation}
	R(\theta)=\prod_{i=1}^L\exp(i\theta \big(\ket{e}_i\!\bra{\eta}+\ket{\eta}_i\!\bra{e}\big)),
\end{equation}
similar to the continuous replica exchange symmetry of free fermion effective models \cite{enriched-phases,fava2024,swann2025,lastres2026geometryfreefermioncommutants}.
This symmetry mixes with the other internal symmetries of the $es\eta$ model, resulting in a global $SU(3)$ symmetry.
Indeed, at $\Delta =1$, $P^{(es\eta)}$ becomes exactly the three-flavor ferromagnetic Heisenberg model, and its global ground states are all obtained by applying global $SU(3)$ rotations on the fully polarized states $\ket{ss \cdots s}$.
Equivalently, the ground state manifold is composed of all states on a single $es\eta$ site:
\begin{equation}
    M_\mrm{GS} = \mb P(\mb C^3) = \mb{CP}^2 \cong SU(3)/U(2).
\end{equation}
Note that this $\Delta=1$ ground state manifold is a superset of the generic $\Delta>1$ $es\eta$ ground state manifold [cf.~Eq.~\eqref{eq:gsman-esn}].
In this way, the former is the projective plane $\mb{CP}^2$ and the latter is the union of two lines (i.e., $\mb{CP}^1$) within it.
Since a property of the projective plane is that any two lines intersect at exactly one point, this provides a different point of view for the shape of the $es\eta$ ground state manifold of Eq.~(\ref{eq:gsman-esn}).
Similar to free-fermion commutants discussed in Sec.~\ref{subsubsec:free-fermion-generic}, this ground state manifold is smooth and homogeneous, and therefore can be parametrized by generalized coherent states~\cite{generalized-coher-st} discussed in Appendix \ref{subsec:matrixcoherent}, which can be used to describe the low-energy physics of the model with greater accuracy compared to the variational ansatz of TDVP, using instead a path-integral approach that generalizes the one discussed in Appendix~\ref{app:path-integral}~\cite{Kochetov1995,Berezin:1978sn}.
From simple geometrical considerations however, we can see that the singularity in the manifold that resulted in void formation is no longer present, and we expect domain wall configurations to melt diffusively as in the Heisenberg model discussed in Sec.~\ref{sec:ex-dw}.
Note that under the Jordan-Wigner transformation,
the single-particle state $\ket{\textsc{sp}_\eta}$ of Eq.~\eqref{eq:esn-initial} maps to
\begin{equation}
    \ket{\textsc{sp}_\eta}\mapsto \dots
    \ot \frac{\ket{s}+\ket{e}}{\sqrt{2}}
    \ot \ket{\eta}
    \ot \frac{\ket{s}-\ket{e}}{\sqrt{2}}
    \ot \dots
\end{equation}
which can be seen as a domain wall configuration within the $e$-$s$ manifold, with a single-particle excitation on top which---unlike in the $\Delta > 1$ case---diffuses freely on any background-field configuration for the $\ket e$ and $\ket s$ states.
The domain-wall part of the configuration should be expected to melt diffusively as well, leading again to a decay of the norm $\norm{\textsc{sp}_\eta} \sim e^{-\sqrt{\kappa t}}$.

A similar homogeneous geometry is expected to describe more generally the ground state manifold of any Brownian free-fermion model (with or without internal symmetries), for any numer of replicas $k$, as seen with various examples in Ref.~\cite{lastres2026geometryfreefermioncommutants}. For example, the $U(1)$-symmetric model of Eq.~\eqref{eq:u1-generators} in the free-fermion limit where $\lambda=0$ in 1d, acquires a Grassmannian ground state manifold $\text{Gr}(2, 4)\cong U(4)/(U(2)\times U(2))$, resulting in dynamics which are expected to be similar to those of the ferromagnetic Heisenberg model.

\section{Entropy Scaling for Superpositions of TDVP Trajectories} % MARK: APP: ENTR
\label{app:entr}
In Sec.~\ref{sec:ex-dw} and~\ref{sec:su2}, we find for some initial conditions that a continuous family of equivalent semiclassical TDVP trajectories exists, and we describe the imaginary time evolution of the state as an equal superposition of all equivalent semiclassical solutions.
In this appendix, we show that a signature of such an ansatz in the models with diffusive dynamics we are interested in, is the logarithmic growth of entanglement entropy of the imaginary-time evolved state.
The initial states we start from are domain walls between two ground states, which lie on two different points of the ground state manifold, and the TDVP trajectories are composed of states derived from different equivalent geodesics that connect the two points on the ground state manifold, see Eqs.~(\ref{eq:heisinitsuperposition}) and (\ref{eq:heis-dw-cont-sup}) for a concrete example from the Heisenberg model.
Here we will assume that, for each trajectory, after time $t$, a region $A$ of width $\ell\sim\sqrt{\kappa t}$ in the middle of the system can be approximated as being fully polarized.
From this assumption we will compute the von Neumann entropy of the reduced density matrix $\rho_A$ of the equal superposition of states obtained from all the trajectories and study its scaling over time.

Starting from the Heisenberg domain wall of Sec.~\ref{sec:ex-dw}, in particular Eq.~\eqref{eq:heis-dw-cont-sup}, the fully-polarized states in each trajectory are of the form
\begin{equation}
    \ket{\varphi}^{\otimes \ell},\quad\text{where}\;\;\ket{\varphi} = \frac{\ket\uparrow + e^{i\varphi}\ket\downarrow}{\sqrt 2},\quad\varphi\in[0,2\pi).
\end{equation}
The reduced density matrix of their superposition can be computed explicitly by decomposing the product state into a sum of spin-patterns $\{\varsigma_i\}_{i\in\{1,...,\ell\}}$ where $\varsigma_i\in\{\up,\dn\}$:
\begin{equation}
    \rho_A = \int_0^{2\pi}\dd\varphi\, (\ketbra{\varphi}{\varphi})^{\otimes \ell} = \int_0^{2\pi}\dd\varphi\sum_{\{\varsigma_i\},\{\varsigma'_i\}}\frac{e^{i(N_{\{\varsigma_i\}}-N_{\{\varsigma'_i\}})\varphi}}{2^\ell}\ketbra{\varsigma_1...\varsigma_\ell}{\varsigma'_1...\varsigma'_L}=\sum_{\{\varsigma_i\},\{\varsigma'_i\}}\frac{\delta_{N_{\{\varsigma_i\}},N_{\{\varsigma'_i\}}}}{2^\ell}\ketbra{\varsigma_1...\varsigma_L}{\varsigma'_1...\varsigma'_\ell}
\end{equation}
where $N_{\{\varsigma_i\}}$ is the number of spin-down states in the configuration.
The eigenstates of this density matrix are the Dicke states $\ket{D_k}$, with eigenvalues $\lambda_k=\frac{1}{2^\ell}\binom{\ell}{k}$ proportional to the size of the associated charge sector:
\begin{equation}
    \ket{D_k}=\frac{1}{\sqrt{\binom{\ell}{k}}}\sum_{\{\varsigma_i\}}\delta_{N_{\{\varsigma_i\}},k}\ket{\{\varsigma_i\}} \;\;\implies\;\; \rho_A = \frac{1}{2^\ell}\sum_{k=0}^\ell \binom{\ell}{k}\ketbra{D_k}{D_k}
    \;\;\implies\;\; S(\rho_A)\approx\frac{1}{2}\log(\frac{\pi e}{2}\ell)\sim \log(\kappa t)
\end{equation}
where in the final step we used the asymptotic expression for the Shannon entropy of the binomial distribution for large $\ell$.
The case for Sec.~\ref{sec:su2}, in particular Eq.~\eqref{eq:superposition-sphere}, can be studied more simply by symmetry arguments.
In this case we are interested in a superposition of the form
\begin{equation}
    \rho_A = \int_{\mb{S}^2}\dd\mu(\vec n)\, \left(\sketbra{\vec n^{\times 4}}{\vec n^{\times 4}}\right)^{\otimes \ell},\quad\text{where}\quad \sket{\vec n^{\times 4}}=\ket{\vec n}\ot\ket{\vec n}^*\ot\ket{\vec n}\ot\ket{\vec n}^*.
\end{equation}
Note that $\rho_A$ is $SU(2)$-symmetric under the action of the diagonal generator on the four copies
\begin{equation}
    \vec S_\mrm{tot}^{(\mrm{diag.})} = \sum_{i}\sum_{a=1}^4 \vec S_i^{[a]}.
\end{equation}
Further,the action of $\rho_A$ is restricted completely within the vector space space spanned by fully-polarized configurations of spin-$2$ states on each replicated physical site (composed of four spin-$\frac{1}{2}$s on each of the copies).
These polarized states span the spin-$2\ell$ irrep of $SU(2)$, and therefore by Schur's lemma, $\rho_A$ must be proportional to the projector onto this irrep, which has dimension $4\ell+1$.
Thus
\begin{equation}
    \rho_A=\frac{1}{4\ell+1}\Pi_\textbf{2s}
    \;\;\implies\;\; S(\rho_A)=\log(4\ell+1)\sim \log(\kappa t).
\end{equation}

\section{Entanglement Growth Predictions} % MARK: APP: CONTI
\label{app:conti}
In this appendix, we provide details of TDVP predictions of overlaps of various replicated initial states on imaginary-time evolved domain wall states, which are relevant for the quantitative prediction of the evolution of purities of the initial states under Brownian dynamics.
\subsection{$es\eta$ Toy Model}
We discuss this for the $es\eta$ toy model, where this overlap does not correspond to any physical Brownian evolution, but is nevertheless very instructive as discussed in Sec.~\ref{sec:esn}.
Using the TDVP solution Eq.~\eqref{eq:esn-sol-dw} for the evolution of the $\ket{\textsc{dw}(t)}$ state in the $es\eta$ model, we can compute the predicted overlap between the evolved state and the following family of fully-polarized states:
\begin{equation}
    \ket{\zeta}^{\ot L},\text{ where }\ket{\zeta}=\cos^2\zeta\ket{s}+\sin^2\zeta\frac{\ket{e}+\ket{s}+\ket{\eta}}{2}.
\end{equation}
These states are normalized such that $\braket{\textsc{dw}(0)}{\zeta}^{\ot L}=1$.
From the parametrization of Eq.~\eqref{eq:rescaled_esn} we find:
\begin{equation}\label{eq:theta-recall}
    \ket{\textsc{dw}(t)}\approx \bigotimes_x \left( \ket{s}+\tan\frac{\theta(x,t)}{2}\ket{\sigma(x)}\right),\quad \theta(x,t)=\frac{\pi}{2}\left|\mrm{erf}\left(\frac{x}{\sqrt{4\kappa t}}\right)\right|\!, \quad \sigma(x)=\begin{cases}
        e\quad &x<0,\\
        \eta\quad &x>0.
    \end{cases}
\end{equation}
Since the overlap between product states is multiplicative along system size, in the continuum limit we can only make predictions about their logarithm, up to an overall non-universal prefactor that depends on details such as lattice spacing, which we ignore.
Then we obtain the following predictions for the log-overlaps $\Omega(\zeta, t)$:
\begin{equation}
\begin{gathered}
    \Omega(\zeta,t)\defeq -\log(\sbra{\textsc{dw}(t)}\zeta\rangle^{\ot L})
    \approx -\log(\prod_{x}\left[ \cos^2\zeta+\frac{1}{2} \sin^2\zeta\left(1+\tan\frac{\theta(x,t)}{2}\right)\right])\approx\\
    \approx\left[-\int\dd x\,\log(
    1-\frac{1}{2} \sin^2\zeta\Big(1-\tan f(x)\Big)
    )\right]\sqrt{\kappa t},
\end{gathered}
\end{equation}
where $f(x)=\frac{\pi}{4}|\mrm{erf}(x/2)|$. The integral in brackets only depends on the angle $\zeta$ and can be estimated numerically; if we call this integral $I(\zeta)$, then this expression allows us to predict the asymptotic ratios of overlaps:
\begin{equation}
    \lim_{t\rt\infty}\frac{\Omega(\zeta_1,t)}{\Omega(\zeta_2,t)} \approx \frac{I(\zeta_1)}{I(\zeta_2)}.
\end{equation}
While due to the variational nature of the TDVP ansatz, we cannot make precise statements about the convergence of this ratio, numerically we find good agreement, as shown in the main text Fig.~\ref{fig:esn-data}.
\subsection{$U(1)$ model}
We now use the explicit solution to the equations of motion for the domain wall cofiguration in the $U(1)$-conserving model Eq.~\eqref{eq:dw-diffeq-both-u1} in order to compute the mean value of the half-chain purity [Eq.~\eqref{eq:renyi-purity}] for an infinite 1d system.
Using the parametrization of Eq.~\eqref{eq:u1-dw-parametrization}, and recalling the solution Eq.~\eqref{eq:esn-sol-dw}, we can write
\begin{multline}
    \sket{A\!:\!\bar A(t)} \approx \bigotimes_{x} \left(\ket{s_0} + \tan^2\frac{\theta(x,t)}{2} \ket{s_1} + \tan\frac{\theta(x,t)}{2} \big(\ket{\sigma(x)_0}+\ket{\sigma(x)_1}\big) \right) +\quad\\
    + \bigotimes_{x} \left(\tan^2\frac{\theta(x,t)}{2} \ket{s_0} +  \ket{s_1} + \tan\frac{\theta(x,t)}{2} \big(\ket{\sigma(x)_0}+\ket{\sigma(x)_1}\big) \right)
\end{multline}
where again $\theta(x,t)$ is as in Eq.~\eqref{eq:theta-recall}.
We will compute the overlap $\sbra{A\!:\!\bar A(t)} e^{-\kappa P^{(2)}t}\ket{\rho}^{\ot 2}$ with generic product states on the spin-$\frac{1}{2}$ chain:
\begin{equation}
    \rho=\ketbra{\psi}{\psi},\qquad \ket{\psi} = \bigotimes_x
    \left(\cos\frac{\zeta(x)}{2} \ket{\up} + e^{i\varphi(x)} \sin\frac{\zeta(x)}{2} \ket{\dn}\right),
\end{equation}
such that, when projected to the $\{\ket{e_p},\ket{s_q},\ket{\eta_r}\}$ basis of interest, they read
\begin{equation}
    \ket{\rho}^{\ot 2} = \bigotimes_x
    \left( \sin^4\frac{\zeta(x)}{2}\ket{s_0} + \cos^4\frac{\zeta(x)}{2}\ket{s_1} + \sin^2\frac{\zeta(x)}{2}\cos^2\frac{\zeta(x)}{2}\big(\ket{e_0}+\ket{e_1}+\ket{\eta_0}+\ket{\eta_1}\big)\right).
\end{equation}
The overlap with $\sket{A\!:\!\bar A(t)}$ is then computed to be
\begin{equation}\label{eq:u1-purity-sol}
\begin{gathered}
    \sbra{A\!:\!\bar A(t)}\rho\rangle^{\ot 2} \approx \prod_x \left(\cos^2\frac{\zeta(x)}{2}+\sin^2\frac{\zeta(x)}{2}\tan\frac{\theta(x,t)}{2}\right)^2
    + \prod_x \left(\sin^2\frac{\zeta(x)}{2}+\cos^2\frac{\zeta(x)}{2}\tan\frac{\theta(x,t)}{2}\right)^2 = \\
    = \frac{ \prod_x \big( 1+ r(x,t) \langle Z_x\rangle\big)^2+
    \prod_x \big( 1- r(x,t) \langle Z_x\rangle\big)^2 }{\prod_x\big(1+r(x,t)\big)^2}
\end{gathered}
\end{equation}
where $\langle Z_x\rangle = \cos\zeta(x)$ is the local mean value of the charge, and
\begin{equation}
    r(x,t) = \tan(\frac{\pi}{4}-\frac{\theta(x,t)}{2})=\tan\left[\frac{\pi}{4}\left(1-\left|\mrm{erf}\left(\frac{x}{\sqrt{4\kappa t}}\right)\right|\right)\right].
\end{equation}
Using the fact that $r(0,t)=1$ for all times, we see that the prediction provides $\sbra{A\!:\!\bar A(t)}\rho\rangle^{\ot 2}=1$ for the frozen states $\zeta(x)\equiv 0$ and $\zeta(x)\equiv \pi$.
In the simple case where $\langle Z_x\rangle$ is constant and non-zero, then Eq.~\eqref{eq:u1-purity-sol} simplifies to Eq.~\eqref{eq:purity-prediction} in the main text, since the numerator is dominated by the term where $r(x,t)\langle Z_x\rangle$ is positive. 
Following a procedure analogous to the one used in the $es\eta$ model in this appendix, in Fig.~\ref{fig:u1-pred-data} we compare the prediction to numerical data, which only spans a short timeframe due to limitations of the numerical data.

\end{document}